\documentclass[12pt]{article}
\usepackage[utf8]{inputenc}
\usepackage[sectionbib]{natbib}

\title{\vspace{-2cm}High-Dimensional Vector Autoregression with Common Response and Predictor Factors}
\author{Di Wang$^\dag$, Xiaoyu Zhang$^{\ddag}$, Guodong Li$^\ddag$ and Ruey S. Tsay$^{\star}$\\\vspace{-0.3cm}\textit{$^\dag${\small School of Mathematical Sciences, Shanghai Jiao Tong University}}\\\textit{$^{\ddag}${\small Department of Statistics and Actuarial Science, University of Hong Kong}}\vspace{-0.3cm}\\\textit{$^\star${\small Booth School of Business, University of Chicago}}}

\usepackage[margin=1in]{geometry}
\usepackage{graphicx}
\usepackage{mathrsfs}
\usepackage{multirow}
\usepackage{amsfonts}
\usepackage{amsmath}
\usepackage{comment}
\usepackage{amsthm}
\usepackage{xr}
\usepackage{xcolor}
\usepackage{mathtools}
\usepackage{algorithm}
\usepackage{sectsty}
\usepackage[colorlinks]{hyperref}
\usepackage{amsmath}
\usepackage{amssymb}
\usepackage[toc,page]{appendix}
\usepackage[mathscr]{euscript}
\usepackage[toc]{appendix}

\let\counterwithin\relax
\usepackage{chngcntr}
\usepackage{apptools}
\usepackage{booktabs}
\usepackage{lscape}
\usepackage{soul}
\usepackage{epstopdf}
\usepackage{array}
\usepackage{diagbox}
\usepackage{mathabx}

\numberwithin{equation}{section}

\newcolumntype{L}[1]{>{\raggedright\let\newline\\\arraybackslash\hspace{0pt}}m{#1}}
\newcolumntype{C}[1]{>{\centering\let\newline\\\arraybackslash\hspace{0pt}}m{#1}}
\newcolumntype{R}[1]{>{\raggedleft\let\newline\\\arraybackslash\hspace{0pt}}m{#1}}

\makeatletter
\newcommand*{\addFileDependency}[1]{
  \typeout{(#1)}
  \@addtofilelist{#1}
  \IfFileExists{#1}{}{\typeout{No file #1.}}
}
\makeatother

\newtheorem{assumption}{Assumption}
\newtheorem{definition}{Definition}

\newtheorem{lemma}{Lemma}

\newtheorem{theorem}{Theorem}

\newtheorem{remark}{Remark}

\hypersetup{
    colorlinks=true,
    linkcolor=blue,
    citecolor=blue,
    filecolor=magenta,
    urlcolor=cyan,
}

\DeclareMathOperator*{\argmin}{arg\,min}

\newcommand{\bm}{\mathbf}
\newcommand{\bbm}{\boldsymbol}
\newcommand{\lb}{\textbf{[}}
\newcommand{\rb}{\textbf{]}}
\newcommand{\cm}[1]{\mbox{\boldmath$\mathscr{#1}$}}

\mathtoolsset{showonlyrefs}

\usepackage{tikz}
\usetikzlibrary{shapes.geometric, arrows}

\tikzstyle{startstop} = [rectangle, rounded corners, minimum width=3cm, minimum height=1cm,text centered, draw=black, fill=purple!30]
\tikzstyle{stop} = [rectangle, rounded corners, minimum width=3.6cm, minimum height=1cm, draw=black, fill=purple!30, text width = 3.6cm]
\tikzstyle{start} = [rectangle, rounded corners, minimum width=2.7cm, minimum height=1cm, draw=black, fill=purple!30, text width = 2.7cm]
\tikzstyle{io} = [trapezium, trapezium left angle=70, trapezium right angle=110, minimum width=3cm, minimum height=1cm, draw=black, fill=blue!30]
\tikzstyle{process} = [rectangle, rounded corners, minimum width=3cm, minimum height=1cm, draw=black, fill=blue!10, text width=10cm]
\tikzstyle{process_new} = [rectangle, rounded corners, minimum width=3cm, minimum height=1cm, draw=black, fill=blue!10, text width=11.5cm]
\tikzstyle{process2} = [rectangle, rounded corners, minimum width=3cm, minimum height=1cm, draw=black, fill=blue!10, text width=10cm]
\tikzstyle{decision} = [diamond, minimum width=3cm, minimum height=1cm, text centered, draw=black, fill=green!30]
\tikzstyle{arrow} = [thick,->,>=stealth]

\DeclareRobustCommand\sampleline[1]{%
    \tikz\draw[#1] (0,0) (0,\the\dimexpr\fontdimen22\textfont2\relax)
    -- (2em,\the\dimexpr\fontdimen22\textfont2\relax);%
}

\begin{document}

\setlength{\parindent}{16pt}

\maketitle

\begin{abstract}
The reduced-rank vector autoregressive (VAR) model can be interpreted as a supervised factor model, where two factor modelings are simultaneously applied to response and predictor spaces. This article introduces a new model, called vector autoregression with common response and predictor factors, to explore further the common structure between the response and predictors in the VAR framework.
The new model can provide better physical interpretations and improve estimation efficiency. In conjunction with the tensor operation, the model can easily be extended to any finite-order VAR model.
A regularization-based method is considered for the high-dimensional estimation with the gradient descent algorithm, and its computational and statistical convergence guarantees are established.
For data with pervasive cross-sectional dependence, a transformation for responses is developed to alleviate the diverging eigenvalue effect. Moreover, we consider additional sparsity structure in factor loading for the case of ultra-high dimension.
Simulation experiments confirm our theoretical findings and a macroeconomic application showcases the appealing properties of the proposed model in structural analysis and forecasting.
  
\end{abstract}

\textit{Keywords}: Factor model, High-dimensional time series, Gradient descent, Matrix factorization, Tensor decomposition

\newpage
\vspace{-1cm}

\linespread{1.53}
\selectfont{}

\section{Introduction}

Due to recent developments in information technologies, high-dimensional data, especially time-dependent data, have been routinely collected from a wide range of scientific areas, including economics, finance, neuroscience, and meteorology, among others \citep{Gorrostieta2012, HL13, DP16}. 
The well-developed statistical methodology for fixed-dimensional data may not be directly applicable to high-dimensional cases, and large-scale data sets often also require scalable and efficient computational algorithms. As a result, it becomes an emerging area of research to develop new statistical methodology and theoretically justified algorithms to analyze the high dimensional data \citep{wainwright2019high,chi2019nonconvex}. In addition, more efforts are needed for high-dimensional time series data due to its complex dynamic dependency; see \cite{pena2021statistical}. 

The vector autoregressive (VAR) model, arguably the most widely used model in 
multivariate time series applications, has been a primary workhorse for analyzing serially dependent data. 
Consider a VAR(1) model for a $p$-dimensional mean-zero time series $\{\bm{y}_t\}$,
\begin{equation}\label{eq:VAR1}
	\bm{y}_t=\bm{A}\bm{y}_{t-1}+\bbm{\varepsilon}_t,~~1\leq t\leq T,
\end{equation}
where $\bbm{\varepsilon}_t\in\mathbb{R}^p$ is a white noise process with mean zero and the covariance matrix $\bm{\Sigma}_{\bbm{\varepsilon}}$, and $\bm{A}\in\mathbb{R}^{p\times p}$ is the parameter matrix providing a straightforward characterization of the interactions between the response $\bm{y}_t$ and predictor $\bm{y}_{t-1}$; see \cite{Tsay14}.
Note that the number of parameters in $\bm{A}$ increases quadratically with the dimension $p$, making it difficult to apply VAR models to high-dimensional data.
To overcome it, a commonly used solution is to assume sparsity in parameter matrices, and many sparsity-imposing or inducing methods can then be employed for estimation and variable selection, including $L_1$ regularization \citep{basu2015regularized,han2015direct} and linear restrictions \citep{guo2016high,wang2021robust}.

Despite its popularity in the literature, there are two concerns about sparse VAR modeling. First, the general sparsity structure cannot guarantee the spectral radius condition for stationarity, so a sparse estimate may result in a non-stationary VAR model. 
Second, for financial and economic time series, one often observes strong dependence among the $p$ scalar series, which often is investigated via factor models
by assuming that the $p$ variables can be decomposed into two parts, factors and errors.
In the vast literature of econometrics and statistics, there are two classes of factor models under various assumptions on factors and errors. The first class assumes that factors are common cross-sectionally and allows for serially dependent idiosyncratic errors \citep{stock2002forecasting,bai2002determining,bai2008large}; while the second class assumes that the dynamic structures along the temporal direction are summarized in the factors, and the errors are temporally uncorrelated \citep{pena1987identifying,lam2012factor,gao2021modeling}.

Another solution for handling high dimensionality is to impose some low-rank structure on the parameter matrices of VAR models, and it leads to the reduced-rank model \citep{velu2013multivariate} or the multilinear low-rank model \citep{wang2019high}. 
This method can circumvent the two concerns mentioned above, and especially for financial and economic data, the fitted models can be interpreted from the perspective of the second class of factor models.
Specifically, assume that $\bm{A}$ is of rank $r$ with $r\ll p$ and, hence, it admits a singular value decomposition (SVD) $\bm{A}=\bm{U}\bm{S}\bm{V}^\top$, where $\bm{U}$ and $\bm{V}$ are $p$-by-$r$ orthonormal matrices.
Accordingly, model \eqref{eq:VAR1} can then be rewritten as
\begin{equation}\label{eq:RRVAR_intepretation}
	\bm{U}^\top\bm{y}_t = \bm{S}\bm{V}^\top\bm{y}_{t-1}+\bm{U}^\top\bbm{\varepsilon}_t,
\end{equation}
and this motivates us to interpret $\bm{U}^\top\bm{y}_t$ and $\bm{V}^\top\bm{y}_{t-1}$ as $r$-dimensional response and predictor factors, respectively, where $\bm{U}$ and $\bm{V}$ are the factor loading matrices. 
The factors defined here can summarize the temporal dynamics in responses and predictors, and should be understood in the sense of the second class of factor models, in which the factors capture all dynamic dependency in the data.
Along this line, we reformulate the VAR model with a low-rank parameter matrix into a form of supervised factor modeling in Section \ref{sec:2.1}.
Specifically, it is equivalent to simultaneously conducting two factor modelings for the $p$ financial or economic variables in a market, where the latent response factors $\bm{U}^\top\bm{y}_t$ can summarize the whole market as in the traditional factor modeling, while the latent predictor factors $\bm{V}^\top\bm{y}_{t-1}$ are the driving forces of the market; see Section \ref{sec:2} for more details.

Although the factor model cannot directly be used for forecasting, a common practice is to apply a low-dimensional model to the factor processes, and then use predictions of the factors and the loading matrices to obtain forecasts of the high-dimensional time series \citep{lam2012factor,gao2021modeling}.
For example, for a $p$-dimensional time series $\bm{y}_t$, consider the factor model in \citet{gao2021modeling}, which can be written as $\bm{y}_t=\bm{\Lambda}\bm{f}_t+\bm{\Gamma}\bm{e}_t$ with $\bm{f}_t\in\mathbb{R}^r$ being a low-dimensional factor process and $\bm{e}_t$ a $(p-r)$-dimensional white noise, and the loading matrix $\lb\bm{\Lambda}~\bm{\Gamma}\rb\in\mathbb{R}^{p\times p}$ 
being orthonormal. 
Assuming a VAR(1) model for $\bm{f}_t$, say $\bm{f}_t=\bm{B}\bm{f}_{t-1}+\bbm{\xi}_t$, it can be shown that $\bm{\Lambda}^\top\bm{y}_t$ follows a VAR(1) process
\begin{equation}\label{eq:FM}
	\bm{\Lambda}^\top\bm{y}_t =\bm{B}\bm{\Lambda}^\top\bm{y}_{t-1}+\bbm{\xi}_t.
\end{equation}
In comparison with model \eqref{eq:RRVAR_intepretation}, the spaces spanned by the response and predictor factors in model \eqref{eq:FM} are identical; see Section \ref{sec:2.1} for more discussions on its relationship to VAR models.
As shown by the empirical example in Section \ref{sec:7}, the setting of factor models may be too restrictive, while the spaces spanned by $\bm{U}$ and $\bm{V}$ from VAR models in \eqref{eq:RRVAR_intepretation} may be overlapped, i.e., there may exist common factors in responses and predictors.
The first main contribution of this article is to propose a VAR model with common response and predictor factors in Section \ref{sec:2.2}, where dynamic dependence in time series is summarized into three types of factors: response-specific, predictor-specific, and common factors.
This enables a better physical interpretation and facilitates the development of more efficient estimation.

We then consider the high-dimensional estimation method and algorithm. A form of matrix or tensor decomposition is considered for the proposed model. However, as the decomposition is not unique and the optimization problem is non-convex, it is challenging to derive computational and statistical guarantees.
To this end, the second contribution of this article is to develop a complete modeling procedure for estimation and parameter selection in Section \ref{sec:3} and further to provide theoretical justifications for both computational and statistical convergence in Section \ref{sec:4}. Specifically, a regularized estimation framework is proposed for high-dimensional VAR models with common response and predictor factors, and a scalable and efficient gradient descent algorithm with spectral initialization is developed accordingly. 
From the computational and statistical convergence analysis, the proposed procedure can effectively and efficiently achieve a statistically optimal rate for estimation errors. Moreover, a data-driven procedure is suggested to determine the numbers of common and specific factors, and its theoretical justifications are also established. 

To adequately address the strong cross-sectional dependence of time series data in the many real applications, in Section \ref{sec:diverging}, we further investigate the case where the largest eigenvalue of $\text{var}(\bm{y}_t)$ may diverge to infinity as $p$ increases. The third contribution of this article is to provide the first solution to deal with the diverging eigenvalue effect, or pervasive cross-sectional dependency, in high-dimensional VAR estimation.
Additionally, in Section \ref{sec:5}, for the case of $p\gg T$, we consider an additional sparsity structure on the factor loading matrices to improve estimation efficiency and to perform variable selection.
Finally, some simulation results and an empirical example are presented in Sections \ref{sec:6} and \ref{sec:7}, respectively. Section \ref{sec:8} gives a short conclusion with discussion. All technical proofs, codes, data, and additional simulation results are given in appendices. 

This work is also related to the vast literature on Bayesian VAR models. \citet{banbura2010large} studied the shrinkage prior for large Bayesian VAR models, and \citet{koop2013forecasting} applied it to the macroeconomic data of medium and large sizes. Bayesian variable selection method for VAR processes was first proposed by \citet{korobilis2013var}. \citet{ghosh2018high} and \citet{ghosh2021strong} studied posterior estimation consistency and strong variable selection consistency of large Bayesian VAR models, respectively.

Throughout this article, we denote vectors by boldface lower case letters, e.g., $\bm{v}$, matrices by boldface capital letters, e.g., $\bm{M}$, and third-order tensors by Euler script letters, e.g., $\cm{T}$. For any vector $\bm{v}$, denote by $\|\bm{v}\|_2$ its Euclidean norm. For any matrix $\bm{M}$, denote by $\bm{M}^\top$, $\|\bm{M}\|_\text{F}$, $\sigma_i(\bm{M})$, $\mathcal{M}(\bm{M})$, and $\mathcal{M}^\perp(\bm{M})$ its transpose, Frobenius norm, $i$-th largest singular value, column space, and orthogonal complement of column space, respectively. For a symmetric matrix $\bm{M}$, denote by $\lambda_{\max}(\bm{M})$ and $\lambda_{\min}(\bm{M})$ its largest and smallest eigenvalue, respectively. For two matrices $\bm{M}_1$ and $\bm{M}_2$, denote by $\lb\bm{M}_1~\bm{M}_2\rb$ and $\langle\bm{M}_1,\bm{M}_2\rangle$ their column-wise matrix concatenation and inner product, respectively. For positive integers $p\geq q$, denote the set of orthonormal matrices by $\mathbb{O}^{p\times q}:=\{\bm{M}\in\mathbb{R}^{p\times q}:\bm{M}^\top\bm{M}=\bm{I}_q\}$. For a third-order tensor $\cm{T}$, denote by $\|\cm{T}\|_\textup{F}$ the Frobenius norm and by $\cm{T}_{(i)}$ its mode-$i$ matricization, for $1\leq i\leq 3$. For a tensor $\cm{T}\in\mathbb{R}^{p_1\times\dots\times p_i\times\dots\times p_d}$ and matrix $\bm{M}\in\mathbb{R}^{q\times p_i}$, denote by $\cm{T}\times_i\bm{M}$ the mode-$i$ tensor-matrix multiplication, for $1\leq i\leq d$. 
Let $C$ denote a generic positive constant. For two real-valued sequences $x_k$ and $y_k$, $x_k\gtrsim y_k$ if there exists a $C>0$ such that $x_k\geq Cy_k$ for all $k$. In addition, we write $x_k\asymp y_k$ if $x_k\gtrsim y_k$ and $y_k\gtrsim x_k$.
Some preliminaries of tensor notation and tensor algebra are presented in Appendix \ref{append:VAR_L}.

\section{VAR with Common Response and Predictor Factors}\label{sec:2}
\subsection{Relationship between reduced-rank VAR and factor models}\label{sec:2.1}

Consider the VAR(1) model in \eqref{eq:VAR1}. Assume that the parameter matrix $\bm{A}$ has a low rank $r$, which is much smaller than $p$, and admits the SVD $\bm{A}=\bm{U}\bm{S}\bm{V}^\top$, where $\bm{U},\bm{V}\in\mathbb{O}^{p\times r}$ are orthonormal matrices, and $\bm{S}\in\mathbb{R}^{r\times r}$ is a diagonal matrix.
As a result, the reduced-rank VAR model can be formulated into
\begin{equation}\label{eq:rrVAR}
	\bm{y}_t=\bm{U}\bm{S}\bm{V}^\top\bm{y}_{t-1}+\bbm{\varepsilon}_t \hspace{5mm}\text{or}\hspace{5mm}
	\bm{U}^\top\bm{y}_t = \bm{S}\bm{V}^\top\bm{y}_{t-1}+\bm{U}^\top\bbm{\varepsilon}_t,
\end{equation}
where $\{\bbm{\varepsilon}_t\}$ are \textit{i.i.d.} with mean zero and finite variance matrix; see \citet{velu2013multivariate}.
Note that the singular vectors $\bm{U}$ and $\bm{V}$ are not unique, as sign switches and column exchanges can be applied. Also, when some of the singular values are identical, their corresponding singular vectors are also not unique. However, the column spaces $\mathcal{M}(\bm{U})$ and $\mathcal{M}(\bm{V})$, as well as the corresponding subspace projectors $\bm{U}\bm{U}^\top$ and $\bm{V}\bm{V}^\top$, can be uniquely defined.
In fact, $\mathcal{M}(\bm{U})$ and $\mathcal{M}(\bm{V})$ are the column and row spaces of $\bm{A}$, respectively.

From model \eqref{eq:rrVAR}, we can interpret $\bm{U}^\top\bm{y}_t$ and $\bm{V}^\top\bm{y}_{t-1}$, respectively, as the response and predictor factors, which correspond to two different factor modelings.
On one hand, for dimension reduction on the response factor space, $\bm{y}_t$ can be projected onto $\mathcal{M}(\bm{U})$ and $\mathcal{M}^\perp(\bm{U})$, where these two parts can be verified to have completely different dynamic structures,  
\begin{equation}\label{eq:dynamic_part}
	\bm{U}\bm{U}^\top\bm{y}_t=(\bm{U}\bm{S}\bm{V}^\top)(\bm{V}\bm{V}^\top\bm{y}_{t-1})+\bm{U}\bm{U}^\top\bbm{\varepsilon}_t 
	\hspace{5mm}\text{and}\hspace{5mm}
	(\bm{I}_p-\bm{U}\bm{U}^\top)\bm{y}_t=(\bm{I}_p-\bm{U}\bm{U}^\top)\bbm{\varepsilon}_t.
\end{equation}
All information of $\bm{y}_t$ related to temporal dynamic structures is collected into $\mathcal{M}(\bm{U})$, and the projection of $\bm{y}_t$ onto $\mathcal{M}^\perp(\bm{U})$ is serially uncorrelated.
In fact, model \eqref{eq:rrVAR} can also simply be rewritten as $
\bm{y}_t=\bm{U}\bm{f}_t+\bbm{\varepsilon}_t$ with $\bm{f}_t=\bm{S}\bm{V}^\top\bm{y}_{t-1}$. Consequently, the time series generated by model \eqref{eq:rrVAR} admits a form of static factor models, and the corresponding factor space is exactly the response factor space $\mathcal{M}(\bm{U})$.

On the other hand, if we project $\bm{y}_{t-1}$ onto $\mathcal{M}(\bm{V})$ and $\mathcal{M}^\perp(\bm{V})$, it holds that
\[
\text{cov}(\bm{y}_t,\bm{V}\bm{V}^\top\bm{y}_{t-1})=(\bm{U}\bm{S}\bm{V}^\top)\text{var}(\bm{V}\bm{V}^\top\bm{y}_{t-1})
\hspace{5mm}\text{and}\hspace{5mm}
\text{cov}(\bm{y}_t,(\bm{I}_p-\bm{V}\bm{V}^\top)\bm{y}_{t-1})=\bm{0}.
\]
All information of $\bm{y}_{t-1}$ that can contribute to predicting $\bm{y}_t$ is summarized into the space $\mathcal{M}(\bm{V})$, and the predictor factors $\bm{V}^\top\bm{y}_{t-1}$ contain all driving forces of the market.
Following the existing work considering dynamically dependent factors and white noise errors \citep{lam2012factor}, we may treat model \eqref{eq:rrVAR} as a supervised factor model, where two different factor modelings are conducted simultaneously and the dynamic dependence of time series is driven by these two types of factors.

Factor modeling, with dynamically dependent factors and white noise errors, is another method to forecast high-dimensional time series in the statistics literature; see \citet{lam2011estimation}, \cite{lam2012factor}, \cite{gao2021modeling}, among others. 
Specifically, consider the factor model in \cite{gao2021modeling}, and assume that  $\bm{y}_t\in\mathbb{R}^p$ has a latent structure of
\begin{equation}\label{eq:structuralFM}
	\bm{y}_t=\bm{\Lambda}\bm{f}_t+\bm{\Gamma}\bm{e}_t=\lb\bm{\Lambda}~\bm{\Gamma}\rb\lb\bm{f}_t^\top~\bm{e}_t^\top\rb^\top,~~1\leq t\leq T
\end{equation}
where $\bm{f}_t\in\mathbb{R}^r$ is a dynamic factor, $\bm{e}_t\in\mathbb{R}^{p-r}$ is a white noise, and $\bm{\Lambda}\in\mathbb{R}^{p\times r}$ and $\bm{\Gamma}\in\mathbb{R}^{p\times(p-r)}$ are full-rank loading matrices for factors and white noise components, respectively.
For the sake of identification, $\bm{\Lambda}$ and $\bm{\Gamma}$ are assumed to be orthonormal, i.e., $\bm{\Lambda}^\top\bm{\Lambda}=\bm{I}_r$ and $\bm{\Gamma}^\top\bm{\Gamma}=\bm{I}_{p-r}$, and $\lb\bm{\Lambda}~\bm{\Gamma}\rb\in\mathbb{R}^{p\times p}$ is of full rank such that $\text{var}(\bm{y}_t)$ is nonsingular.
Furthermore, suppose that the factors in \eqref{eq:structuralFM} follow a VAR model,
$\bm{f}_t=\bm{B}\bm{f}_{t-1}+\bbm{\xi}_t$,
where $\bm{B}\in\mathbb{R}^{r\times r}$ is a coefficient matrix, and $\bbm{\xi}_t\in\mathbb{R}^{r}$ is a white noise uncorrelated with $\{\bm{e}_t\}$ in all leads and lags. 

Denote $\bbm{\eta}_t=\bm{\Lambda}\bbm{\xi}_t+\bm{\Gamma}\bm{e}_t\in\mathbb{R}^p$, which is serially uncorrelated. 
Note that $(\bm{I}_p-\bm{\Lambda}\bm{\Lambda}^\top)\bbm{\eta}_t=(\bm{I}_p-\bm{\Lambda}\bm{\Lambda}^\top)\bm{\Gamma}\bm{e}_t$ and, for the projection of $\bm{y}_t$ onto $\mathcal{M}^\perp(\bm{\Lambda})$, it holds that
$(\bm{I}_p-\bm{\Lambda}\bm{\Lambda}^\top)\bm{y}_t=(\bm{I}_p-\bm{\Lambda}\bm{\Lambda}^\top)\bbm{\eta}_t$, which is also serially uncorrelated.
The projection of $\bm{y}_t$ onto $\mathcal{M}(\bm{\Lambda})$ follows
\begin{equation}\label{eq:SFM_dynamic}
\bm{\Lambda}\bm{\Lambda}^\top\bm{y}_t =(\bm{\Lambda}\bm{B}\bm{\Lambda}^\top)(\bm{\Lambda}\bm{\Lambda}^\top\bm{y}_{t-1}) + \bm{\Lambda}\bm{\Lambda}^\top\bbm{\eta}_t - \bm{\Lambda}\bm{\Lambda}^\top\lb \bm{\Lambda}\bm{B}\bm{\Lambda}^\top\bm{\Gamma}~\bm{0}_{p\times r} \rb
\lb\bm{\Gamma}~\bm{\Lambda}\rb^{-1}\bbm{\eta}_{t-1},
\end{equation}
i.e. a form of vector autoregressive and moving average (VARMA) models \citep{Tsay14}.
Moreover, when $\bm{\Lambda}^\top\bm{\Gamma}=\bm{0}_{r\times(p-r)}$, it reduces to a VAR(1) process.
In comparison with \eqref{eq:dynamic_part}, the response and predictor spaces in \eqref{eq:SFM_dynamic} are identical, and this may be too restrictive; see Section \ref{sec:7} for empirical evidence.

We finally consider model \eqref{eq:rrVAR} with the same response and predictor spaces, i.e., $\mathcal{M}(\bm{U})=\mathcal{M}(\bm{V})$. There exists an orthogonal matrix $\bm{O}\in\mathbb{O}^{r\times r}$ such that $\bm{U}\bm{O}=\bm{V}$. As a result, model \eqref{eq:rrVAR} can be rewritten as $\bm{y}_t=\bm{U}\bm{Q}\bm{U}^\top\bm{y}_{t-1}+\bbm{\varepsilon}_t$, where $\bm{Q}=\bm{S}\bm{O}^\top$,
\begin{equation}
	\bm{U}\bm{U}^\top\bm{y}_t=(\bm{U}\bm{Q}\bm{U}^\top)(\bm{U}\bm{U}^\top\bm{y}_{t-1})+\bm{U}\bm{U}^\top\bbm{\varepsilon}_t\quad\text{and}\quad(\bm{I}_p-\bm{U}\bm{U}^\top)\bm{y}_t=(\bm{I}_p-\bm{U}\bm{U}^\top)\bbm{\varepsilon}_t,
\end{equation}
which remarkably coincide with dynamic structures of the above-mentioned dynamic factor model with $\bm{U}\bm{U}^\top=\bm{\Lambda}\bm{\Lambda}^\top$, $\bbm{\varepsilon}_t=\bbm{\eta}_t$, $\bm{U}\bm{Q}\bm{U}^\top=\bm{\Lambda}\bm{B}\bm{\Lambda}^\top$, and $\bm{\Lambda}^\top\bm{\Gamma}=\bm{0}_{r\times(p-r)}$.

\subsection{Common response and predictor factors}\label{sec:2.2}

For model \eqref{eq:rrVAR}, consider its response and predictor spaces, $\mathcal{M}(\bm{U})$ and $\mathcal{M}(\bm{V})$, and suppose that their intersection is of dimension $d$, i.e. $\text{rank}(\lb\bm{U}~\bm{V}\rb)=2r-d$, with $0\leq d\leq r$. 
As a result, when $d\geq 1$, there exist two orthogonal matrices $\bm{O}_1,\bm{O}_2\in\mathbb{O}^{r\times r}$ such that
\begin{equation}
\bm{U}\bm{O}_1=[\bm{C}~\bm{R}\textbf{]}\in\mathbb{R}^{p\times r},~~ \bm{V}\bm{O}_2=\textbf{[}\bm{C}~\bm{P}\textbf{]}\in\mathbb{R}^{p\times r}~~\text{and}~~\bm{C}^\top\bm{R}=\bm{C}^\top\bm{P}=\bm{0}_{d\times (r-d)},
\end{equation}
where the orthonormal matrices $\bm{R},\bm{P}\in\mathbb{O}^{p\times (r-d)}$ represent the \textit{response-specific} and \textit{predictor-specific subspaces} of dimension $r-d$, respectively, the orthonormal matrix $\bm{C}\in\mathbb{O}^{p\times d}$ represents the \textit{common subspace} of dimension $d$, and $\bm{C}$ is orthogonal to $\bm{R}$ and $\bm{P}$.

Let $\bm{D}=\bm{O}_1^\top\bm{S}\bm{O}_2$, and model \eqref{eq:rrVAR} can be rewritten as
\begin{equation}\label{eq:equivalent_form}
\bm{y}_t=\textbf{[}\bm{C}~\bm{R}\textbf{]}\bm{D}\textbf{[}\bm{C}~\bm{P}\textbf{]}^\top\bm{y}_{t-1}+\bbm{\varepsilon}_t.
\end{equation}
Let $\bm{c}_t=\bm{C}^\top\bm{y}_t\in\mathbb{R}^d$, $\bm{r}_t=\bm{R}^\top\bm{y}_t\in\mathbb{R}^{r-d}$, and $\bm{p}_t=\bm{P}^\top\bm{y}_t\in\mathbb{R}^{r-d}$. Then model \eqref{eq:equivalent_form} implies
\begin{equation}\label{eq:var1}
\renewcommand{\arraystretch}{0.6}
\begin{bmatrix}
\bm{c}_t \\ \bm{r}_t
\end{bmatrix} = \bm{D}\begin{bmatrix}
\bm{c}_{t-1} \\ \bm{p}_{t-1}
\end{bmatrix}+\begin{bmatrix}
\bm{C}^\top\bbm{\varepsilon}_t \\ \bm{R}^\top\bbm{\varepsilon}_t
\end{bmatrix}.
\end{equation}
We call the model in \eqref{eq:equivalent_form} or equivalently 
in \eqref{eq:var1} the \textit{vector autoregression with common response and predictor factors}, and $\bm{c}_t$, $\bm{r}_t$, and $\bm{p}_t$ are referred to as the  \textit{common, response-specific, and predictor-specific factors}, respectively.

The proposed model in \eqref{eq:equivalent_form} 
can provide a better physical interpretation, especially for financial and economic series, than reduced-rank VAR models in \eqref{eq:rrVAR} by distinguishing these three types of factors; see Section \ref{sec:7} for empirical evidence.
Moreover, the reduced-rank model in \eqref{eq:rrVAR} has $d_\text{RR}(p,r)=r(2p-r)$ parameters, whereas the proposed model has $d_\text{CS}(p,r,d)=r(2p-r)-d(p-(d+1)/2)$ parameters. When $r$ and $d$ are much smaller than $p$, the model complexity is roughly reduced from $2pr$ to $2pr-pd$, and hence the corresponding estimation efficiency can be improved; see simulation experiments in Section \ref{sec:6}.

Using tensor operations, we extend the proposed model to general VAR($\ell$) processes,
\begin{equation}\label{eq:VAR_ell}
\bm{y}_t=\bm{A}_1\bm{y}_{t-1}+\cdots+\bm{A}_\ell\bm{y}_{t-\ell}+\bbm{\varepsilon}_t.
\end{equation}
The parameter matrices are first rearranged into a tensor $\cm{A}\in\mathbb{R}^{p\times p\times \ell}$ such that its mode-1 matricization is  $\cm{A}_{(1)}=\lb\bm{A}_1~\bm{A}_2~\cdots~\bm{A}_\ell\rb$, and its mode-2 matricization assumes the form  $\cm{A}_{(2)}=\lb\bm{A}_1^\top~\bm{A}_2^\top~\cdots~\bm{A}_\ell^\top\rb$.
Note that the column spaces of $\cm{A}_{(1)}$ and $\cm{A}_{(2)}$ are the column and row spaces of all parameter matrices, respectively. Suppose that they are of dimensions $r_1$ and $r_2$, respectively; that is, $\text{rank}(\cm{A}_{(1)})=r_1$ and $\text{rank}(\cm{A}_{(2)})=r_2$, where $r_1$ and $r_2$ may not be equal. We have a Tucker decomposition via higher-order singular value decomposition (HOSVD) \citep{de2000multilinear},
$$\cm{A}=\cm{S}\times_1\bm{U} \times_2\bm{V},$$
where $\bm{U}\in\mathbb{O}^{p\times r_1}$ and $\bm{V}\in\mathbb{O}^{p\times r_2}$ consist of the top $r_1$ and $r_2$ left singular vectors of $\cm{A}_{(1)}$ and $\cm{A}_{(2)}$, respectively, the core tensor $\cm{S}=\cm{A}\times_1\bm{U}^{\top} \times_2\bm{V}^{\top}\in\mathbb{R}^{r_1\times r_2\times \ell}$, and $\times_i$ is the tensor-matrix mode-$i$ multiplication defined in Appendix \ref{append:VAR_L}. 

Split the mode-1 matricization of $\cm{S}$ into $\lb\bm{S}_1~\bm{S}_2~\cdots~\bm{S}_\ell\rb$ with each $\bm{S}_k\in\mathbb{R}^{r_1\times r_2}$, and then model \eqref{eq:VAR_ell} becomes
\begin{equation}\label{eq:var-factor}
	\bm{y}_t=\bm{U}\sum_{k=1}^{\ell}\bm{S}_k \bm{V}^{\top} \bm{y}_{t-k}+\bbm{\varepsilon}_t \hspace{5mm}\text{or}\hspace{5mm}
	\bm{U}^{\top}\bm{y}_t=\sum_{k=1}^{\ell}\bm{S}_k \bm{V}^{\top} \bm{y}_{t-k}+\bm{U}^{\top}\bbm{\varepsilon}_t,
\end{equation}
where $\bm{U}^{\top}\bm{y}_t$ and $\bm{V}^{\top} \bm{y}_{t-k}$ are the response and predictor factors, respectively.
Note that $\bm{U}$ and $\bm{V}$ are not unique, but the subspaces $\mathcal{M}(\bm{U})$ and $\mathcal{M}(\bm{V})$, together with their projectors $\bm{U}\bm{U}^\top$ and $\bm{V}\bm{V}^\top$, can be uniquely defined.
As in VAR(1), we can interpret model \eqref{eq:var-factor} as a supervised factor model with $\mathcal{M}(\bm{U})$ and $\mathcal{M}(\bm{V})$ being the response and predictor factor spaces, respectively. 
\begin{remark}
    Let $\bm{\Lambda}_k=\bm{U}\bm{S}_{k+1}$ and $\bm{f}_{t-k}=\bm{V}^\top\bm{y}_{t-k-1}$, and then model \eqref{eq:var-factor} can be rewritten into $\bm{y}_t=\sum_{k=0}^{\ell-1}\bm{\Lambda}_k\bm{f}_{t-k}+\bbm{\varepsilon}_t$, i.e. it admits a generalized dynamic factor modeling form in \citet{forni2000generalized,forni2005generalized}. Note that the proposed model is for a supervised problem, while the generalized dynamic factor modeling is fundamentally for an unsupervised one. In addition, the factors and errors in \citet{forni2000generalized,forni2005generalized} are assumed to be uncorrelated in all leads and lags, but those in our model are not.
\end{remark}

Suppose that the response and predictor subspaces share a common subspace of dimension $d$, i.e., $\text{rank}(\lb\bm{U} ~\bm{V}\rb)=r_1+r_2-d$, with $0\leq d\leq \min(r_1,r_2)$.
Then there exist two matrices, $\bm{O}_1\in\mathbb{O}^{r_1\times r_1}$ and $\bm{O}_2\in\mathbb{O}^{r_2\times r_2}$, such that
$\bm{U}\bm{O}_1=[\bm{C}~\bm{R}\textbf{]}\in\mathbb{R}^{p\times r_1}$, $\bm{V}\bm{O}_2=\textbf{[}\bm{C}~\bm{P}\textbf{]}\in\mathbb{R}^{p\times r_2}$, and $\bm{C}^\top\bm{R}=\bm{C}^\top\bm{P}=\bm{0}_{d\times (r-d)}$,
where $\bm{R}\in\mathbb{O}^{p\times (r_1-d)}$, $\bm{P}\in\mathbb{O}^{p\times (r_2-d)}$, and $\bm{C}\in\mathbb{O}^{p\times d}$ are the response-specific, predictor-specific, and common subspaces of dimensions $r_1-d$, $r_2-d$, and $d$, respectively.
As a result, the parameter tensor can be formulated into
$\cm{A}=\cm{D}\times_1\lb\bm{C}~\bm{R}\rb\times_2\lb\bm{C}~\bm{P}\rb$,
where $\cm{D}=\cm{S}\times_1 \bm{O}_1^\top\times_2\bm{O}_2^\top$.
Furthermore, let $\bm{c}_t=\bm{C}^\top\bm{y}_t\in\mathbb{R}^d$, $\bm{r}_t=\bm{R}^\top\bm{y}_t\in\mathbb{R}^{r_1-d}$, and $\bm{p}_t=\bm{P}^\top\bm{y}_t\in\mathbb{R}^{r_2-d}$, and then model \eqref{eq:var-factor} has the form 
\begin{equation}\label{eq:FA_interpretation}
	\bm{y}_t=\lb\bm{C}~\bm{R}\rb \sum_{k=1}^{\ell}\bm{D}_k \lb\bm{C}~\bm{P}\rb^{\top} \bm{y}_{t-k}+\bbm{\varepsilon}_t 
	\hspace{5mm}\text{or}\hspace{5mm}
\renewcommand{\arraystretch}{0.6}
\begin{bmatrix}
\bm{c}_t \\ \bm{r}_t
\end{bmatrix} = \sum_{k=1}^{\ell}\bm{D}_k\begin{bmatrix}
\bm{c}_{t-k} \\ \bm{p}_{t-k}
\end{bmatrix} +\begin{bmatrix}
\bm{C}^\top\bbm{\varepsilon}_t \\ \bm{R}^\top\bbm{\varepsilon}_t
\end{bmatrix},
\end{equation}
where each $\bm{D}_k$ is a $r_1$-by-$r_2$ matrix such that $\cm{D}_{(1)}=\lb\bm{D}_1~\bm{D}_2~\cdots~\bm{D}_\ell\rb$. 
The model \eqref{eq:FA_interpretation} defines a general vector autoregression with common response and predictor factors, and $\bm{c}_t$, $\bm{r}_t$, and $\bm{p}_t$ are the common, response-specific, and predictor-specific factors, respectively.

For the parameter tensor $\cm{A}\in\mathbb{R}^{p\times p\times \ell}$, dimension reduction is conducted along the first two modes, and when the lag order $\ell$ is large, it is also of interest to further restrict the parameter space along the third mode. Specifically, assume that $\text{rank}(\cm{A}_{(3)})=r_3$, and we have the Tucker decomposition: $\cm{D}=\cm{G}\times_3 \bm{L}$ and
\begin{equation}\label{eq:common_tensor_decomp2}
    \cm{A}=\cm{G}\times_1\lb\bm{C}~\bm{R}\rb\times_2\lb\bm{C}~\bm{P}\rb\times_3\bm{L},
\end{equation}
where $\cm{G}\in\mathbb{R}^{r_1\times r_2\times r_3}$ is the core tensor and $\bm{L}\in\mathbb{O}^{\ell\times r_3}$ is the lag factor matrix.
Similarly, this additional low-rankness along the third mode would lead to a lag-specific factor.
The number of parameters under the low-rank structure is $d_\text{CS}(p,\ell,r_1,r_2,r_3,d)=r_1r_2r_3+r_1(p-r_1)+r_2(p-r_2)+r_3(\ell-r_3)-d(p-(d+1)/2)$, while model \eqref{eq:VAR_ell} has $p^2\ell$ parameters.

\section{High-Dimensional Estimation Methods}\label{sec:3}


\subsection{Regularized estimation and gradient descent algorithm}\label{sec:3.1}

Consider the observed sequence, $\{\bm{y}_0,\bm{y}_1,\ldots,\bm{y}_T\}$, generated by the VAR(1) model in \eqref{eq:equivalent_form}, and suppose that both $r$ and $d$ are known.
Our aim is to estimate the parameter matrix
\begin{equation*}
	\bm{A}:=\bm{A}(\bm{C},\bm{R},\bm{P},\bm{D})=\textbf{[}\bm{C}~\bm{R}\textbf{]}\bm{D}\textbf{[}\bm{C}~\bm{P}\textbf{]}^\top=\lb\bm{C}~\bm{R}\rb
    \renewcommand{\arraystretch}{0.6}
	\begin{bmatrix}
		\bm{D}_{11} & \bm{D}_{12}\\
		\bm{D}_{21} & \bm{D}_{22}
	\end{bmatrix}\begin{bmatrix}
	\bm{C}^\top \\ \bm{P}^{\top}
\end{bmatrix},
\end{equation*}
where $\bm{D}_{11}\in\mathbb{R}^{d\times d}$, $\bm{D}_{12}\in\mathbb{R}^{d\times(r-d)}$, $\bm{D}_{21}\in\mathbb{R}^{(r-d)\times d}$, and $\bm{D}_{22}\in\mathbb{R}^{(r-d)\times (r-d)}$ are four blocks of $\bm{D}$.
Let $\bm{Y}=\lb\bm{y}_1~\cdots~\bm{y}_T\rb$ and $\bm{X}=\lb\bm{y}_0~\cdots~\bm{y}_{T-1}\rb$, and the squared loss function is
\begin{equation}
\mathcal{L}(\bm{C},\bm{R},\bm{P},\bm{D})
=\frac{1}{2T}\sum_{t=1}^T\left\|\bm{y}_t-\lb\bm{C}~\bm{R}\rb
\bm{D}\lb\bm{C}~\bm{P}\rb^\top\bm{y}_{t-1}\right\|_2^2
=\frac{1}{2T}\left\|\bm{Y}-\lb\bm{C}~\bm{R}\rb
\bm{D}\lb\bm{C}~\bm{P}\rb^\top\bm{X}\right\|_\text{F}^2.
\end{equation}
With $a,b>0$ being regularization parameters, the components in \eqref{eq:equivalent_form} can be estimated by
\begin{equation}\label{eq:extra_decompose}
	\begin{split}
		\left(\widehat{\bm{C}},\widehat{\bm{R}},\widehat{\bm{P}},\widehat{\bm{D}}\right)
		= & \argmin_{\substack{\bm{C}\in\mathbb{R}^{p\times d},\bm{D}\in\mathbb{R}^{r\times r},\\ \bm{R},\bm{P}\in\mathbb{R}^{p\times(r-d)}}} \Bigg\{\mathcal{L}(\bm{C},\bm{R},\bm{P},\bm{D}) \\
        &+  \frac{a}{2}\left\|\lb\bm{C}~\bm{R}\rb^\top\lb\bm{C}~\bm{R}\rb-b^2\bm{I}_r\right\|_\text{F}^2
		 + \frac{a}{2}\left\|\lb\bm{C}~\bm{P}\rb^\top\lb\bm{C}~\bm{P}\rb-b^2\bm{I}_r\right\|_\text{F}^2\Bigg\}.
	\end{split}
\end{equation}

The above estimation method is motivated by \citet{han2020optimal} for low-rank tensor estimation, and the regularization terms $\|\lb\bm{C}~\bm{R}\rb^\top\lb\bm{C}~\bm{R}\rb-b^2\bm{I}_r\|_\text{F}^2$ and $\|\lb\bm{C}~\bm{P}\rb^\top\lb\bm{C}~\bm{P}\rb-b^2\bm{I}_r\|_\text{F}^2$ are used to keep $\lb\bm{C}~\bm{R}\rb$ and $\lb\bm{C}~\bm{P}\rb$ from being singular and to balance the scaling of these components. 
It is noteworthy that $\lb\widehat{\bm{C}}~\widehat{\bm{R}}\rb^\top\lb\widehat{\bm{C}}~\widehat{\bm{R}}\rb=\lb\widehat{\bm{C}}~\widehat{\bm{P}}\rb^\top\lb\widehat{\bm{C}}~\widehat{\bm{P}}\rb=b^2\bm{I}_r$.
Moreover, the estimated parameter matrix $\bm{A}(\widehat{\bm{C}},\widehat{\bm{R}},\widehat{\bm{P}},\widehat{\bm{D}})$ is not sensitive to the choices of regularization parameters $a$ and $b$, and they are set to one in all our numerical analysis.

We use the gradient descent method to solve the optimization problem in \eqref{eq:extra_decompose}. Specifically, the partial derivatives can be calculated as
\begin{equation}
	\begin{split}
		\nabla_{\bm{C}}\mathcal{L}=&\nabla\mathcal{L}(\bm{A})(\bm{C}\bm{D}_{11}^\top+\bm{P}\bm{D}_{12}^\top)+\nabla\mathcal{L}(\bm{A})^\top(\bm{C}\bm{D}_{11}+\bm{R}\bm{D}_{21}),\nabla_{\bm{D}}\mathcal{L}=\lb\bm{C}~\bm{R}\rb^\top\nabla\mathcal{L}(\bm{A})\lb\bm{C}~\bm{P}\rb,\\
        \nabla_{\bm{R}}\mathcal{L}=&\nabla\mathcal{L}(\bm{A})\lb\bm{C}~\bm{P}\rb\lb\bm{D}_{21}~\bm{D}_{22}\rb^\top,~\text{and}~
		\nabla_{\bm{P}}\mathcal{L}=\nabla\mathcal{L}(\bm{A})^\top\lb\bm{C}~\bm{R}\rb\lb\bm{D}_{12}^\top~\bm{D}_{22}^\top\rb^\top,
	\end{split}
\end{equation}
where $\nabla\mathcal{L}(\bm{A})=T^{-1}(\bm{A}\bm{X}-\bm{Y})\bm{X}^\top=T^{-1}(\lb\bm{C}~\bm{R}\rb
\bm{D}\lb\bm{C}~\bm{P}\rb^\top\bm{X}-\bm{Y})\bm{X}^\top$.
Given an initial estimator $\bm{A}^{(0)}=\lb\bm{C}^{(0)}~\bm{R}^{(0)}\rb\bm{D}^{(0)}\lb\bm{C}^{(0)}~\bm{P}^{(0)}\rb^\top$ and a step size $\eta$, we can then design a gradient descent algorithm to search for the minimizer of \eqref{eq:extra_decompose}; see Algorithm \ref{alg:GD_LRP}.

\begin{algorithm}[!htp]
	\caption{Gradient descent algorithm for VAR(1) model with known $r$ and $d$}
	\label{alg:GD_LRP}
	1: \textbf{Input}: $\bm{Y}$, $\bm{X}$, step size $\eta$, number of iteration $I$, initial values $\bm{C}^{(0)}$, $\bm{R}^{(0)}$, $\bm{P}^{(0)}$, and $\bm{D}^{(0)}$\\[-0.5em]
	2: \textbf{for} $i=0,\dots,I-1$\\[-0.3em]
	3: \hspace*{0.7cm} $\bm{C}^{(i+1)}=\bm{C}^{(i)}-\eta\nabla_{\bm{C}}\mathcal{L}^{(i)}-\eta a\big[2\bm{C}^{(i)}(\bm{C}^{(i)\top}\bm{C}^{(i)}-b^2\bm{I}_d)+\bm{R}^{(i)}\bm{R}^{(i)\top}\bm{C}^{(i)}+\bm{P}^{(i)}\bm{P}^{(i)\top}\bm{C}^{(i)}\big]$\\[-0.3em]
	4: \hspace*{0.7cm} $\bm{R}^{(i+1)}=\bm{R}^{(i)}-\eta\nabla_{\bm{R}}\mathcal{L}^{(i)}-\eta a\big[\bm{R}^{(i)}(\bm{R}^{(i)\top}\bm{R}^{(i)}-b^2\bm{I}_{r-d})+\bm{C}^{(i)}\bm{C}^{(i)\top}\bm{R}^{(i)}\big]$\\[-0.3em]
	5: \hspace*{0.7cm} $\bm{P}^{(i+1)}=\bm{P}^{(i)}-\eta\nabla_{\bm{P}}\mathcal{L}^{(i)}-\eta a\big[\bm{P}^{(i)}(\bm{P}^{(i)\top}\bm{P}^{(i)}-b^2\bm{I}_{r-d})+\bm{C}^{(i)}\bm{C}^{(i)\top}\bm{P}^{(i)}\big]$\\[-0.3em]
	6: \hspace*{0.7cm}
	$\bm{D}^{(i+1)}=\bm{D}^{(i)}-\eta\nabla_{\bm{D}}\mathcal{L}^{(i)}$\\[-0.3em]
	7: \textbf{end for}\\[-0.3em]
	8: \textbf{Return}: $\bm{A}^{(I)}=\lb\bm{C}^{(I)}~\bm{R}^{(I)}\rb\bm{D}^{(I)}\lb\bm{C}^{(I)}~\bm{P}^{(I)}\rb^\top$
\end{algorithm}

\subsection{Initialization of the algorithm}\label{sec:3.2}

The problem in \eqref{eq:extra_decompose} is non-convex, and the initial values $(\bm{C}^{(0)},\bm{R}^{(0)},\bm{P}^{(0)},\bm{D}^{(0)})$ play important roles in the algorithm. Hence, we provide a simple spectral initialization method.

Consider model \eqref{eq:rrVAR} with $\text{rank}(\lb\bm{U}~\bm{V}\rb)=2r-d$ and its equivalent form in \eqref{eq:equivalent_form} with $\lb\bm{C}~\bm{R}\rb$ and $\lb\bm{C}~\bm{P}\rb$ being orthonormal matrices. It then holds that $\bm{U}\bm{U}^\top=\bm{C}\bm{C}^\top+\bm{R}\bm{R}^\top$, $\bm{V}\bm{V}^\top=\bm{C}\bm{C}^\top+\bm{P}\bm{P}^\top$, and
$\bm{U}\bm{U}^\top-\bm{V}\bm{V}^\top=\bm{R}\bm{R}^\top-\bm{P}\bm{P}^\top$. Moreover, since $\bm{C}$ is orthogonal to $\bm{R}$ and $\bm{P}$, we have
$\bm{U}\bm{U}^\top(\bm{I}_p-\bm{V}\bm{V}^\top)=\bm{R}\bm{R}^\top(\bm{I}_p-\bm{P}\bm{P}^\top)$
and
$\bm{V}\bm{V}^\top(\bm{I}_p-\bm{U}\bm{U}^\top)=\bm{P}\bm{P}^\top(\bm{I}_p-\bm{R}\bm{R}^\top)$.
It implies that $\mathcal{M}(\bm{R})$ and $\mathcal{M}(\bm{P})$ are the subspaces spanned by the first $r-d$ left singular vectors of $\bm{U}\bm{U}^\top(\bm{I}_p-\bm{V}\bm{V}^\top)$ and $\bm{V}\bm{V}^\top(\bm{I}_p-\bm{U}\bm{U}^\top)$. In addition, $\bm{D}=\lb\bm{C}~\bm{R}\rb^\top\bm{A}\lb\bm{C}~\bm{P}\rb$ and
$(\bm{I}_p-\bm{R}\bm{R}^\top)(\bm{I}_p-\bm{P}\bm{P}^\top)(\bm{U}\bm{U}^\top+\bm{V}\bm{V}^\top)(\bm{I}_p-\bm{R}\bm{R}^\top)(\bm{I}_p-\bm{P}\bm{P}^\top)=2\bm{C}\bm{C}^\top$.

The above finding motivates us to use a reduced-rank VAR estimation \citep{velu2013multivariate} to construct an initialization. Specifically, denote by $\widehat{\bm{H}}\in\mathbb{O}^{p\times r}$ the first $r$ eigenvectors of $\bm{Y}\bm{X}^\top(\bm{X}\bm{X}^\top)^{-1}\bm{X}\bm{Y}^\top$, corresponding to the $r$ largest eigenvalues in the decreasing order, and then the reduced-rank VAR estimation  has an explicit form of
\begin{equation}\label{eq:rr_init}
	\widetilde{\bm{A}}_\text{RR}(r)=\argmin_{\text{rank}(\bm{A})\leq r}\|\bm{Y}-\bm{A}\bm{X}\|_\text{F}^2=\widehat{\bm{H}}\widehat{\bm{H}}^\top\bm{Y}\bm{X}^\top(\bm{X}\bm{X}^\top)^{-1}.
\end{equation}
As a result, the following procedure is suggested for initialization:
\begin{itemize}
	\item[(i.)] Conduct SVD to the reduced-rank VAR estimator: $\widetilde{\bm{A}}_\text{RR}(r)=\widetilde{\bm{U}}\widetilde{\bm{S}}\widetilde{\bm{V}}^\top$;\vspace{-0.3cm}
	\item[(ii.)] Calculate the top $r-d$ left singular vectors of $\widetilde{\bm{U}}\widetilde{\bm{U}}^\top(\bm{I}_p-\widetilde{\bm{V}}\widetilde{\bm{V}}^\top)$ and $\widetilde{\bm{V}}\widetilde{\bm{V}}^\top(\bm{I}_p-\widetilde{\bm{U}}\widetilde{\bm{U}}^\top)$, and denote them by $\widetilde{\bm{R}}$ and $\widetilde{\bm{P}}$, respectively;\vspace{-0.3cm}
	\item[(iii.)] Calculate the top $d$ eigenvectors of $(\bm{I}_p-\widetilde{\bm{R}}\widetilde{\bm{R}}^\top)(\bm{I}_p-\widetilde{\bm{P}}\widetilde{\bm{P}}^\top)(\widetilde{\bm{U}}\widetilde{\bm{U}}^\top+\widetilde{\bm{V}}\widetilde{\bm{V}}^\top)(\bm{I}_p-\widetilde{\bm{R}}\widetilde{\bm{R}}^\top)(\bm{I}_p-\widetilde{\bm{P}}\widetilde{\bm{P}}^\top)$, and denote it by $\widetilde{\bm{C}}$;\vspace{-0.3cm}
	\item[(iv.)] Calculate $\widetilde{\bm{D}}=\lb\widetilde{\bm{C}}~\widetilde{\bm{R}}\rb^\top\widetilde{\bm{A}}_\text{RR}(r)\lb\widetilde{\bm{C}}~\widetilde{\bm{P}}\rb$;\vspace{-0.3cm}
	\item[(v.)] Set the initialization to $\bm{C}^{(0)}=b\widetilde{\bm{C}}$, $\bm{R}^{(0)}=b\widetilde{\bm{R}}$, $\bm{P}^{(0)}=b\widetilde{\bm{P}}$, and $\bm{D}^{(0)}=b^{-2}\widetilde{\bm{D}}$.
\end{itemize}

\subsection{Rank selection and common dimension selection}\label{sec:3.3}

The rank $r$ and common dimension $d$ are assumed to be known in the previous two subsections, but they are 
unknown in most real applications. Here we propose a two-stage selection procedure to select them and relegate its theoretical justification to Section \ref{sec:4}.

A ridge-type ratio method \citep{xia2015consistently} is first introduced to select the rank $r$, regardless of the existence of the common subspace. Specifically, we first give a pre-specified upper bound $\bar{r}=c\cdot r\ll p$ for some $c>1$, and then calculate the estimate $\widetilde{\bm{A}}_\text{RR}(\bar{r})$ in \eqref{eq:rr_init}.  
Denote by $\widetilde{\sigma}_1\geq \widetilde{\sigma}_2\geq\cdots\geq\widetilde{\sigma}_{\bar{r}}$ its singular values, and then the rank $r$ can be selected by
\begin{equation}\label{eq:ridge_ratio}
	\widehat{r}=\argmin_{1\leq i\leq \bar{r}-1}\frac{\widetilde{\sigma}_{i+1}+s(p,T)}{\widetilde{\sigma}_{i}+s(p,T)},
\end{equation}
where the ridge parameter $s(p,T)$ is a positive sequence depending on $p$ and $T$.

The proposed method is not sensitive to the choice of $\bar{r}$ as long as it is greater than $r$. Thus, for large datasets with a large dimension $p$, we can choose the upper bound $\bar{r}$ to be reasonably large but much smaller than $p$. When $p$ is small, we may even simply set $\bar{r}$ to $p$. On the other hand, the ridge parameter $s(p,T)$ is essential for consistent rank selection. We suggest using  $s(p,T)=\sqrt{p\log(T)/(10T)}$, according to Theorem \ref{thm:rank_selection} in Section \ref{sec:4.3}, and its satisfactory performance is observed in our simulation experiments of Section \ref{sec:6}. 

\begin{remark}
    Note that we do not consider the case of $r=0$ throughout this article as it implies that the time series data is a pure white noise sequence, and the ratio estimator naturally rules it out. To formally test whether $r=0$, one may apply the high-dimensional white noise test \citep{li2019testing,tsay2020test}.
\end{remark}

Next, we consider the selection of the common 
dimension $d$ in model \eqref{eq:equivalent_form}.
Denote by $\widehat{\bm{A}}(r,d)$ the estimator obtained from Algorithm \ref{alg:GD_LRP} with the rank $r$ and common dimension $d$, and then the Bayesian information criterion (BIC) can be constructed below,
\begin{equation}\label{eq:BIC}
\text{BIC}(r,d) = Tp\log(\|\bm{Y}-\widehat{\bm{A}}(r,d)\bm{X}\|_\text{F}^2) + d_{\text{CS}}(p,r,d)\log(T),
\end{equation}
where $d_{\text{CS}}(p,r,d)=r(2p-r)-d(p-(d+1)/2)$ is the number of free parameters.
As a result, given $r$, the common dimension $d$ can be selected by
$\widehat{d} = \argmin_{0\leq k\leq r}\text{BIC}(r,d)$.
Note that the BIC in \eqref{eq:BIC} can also be used to select $r$ and $d$ simultaneously, but it would be  time-consuming in practice.


\subsection{The case of VAR($\ell$) models}
\label{sec:3.4}

This subsection extends the proposed methodology to VAR($\ell$) models with common response and predictor factors.
Suppose that $(r_1,r_2,r_3)$ and $d$ are known. To estimate the parameter tensor $\cm{A}=\cm{G}\times_1\lb\bm{C}~\bm{R}\rb\times_2\lb\bm{C} ~\bm{P}\rb\times_3\bm{L}$, the loss function is
$\mathcal{L}(\bm{C},\bm{R},\bm{P},\bm{L},\cm{G})
=(2T)^{-1}\sum_{t=1}^T\|\bm{y}_t-(\cm{G}\times_1\lb\bm{C}~\bm{R}\rb\times_2\lb\bm{C}~\bm{P}\rb\times_3\bm{L})_{(1)}\bm{x}_{t}\|_2^2$,
where $\bm{x}_t=\lb\bm{y}_{t-1}^\top~\dots~\bm{y}_{t-\ell}^\top\rb^\top$. With regularization parameters $a,b>0$, we can use a gradient descent algorithm to find the following estimators
\begin{align*}\label{eq:extra_decompose_2}
		\left(\widehat{\bm{C}},\widehat{\bm{R}},\widehat{\bm{P}},\widehat{\bm{L}},\cm{\widehat{G}}\right)
		= &\argmin_{\substack{\bm{C}\in\mathbb{R}^{p\times d},\bm{R}\in\mathbb{R}^{p\times (r_1-d)},\\ \bm{P}\in\mathbb{R}^{p\times(r_2-d)},\bm{L}\in\mathbb{R}^{\ell\times r_3},\\\scalebox{0.75}{\cm{G}}\in\mathbb{R}^{r_1\times r_2\times r_3}}} \Bigg\{\mathcal{L}(\bm{C},\bm{R},\bm{P},\bm{L},\cm{G})+  \frac{a}{2}\left\|\lb\bm{C}~\bm{R}\rb^\top\lb\bm{C}~\bm{R}\rb-b^2\bm{I}_{r_1}\right\|_\text{F}^2\\
		& \hspace{35mm}+ \frac{a}{2}\left\|\lb\bm{C}~\bm{P}\rb^\top\lb\bm{C}~\bm{P}\rb-b^2\bm{I}_{r_2}\right\|_\text{F}^2 + \frac{a}{2}\left\|\bm{L}^\top\bm{L}-b^2\bm{I}_{r_3}\right\|_\text{F}^2\Bigg\}.
\end{align*}
For initialization, consider the rank-constrained estimator \citep{wang2019high}
$$\cm{\widetilde{A}}_\text{RR}=\argmin_{\text{rank}(\scalebox{0.7}{\cm{A}}_{(i)})=r_i}(2T)^{-1}\sum_{t=1}^T\|\bm{y}_t-\cm{A}_{(1)}\bm{x}_t\|_2^2,$$
and apply the similar initialization method in Section \ref{sec:3.2} to obtain $(\bm{C}^{(0)},\bm{R}^{(0)},\bm{P}^{(0)},\bm{L}^{(0)},\cm{G}^{(0)})$. In addition, the ridge-type rank selection and the common dimension selection via BIC can also be extended to VAR($\ell$) models. For brevity, the algorithm and implementation details are relegated to Appendix \ref{append:D.3}.

\section{Computational and Statistical Convergence Analysis}\label{sec:4}

Sections \ref{sec:4.1} and \ref{sec:4.2} establish the computational and statistical convergence for the VAR(1) model, respectively. Section \ref{sec:4.3} studies the consistency of the rank and common dimension selection, and Section \ref{sec:4.4} 
provides the theoretical justification for the VAR($\ell$) model. In what follows, we denote $\bm{A}^*$ and $\cm{A}^*$ as the ground truth of the parameter matrix and tensor.

\subsection{Computational convergence analysis}\label{sec:4.1}

The optimization problem in \eqref{eq:extra_decompose} is non-convex, and it is challenging to establish the convergence analysis of Algorithm \ref{alg:GD_LRP}.
To solve it, we introduce some regulatory conditions.

\begin{definition}\label{def:RSC_RSS}
    A function $\mathcal{L}(\cdot): \mathbb{R}^{p\times p} \rightarrow \mathbb{R}$ is restricted strongly convex with parameter $\alpha$ and restricted strongly smooth with parameter $\beta$, if for any matrices $\bm{A},\bm{A}'\in\mathbb{R}^{p\times p}$ of rank $r$,
    \begin{equation}
        \frac{\alpha}{2}\|\bm{A}-\bm{A}'\|_\textup{F}^2 \leq \mathcal{L}(\bm{A})-\mathcal{L}(\bm{A}')-\left\langle\nabla\mathcal{L}(\bm{A}'),\bm{A}-\bm{A}'\right\rangle \leq \frac{\beta}{2}\|\bm{A}-\bm{A}'\|_\textup{F}^2.
    \end{equation}
\end{definition}

\begin{definition}\label{def:deviation}
    For the given rank $r$, common dimension $d$, and the true parameter matrix $\bm{A}^*$, the deviation bound is defined as
    \begin{equation}
        \xi(r,d)=\sup_{\substack{\textup{\bf{[}}\bm{C}~\bm{R}\textup{\bf{]}},\textup{\bf{[}}\bm{C}~\bm{P}\textup{\bf{]}}\in\mathbb{O}^{p\times r}, \bm{D}\in\mathbb{R}^{r\times r},\|\bm{D}\|_\textup{F}=1}}\left\langle\nabla\mathcal{L}(\bm{A}^*),\textup{\bf{[}}\bm{C}~\bm{R}\textup{\bf{]}}\bm{D}\textup{\bf{[}}\bm{C}~\bm{P}\textup{\bf{]}}^\top\right\rangle.
    \end{equation}
\end{definition}

The restricted strong convexity and smoothness of Definition \ref{def:RSC_RSS} are essential in establishing the convergence analysis for many non-convex optimization problems; see \cite{jain2017non} and references therein. 
The deviation bound $\xi(r,d)$ in Definition \ref{def:deviation} characterizes the magnitude of statistical noises projected onto a low-dimensional space of matrices with rank $r$ and common dimension $d$, and we can treat it as a statistical error as in \cite{han2020optimal}. 

For the true parameter matrix $\bm{A}^*$, denote its largest and smallest singular values and its condition number by $\sigma_1=\sigma_1(\bm{A}^*)$, $\sigma_{r}=\sigma_{r}(\bm{A}^*)$, and $\kappa=\sigma_1/\sigma_{r}$, respectively.
Assuming that both $r$ and $d$ are known, we state the convergence analysis of Algorithm \ref{alg:GD_LRP} below.

\begin{theorem}
    \label{thm:gd} 
    Suppose that the loss function $\mathcal{L}(\cdot)$ satisfies the restricted strong convexity and smoothness of Definition \ref{def:RSC_RSS} and the deviation bound in Definition \ref{def:deviation}.
    If $\|\bm{A}^{(0)}-\bm{A}^*\|_\textup{F}\lesssim \sigma_{r}$, $a\asymp\alpha\sigma_1^{2/3}\kappa^{-2}$, $b\asymp\sigma_1^{1/3}$, and $\eta=\eta_0\alpha^{-1} \kappa^2\sigma_1^{-4/3}$ with $\eta_0$ being a positive constant not greater than $1/260$, then it holds that, for all $i\geq 1$,
    \begin{equation}\label{eq:convegence_init}
    \|\bm{A}^{(i)}-\bm{A}^*\|_\textup{F}^2 \lesssim \kappa^2(1-C\eta_0\alpha\beta^{-1}\kappa^{-2})^i\|\bm{A}^{(0)}-\bm{A}^*\|_\textup{F}^2 + \kappa^2\alpha^{-2}\xi^2(r,d).
    \end{equation}
\end{theorem}\vspace{-0.6cm}
For the upper bound in Theorem \ref{thm:gd}, the first term corresponds to optimization errors, while the second term is related to statistical errors.
From Theorem \ref{thm:gd}, the estimation error decreases toward a statistical limit exponentially with respect to iterations.
Moreover, when the parameters $\alpha$, $\sigma_1$, and $\kappa$ are bounded away from zero and infinity, the tuning parameters $a$, $b$, and $\eta$ would be at a constant level and, hence, do not depend on $p$ and $T$.

\subsection{Statistical convergence analysis}\label{sec:4.2}

Consider VAR(1) models in \eqref{eq:rrVAR} and \eqref{eq:equivalent_form}. We first state some general conditions.

\vspace{-0.1cm}
\begin{assumption}\label{asmp:1}
	The parameter matrix $\bm{A}^*$ has a spectral radius strictly less than one.\vspace{-0.1cm}
\end{assumption}

\begin{assumption}\label{asmp:2}
	The error term is  $\bbm{\varepsilon}_t=\bm{\Sigma}_{\bbm{\varepsilon}}^{1/2}\bbm{\zeta}_t$, where $\{\bbm{\zeta}_t\}$ are $i.i.d.$ random vectors with $\mathbb{E}(\bbm{\zeta}_t)=\bm{0}$ and $\textup{var}(\bbm{\zeta}_t)=\bm{I}_p$. Moreover, the entries $(\bbm{\zeta}_{it})_{1\leq i\leq p}$ of $\bbm{\zeta}_t$ are mutually independent and $\tau^2$-sub-Gaussian, i.e. $\mathbb{E}[\exp(\mu\bbm{\zeta}_{it})]\leq \exp(\tau^2\mu^2/2)$ for any $\mu\in\mathbb{R}$ and $1\leq i\leq p$.
\end{assumption}
\vspace{-0.1cm}


\begin{remark}
    Assumption \ref{asmp:1} is sufficient and necessary for the existence of a unique strictly stationary solution to model \eqref{eq:rrVAR} with any finite $p$, and this is consistent with the non-asymptotic framework used in this article. For the case with $p\rightarrow\infty$, we may refer to \cite{zhu2017network} for the definition of strict stationarity, which is given via a mechanism similar to the Cramer-Wold device. Moreover, the Gaussian condition is commonly used in the literature of high-dimensional time series \citep{basu2015regularized}, while the sub-Gaussian condition in Assumption \ref{asmp:2} is more general.
\end{remark}

In the decomposition in \eqref{eq:equivalent_form}, intuitively, the response-specific and predictor-specific subspaces $\bm{R}\in\mathbb{O}^{p\times (r-d)}$ and $\bm{P}\in\mathbb{O}^{p\times (r-d)}$ cannot be too close so that we can separate the common subspace $\bm{C}\in\mathbb{O}^{p\times d}$ out successfully.
Here we use the $\sin\theta$ distance for two spaces.
Specifically, let $s_1\geq \dots\geq s_{r-d}\geq0$ be the singular values of $\bm{R}^\top\bm{P}$. Then, the canonical angles between $\mathcal{M}(\bm{R})$ and $\mathcal{M}(\bm{P})$ can be defined as
$\theta_i(\bm{R},\bm{P})=\arccos(s_i)$ for $1\leq i\leq r-d$.
The following condition is added to the smallest canonical angle between $\mathcal{M}(\bm{R})$ and $\mathcal{M}(\bm{P})$.

\begin{assumption}\label{asmp:3}
    There exists a constant $g_{\min}>0$ such that $\sin\theta_{1}(\bm{R}^*,\bm{P}^*)\geq g_{\min}$.
\end{assumption}

Furthermore, we quantify the temporal and cross-sectional dependency as in \cite{basu2015regularized}.
For any $z\in\mathbb{C}$, let $\mathcal{A}(z)=\bm{I}_p-\bm{A}^*z$ be the matrix polynomial, where $\mathbb{C}$ is the set of all complex numbers. Let $\mu_{\min}(\mathcal{A})=\min_{|z|=1}\lambda_{\min}(\mathcal{A}^\dagger(z)\mathcal{A}(z))$ and $\mu_{\max}(\mathcal{A})=\max_{|z|=1}\lambda_{\max}(\mathcal{A}^\dagger(z)\mathcal{A}(z))$, where $\mathcal{A}^\dagger(z)$ is the conjugate transpose of $\mathcal{A}(z)$. Moreover, denote
\begin{equation}
    \alpha_\textup{RSC}=\frac{\lambda_{\min}(\bm{\Sigma}_{\bbm{\varepsilon}})}{2\mu_{\max}(\mathcal{A})},~\beta_\textup{RSS}=\frac{3\lambda_{\max}(\bm{\Sigma}_{\bbm{\varepsilon}})}{2\mu_{\min}(\mathcal{A})},~M_1=\frac{\lambda_{\max}(\bm{\Sigma}_{\bbm{\varepsilon}})}{\mu_{\min}^{1/2}(\mathcal{A})},~\text{and}~M_2=\frac{\lambda_{\min}(\bm{\Sigma}_{\bbm{\varepsilon}})\mu_{\max}(\mathcal{A})}{\lambda_{\max}(\bm{\Sigma}_{\bbm{\varepsilon}})\mu_{\min}(\mathcal{A})}.
\end{equation}
Based on them, we have the following statistical convergence analysis for Algorithm \ref{alg:GD_LRP}. 

\begin{theorem}\label{thm:stat}
    Suppose that Assumptions \ref{asmp:1}--\ref{asmp:3} hold, 
    $T\gtrsim\max(\tau^4,\tau^2)M_2^{-2}p$, and the conditions in Theorem \ref{thm:gd} are satisfied with $\alpha=\alpha_\textup{RSC}$ and $\beta=\beta_\textup{RSS}$. Then, after $I$-th iteration with $
    I \gtrsim \log(\kappa^{-1}\sigma_1^{-1/3}g_{\min})/\log(1-C\eta_0\alpha_\textup{RSC}\beta_\textup{RSS}^{-1}\kappa^{-2})$, with probability at least $1-4\exp[-CM_2^2\min(\tau^{-2},\tau^{-4})T]-2\exp(-Cp)$,
	\begin{equation}
    \|\bm{A}^{(I)}-\bm{A}^*\|_\textup{F}\lesssim \kappa\alpha_\textup{RSC}^{-1}\tau^2M_1\sqrt{d_\textup{CS}(p,r,d)/T}.
    \end{equation}
\end{theorem}\vspace{-0cm}

The above theorem gives an estimation error bound after a sufficiently large number of iterations. 
When the quantities of $\kappa,\sigma_1,g_{\min},\alpha_\textup{RSC}$ and $\beta_\textup{RSS}$ are bounded away from zero and infinity, the required number of iterations does not depend on the dimension $p$ or the sample size $T$, and this makes sure that the proposed algorithm can be applied to large datasets without any difficulty.
Moreover, the estimated parameter matrix from Algorithm \ref{alg:GD_LRP} has the convergence rate of $\sqrt{d_\text{CS}(p,r,d)/T}$, while the  reduced-rank VAR estimation has the rate of $\sqrt{d_\text{RR}(p,r)/T}$ \citep{negahban2011estimation}. Note that $d_\text{RR}(p,r)-d_\text{CS}(p,r,d)$ roughly equals to $pd$ when both $r$ and $d$ are much smaller than $p$. This  efficiency improvement is due to the fact that the proposed methodology takes into account the possible common subspace.\vspace{-0.1cm}

\subsection{Rank and common dimension selection consistency}\label{sec:4.3}

In this section, we provide theoretical justifications for rank and common dimension selection. 
First, we establish the rank selection consistency for the ridge-type ratio in \eqref{eq:ridge_ratio}.

\begin{theorem}\label{thm:rank_selection}
	Under Assumptions \ref{asmp:1} and \ref{asmp:2}, if $T\gtrsim\max(\tau^4,\tau^2)M_2^{-2}p$, $\alpha_\textup{RSC}^{-1}\tau^2M_1\sqrt{p\bar{r}/T}=o(s(p,T))$,
	$s(p,T)=o(\sigma_r^{-1}\min_{1\leq i\leq r-1}\sigma_{j+1}/\sigma_j)$,
	and $r <\bar{r}$, then $\mathbb{P}(\widehat{r}=r)\to1$ as $T\to\infty$.
\end{theorem}

The conditions in this theorem reduce to $s^{-1}(p,T)\sqrt{p/T}\to 0$ and $s(p,T)\to0$ as $T\to\infty$, when $\sigma_1$, $\sigma_r^{-1}$, $\alpha_\textup{RSC}^{-1}$, $\tau$, and $M_1$ are bounded.
Moreover, the required sample size in Theorem \ref{thm:rank_selection} is the same as that for the  estimation consistency in Theorem \ref{thm:stat}. 

Given that the rank $r$ is known, the following theorem provides theoretical justifications for the proposed BIC in \eqref{eq:BIC}.

\begin{theorem}\label{thm:d_consistency}
	Suppose the conditions in Theorem \ref{thm:stat} hold. Then, $\mathbb{P}(\widehat{d}=d)\to1$ as $T\to\infty$.
\end{theorem}

\subsection{Convergence analysis for VAR($\ell$) models}\label{sec:4.4}

This subsection extends the convergence analysis of the gradient descent algorithm for VAR(1) to VAR($\ell$). We refer the readers to Appendix \ref{append:VAR_L} for the detailed algorithm and implementation.
The computational convergence analysis can be extended from that in Section \ref{sec:4.1}, and is omitted to save space. Here, we focus on the statistical convergence. 

For the VAR ($\ell$) model in \eqref{eq:VAR_ell}, define the matrix polynomial $\mathcal{A}(z)=\bm{I}_p-\bm{A}_1^*z-\bm{A}_2^*z^2-\cdots-\bm{A}_\ell^* z^\ell$, where $z\in\mathbb{C}$, and its stationarity condition is given below. 

\begin{assumption}\label{asmp:4}
	The determinant of $\mathcal{A}(z)$ is not equal to zero for all $|z|<1$.
\end{assumption}

Denote by $\bar{\sigma}=\max_{1\leq i\leq 3}\sigma_1(\cm{A}^*_{(i)})$, $\underline{\sigma}=\min_{1\leq i\leq 3}\sigma_{r_i}(\cm{A}^*_{(i)})$, and $\kappa=\bar{\sigma}/\underline{\sigma}$ the largest and smallest singular values and the condition number of the true parameter tensor, respectively.
As in Section \ref{sec:4.2}, we can similarly define the quantities, $\mu_{\min}(\mathcal{A})$, $\mu_{\max}(\mathcal{A})$, $\alpha_\textup{RSC}$, $\beta_\textup{RSS}$, $M_1$ and $M_2$.
The statistical convergence analysis is given below.

\begin{theorem}\label{thm:VAR_L}
	Suppose that Assumptions \ref{asmp:2}--\ref{asmp:4} hold, $T\gtrsim\max(\tau^4,\tau^2)M_2^{-2}(pr_1+pr_2+\ell r_3)$, $a\asymp\alpha_\textup{RSC}\bar{\sigma}^{3/4}\kappa^{-2}$, $b\asymp\bar{\sigma}^{1/4}$, and $\eta=\eta_0\alpha_\textup{RSC}^{-1}\kappa^2\bar{\sigma}^{-3/2}$. Then, after $I$-th iteration of the gradient descent with
	$I\gtrsim \log(\kappa^{-1}\bar{\sigma}^{-3/4}g_{\min})/\log(1-C\eta_0\alpha_\textup{RSC}\beta_\textup{RSS}^{-1}\kappa^{-2})$, with probability at least $1-4\exp[-CM_2^2\min(\tau^{-2},\tau^{-4})T]-2\exp[-C(r_1r_2r_3+pr_1+pr_2+\ell r_3)]$,
	\begin{equation}
		\|\cm{A}^{(I)}-\cm{A}^*\|_\textup{F}\lesssim\kappa\alpha_\textup{RSC}^{-1}\tau^2M_1\sqrt{d_\textup{CS}(p,\ell,r_1,r_2,r_3,d)/T}.
	\end{equation}
\end{theorem}

From the theorem, the estimation efficiency can be achieved by considering the common structure between $\mathcal{M}(\cm{A}_{(1)})$ and $\mathcal{M}(\cm{A}_{(2)})$.
We can also establish the consistency for rank and common dimension selection, but it is omitted for brevity.

\section{Diverging Eigenvalue Effect}\label{sec:diverging}

\subsection{Diverging eigenvalue and elimination transformation}

In many high-dimensional time series data, it is common to observe the diverging eigenvalue effect in $\bm{\Sigma}_{\bm{y}}=\text{var}(\bm{y}_t)$: all diagonal entries in $\bm{\Sigma}_{\bm{y}}$ are bounded, but the leading eigenvalues of $\bm{\Sigma}_{\bm{y}}$ are diverging to infinity with $p$ increasing. This phenomenon implies the pervasive cross-sectional dependency and has been well studied in the econometrics literature of factor modeling with common factors and idiosyncratic errors \citep{bai2008large}.

If we model the data with the diverging eigenvalue effect via a VAR(1) model in \eqref{eq:VAR1}, the relationship between $\bm{\Sigma}_{\bm{y}}$ and $\bm{\Sigma}_{\bbm{\varepsilon}}$, namely
$\bm{\Sigma}_{\bm{y}}=\bm{A}\bm{\Sigma}_{\bm{y}}\bm{A}^\top+\bm{\Sigma}_{\bbm{\varepsilon}}$,
implies that the diverging eigenvalues of $\bm{\Sigma}_{\bm{y}}$ may be splitted into the autoregression part and white noise part. In other words, at least one of the conditional expectation of the response and white noise have strong cross-sectional dependence.
Under Assumption \ref{asmp:1} for the eigenvalues of $\bm{A}$, some singular values of $\bm{A}$ are allowed to diverge. However, if $\lambda_{\max}(\bm{\Sigma}_{\bbm{\varepsilon}})$ is diverging, the estimation error bound in Theorem \ref{thm:stat} is $O_p(\mu_{\max}(\mathcal{A})\lambda_{\max}(\bm{\Sigma}_{\bbm{\varepsilon}})\sqrt{pr/T})$, resulting in a much larger sample size requirement. To this end, we propose an elimination transformation to remove the diverging eigenvalue effect in white noise errors.

\begin{remark}\label{rml:diveging_singular_value}
    For $\bm{A}=\bm{U}\bm{S}\bm{V}^\top$, the diverging singular values exist typically when $\bm{U}$ is pervasive such that the response factor is related to most or even all of the variables, but the predictor loading $\bm{V}$ is highly sparse. For example, consider the rank-1 $\bm{A}=0.9\bm{1}_p(1,0,0,\cdots,0)^\top$, where $\sigma_1(\bm{A})=0.9\sqrt{p}$ and the nonzero eigenvalue of $\bm{A}$ is 0.9. In this case, the diverging eigenvalues in $\bm{\Sigma}_{\bm{y}}$ may come from the pervasive dependence on the conditional expectation of the response;
    see the empirical evidence in Section \ref{sec:7}.
\end{remark}

Based on the decomposition in \eqref{eq:dynamic_part}, $\bbm{\varepsilon}_{1t}:=\bm{U}^\top\bbm{\varepsilon}_t$ is involved in  the low-dimensional autoregressive model, and $\bbm{\varepsilon}_{2t}:=\bm{U}_\perp^\top\bbm{\varepsilon}_t$, where $\bm{U}_\perp\in\mathbb{O}^{p\times(p-r)}$ such that $\bm{U}^\top\bm{U}_\perp=\bm{0}$, is not related to the parameter matrix $\bm{A}$. Hence, when estimating $\bm{A}$, it is beneficial to eliminate the diverging eigenvalue effect in $\bbm{\varepsilon}_{2t}$ and preserve information in $\mathcal{M}({\bm{U}})$ to avoid model bias.
Suppose $\bm{\Sigma}_{\bbm{\varepsilon}_2}:=\text{var}(\bbm{\varepsilon}_{2t})$ has $K$ diverging eigenvalues, i.e.,
    $\bm{\Sigma}_{\bbm{\varepsilon}_2}=\bm{K}\bm{\Lambda}_{\bbm{\varepsilon}_2}^K\bm{K}^\top+\bm{K}_\perp\bm{\Lambda}_{\bbm{\varepsilon}_2}^{-K}\bm{K}^\top_\perp$, 
where $\bm{\Lambda}_{\bbm{\varepsilon}_2}^K=\text{diag}(\lambda_1(\bm{\Sigma}_{\bbm{\varepsilon}_2}),\dots,\lambda_K(\bm{\Sigma}_{\bbm{\varepsilon}_2}))$ and $\bm{\Lambda}_{\bbm{\varepsilon}_2}^{-K}=\text{diag}(\lambda_{K+1}(\bm{\Sigma}_{\bbm{\varepsilon}_2}),\dots,\lambda_{p-r}(\bm{\Sigma}_{\bbm{\varepsilon}_2}))$ contains the diverging and bounded eigenvalues, respectively, and $\bm{K}$ and $\bm{K}_\perp$ are the corresponding eigenvectors. The transformation 
$\bm{T}_K=\bm{K}(\bm{\Lambda}_{\bbm{\varepsilon}_2}^{K})^{-1/2}\bm{K}^\top+\bm{K}_\perp\bm{K}_\perp^\top$ can remove the diverging eigenvalue effect in $\bm{\Sigma}_{\bbm{\varepsilon}_2}$, since
\begin{equation}
    \text{var}(\bm{T}_K\bbm{\varepsilon}_{2t})=\bm{T}_K\bm{\Sigma}_{\bbm{\varepsilon}_2}\bm{T}_K=\bm{K}\bm{K}^\top+\bm{K}_\perp\bm{\Lambda}_{\bbm{\varepsilon}_2}^{-K}\bm{K}_\perp^\top=[\bm{K}~\bm{K}_\perp]\begin{bmatrix}\bm{I}&\bm{0}_{K\times(p-K)}\\\bm{0}_{(p-K)\times K}&\bm{\Lambda}_{\bbm{\varepsilon}_2}^{-K}\end{bmatrix}[\bm{K}~\bm{K}_\perp]^\top.
\end{equation}

In order to preserve the informative factor loading in $\bm{U}$ when estimating $\bm{A}$, we consider the $\bm{U}$-preserved transformation
$\bm{T}_U=\bm{U}\bm{U}^\top+\bm{U}_\perp\bm{T}_K\bm{U}_\perp^\top$
and denote $\widebar{\bbm{\varepsilon}}_t=\bm{T}_U\bbm{\varepsilon}_t$ such that
\begin{equation}
    \text{var}(\widebar{\bbm{\varepsilon}}_t)=\text{var}(\bm{U}\bbm{\varepsilon}_{1t}+\bm{U}_\perp\bm{T}_K\bbm{\varepsilon}_{2t})=[\bm{U}~\bm{U}_\perp]\begin{bmatrix}\bm{\Sigma}_{\bbm{\varepsilon}_1} & \mathbb{E}[\bbm{\varepsilon}_{1t}\bbm{\varepsilon}_{2t}^\top\bm{T}_K]\\ \mathbb{E}[\bm{T}_K\bbm{\varepsilon}_{2t}\bbm{\varepsilon}_{1t}^\top] & \bm{T}_K\bm{\Sigma}_{\bbm{\varepsilon}_2}\bm{T}_K
    \end{bmatrix}[\bm{U}~\bm{U}_\perp]^\top.
\end{equation}
If the eigenvalues of $\bm{\Sigma}_{\bbm{\varepsilon}_1}$ are bounded, all eigenvalues of $\text{var}(\widebar{\bbm{\varepsilon}}_t)$ are bounded. In other words, the diverging eigenvalue effects of $\bm{\Sigma}_{\bbm{\varepsilon}}$ in $\mathcal{M}^\perp(\bm{U})$ can be removed.
By the $\bm{U}$-preseved property of the transformation $\bm{T}_U$, we have that $\bm{A}=\bm{T}_U\bm{A}$. Thus, the reduced-rank VAR model in \eqref{eq:rrVAR} implies that
$\widebar{\bm{y}}_t\equiv\bm{T}_U\bm{y}_t=\bm{T}_U\bm{A}\bm{y}_{t-1}+\bm{T}_U\bbm{\varepsilon}_t=\bm{A}\bm{y}_{t-1}+\widebar{\bbm{\varepsilon}}_t$,
in which the diverging eigenvalue effects in the white noise innovations are alleviated. 

\subsection{Estimation methodology}

First, the reduced-rank VAR(1) model can be formulated to the factor model in \eqref{eq:structuralFM}
\begin{equation}
    \label{eq:VAR_in_FM}
    \bm{y}_t=\bm{U}\bm{S}\bm{V}^\top\bm{y}_{t-1}+\bbm{\varepsilon}_t
    =\bm{U}(\bm{S}\bm{V}^\top\bm{y}_{t-1}+\bbm{\varepsilon}_{1t})+\bm{U}_\perp\bbm{\varepsilon}_{2t}:=\bm{U}\bm{f}_t+\bm{U}_\perp\bbm{\varepsilon}_{2t},
\end{equation}
where $\bm{f}_t$ is an $r$-dimensional dynamic factor and $\bbm{\varepsilon}_{2t}$ is a white noise.
Following the literature of factor models with dynamically dependent factors and white noise errors \citep{lam2012factor,gao2021modeling}, we consider the autocovariance matrices $\bm{\Sigma}_\bm{y}(j)=\mathbb{E}[\bm{y}_t\bm{y}_{t-j}^\top]$, for $j\geq0$. It follows from \eqref{eq:VAR_in_FM} that
    $\bm{\Sigma}_{\bm{y}}(j)=\bm{U}\bm{\Sigma}_{\bm{f}}(j)\bm{U}^\top+\bm{U}\bm{\Sigma}_{\bm{f}\bbm{\varepsilon}_2}(j)\bm{U}_\perp^\top$, for $j\geq1$,
where $\bm{\Sigma}_{\bm{f}}(j)=\mathbb{E}[\bm{f}_t\bm{f}_{t-j}^\top]$ and $\bm{\Sigma}_{\bm{f}\bbm{\varepsilon}_2}(j)=\mathbb{E}[\bm{f}_t\bbm{\varepsilon}_{2,t-j}^\top]$. For a prespecified integer $N>0$, define $\bm{M}=\sum_{j=1}^N\bm{\Sigma}_{\bm{y}}(j)\bm{\Sigma}_{\bm{y}}(j)^\top$,
and $\mathcal{M}(\bm{U})$ is the subspace spanned by the first $r$ eigenvectors of $\bm{M}$. Given $\bm{U}_\perp$, the covariance matrix of $\bbm{\varepsilon}_{2t}$ is $\bm{U}_\perp^\top\bm{\Sigma}_{\bm{y}}\bm{U}_\perp$ and its first $K$ eigenvectors are $\bm{K}$. 

Therefore, we can first obtain the estimates $\widehat{\bm{U}}$ and $\widehat{\bm{U}}_\perp$ by calculating the first $r$ and last $p-r$ eigenvectors of $\widehat{\bm{M}}=\sum_{j=1}^N\widehat{\bm{\Sigma}}_{\bm{y}}(j)\widehat{\bm{\Sigma}}_{\bm{y}}(j)^\top$, respectively, where each sample autocovariance is calculated via $\widehat{\bm{\Sigma}}_{\bm{y}}(j)=(T-j)^{-1}\sum_{t=j+1}^T\bm{y}_t\bm{y}_{t-j}^\top$. The diverging eigenvalues and the corresponding eigenvectors of $\text{var}(\bbm{\varepsilon}_{2t})$ can be estimated by the first $K$ eigenvalues and eigenvectors of $\widehat{\bm{U}}_\perp^\top\widehat{\bm{\Sigma}}_{\bm{y}}\widehat{\bm{U}}_\perp$, denoted by $\widehat{\bm{\Lambda}}^K_{\bbm{\varepsilon}_2}$ and $\widehat{\bm{K}}$.
Then, we can obtain the estimated transformations $\widehat{\bm{T}}_K=\widehat{\bm{K}}(\widehat{\bm{\Lambda}}^K_{\bbm{\varepsilon}_2})^{-1/2}\widehat{\bm{K}}^\top+\widehat{\bm{K}}_\perp\widehat{\bm{K}}_\perp^\top$ and $\widehat{\bm{T}}_U=\widehat{\bm{U}}\widehat{\bm{U}}^\top+\widehat{\bm{U}}_\perp\widehat{\bm{T}}_K\widehat{\bm{U}}^\top_\perp$, and apply the transformation to obtain $\widetilde{\bm{y}}_t=\widehat{\bm{T}}_U\bm{y}_t$. Let the matrix $\widetilde{\bm{Y}}=[\widetilde{\bm{y}}_1~\cdots~\widetilde{\bm{y}}_T]\in\mathbb{R}^{p\times T}$ collect all transformed response vectors, and we can use the transformed response $\widetilde{\bm{Y}}$ and the original predictor $\bm{X}$ in the methods described in Section \ref{sec:3} to complete the estimation procedure.

When the number of factors $r$ in \eqref{eq:VAR_in_FM} is unknown, we may use the eigenvalue ridge-type ratios \citep{xia2015consistently} of $\widehat{\bm{M}}$ to estimate it numerically. To determine the number of diverging eigenvalues $K$, we may also calculate the eigenvalue ridge-type ratios of $\widehat{\bm{U}}_\perp^\top\widehat{\bm{\Sigma}}_{\bm{y}}\widehat{\bm{U}}_\perp$, or select the diverging eigenvalues based on a threshold $p^\delta$ for some $\delta\in(0,1)$.

Finally, for the general VAR($\ell$) model, we can also use the same factor model estimation method to obtain $\widehat{\bm{T}}_U$ and $\widetilde{\bm{y}}_t$, and apply the transformed response $\widetilde{\bm{y}}_t$ and the original predictors $\bm{y}_{t-1},\cdots,\bm{y}_{t-\ell}$ to the proposed estimation procedure in Section \ref{sec:3.4}.

\subsection{Theoretical results}

In this subsection, we focus on the VAR(1) model, as the results can  easily be extended to the general VAR($\ell$) models.
When the singular values of $\bm{A}^*$ are diverging, the spectral measurements $\mu_{\max}(\mathcal{A})$ and $\mu^{-1}_{\min}(\mathcal{A})$ may be diverging as well.
In the setting with diverging eigenvalue effect, representing the cross-sectional and temporal dependency in $\bm{y}_t$ by the spectral measurements in Section \ref{sec:4} may result in a loose result. As we assume the leading eigenvalues of $\bm{\Sigma}_{\bm{y}}$ are diverging, 
it is more natural to impose the following assumptions on the explicit diverging rates of the specific components in the model.
\begin{assumption}\label{asmp:5}
    All nonzero singular values of $\bm{A}^*$ scale as $p^{\delta_s}$ for some $\delta_s\in[0,1/2]$.
\end{assumption}

\begin{assumption}\label{asmp:6}
    The first $r$ eigenvalues of $\bm{U}^{*\top}\bm{\Sigma}_{\bbm{\varepsilon}}\bm{U}^*$ scale as $p^{\delta_u}$ for some $\delta_u\in[0,1]$. The first $K$ eigenvalues of $\bm{U}_\perp^{*\top}\bm{\Sigma}_{\bbm{\varepsilon}}\bm{U}_\perp^{*}$ scale as $p^{\delta_u'}$ for some $\delta_u'\in(0,1]$, and the other eigenvalues are bounded. In addition, the first $r$ eigenvalues of $\bm{V}^{*\top}\bm{\Sigma}_{\bm{y}}\bm{V}^*$ scale as $p^{\delta_v}$ for some $\delta_v\in[0,1]$.
\end{assumption}

The term $\delta_s$ in Assumption \ref{asmp:5} characterizes the strength of diverging singular values in $\bm{A}^*$. Note that it is possible that $\rho(\bm{A}^*)<1$ but $\sigma_1(\bm{A}^*)$ diverges to infinity. However, when $d=r$, we must have $\delta_s=0$ to ensure stationarity. The terms $\delta_u$ and $\delta_u'$ in Assumption \ref{asmp:6} represent the diverging eigenvalue effects of the white noise errors in the response subspace $\mathcal{M}(\bm{U}^*)$ and its orthogonal complement, respectively. The term $\delta_v$ characterizes the signal strength in the predictor factor $\bm{V}^{*\top}\bm{y}_t$. Based on these diverging rates, denote 
$\alpha_{\textup{RSC}}'=\lambda_{\min}(\bm{V}^{*\top}\bm{\Sigma}_{\bm{y}}\bm{V}^*)/2$, $\beta_{\textup{RSS}}'=2\lambda_{\max}(\bm{V}^{*\top}\bm{\Sigma}_{\bm{y}}\bm{V}^*)$ and $M_1'= Cp^{(\delta_u+\delta_v)/2}$
as the variants of $\alpha_\text{RSC}$, $\beta_\text{RSS}$ and $M_1$ defined in Section \ref{sec:4}. These quantities are related to the explicit diverging rates and are more suitable to derive the theory here.

Suppose that the number of diverging eigenvalues $K$ are known, we present the theoretical guarantees for the estimation procedures with $\widetilde{\bm{Y}}$ and $\bm{X}$ used in Algorithm \ref{alg:GD_LRP}.

\begin{theorem}\label{thm:diverging}
    Suppose that Assumptions \ref{asmp:1}-\ref{asmp:3}, \ref{asmp:5} and \ref{asmp:6} hold, $T\gtrsim\max(\tau^4,\tau^2)p$, $\max(\delta_u,\delta_v+2\delta_s)\geq\delta_u'$, and the conditions in Theorem \ref{thm:gd} are satisfied with $\alpha=\alpha'_\textup{RSC}$ and $\beta=\beta'_\textup{RSS}$. Then, after the $I$-th iteration with $I\gtrsim\log(p^{\delta_s/3}g_{\min})/\log(1-C\eta_0)$, with probability at least $1-4\exp[-C\min(\tau^{-2},\tau^{-4})T]-C\exp(-Cp)$,
    \begin{equation}
        \|\bm{A}^{(I)}-\bm{A}^*\|_\textup{F}\lesssim (\alpha'_\textup{RSC})^{-1}\tau^2M_1'\sqrt{\frac{d_\textup{CS}(p,r,d)}{T}} = \tau^2p^{(\delta_u-\delta_v)/2}\sqrt{\frac{d_\textup{CS}(p,r,d)}{T}}.
    \end{equation}
\end{theorem}

Theorem \ref{thm:diverging} presents the statistical convergence rate of the proposed estimator for the data with the diverging eigenvalue effect. First, the additional signal strengh condition, $\max(\delta_u,\delta_v+2\delta_s)\geq\delta_u'$, implies that the signal strength in the low-dimensional dynamic part $\bm{U}^{*\top}\bm{y}_t=\bm{S}^*\bm{V}^{*\top}\bm{y}_{t-1}+\bm{U}^{*\top}\bbm{\varepsilon}_t$ is not weaker than that in the white noise part $\bm{U}_\perp^{*\top}\bm{y}_t=\bm{U}_\perp^{*\top}\bbm{\varepsilon}_t$, such that the proposed factor modeling method can work. Second, if $\tau$ is fixed, the upper bound scales as $O_p(\sqrt{p^{1+\delta_u-\delta_v}/T})$. If $\delta_u\leq\delta_v$, i.e., the strength of the predictor factors is stronger than the that of white noise errors in $\mathcal{M}(\bm{U}^*)$, the rate is even faster than that in Theorem \ref{thm:stat}. Third, if we ignore the diverging effect in the data and apply the standard estimation procedure in Section \ref{sec:3}, the resulting rate in Theorem \ref{thm:stat}, scaling as $O_p(\mu_{\max}(\mathcal{A})\lambda_{\max}(\bm{\Sigma}_{\bbm{\varepsilon}})\sqrt{p/T})$, can be much larger than that in Theorem \ref{thm:diverging}, which confirms the efficacy of the proposed transformation method in removing the diverging effect. Finally, when $\delta_s>0$, i.e., $\|\bm{A}^*\|_\text{F}$ diverges with $p$, we may also consider the relative estimation error $\|\bm{A}^{(I)}-\bm{A}^*\|_\text{F}/\|\bm{A}^*\|_\text{F}\lesssim \tau^2\sqrt{p^{1+\delta_u-\delta_v-2\delta_s}/T}$ for the estimation consistency of factor loadings.

\section{Sparsity on Factor Loading Matrices}\label{sec:5}

The convergence analysis in Sections \ref{sec:4} and \ref{sec:diverging} requires $T\gtrsim p$ or $T\gtrsim p^{1+\delta_u-\delta_v}$; however, the number of series $p$ could be comparable to or even larger than the sample size $T$ for some real applications.
For this case, the common and specific factors are related to only a small subset of variables, while many other variables have no contribution in extracting factors.
Thus, in order to improve the estimation efficiency and model interpretation, it is of interest to further consider additional row-wise sparsity structure to factor loading matrices. If the factor loadings are sparse, it is unlikely to have strong cross-sectional dependence, so in this section, the diverging eigenvalue effect is not considered.

We consider the case of VAR(1) models and the result can be extended to that of general VAR($\ell$) models. Suppose that there are at most $s_c$, $s_r$, and $s_p$ variables related to the common, response-specific, and predictor-specific factors, respectively.
Let $\mathbb{S}(p,q,s)=\{\bm{M}\in\mathbb{R}^{p\times q}: \text{the number of non-zero rows of }\bm{M}\text{ is at most }s\}$. We can extend the regularized estimation method in \eqref{eq:extra_decompose} to encourage the row-wise sparsity on $\bm{C}$, $\bm{R}$, and $\bm{P}$,
\begin{equation}
\begin{split}
\left(\widehat{\bm{C}},\widehat{\bm{R}},\widehat{\bm{P}},\widehat{\bm{D}}\right)
= \argmin_{\substack{\bm{C}\in\mathbb{S}(p,d,s_c),\bm{D}\in\mathbb{R}^{r\times r},\\ \bm{R}\in\mathbb{S}(p,r-d,s_r),\bm{P}\in\mathbb{S}(p,r-d,s_p)}}& \Bigg\{\mathcal{L}(\bm{C},\bm{R},\bm{P},\bm{D})+  \frac{a}{2}\left\|\lb\bm{C}~\bm{R}\rb^\top\lb\bm{C}~\bm{R}\rb-b^2\bm{I}_r\right\|_\text{F}^2\\
&\hspace{29mm} + \frac{a}{2}\left\|\lb\bm{C}~\bm{P}\rb^\top\lb\bm{C}~\bm{P}\rb-b^2\bm{I}_r\right\|_\text{F}^2\Bigg\}.
\end{split}
\end{equation}
Accordingly, for the row-wise sparsity constraint, the hard thresholding operation $\text{HT}(\bm{M},s)$ can be added to the gradient descent algorithm, where $\text{HT}(\bm{M},s)$ projects the matrix $\bm{M}\in\mathbb{R}^{p\times q}$ onto $\mathbb{S}(p,q,s)$ by keeping the top $s$ largest rows of $\bm{M}$ in terms of Euclidean norm and truncating the rest to zeros; see Algorithm \ref{alg:HT_GD}.
Note that the row-wise sparsity structure is invariant with respect to rotation. 

\begin{algorithm}[!htp]
    \caption{Hard thresholding gradient descent algorithm for sparse estimation}
    \label{alg:HT_GD}
    1: \textbf{Input}: $\bm{Y}$, $\bm{X}$, $\eta$, $I$, $\bm{C}^{(0)}$, $\bm{R}^{(0)}$, $\bm{P}^{(0)}$, $\bm{D}^{(0)}$, and sparsity level $(s_c,s_r,s_p)$.\\[-0.3em]
    2: \textbf{for} $i=0,\dots,I-1$\\[-0.3em]
    3: \hspace*{0.7cm} Use lines 3-6 in Algorithm \ref{alg:GD_LRP} to obtain $\widetilde{\bm{C}}^{(i+1)}$, $\widetilde{\bm{R}}^{(i+1)}$, $\widetilde{\bm{P}}^{(i+1)}$, and $\bm{D}^{(i+1)}$\\[-0.3em]
    4: \hspace*{0.7cm} $\bm{C}^{(i+1)}=\text{HT}(\widetilde{\bm{C}}^{(i+1)},s_c)$, $\bm{R}^{(i+1)}=\text{HT}(\widetilde{\bm{R}}^{(i+1)},s_r)$, and $\bm{P}^{(i+1)}=\text{HT}(\widetilde{\bm{P}}^{(i+1)},s_p)$\\[-0.3em]
    5: \textbf{end for}\\[-0.3em]
    6: \textbf{Return}: $\bm{A}^{(I)}=\lb\bm{C}^{(I)}~\bm{R}^{(I)}\rb\bm{D}^{(I)}\lb\bm{C}^{(I)}~\bm{P}^{(I)}\rb^\top$
\end{algorithm}

For initialization of Algorithm \ref{alg:HT_GD},  we conduct the $L_1$ regularized least squares estimation
$\widetilde{\bm{A}}_{L_1} = \argmin\{(2T)^{-1}\sum_{t=1}^T\|\bm{y}_t-\bm{A}\bm{y}_{t-1}\|_2^2 + \lambda\|\bm{A}\|_1\}$.
Denote by $\widetilde{\bm{U}}$ and $\widetilde{\bm{V}}$ the first $r$ left and right singular vectors of $\widetilde{\bm{A}}_{L_1}$, respectively, and then the spectral initialization method in Section \ref{sec:3.2} can be used to obtain $(\widetilde{\bm{C}},\widetilde{\bm{R}},\widetilde{\bm{P}},\widetilde{\bm{D}})$. Finally we set $\bm{C}^{(0)}=\text{HT}(b\widetilde{\bm{C}},s_c)$, $\bm{R}^{(0)}=\text{HT}(b\widetilde{\bm{R}},s_r)$, $\bm{P}^{(0)}=\text{HT}(b\widetilde{\bm{P}},s_p)$, and $\bm{D}^{(0)}=b^{-2}\widetilde{\bm{D}}$. 
The rank and common dimension can similarly be selected by the proposed methods in Section \ref{sec:3.3}.
The sparsity level $(s_c,s_r,s_p)$ can be determined by the domain knowledge or estimated based on $(\widetilde{\bm{C}},\widetilde{\bm{R}},\widetilde{\bm{P}})$. 

Assume that the numbers of nonzero rows in the ground truth $\bm{C}^*$, $\bm{R}^*$, and $\bm{P}^*$ are $s_c^*$, $s_r^*$, and $s_p^*$, respectively, and the sparsity levels in Algorithm \ref{alg:HT_GD} satisfy that $s_c\geq (1+\gamma)s_c^*$, $s_r\geq (1+\gamma)s_r^*$, and $s_p\geq (1+\gamma)s_p^*$, for some constant $\gamma>0$.

\begin{theorem}\label{thm:sparse}

    Let $s=\max(s_c+s_r,s_c+s_p)$ and $S^*=(s_c^*+s_r^*)(s_c^*+s_p^*)$. Suppose that $T\gtrsim\max(\tau^4,\tau^2)M_2^{-2}\{sr+r^2+s\min(\log(p),\log(ep/s))+S^*\log(p)\}$, $\lambda\asymp\tau^2M_1\sqrt{S^*\log(p)/T}$, $\gamma\gtrsim \alpha_\textup{RSC}^{-2}\beta_\textup{RSS}^2\kappa^4$, and other conditions in Theorem \ref{thm:stat} hold. Then, after $I$-th iteration of Algorithm \ref{alg:HT_GD} with
    $I\gtrsim\log(\kappa^{-1}\sigma_1^{-1/3}g_{\min})/\log(1-C\eta_0\alpha_\textup{RSC}^2\beta_\textup{RSS}^{-2}\kappa^{-4})$, with probability at least $1-4\exp[-CM_2^2\min(\tau^{-2},\tau^{-4})T]-2\exp[-C\log(p)]$,
	\begin{equation}
		\|\bm{A}^{(I)}-\bm{A}^*\|_\textup{F}\lesssim\kappa\alpha_\textup{RSC}^{-1}\tau^2M_1\sqrt{[sr+r^2+s\min(\log(p),\log(ep/s))]/T}.
	\end{equation}

\end{theorem}

From the above theorem, the estimation efficiency of the sparsity-constrained estimator is improved significantly. Specifically, when both $s$ and $r$ are much smaller than $p$, the required sample size is reduced from $T\gtrsim p$ to $T\gtrsim s\log(p)$. In other words, the proposed sparsity-constrained method can be applied to the case with $T\ll p$.

\section{Simulation Studies}\label{sec:6}

\subsection{VAR with common factors}

We conduct two simulation experiments to evaluate the finite-sample performance of the proposed estimation methods. The number of replications is set to 500 for each experiment.

In the first experiment, the VAR(1) model in \eqref{eq:equivalent_form} is considered with $r=3$, $d=0,1,2,3$, and $\bm{\Sigma}_{\bbm{\varepsilon}}=\bm{I}_p$. 
The dimension is $p=40$ or 100, and we consider $T\in\{500,600,700,800\}$ for $p=40$ and $T\in\{1000,1200,1400,1600\}$ for $p=100$. 
The orthonormal matrices $\bm{C},\bm{R},\bm{P}$ and $\bm{O}_1,\bm{O}_2\in\mathbb{O}^{r\times r}$are generated randomly in each replication such that $\bm{C}^\top\bm{R}=\bm{C}^\top\bm{P}=\bm{0}_{d\times(r-d)}$.
Moreover, let $\bm{S}=\text{diag}(c_1,c_2,c_3)$ with each $c_i\sim_{i.i.d.}\text{Unif}(0.8,1.5)$ in each replication.
As a result, from \eqref{eq:rrVAR} and \eqref{eq:equivalent_form}, the parameter matrix assumes the form $\bm{A}=\lb\bm{C}~\bm{R}\rb\bm{O}_1^{\top}\bm{S}\bm{O}_2\lb\bm{C}~\bm{P}\rb^\top$.

The proposed methodology in Sections \ref{sec:3.1}-\ref{sec:3.3} is applied to the generated data. 
The ridge-type ratio method with $\bar{r}=10$ and $s(p,T)=\sqrt{p\log(T)/(10T)}$ is used to select $r$, and BIC in \eqref{eq:BIC} is used to select $d$. 
Table \ref{tbl:r&d_ex1} lists the percentages of correct rank and common dimension selection.
It can be seen that both can be correctly selected almost for all cases, and the percentages of correct selection in rank and common dimension both increase as the sample size $T$ increases.
This confirms the selection consistency derived in Section \ref{sec:4.3}.

We next compare the estimation efficiency between two models: the proposed model with a common subspace (CS) in \eqref{eq:equivalent_form} and the reduced-rank (RR) model in \eqref{eq:rrVAR}.
The median of estimation errors, $\|\widehat{\bm{A}}-\bm{A}^*\|_\text{F}$, over 500 replications is presented in Figure \ref{fig:sim1}, and the 0.75- and 0.25-th quantiles are also given in terms of error bars.
The two models have similar performances when $d=0$. However, when $d\geq 1$, the proposed model is more efficient, and the efficiency gain increases as the common dimension $d$ becomes larger. This result shows that it pays to explore 
 the common subspace between response and predictor spaces. 

The data generating process of the second experiment is a VAR($\ell$) model in \eqref{eq:VAR_ell} with $\ell=5$, and its parameter tensor is in the form of \eqref{eq:common_tensor_decomp2} with ranks $r_1=r_2=r_3=3$ and common dimension $0\leq d\leq 3$. 
The dimension is $p=30$ or 50, and we consider $T\in\{500,600,700,800\}$ for $p=30$ and $T\in\{1000,1200,1400,1600\}$ for $p=50$. 
We generate the orthonormal matrices, $\bm{C}$, $\bm{R}$, $\bm{P}$, $\bm{L}$, $\bm{O}_1$, $\bm{O}_2$ and $\bm{O}_3$, randomly for each replication.
Let $\cm{S}\in\mathbb{R}^{3\times 3\times 3}$ be a super-diagonal tensor with diagonal entries $\{c_1,c_2,c_3\}$, and $c_i$'s are generated by 
the same method as that of the first experiment.
The parameter tensor has the form  $\cm{A}=\cm{G}\times_1\lb\bm{C}~\bm{R}\rb\times_2\lb\bm{C}~\bm{P}\rb\times_3\bm{L}$, where $\cm{G}=\cm{S}\times_1\bm{O}_1\times_2 \bm{O}_2\times_3\bm{O}_3$.
The proposed methodology in Section \ref{sec:3.4} and Appendix \ref{append:D.3} is applied to each generated sample.

Table \ref{tbl:r&d_ex2} gives the percentages of correct rank and common dimension selection, respectively, and Figure \ref{fig:sim2} presents the median, 0.75- and 0.25-th quantiles of estimation errors, $\|\cm{\widehat{A}}-\cm{A}^*\|_\text{F}$, from our model (CS) in \eqref{eq:common_tensor_decomp2} and the reduced-rank model (RR) in \citet{wang2019high}.
From Table \ref{tbl:r&d_ex2}, the percentages of correct rank selection are obviously smaller than those of correct common dimension selection.
This is mainly due to the fact that, when tensor ranks are over-selected, the proposed method may still correctly select the common dimension.
All other findings are similar to those of the first experiment.

\subsection{VAR with diverging eigenvalue effect}

Next, we conduct a simulation experiment to investigate how the diverging eigenvalues may affect the estimation procedure and support the proposed methodology in Section \ref{sec:diverging}. The DGP is the same as the first experiment with $p=40$, except for $\bm{\Sigma}_{\bbm{\varepsilon}}=0.5\bm{I}_p+0.5\bm{1}_p\bm{1}_p^\top$, where $\bm{1}_p$ is the $p$-dimensional vector whose entries are all one. Hence, $\bm{\Sigma}_{\bbm{\varepsilon}}$ has one diverging eigenvalue as $\lambda_{\max}(\bm{\Sigma}_{\bbm{\varepsilon}})=0.5p+0.5$. 

In this experiment, we apply the proposed estimation procedure with diverging eigenvalue effect, including common subspace with diverging eigenvalues (CS-DE) and reduced-rank with diverging eigenvalues (RR-DE), and compare them with the standard versions CS and RR. Table \ref{tbl:r&d_ex3} contains the percentages of correct rank and common dimension selection for CS and CS-DE. When $\bm{\Sigma}_{\bbm{\varepsilon}}$ has a diverging eigenvalue, the CS method fails to find the correct rank while CS-DE can consistently estimate the rank and common dimension.
The estimation errors over 500 replications are presented in Figure \ref{fig:exp3}. In all cases of $d$, the errors of RR-DE and CS-DE are significantly smaller than those of the standard methods. When $d>0$, the CS-DE performs better than RR-DE, which supports the theoretical results in Section \ref{sec:diverging}.

\section{An Empirical Example}\label{sec:7}

We apply the proposed methodology to a macroeconomic time series data set with $p=40$ variables of the United States. The data  consist of quarterly economic series from Q3-1959 to Q4-2007 with the length $T=194$. 
These macroeconomic variables, selected by \citet{koop2013forecasting}, can be classified into eight categories: GDP decomposition, NAPM indices, industrial production, housing, interest rates, employment, prices, and others.
All series have been transformed to stationary series and standardized, and seasonal adjustment has also been conducted to all variables except the financial series. 
More information of these macroeconomic 
series can be found in Appendix \ref{append:real_data}.

This dataset has been well studied in various factor models, and the largest eigenvalues of the sample covariance matrix $\widehat{\bm{\Sigma}}_{\bm{y}}$ is 11.68. Hence, we first apply the factor modeling procedure in \citet{gao2021modeling}, but no diverging eigenvalue in the white noise is found.
Following \citet{koop2013forecasting}, we consider a VAR(4) model and the parameter tensor is specified as in \eqref{eq:common_tensor_decomp2}.
The modeling procedure in Section \ref{sec:3.4} is applied to the above high-dimensional macroeconomic time series. 
The estimated ranks are $(\widehat{r}_1,\widehat{r}_2,\widehat{r}_3)=(4,3,2)$, while the selected common dimension is $\widehat{d}=2$, i.e., there are two common, two response-specific, and one predictor-specific factors. 
The singular values of $\cm{\widehat{A}}_{(1)}$ are 8.87, 3.70, 1.15, and 0.56, and the first two singular values can explain the diverging eigenvalue in $\widehat{\bm{\Sigma}}_{\bm{y}}$.
For the factor interpretation, since the tensor decomposition in \eqref{eq:common_tensor_decomp2} is not unique, we standardize the matrices $\widehat{\bm{C}}$, $\widehat{\bm{R}}$ and $\widehat{\bm{P}}$ to be orthonormal, and calculate their projection matrices $\widehat{\bm{C}}\widehat{\bm{C}}^\top$, $\widehat{\bm{R}}\widehat{\bm{R}}^\top$ and $\widehat{\bm{P}}\widehat{\bm{P}}^\top$, which are uniquely defined and can be used to represent the subspaces $\mathcal{M}(\widehat{\bm{C}})$, $\mathcal{M}(\widehat{\bm{R}})$, and $\mathcal{M}(\widehat{\bm{P}})$, respectively; see Section \ref{sec:2} for more details.

Figure \ref{fig:subspace_macro40} plots the calculated projection matrices $\widehat{\bm{C}}\widehat{\bm{C}}^\top$, $\widehat{\bm{R}}\widehat{\bm{R}}^\top$ and $\widehat{\bm{P}}\widehat{\bm{P}}^\top$.
It can be seen that both $\widehat{\bm{P}}\widehat{\bm{P}}^\top$ and $\widehat{\bm{C}}\widehat{\bm{C}}^\top$ are highly sparse, and these nonzero entries exhibit certain clustering pattern, which is consistent with the classification of macroeconomic variables.
Specifically, almost all significant entries in the projection matrix of predictor-specific factors can be observed for the first two categories of variables, GDP decomposition and NAPM indices, while those of common factors are from NAMP indices, industrial production, and housing. 
We may argue that the three fitted predictor factors mainly extract information from four classes of variables, including GDP, NAPM indices, industrial production and housing, for the sake of predicting the future values of all series. 
The predictability of these variables is consistent with our empirical experience: GDP is the most important measure of the current status of an economy, and purchasing manager indices, industrial production indices, and housing starts are widely recognized as leading indicators of economic activities. 
Nevertheless, the estimated response-specific projection matrix is much denser, indicating that almost all economic variables are related to the response-specific factors. The patterns in these projection matrices can also help us interpret the diverging eigenvalues in $\widehat{\bm{\Sigma}}_{\bm{y}}$. Since the response-specific loading $\widehat{\bm{R}}\in\mathbb{R}^{40\times 2}$ is pervasive, strong cross-sectional dependency may exist in the conditional expectation of the response, but the common and predictor-specific factors are only related to a small subset of variables. As discussed in Remark \ref{rml:diveging_singular_value}, the pervasive response loadings and sparse predictor loadings may lead to the large singular values of $\cm{\widehat{A}}_{(1)}$.

On the other hand, it is interesting to observe that the upper left corner of the common projection matrix is almost sparse, while those in predictor-specific and response-specific projection matrices are not. We may argue that the information of GDP extracted for predictors and responses are different in general. The GDP components in the responses are positively correlated, as shown in the green upper-left block in the response-specific projection matrix, whereas in the predictor-specific projection matrix, the first variable, real GDP, is negatively correlated to real personal consumption, private domestic investment, real exports, and government consumption and investment. By definition, as the real GDP is the summation of its decomposition, the negative relationship in the predictor subspace partially cancels out the double counting.

Finally, we compare the proposed model with two other commonly used models, the rank-constrained VAR(4) model (VAR-RR) in \citet{wang2019high} with ranks  $(r_1,r_2,r_3)=(4,3,2)$ and the dynamic factor model in \citet{lam2012factor} with a low-dimensional VAR(4) for factors (DFM-VAR), in terms of rolling forecast.
Specifically, from the time point of Q1-2000 ($t=163$) to Q2-2007 ($t=192$), we fit these three models utilizing all available historical data until time $t-1$ and obtain one-, two- and three-step-ahead forecasts. 
We consider two forecasting tasks: one is to predict all forty variables, and the other is to only forecast the 34th series, CPI for all items, as the inflation rate is one of the typical macroeconomic variables of interest in forecasting.
The average rolling forecast errors for both tasks are summarized in Table \ref{tbl:forecast_macro40}, and it can be seen that our model has the smallest errors in both tasks, especially the overall forecasting. This is due to the fact that, compared to the dynamic factor modeling, our model is able to flexibly extract useful information for responses and predictors. In the meanwhile, in the proposed model, substantial dimension reduction can be further achieved by exploring the possible common subspace between response and predictor factor spaces of the rank-constrained model. The CPI forecasting errors of both methods are quite close, possibly because CPI is only involved in response-specific factor loading and DFM can also estimate it consistently.

\section{Conclusion and Discussion}\label{sec:8}

Vector autoregressive and factor models are two mainstream modeling frameworks for high-dimensional time series, and they have their own strengths in real applications.
This article proposed a new model by focusing on the dependent factor structure of the series, and it was shown by simulation experiments and an empirical example that the proposed 
model enjoys advantages over both VAR and factor models. Theoretical justifications are established for both computational and statistical convergence of the 
new model. 

The research of this article can be extended in two directions.
Firstly, heavy-tailed distributions and outliers are commonly observed in empirical data sets, which violates Assumption \ref{asmp:2}. Robust estimation methods against the heavy-tailed distribution for high-dimensional VAR models have been investigated recently \citep{wang2021robust}, and it is of practical importance to investigate the robust methods for the proposed model.
Secondly, inspired by the emerging literature on matrix and tensor-valued time series \citep{chen2018autoregressive,chen2019factor,wang2021high}, we may generalize the proposed model, methodology, and theory to autoregressive models for matrix and tensor-valued time series.

\linespread{1.5}
\selectfont{}

\setlength{\bibsep}{1pt}
\bibliography{mybib}

\vspace{0.7cm}

\begin{figure}[!htp]
    \begin{center}
        \includegraphics[width=0.7\textwidth]{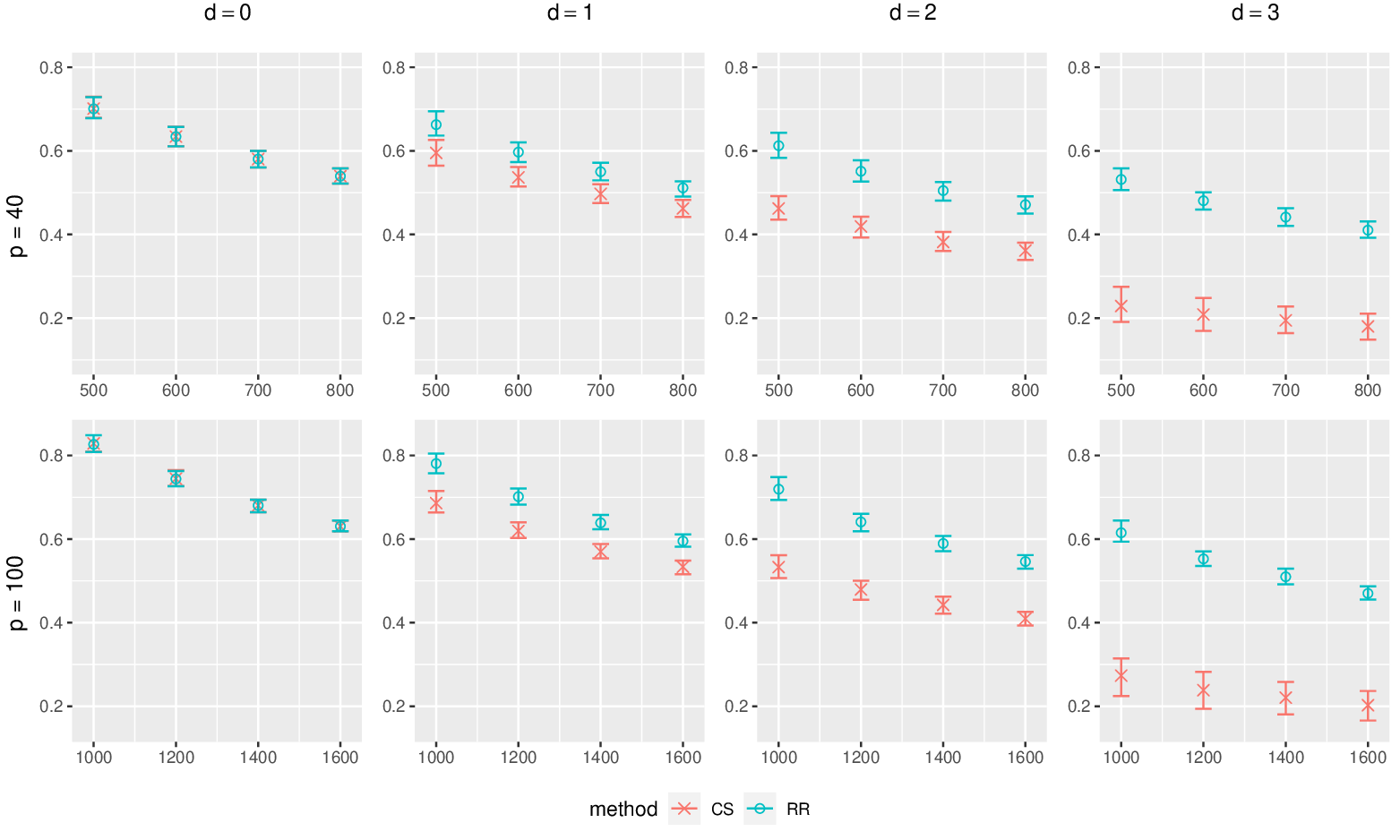}
        \vspace{-1cm}
    \end{center}
    \caption{\small{Plots of estimation errors $\|\widehat{\bm{A}}-\bm{A}^*\|_\text{F}$ of VAR(1) model by common subspace (CS) and reduced-rank (RR) methods. The dimension is $p=40$ (upper panel) or 100 (lower panel).}}
    \label{fig:sim1}
\end{figure}

\vspace{-0.7cm}

\begin{table}[!htp]
    \begin{center}
        \renewcommand{\arraystretch}{1.1}
        \vspace{-0.8cm}
        \caption{\small{Percentages of correct rank and common dimension selection for VAR($1$) models .}}
        \label{tbl:r&d_ex1}
        \renewcommand{\arraystretch}{0.75}
        \vspace{0.2cm}
        \small{\begin{tabular}{r C{1cm} C{1cm} C{1cm} C{1cm} c C{1cm} C{1cm} C{1cm} C{1cm}}
                \hline\hline
                  & \multicolumn{4}{c}{Rank selection} && \multicolumn{4}{c}{Common dimension selection}\\ 
                  \cline{2-5}\cline{7-10}
                $d$ & 0 & 1 & 2 & 3 && 0    & 1    & 2    & 3 \\
                \hline
                &\multicolumn{9}{c}{dimension $p=40$}\\
                 $T=500$ & 97.8 & 97.2 & 96.2 & 98.4  && 98.6 & 97.6 & 95.8 & 98.4 \\
                     600 & 99.8 & 99.6 & 98.6 & 99.8  && 99.4 & 98.8 & 98.4 & 99.8 \\
                     700 & 99.8 & 99.8 & 99.8 & 100   && 100  & 99.8 & 99.8 & 100  \\
                     800 & 99.8 & 100  & 100  & 100   && 100  & 100  & 99.8 & 99.8 \\
                \hline
                &\multicolumn{9}{c}{dimension $p=100$}\\
                $T=1000$ & 94.8 & 92.4 & 92.2 & 93.0 && 93.4 & 91.8 & 85.8 & 93.0 \\
                    1200 & 96.6 & 96.6 & 95.8 & 98.0 && 98.4 & 95.8 & 95.6 & 98.0 \\
                    1400 & 99.0 & 99.4 & 97.8 & 99.6 && 99.8 & 99.2 & 97.6 & 99.2 \\
                    1600 & 99.4 & 98.6 & 98.8 & 99.8 && 100  & 98.6 & 98.8 & 99.8 \\
                \hline
        \end{tabular}}
    \end{center}
\end{table}

\vspace{-2.3cm}

\begin{figure}[!htp]
    \begin{center}
        \includegraphics[width=0.78\textwidth]{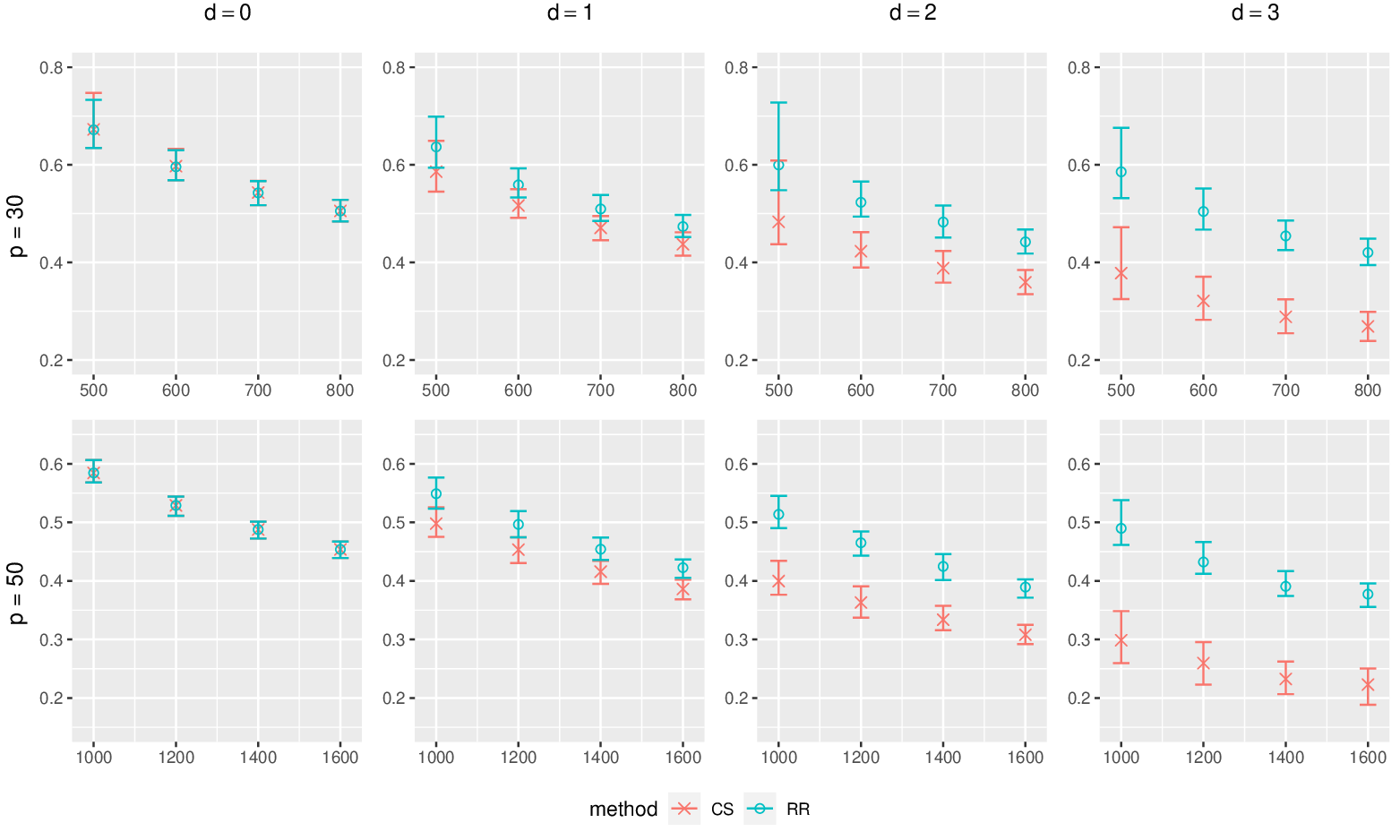}
        \vspace{-1cm}
    \end{center}
    \caption{\small{Plots of estimation errors $\|\cm{\widehat{A}}-\cm{A}^*\|_\text{F}$ of VAR(5) model by common subspace (CS) and reduced-rank (RR) methods. The dimension is $p=30$ (upper panel) or 50 (lower panel).}}
    \label{fig:sim2}
\end{figure}

\vspace{-2.5cm}

\begin{figure}[!htp]
    \begin{center}
        \includegraphics[width=0.78\textwidth]{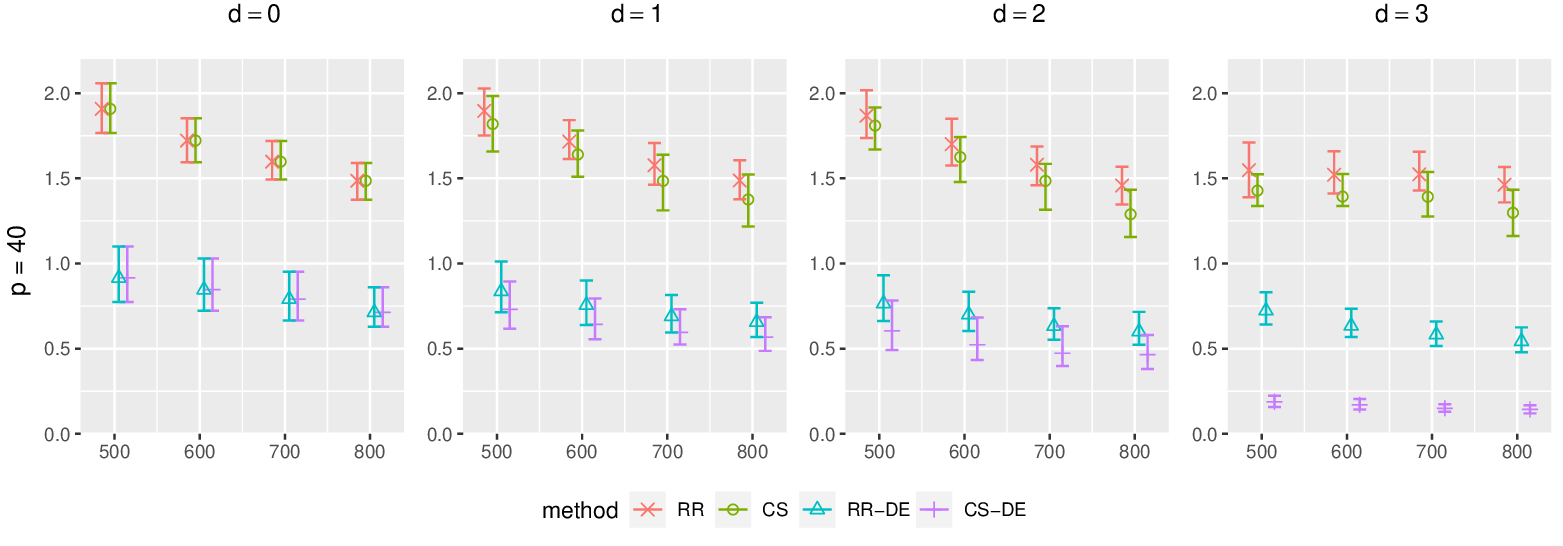}
        \vspace{-1cm}
    \end{center}
    \caption{\small{Plots of estimation errors $\|\cm{\widehat{A}}-\cm{A}^*\|_\text{F}$ of VAR(1) model with diverging eigenvalue by RR, CS, RR-DE and CS-DE methods.}}
    \label{fig:exp3}
\end{figure}
\vspace{-3cm}

\begin{table}[!htp]
    \begin{center}
        \renewcommand{\arraystretch}{0.75}
        \caption{\small{Percentages of correct rank and common dimension selection for VAR($5$) models.}}
        \label{tbl:r&d_ex2}
        \vspace{0.2cm}
        \small{\begin{tabular}{r C{1cm} C{1cm} C{1cm} C{1cm} c C{1cm} C{1cm} C{1cm} C{1cm}}
                \hline\hline
                  & \multicolumn{4}{c}{Rank selection} && \multicolumn{4}{c}{Common dimension selection}\\ 
                  \cline{2-5}\cline{7-10}
                $d$ & 0 & 1 & 2 & 3 && 0    & 1    & 2    & 3 \\
                \hline
                &\multicolumn{9}{c}{dimension $p=30$}\\
                $T=500$ & 76.6 & 75.6 & 71.4 & 68.8  && 89.0 & 87.0 & 77.6 & 80.4 \\
                    600 & 88.0 & 89.4 & 88.4 & 84.8  && 96.6 & 94.4 & 92.0 & 90.6 \\
                    700 & 97.6 & 94.8 & 93.6 & 91.6  && 98.6 & 97.4 & 95.0 & 94.4 \\
                    800 & 98.0 & 97.4 & 98.2 & 97.2  && 99.2 & 99.0 & 98.6 & 98.4 \\
                \hline
                &\multicolumn{9}{c}{dimension $p=50$}\\
                $T=1000$ & 94.6 & 91.6 & 94.6 & 86.8 && 98.2 & 96.0 & 95.8 & 90.8 \\
                    1200 & 98.2 & 97.4 & 97.4 & 98.2 && 99.2 & 98.2 & 98.4 & 99.0 \\
                    1400 & 99.0 & 99.6 & 99.2 & 99.8 && 99.8 & 99.8 & 99.6 & 99.8 \\
                    1600 & 99.8 & 100  & 99.8 & 100  && 100  & 100  & 100  & 100  \\
                \hline
        \end{tabular}}
    \end{center}
\end{table}

\vspace{-1cm}

\begin{table}[!htp]
    \begin{center}
        \renewcommand{\arraystretch}{0.75}
        \caption{\small{Percentages of correct rank and common dimension selection for VAR(1) model with diverging eigenvalue effect.}}
        \label{tbl:r&d_ex3}
        \vspace{0.2cm}
        \small{\begin{tabular}{r C{1cm} C{1cm} C{1cm} C{1cm} c C{1cm} C{1cm} C{1cm} C{1cm}}
                \hline\hline
                  & \multicolumn{4}{c}{Rank selection} && \multicolumn{4}{c}{Common dimension selection}\\ 
                  \cline{2-5}\cline{7-10}
                $d$ & 0 & 1 & 2 & 3 && 0    & 1    & 2    & 3 \\
                \hline
                &\multicolumn{9}{c}{Method: CS}\\
                $T=500$ &  0.4 & 0.2 & 0.0 & 0.0 && 100 & 76.0 & 73.8 & 49.6\\
                    600 &  0.2 & 0.0 & 0.2 & 0.0 && 100 & 78.0 & 76.6 & 66.4\\
                    700 &  0.8 & 0.2 & 0.0 & 0.0 && 100 & 80.6 & 84.8 & 84.4\\
                    800 &  1.2 & 0.4 & 1.6 & 0.0 && 99.8 & 88.6 & 83.2 & 89.4\\
                \hline
                &\multicolumn{9}{c}{Method: CS-DE}\\
                $T=500$ &  99.6 & 98.8 & 97.8 & 99.4 && 100 & 84.2 & 84.8 & 81.6\\
                    600 &  99.4 & 99.8 & 98.8 & 99.8 && 100 & 85.0 & 84.4 & 82.8\\
                    700 &  99.8 & 99.4 & 98.8 & 100 && 100 & 84.2 & 86.6 & 89.8\\
                    800 &  100 & 99.6 & 98.2 & 100 && 100 & 86.8 & 89.0 & 90.4\\
                \hline
        \end{tabular}}
    \end{center}
\end{table}

\vspace{-1cm}

\begin{table}[!htp]
	\begin{center}
		\renewcommand{\arraystretch}{0.8}
		\caption{\small{One-, two-, and three-step ahead errors of the overall and CPI forecasting from our model (VAR-CS), rank-constrained model (VAR-RR), and dynamic factor modeling (DFM-VAR).}}
		\label{tbl:forecast_macro40}
		\vspace{0.2cm}
		\small{\begin{tabular}{cccccccc}
				\hline\hline
                \multirow{2}{*}{Model} & \multicolumn{3}{c}{Overall forecast} && \multicolumn{3}{c}{CPI forecast}\\
                \cline{2-4}\cline{6-8}
				 & One-step & Two-step & Three-step && One-step & Two-step & Three-step\\
				\hline
				VAR-CS & \textbf{4.889} & \textbf{5.156} & \textbf{5.254} && \textbf{0.958} & \textbf{0.929} & \textbf{0.967} \\
				VAR-RR & 5.622 & 5.702 & 5.593 && 1.087 & 1.006 & 1.019 \\
				DFM-VAR & 5.104 & 5.283 & 5.330 && 0.967 & 0.997 & 0.980 \\
				\hline
		\end{tabular}}
	\end{center}
\end{table}

\begin{figure}
	\begin{center}
		\includegraphics[height=0.9\textheight]{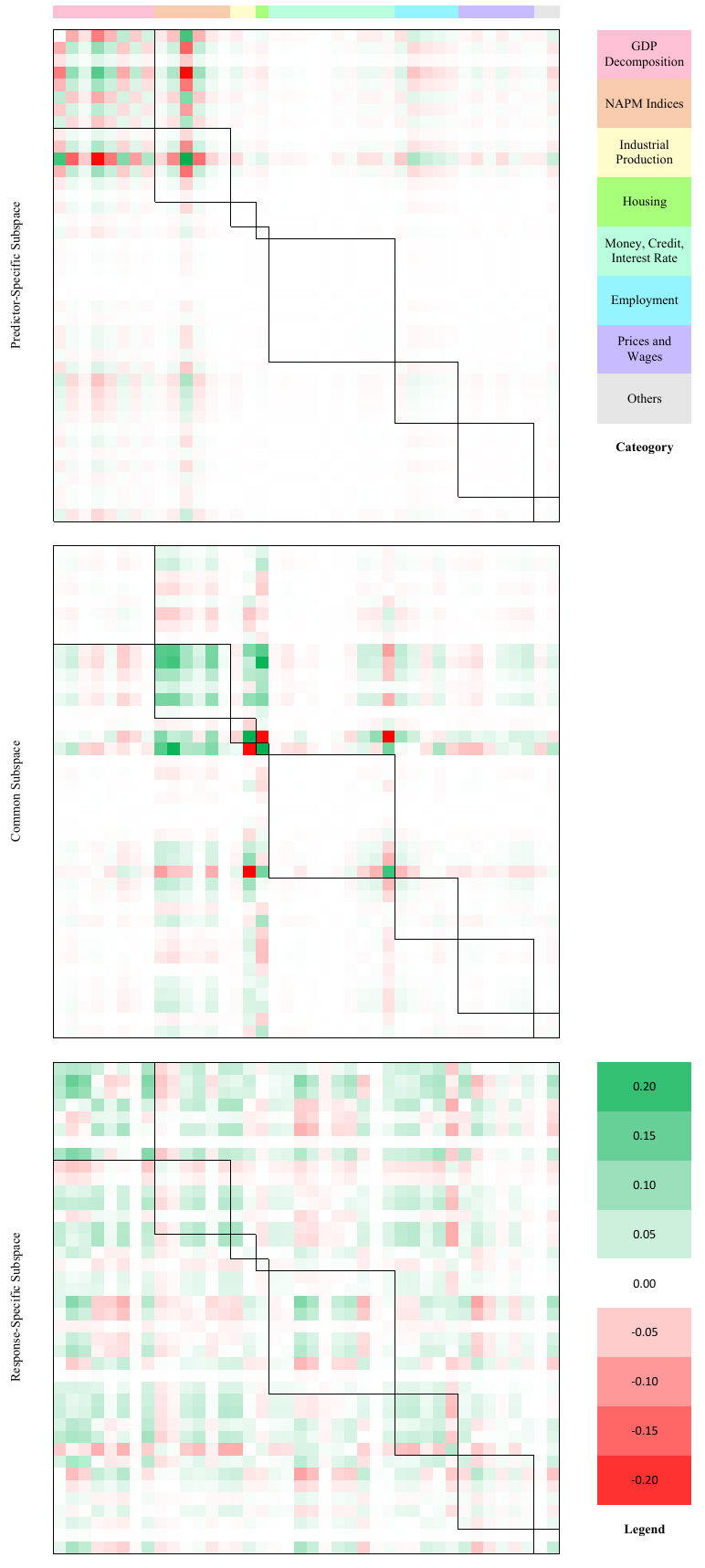}
		\caption{\small{Estimated projection matrices for predictor-specific (upper panel), common (middle panel), and response-specific (lower panel) subspaces, respectively.}}
		\label{fig:subspace_macro40}
	\end{center}
\end{figure}

\newpage

\AtAppendix{\counterwithin{lemma}{section}}
\AtAppendix{\counterwithin{definition}{section}}

\begin{appendix}

\section{Computational convergence analysis of gradient descent}\label{append:comp}

In this appendix, we present the proof of Theorem \ref{thm:gd} and necessary lemmas for the deterministic convergence analysis.

\subsection{Proof of Theorem \ref{thm:gd}}

\begin{proof}
    
    The proof consists of five steps. In the first step, we introduce some notations and conditions essential to the convergence analysis. In the second to fourth steps, we provide a deterministic convergence result for the iterates, given that some regulatory conditions are satisfied. Finally, in the last step, we show that these regulatory conditions hold iteratively.\\
    
    \noindent\textit{Step 1.} (Notations and conditions)\\
    We begin by introducing some notations and conditions for the convergence analysis.
    Denote the empirical least squares loss function as
    \begin{equation}
        \mathcal{L}(\bm{A}) = \frac{1}{2T}\sum_{t=1}^T\|\bm{y}_t-\bm{A}\bm{y}_{t-1}\|_2^2.
    \end{equation}
    
    As the matrix decomposition is not unique, for the iterate at the step $i$, define the combined estimation errors of $(\bm{C},\bm{R},\bm{P},\bm{D})$ up to the optimal rotations as 
    \begin{equation}
        \begin{split}
            E^{(i)}= \min_{\substack{\bm{O}_c\in\mathbb{O}^{d\times d}\\\bm{O}_r,\bm{O}_p\in\mathbb{O}^{(r-d)\times(r-d)}}}&\Big\{\|\bm{C}^{(i)}-\bm{C}^*\bm{O}_c\|_\text{F}^2 + \|\bm{R}^{(i)}-\bm{R}^*\bm{O}_r\|_\text{F}^2+\|\bm{P}^{(i)}-\bm{P}^*\bm{O}_p\|_\text{F}^2\\
            &+\left\|\bm{D}^{(i)}-\text{diag}(\bm{O}_c,\bm{O}_r)^\top\bm{D}^*\text{diag}(\bm{O}_c,\bm{O}_p)\right\|_\textup{F}^2\Big\}
        \end{split}
    \end{equation}
    and the corresponding optimal rotations as $(\bm{O}_c^{(i)},\bm{O}_r^{(i)},\bm{O}_p^{(i)})$. For simplicity in presentation, denote $\bm{O}_1^{(i)}=\text{diag}(\bm{O}_c^{(i)},\bm{O}_r^{(i)})$ and $\bm{O}_2^{(i)}=\text{diag}(\bm{O}_c^{(i)},\bm{O}_p^{(i)})$.
    
    By Definition \ref{def:RSC_RSS}, for the given sample size $T$, $\mathcal{L}$ is restricted strongly convex (RSC) with parameter $\alpha$ and restricted strongly smooth (RSS) with parameter $\beta$, such that for any rank-$r$ matrices $\bm{A},\bm{A}'\in\mathbb{R}^{p\times p'}$,
    \begin{equation}
        \frac{\alpha}{2}\|\bm{A}-\bm{A}'\|_\textup{F}^2 \leq \mathcal{L}(\bm{A})-\mathcal{L}(\bm{A}')-\langle\nabla\mathcal{L}(\bm{A}'),\bm{A}-\bm{A}'\rangle \leq \frac{\beta}{2}\|\bm{A}-\bm{A}'\|_\textup{F}^2.
    \end{equation}
    The $\alpha$-RSC condition implies that
    \begin{equation}
        \mathcal{L}(\bm{A})\geq \mathcal{L}(\bm{A}')+\langle\nabla\mathcal{L}(\bm{A}'),\bm{A}-\bm{A}'\rangle+\frac{\alpha}{2}\|\bm{A}-\bm{A}'\|_\text{F}^2,
    \end{equation}
    and as in \citet{nesterov2003introductory}, the convexity and $\beta$-RSS condition jointly imply that
    \begin{equation}
        \mathcal{L}(\bm{A}')-\mathcal{L}(\bm{A})\geq \langle\nabla\mathcal{L}(\bm{A}),\bm{A}'-\bm{A}\rangle+\frac{1}{2\beta}\|\nabla\mathcal{L}(\bm{A}')-\nabla\mathcal{L}(\bm{A})\|_\text{F}^2.
    \end{equation}
    Combining these two inequalities, we have that
    \begin{equation}
        \langle\nabla\mathcal{L}(\bm{A})-\nabla\mathcal{L}(\bm{A}'),\bm{A}-\bm{A}'\rangle \geq \frac{\alpha}{2}\|\bm{A}-\bm{A}'\|_\text{F}^2 + \frac{1}{2\beta}\|\nabla\mathcal{L}(\bm{A})-\nabla\mathcal{L}(\bm{A}')\|_\text{F}^2,
    \end{equation}
    which is also known as the restricted correlated gradient condition in \citet{han2020optimal}. Moreover, by definition, we immediately have that $\alpha\leq \beta$.
        
    In addition, by Definition \ref{def:deviation}, we assume that given the sample,
    \begin{equation}
        \xi(r,d)=\sup_{\substack{\textup{\bf{[}}\bm{C}~\bm{R}\textup{\bf{]}},\textup{\bf{[}}\bm{C}~\bm{P}\textup{\bf{]}}\in\mathbb{O}^{p\times r},\\ \bm{D}\in\mathbb{R}^{r\times r},\|\bm{D}\|_\textup{F}=1}}\left\langle\nabla\mathcal{L}(\bm{A}^*),\textup{\bf{[}}\bm{C}~\bm{R}\textup{\bf{]}}\bm{D}\textup{\bf{[}}\bm{C}~\bm{P}\textup{\bf{]}}^\top\right\rangle.
    \end{equation}

    For simplicity, we assume $b=\sigma_1^{1/3}$ and $a=C\alpha\sigma_1^{2/3}\kappa^{-2}$, and the proof can  readily be extended to the case with $b\asymp\sigma_1^{1/3}$. Before starting the proof, we also assume that the following conditions hold and will verify them in the last step. For any $i=0,1,2,\dots$, we assume that
    \begin{equation}\label{eq:UVDopbound}
        \begin{split}
            \|\lb\bm{C}^{(i)}~\bm{R}^{(i)}\rb\|_\text{op}\leq 1.1b, ~~ \|\lb\bm{C}^{(i)}~\bm{P}^{(i)}\rb\|_\text{op}\leq 1.1b, ~~
            \textup{and}~~\|\bm{D}^{(i)}\|_\text{op}\leq \frac{1.1\sigma_1}{b^2},
        \end{split}
    \end{equation}
    which obviously implies that
    \begin{equation}
        \|\bm{C}^{(i)}\|_\text{op}\leq 1.1b,~~\|\bm{R}^{(i)}\|_\text{op}\leq 1.1b,~~\text{and}~~\|\bm{P}^{(i)}\|_\text{op}\leq 1.1b.
    \end{equation}
    Note that the constant 1.1 can be replaced by any arbitrary constant greater than 1. In addition, we assume that for any $i=0,1,2,\dots$, $E^{(i)} \leq C\sigma_1^{2/3}\alpha\beta^{-1}\kappa^{-2}$.\\
    
    \noindent\textit{Step 2.} (Upper bound of $E^{(i+1)}-E^{(i)}$)\\
    By definition,
    \begin{equation}
        \begin{split}
            &E^{(i+1)} \\ = & \|\bm{R}^{(i+1)}-\bm{R}^*\bm{O}_r^{(i+1)}\|_\textup{F}^2+\|\bm{P}^{(i+1)}-\bm{P}^*\bm{O}_p^{(i+1)}\|_\textup{F}^2 + \|\bm{C}^{(i+1)}-\bm{C}^*\bm{O}_c^{(i+1)}\|_\textup{F}^2\\
            +&\|\bm{D}^{(i+1)}-\bm{O}_1^{(i+1)\top}\bm{D}^*\bm{O}_2^{(i+1)}\|_\textup{F}^2\\
            \leq & \|\bm{R}^{(i+1)}-\bm{R}^*\bm{O}_r^{(i)}\|_\textup{F}^2+\|\bm{P}^{(i+1)}-\bm{P}^*\bm{O}_p^{(i)}\|_\textup{F}^2 + \|\bm{C}^{(i+1)}-\bm{C}^*\bm{O}_c^{(i)}\|_\textup{F}^2\\ +&\|\bm{D}^{(i+1)}-\bm{O}_1^{(i)\top}\bm{D}^*\bm{O}_2^{(i)}\|_\textup{F}^2.
        \end{split}
    \end{equation}
    \noindent\textit{Step 2.1.} ($\bm{R}$ and $\bm{P}$ steps)\\
    By definition, $\bm{R}^{(i+1)}=\bm{R}^{(i)}-\eta\nabla_{\bm{R}}\mathcal{L}^{(i)}-\eta a[\bm{R}^{(i)}(\bm{R}^{(i)\top}\bm{R}^{(i)}-b^2\bm{I}_{r-d})+\bm{C}^{(i)}\bm{C}^{(i)\top}\bm{R}^{(i)}]$. Thus, we have
    \begin{equation}\label{eq:U_bound}
        \begin{split}
            &\|\bm{R}^{(i+1)}-\bm{R}^*\bm{O}_r^{(i)}\|_\text{F}^2\\
            =&\Big\|\bm{R}^{(i)}-\bm{R}^*\bm{O}_r^{(i)}-\eta\Big(\nabla\mathcal{L}(\bm{A}^{(i)})\lb\bm{C}^{(i)}~\bm{P}^{(i)}\rb\lb\bm{D}_{21}^{(i)}~\bm{D}_{22}^{(i)}\rb^\top\\
            &+a\bm{R}^{(i)}(\bm{R}^{(i)\top}\bm{R}^{(i)}-b^2\bm{I}_{r-d})+a\bm{C}^{(i)}\bm{C}^{(i)\top}\bm{R}^{(i)}\Big)\Big\|_\text{F}^2\\
            =&\|\bm{R}^{(i)}-\bm{R}^*\bm{O}_r^{(i)}\|_\text{F}^2\\
            +&\eta^2\|\nabla\mathcal{L}(\bm{A}^{(i)})\lb\bm{C}^{(i)}~\bm{P}^{(i)}\rb\lb\bm{D}_{21}^{(i)}~\bm{D}_{22}^{(i)}\rb^\top+a\bm{R}^{(i)}(\bm{R}^{(i)\top}\bm{R}^{(i)}-b^2\bm{I}_{r-d})+a\bm{C}^{(i)}\bm{C}^{(i)\top}\bm{R}^{(i)}\|_\text{F}^2\\
            -&2\eta\left\langle\bm{R}^{(i)}-\bm{R}^*\bm{O}_r^{(i)},\nabla\mathcal{L}(\bm{A}^{(i)})\lb\bm{C}^{(i)}~\bm{P}^{(i)}\rb\lb\bm{D}_{21}^{(i)}~\bm{D}_{22}^{(i)}\rb^\top\right\rangle\\
            -&2a\eta\left\langle\bm{R}^{(i)}-\bm{R}^*\bm{O}_r^{(i)},\bm{R}^{(i)}(\bm{R}^{(i)\top}\bm{R}^{(i)}-b^2\bm{I}_{r-d})\right\rangle\\
            -&2a\eta\left\langle\bm{R}^{(i)}-\bm{R}^*\bm{O}_r^{(i)},\bm{C}^{(i)}\bm{C}^{(i)\top}\bm{R}^{(i)}\right\rangle.
        \end{split}
    \end{equation}
    
    First, for the second term in the right hand side of \eqref{eq:U_bound}, by Cauchy's inequality,
    \begin{equation}
        \begin{split}
            &\left\|\nabla\mathcal{L}(\bm{A}^{(i)})\lb\bm{C}^{(i)}~\bm{P}^{(i)}\rb\lb\bm{D}_{21}~\bm{D}_{22}^{(i)}\rb^\top+a\bm{R}^{(i)}(\bm{R}^{(i)\top}\bm{R}^{(i)}-b^2\bm{I}_{r-d})+a\bm{C}^{(i)}\bm{C}^{(i)\top}\bm{R}^{(i)}\right\|_\text{F}^2\\
            \leq &3\|\nabla\mathcal{L}(\bm{A}^{(i)})\lb\bm{C}^{(i)}~\bm{P}^{(i)}\rb\lb\bm{D}_{21}~\bm{D}_{22}^{(i)}\rb^\top\|_\text{F}^2+3a^2\|\bm{R}^{(i)}(\bm{R}^{(i)\top}\bm{R}^{(i)}-b^2\bm{I}_{r-d})\|_\text{F}^2\\
            +&3a^2\|\bm{C}^{(i)}\bm{C}^{(i)\top}\bm{R}^{(i)}\|_\text{F}^2,
        \end{split}
    \end{equation}
    where the first term can be bounded by mean inequality
    \begin{equation}
        \begin{split}
            &\|\nabla\mathcal{L}(\bm{A}^{(i)})\lb\bm{C}^{(i)}~\bm{P}^{(i)}\rb\lb\bm{D}_{21}^{(i)}~\bm{D}_{22}^{(i)}\rb^\top\|_\text{F}^2\\
            \leq & 2\|\nabla\mathcal{L}(\bm{A}^{*})\lb\bm{C}^{(i)}~\bm{P}^{(i)}\rb\lb\bm{D}_{21}^{(i)}~\bm{D}_{22}^{(i)}\rb^\top\|_\text{F}^2\\
            +&2\|[\nabla\mathcal{L}(\bm{A}^{(i)})-\nabla\mathcal{L}(\bm{A}^{*})]\lb\bm{C}^{(i)}~\bm{P}^{(i)}\rb\lb\bm{D}_{21}^{(i)}~\bm{D}_{22}^{(i)}\rb^\top\|_\text{F}^2.
        \end{split}
    \end{equation}

    By the duality of the Frobenius norm, we have
    \begin{equation}
        \begin{split}
            & \|\nabla\mathcal{L}(\bm{A}^{*})\lb\bm{C}^{(i)}~\bm{P}^{(i)}\rb\lb\bm{D}_{21}^{(i)}~\bm{D}_{22}^{(i)}\rb^\top\|_\text{F}\\
            = & \sup_{\bm{W}\in\mathbb{R}^{p\times(r-d)},\|\bm{W}\|_\text{F}=1}\langle\nabla\mathcal{L}(\bm{A}^{*})\lb\bm{C}^{(i)}~\bm{P}^{(i)}\rb\lb\bm{D}_{21}^{(i)}~\bm{D}_{22}^{(i)}\rb^\top,\bm{W}\rangle\\
            = & \sup_{\bm{W}\in\mathbb{R}^{p\times(r-d)},\|\bm{W}\|_\text{F}=1}\langle\nabla\mathcal{L}(\bm{A}^{*}),\bm{W}\lb\bm{D}_{21}^{(i)}~\bm{D}_{22}^{(i)}\rb\lb\bm{C}^{(i)}~\bm{P}^{(i)}\rb^\top\rangle\\
            = & \sup_{\bm{W}\in\mathbb{R}^{p\times(r-d)},\|\bm{W}\|_\text{F}=1}\left\langle\nabla\mathcal{L}(\bm{A}^{*}),\lb\bm{C}^{(i)}~\bm{W}\rb\begin{bmatrix}\bm{0} & \bm{0}\\\bm{D}_{21}^{(i)}&\bm{D}_{22}^{(i)}\end{bmatrix}\lb\bm{C}^{(i)}~\bm{P}^{(i)}\rb^\top\right\rangle\\
            \leq & \|\lb\bm{D}_{21}^{(i)}~\bm{D}_{22}^{(i)}\rb\|_\textup{op}\cdot\|\lb\bm{C}^{(i)}~\bm{P}^{(i)}\rb\|_\textup{op}\cdot\xi(r,d)
        \end{split}
    \end{equation}
    and the first term can be bounded as
    \begin{equation}
        \begin{split}
            &\|\nabla\mathcal{L}(\bm{A}^{(i)})\lb\bm{C}^{(i)}~\bm{P}^{(i)}\rb\lb\bm{D}_{21}^{(i)}~\bm{D}_{22}^{(i)}\rb^\top\|_\text{F}^2\\
            \leq & 2\|\lb\bm{C}^{(i)}~\bm{P}^{(i)}\rb\|_\text{op}^2\cdot\|\lb\bm{D}^{(i)}_{21}~\bm{D}^{(i)}_{22}\rb\|_\text{op}^2\cdot\xi^2(T,\delta)\\
            +&2\|\lb\bm{C}^{(i)}~\bm{P}^{(i)}\rb\|_\text{op}^2\cdot\|\lb\bm{D}^{(i)}_{21}~\bm{D}^{(i)}_{22}\rb\|_\text{op}^2\cdot\|\nabla\mathcal{L}(\bm{A}^{(i)})-\nabla\mathcal{L}(\bm{A}^{*})\|_\text{F}^2\\
            \leq & 4b^{-2}\sigma_1^2\left[\xi^2(r,d)+\|\nabla\mathcal{L}(\bm{A}^{(i)})-\nabla\mathcal{L}(\bm{A}^*)\|_\text{F}^2\right],
        \end{split}
    \end{equation}
    the second term can be bounded as
    \begin{equation}
        \begin{split}
            &a^2\|\bm{R}^{(i)}(\bm{R}^{(i)\top}\bm{R}^{(i)}-b^2\bm{I}_{r-d})\|_\text{F}^2\leq a^2\|\bm{R}^{(i)}\|_\text{op}^2\|\bm{R}^{(i)\top}\bm{R}^{(i)}-b^2\bm{I}_{r-d}\|_\text{F}^2\\
            \leq& 2a^2b^2\|\bm{R}^{(i)\top}\bm{R}^{(i)}-b^2\bm{I}_{r-d}\|_\text{F}^2,
        \end{split}
    \end{equation}
    and the third term can be bounded as
    \begin{equation}
        a^2\|\bm{C}^{(i)}\bm{C}^{(i)\top}\bm{R}^{(i)}\|_\text{F}^2 \leq 2a^2b^2\|\bm{C}^{(i)\top}\bm{R}^{(i)}\|_\text{F}^2.
    \end{equation}
    Thus, we have
    \begin{equation}
        \begin{split}
            & \|\nabla\mathcal{L}(\bm{A}^{(i)})\lb\bm{C}^{(i)}~\bm{P}^{(i)}\rb\lb\bm{D}_{21}^{(i)}~\bm{D}_{22}^{(i)}\rb^\top+a\bm{R}^{(i)}(\bm{R}^{(i)\top}\bm{R}^{(i)}-b^2\bm{I}_{r-d})+2a\bm{C}^{(i)}\bm{C}^{(i)\top}\bm{R}^{(i)}\|_\text{F}^2\\
            \leq & 12b^{-2}\sigma_1^2\left[\xi^2(r,d)+\|\nabla\mathcal{L}(\bm{A}^{(i)})-\nabla\mathcal{L}(\bm{A}^*)\|_\text{F}^2\right] \\
            + & 6a^2b^2[\|\bm{R}^{(i)\top}\bm{R}^{(i)}-b^2\bm{I}_{r-d}\|_\text{F}^2+\|\bm{C}^{(i)\top}\bm{R}^{(i)}\|_\text{F}^2]:=Q_{\text{R},2}.
        \end{split}
    \end{equation}
    
    For the third term in \eqref{eq:U_bound}, denote
    \begin{equation}
        \begin{split}
            & \left\langle\bm{R}^{(i)}-\bm{R}^*\bm{O}_r^{(i)},\nabla\mathcal{L}(\bm{A}^{(i)})\lb\bm{C}^{(i)}~\bm{P}^{(i)}\rb\lb\bm{D}_{21}^{(i)}~\bm{D}_{22}^{(i)}\rb^\top\right\rangle\\
            =&\left\langle\bm{R}^{(i)}\lb\bm{D}_{21}^{(i)}~\bm{D}_{22}^{(i)}\rb\lb\bm{C}^{(i)}~\bm{P}^{(i)}\rb^\top-\bm{R}^*\bm{O}_r^{(i)}\lb\bm{D}_{21}^{(i)}~\bm{D}_{22}^{(i)}\rb\lb\bm{C}^{(i)}~\bm{P}^{(i)}\rb^\top,\nabla\mathcal{L}(\bm{A}^{(i)})\right\rangle\\
            =&\left\langle\bm{A}_\text{R}^{(i)},\nabla\mathcal{L}(\bm{A}^{(i)})\right\rangle,
        \end{split}
    \end{equation}
    where $\bm{A}_\text{R}^{(i)}=(\bm{R}^{(i)}-\bm{R}^*\bm{O}_r^{(i)})\lb\bm{D}_{21}^{(i)}~\bm{D}_{22}^{(i)}\rb\lb\bm{C}^{(i)}~\bm{P}^{(i)}\rb^\top$.

    For the fourth and fifth terms in \eqref{eq:U_bound}, denote
    \begin{equation}
        \begin{split}
            T_\text{R} & = \langle\bm{R}^{(i)}-\bm{R}^*\bm{O}_r^{(i)},\bm{R}^{(i)}(\bm{R}^{(i)\top}\bm{R}^{(i)}-b^2\bm{I}_{r-d}) + \bm{C}^{(i)}\bm{C}^{(i)\top}\bm{R}^{(i)}\rangle
        \end{split}
    \end{equation}

    Therefore, we can rewrite the last three terms in \eqref{eq:U_bound} as
    \begin{equation}
        \begin{split}
            &\left\langle\bm{R}^{(i)}-\bm{R}^*\bm{O}_r^{(i)},\nabla\mathcal{L}(\bm{A}^{(i)})\lb\bm{C}^{(i)}~\bm{P}^{(i)}\rb\lb\bm{D}^{(i)}_{21}~\bm{D}^{(i)}_{22}\rb\right\rangle\\
            +&a\left\langle\bm{R}^{(i)}-\bm{R}^*\bm{O}_r^{(i)},\bm{R}^{(i)}(\bm{R}^{(i)\top}\bm{R}^{(i)}-b^2\bm{I}_{r-d})\right\rangle\\
            +&a\langle\bm{R}^{(i)}-\bm{R}^*\bm{O}_r^{(i)},\bm{C}^{(i)}\bm{C}^{(i)\top}\bm{R}^{(i)}\rangle\\
            =&\left\langle\bm{A}_\text{R}^{(i)},\nabla\mathcal{L}(\bm{A}^{(i)})\right\rangle+a T_\text{R}:=Q_{\text{R},1}.
        \end{split}
    \end{equation}
    Combining the bounds for the terms in \eqref{eq:U_bound}, we have
    \begin{equation}
        \begin{split}
            &\|\bm{R}^{(i+1)}-\bm{R}^*\bm{O}_r^{(i)}\|_\text{F}^2-\|\bm{R}^{(i)}-\bm{R}^*\bm{O}_r^{(i)}\|_\text{F}^2
            \leq -2\eta Q_{\text{R},1}+\eta^2Q_{\text{R},2}.
        \end{split}
    \end{equation}

    Similarly, for $\bm{P}^{(i+1)}$, we can define similar 
    quantities $Q_{\text{V},1}$ and $Q_{\text{V},2}$, and show that
    \begin{equation}
        \begin{split}
            &\|\bm{P}^{(i+1)}-\bm{P}^*\bm{O}_p^{(i)}\|_\text{F}^2-\|\bm{P}^{(i)}-\bm{P}^*\bm{O}_p^{(i)}\|_\text{F}^2
            \leq -2\eta Q_{\text{P},1}+\eta^2Q_{\text{P},2}.
        \end{split}
    \end{equation}
    
    \noindent\textit{Step 2.2.} ($\bm{C}$ step)\\
    For $\bm{C}^{(i+1)}$, note that 
    \begin{equation}
        \bm{C}^{(i+1)}=\bm{C}^{(i)}-\eta\nabla_{\bm{C}}\mathcal{L}^{(i)}-\eta a[2\bm{C}^{(i)}(\bm{C}^{(i)\top}\bm{C}^{(i)}-b^2\bm{I}_d)+\bm{U}^{(i)}\bm{U}^{(i)\top}\bm{C}^{(i)}+\bm{V}^{(i)}\bm{V}^{(i)\top}\bm{C}^{(i)}].
    \end{equation}
    Thus, we have
    \begin{equation}\label{eq:C_bound}
        \begin{split}
            & \|\bm{C}^{(i+1)}-\bm{C}^*\bm{O}_c^{(i)}\|_\text{F}^2\\
            = & \Big\| \bm{C}^{(i+1)} - \bm{C}^*\bm{O}_c^{(i)}-\eta\Big\{\nabla_{\bm{C}}\mathcal{L}^{(i)} + 2a\bm{C}^{(i)}(\bm{C}^{(i)\top}\bm{C}^{(i)}-b^2\bm{I}_d) + a\bm{R}^{(i)}\bm{R}^{(i)\top}\bm{C}^{(i)}\\
            & + a\bm{P}^{(i)}\bm{P}^{(i)\top}\bm{C}^{(i)}\Big\}\Big\|_\text{F}^2\\
            = & \|\bm{C}^{(i)} - \bm{C}^*\bm{O}_c^{(i)}\|_\text{F}^2\\
            + & \eta^2\Big\|\nabla_{\bm{C}}\mathcal{L}^{(i)} + 2a\bm{C}^{(i)}(\bm{C}^{(i)\top}\bm{C}^{(i)}-b^2\bm{I}_d) + a\bm{R}^{(i)}\bm{R}^{(i)\top}\bm{C}^{(i)}+a\bm{P}^{(i)}\bm{P}^{(i)\top}\bm{C}^{(i)}\Big\|_\text{F}^2\\
            - & 2\eta\left\langle\bm{C}^{(i)} - \bm{C}^*\bm{O}_c^{(i)},\nabla_{\bm{C}}\mathcal{L}^{(i)}\right\rangle\\
            - & 2a\eta\left\langle\bm{C}^{(i)} - \bm{C}^*\bm{O}_c^{(i)},2\bm{C}^{(i)}(\bm{C}^{(i)\top}\bm{C}^{(i)}-b^2\bm{I}_d)\right\rangle\\
            - & 2a\eta\left\langle\bm{C}^{(i)} - \bm{C}^*\bm{O}_c^{(i)},\bm{R}^{(i)}\bm{R}^{(i)\top}\bm{C}^{(i)} + \bm{P}^{(i)}\bm{P}^{(i)\top}\bm{C}^{(i)}\right\rangle,
        \end{split}
    \end{equation}
    where 
    \begin{equation}
        \begin{split}
            \nabla_{\bm{C}}\mathcal{L}^{(i)} = & \nabla\mathcal{L}(\bm{A}^{(i)})\lb\bm{C}^{(i)}~\bm{P}^{(i)}\rb\lb\bm{D}_{11}^{(i)}~\bm{D}_{12}^{(i)}\rb^\top + \nabla\mathcal{L}(\bm{A}^{(i)})^\top\lb\bm{C}^{(i)}~\bm{R}^{(i)}\rb\lb\bm{D}_{11}^\top~\bm{D}_{21}^\top\rb^\top .
        \end{split}
    \end{equation}

    First, we have
    \begin{equation}
        \begin{split}
            &\Big\|\nabla_{\bm{C}}\mathcal{L}^{(i)}+ 2a\bm{C}^{(i)}(\bm{C}^{(i)\top}\bm{C}^{(i)}-b^2\bm{I}_d) + a\bm{R}^{(i)}\bm{R}^{(i)\top}\bm{C}^{(i)}+a\bm{P}^{(i)}\bm{P}^{(i)\top}\bm{C}^{(i)}\Big\|_\text{F}^2\\
            \leq & 4\|\nabla_{\bm{C}}\mathcal{L}^{(i)}\|_\text{F}^2+ 16a^2\|\bm{C}^{(i)}(\bm{C}^{(i)\top}\bm{C}^{(i)}-b^2\bm{I}_d)\|_\text{F}^2\\
            + & 4a^2\|\bm{R}^{(i)}\bm{R}^{(i)\top}\bm{C}^{(i)}\|_\text{F}^2 + 4a^2\|\bm{P}^{(i)}\bm{P}^{(i)\top}\bm{C}^{(i)}\|_\text{F}^2,
        \end{split}
    \end{equation}
    where the first term can be bounded as
    \begin{equation}
        \begin{split}
            & \|\nabla_{\bm{C}}\mathcal{L}^{(i)}\|_\text{F}^2\\
            \leq & 2\|\nabla\mathcal{L}(\bm{A}^{(i)})\lb\bm{C}^{(i)}~\bm{P}^{(i)}\rb\lb\bm{D}_{11}^{(i)}~\bm{D}_{12}^{(i)}\rb^\top\|_\text{F}^2 + 2\|\nabla\mathcal{L}(\bm{A}^{(i)})^\top\lb\bm{C}^{(i)}~\bm{R}^{(i)}\rb\lb\bm{D}_{11}^\top~\bm{D}_{21}^\top\rb^\top\|_\text{F}^2\\
            \leq & 4\|\nabla\mathcal{L}(\bm{A}^{*})\lb\bm{C}^{(i)}~\bm{P}^{(i)}\rb\lb\bm{D}_{11}^{(i)}~\bm{D}_{12}^{(i)}\rb^\top\|_\text{F}^2 + 4\|\nabla\mathcal{L}(\bm{A}^{*})^\top\lb\bm{C}^{(i)}~\bm{R}^{(i)}\rb\lb\bm{D}_{11}^\top~\bm{D}_{21}^\top\rb^\top\|_\text{F}^2\\
            + & 4\|[\nabla\mathcal{L}(\bm{A}^{*})-\nabla\mathcal{L}(\bm{A}^{(i)})]\lb\bm{C}^{(i)}~\bm{P}^{(i)}\rb\lb\bm{D}_{11}^{(i)}~\bm{D}_{12}^{(i)}\rb^\top\|_\text{F}^2\\
            + & 4\|[\nabla\mathcal{L}(\bm{A}^{*})-\nabla\mathcal{L}(\bm{A}^{(i)})]^\top\lb\bm{C}^{(i)}~\bm{R}^{(i)}\rb\lb\bm{D}_{11}^\top~\bm{D}_{21}^\top\rb^\top\|_\text{F}^2\\
            \leq & 8b^{-2}\sigma_1^2[\xi^2(r,d)+\|\nabla\mathcal{L}(\bm{A}^{(i)})-\nabla\mathcal{L}(\bm{A}^{*})\|_\text{F}^2]
        \end{split}
    \end{equation}
    and the other three terms can be bounded as
    \begin{equation}
        a^2\|\bm{C}^{(i)}(\bm{C}^{(i)\top}\bm{C}^{(i)}-b^2\bm{I}_d)\|_\text{F}^2 \leq 2a^2b^2\|\bm{C}^{(i)\top}\bm{C}^{(i)}-b^2\bm{I}_d\|_
        \text{F}^2,
    \end{equation}
    \begin{equation}
        a^2\|\bm{R}^{(i)}\bm{R}^{(i)\top}\bm{C}^{(i)}\|_\text{F}^2 \leq 2a^2b^2\|\bm{R}^{(i)\top}\bm{C}^{(i)}\|_\text{F}^2,
    \end{equation}
    \begin{equation}
        a^2\|\bm{P}^{(i)}\bm{P}^{(i)\top}\bm{C}^{(i)}\|_\text{F}^2 \leq 2a^2b^2\|\bm{P}^{(i)\top}\bm{C}^{(i)}\|_\text{F}^2.
    \end{equation}
    Thus, the second term in \eqref{eq:C_bound} can be bounded as
    \begin{equation}
        \begin{split}
            & \Big\|\nabla_{\bm{C}}\mathcal{L} + 2a\bm{C}^{(i)}(\bm{C}^{(i)\top}\bm{C}^{(i)}-b^2\bm{I}_d) + 2a\bm{R}^{(i)}\bm{R}^{(i)\top}\bm{C}^{(i)}+2a\bm{P}^{(i)}\bm{P}^{(i)\top}\bm{C}^{(i)}\Big\|_\text{F}^2\\
            \leq & 32b^{-2}\sigma_1^2[\xi^2(r,d)+\|\nabla\mathcal{L}(\bm{A}^{(i)})-\nabla\mathcal{L}(\bm{A}^*)\|_\text{F}^2]\\
            + & 8a^2b^2(4\|\bm{C}^{(i)\top}\bm{C}^{(i)}-b^2\bm{I}_d\|_
            \text{F}^2 + \|\bm{R}^{(i)\top}\bm{C}^{(i)}\|_\text{F}^2 + \|\bm{P}^{(i)\top}\bm{C}^{(i)}\|_\text{F}^2) := Q_{\text{C},2}.
        \end{split}
    \end{equation}

    For the third term in \eqref{eq:C_bound},
    \begin{equation}
        \begin{split}
            & \left\langle \bm{C}^{(i)}-\bm{C}^*\bm{O}_c^{(i)}, \nabla_{\bm{C}}\mathcal{L}\right\rangle\\
            = & \left\langle\bm{R}^{(i)}\bm{D}_{12}^{(i)}\bm{C}^{(i)\top}-\bm{R}^{(i)}\bm{D}_{12}^{(i)}\bm{O}_c^{(i)\top}\bm{C}^{*\top},\nabla\mathcal{L}(\bm{A}^{(i)})\right\rangle\\
            + & \left\langle\bm{C}^{(i)}\bm{D}_{21}^{(i)}\bm{P}^{(i)\top}-\bm{C}^{*}\bm{O}_c^{(i)}\bm{D}_{21}^{(i)}\bm{P}^{(i)\top},\nabla\mathcal{L}(\bm{A}^{(i)})\right\rangle\\
            + & \left\langle\bm{C}^{(i)}\bm{D}_{11}^{(i)}\bm{C}^{(i)\top}-\bm{C}^{*}\bm{O}_c^{(i)}\bm{D}_{11}^{(i)}\bm{C}^{(i)\top},\nabla\mathcal{L}(\bm{A}^{(i)})\right\rangle\\
            + & \left\langle\bm{C}^{(i)}\bm{D}_{11}^{(i)}\bm{C}^{(i)\top}-\bm{C}^{(i)}\bm{D}_{11}^{(i)}\bm{O}_c^{(i)}\bm{C}^{*\top},\nabla\mathcal{L}(\bm{A}^{(i)})\right\rangle\\
            := & \left\langle\bm{A}_\text{C}^{(i)},\nabla\mathcal{L}(\bm{A}^{(i)})\right\rangle.
        \end{split}
    \end{equation}

    For the fourth and fifth terms in \eqref{eq:C_bound}, denote
    \begin{equation}
        \langle\bm{C}^{(i)}-\bm{C}^*\bm{O}_c^{(i)},2\bm{C}^{(i)}(\bm{C}^{(i)\top}\bm{C}^{(i)}-b^2\bm{I}_d)+\bm{R}^{(i)}\bm{R}^{(i)\top}\bm{C}^{(i)}+\bm{P}^{(i)}\bm{P}^{(i)\top}\bm{C}^{(i)}\rangle:=T_\text{C}.
    \end{equation}

    Hence, we can bound the last three terms in \eqref{eq:C_bound} as
    \begin{equation}
        \begin{split}
            & \eta\left\langle\bm{C}^{(i)} - \bm{C}^*\bm{O}_c^{(i)},\nabla_{\bm{C}}\mathcal{L}\right\rangle+ a\eta\left\langle\bm{C}^{(i)} - \bm{C}^*\bm{O}_c^{(i)},\bm{C}^{(i)}(\bm{C}^{(i)\top}\bm{C}^{(i)}-b^2\bm{I}_d)\right\rangle\\
            + & a\eta\left\langle\bm{C}^{(i)} - \bm{C}^*\bm{O}_c^{(i)},\bm{R}^{(i)}\bm{R}^{(i)\top}\bm{C}^{(i)} + \bm{P}^{(i)}\bm{P}^{(i)\top}\bm{C}^{(i)}\right\rangle\\
            \geq & \left\langle\bm{A}_\text{C}^{(i)},\nabla\mathcal{L}(\bm{A}^{(i)})\right\rangle+aT_\text{C}:=Q_{\text{C},1}.
        \end{split}
    \end{equation}
    Combining these bounds, we have
    \begin{equation}
        \|\bm{C}^{(i+1)}-\bm{C}^*\bm{O}_c^{(i)}\|_\text{F}^2 - \|\bm{C}^{(i)}-\bm{C}^*\bm{O}_c^{(i)}\|_\text{F}^2 \leq -2\eta Q_{\text{C},1} + \eta^2 Q_{\text{C},2}.
    \end{equation}
    
    \noindent\textit{Step 2.3.} ($\bm{D}$ step)\\
    For $\bm{D}^{(i)}$, we consider the following decomposition
    \begin{equation}
        \begin{split}
            & \|\bm{D}^{(i+1)}-\bm{O}_1^{(i)\top}\bm{D}^*\bm{O}_2^{(i)}\|_\text{F}^2\\
            = & \|\bm{D}^{(i)}-\bm{O}_1^{(i)\top}\bm{D}^*\bm{O}_2^{(i)}-\eta\lb\bm{C}^{(i)}~\bm{R}^{(i)}\rb^\top\nabla\mathcal{L}(\bm{A}^{(i)})\lb\bm{C}^{(i)}~\bm{P}^{(i)}\rb\|_\text{F}^2\\
            = & \|\bm{D}^{(i)}-\bm{O}_1^{(i)\top}\bm{D}^*\bm{O}_2^{(i)}\|_\text{F}^2 + \eta^2\|\lb\bm{C}^{(i)}~\bm{R}^{(i)}\rb^\top\nabla\mathcal{L}(\bm{A}^{(i)})\lb\bm{C}^{(i)}~\bm{P}^{(i)}\rb\|_\text{F}^2\\
            - & 2\eta\left\langle\bm{D}^{(i)}-\bm{O}_1^{(i)\top}\bm{D}^*\bm{O}_2^{(i)},\lb\bm{C}^{(i)}~\bm{R}^{(i)}\rb^\top\nabla\mathcal{L}(\bm{A}^{(i)})\lb\bm{C}^{(i)}~\bm{P}^{(i)}\rb\right\rangle.
        \end{split}
    \end{equation}
    For the third term, we have
    \begin{equation}
        \begin{split}
            & \langle\bm{D}^{(i)}-\bm{O}_1^{(i)\top}\bm{D}^*\bm{O}_2^{(i)}, \bm{R}^{(i)\top}\nabla\mathcal{L}(\bm{A}^{(i)})\bm{P}^{(i)}\rangle\\
            =&\langle\bm{A}^{(i)}-\lb\bm{C}^{(i)}~\bm{R}^{(i)}\rb\bm{O}_1^{(i)\top}\bm{D}^*\bm{O}_2^{(i)}\lb\bm{C}^{(i)}~\bm{P}^{(i)}\rb^\top,\nabla\mathcal{L}(\bm{A}^{(i)})\rangle\\
            =&\langle\bm{A}_\text{D}^{(i)},\nabla\mathcal{L}(\bm{A}^{(i)})\rangle:=Q_{\text{D},1}
        \end{split}
    \end{equation}
    In addition,
    \begin{equation}
        \begin{split}
            &\|\lb\bm{C}^{(i)}~\bm{R}^{(i)}\rb^\top\nabla\mathcal{L}(\bm{A}^{(i)})\lb\bm{C}^{(i)}~\bm{P}^{(i)}\rb\|_\text{F}^2\\
            \leq & 2\|\lb\bm{C}^{(i)}~\bm{R}^{(i)}\rb^\top\nabla\mathcal{L}(\bm{A}^*)\lb\bm{C}^{(i)}~\bm{P}^{(i)}\rb\|_\text{F}^2\\
            +&2\|\lb\bm{C}^{(i)}~\bm{R}^{(i)}\rb^\top[\nabla\mathcal{L}(\bm{A}^{(i)})-\nabla\mathcal{L}(\bm{A}^*)]\lb\bm{C}^{(i)}~\bm{P}^{(i)}\rb\|_\text{F}^2\\
            \leq & 2\|\lb\bm{C}^{(i)}~\bm{R}^{(i)}\rb\|_\text{op}^2\cdot\|\lb\bm{C}^{(i)}~\bm{P}^{(i)}\rb\|_\text{op}^2\cdot[\xi^2(T,\delta)+\|\nabla\mathcal{L}(\bm{A}^{(i)})-\nabla\mathcal{L}(\bm{A}^*)\|_\text{F}^2]\\
            =&4b^2[\xi^2(r,d)+\|\nabla\mathcal{L}(\bm{A}^{(i)})-\nabla\mathcal{L}(\bm{A}^*)\|_\text{F}^2]:=Q_{\text{D},2}.
        \end{split}
    \end{equation}
    Hence, we have
    \begin{equation}
        \|\bm{D}^{(i+1)}-\bm{O}_1^{(i)\top}\bm{D}^*\bm{O}_2^{(i)}\|_\text{F}^2-\|\bm{D}^{(i)}-\bm{O}_1^{(i)\top}\bm{D}^*\bm{O}_2^{(i)}\|_\text{F}^2\leq -2\eta Q_{\text{D},1}+\eta^2Q_{\text{D},2}.
    \end{equation}
    Together, we have that
    \begin{equation}
        \begin{split}
            E^{(i+1)}&\leq \|\bm{R}^{(i)}-\bm{R}^*\bm{O}_r^{(i)}\|_\text{F}^2 - 2\eta Q_{\text{R},1}+\eta^2Q_{\text{R},2}\\
            &+\|\bm{P}^{(i)}-\bm{P}^*\bm{O}_p^{(i)}\|_\text{F}^2 - 2\eta Q_{\text{P},1} + \eta^2 Q_{\text{P},2}\\
            &+\|\bm{C}^{(i)}-\bm{C}^*\bm{O}_p^{(i)}\|_\text{F}^2 - 2\eta Q_{\text{C},1} + \eta^2 Q_{\text{C},2}\\
            &+\|\bm{D}^{(i)}-\bm{O}_1^{(i)\top}\bm{D}^*\bm{O}_2^{(i)}\|_\text{F}^2-2\eta Q_{\text{D},1}+\eta^2Q_{\text{D},2}\\
            & \leq E^{(i)} -2\eta(Q_{\text{D},1}+Q_{\text{R},1}+Q_{\text{P},1}+Q_{\text{C},1})+\eta^2(Q_{\text{D},2}+Q_{\text{R},2}+Q_{\text{P},2}+Q_{\text{C},2}).
        \end{split}
    \end{equation}~\\

    \noindent\textit{Step 3.} (Lower bound of $Q_{\textup{D},1}+Q_{\textup{R},1}+Q_{\textup{P},1}+Q_{\textup{C},1}$)\\
    In the third step, we develop a lower bound for $Q_{\text{D},1}+Q_{\text{R},1}+Q_{\text{P},1}+Q_{\text{C},1}$. By definition, 
    \begin{equation}
        \label{eq:Q_bound}
        \begin{split}
            & Q_{\text{D},1}+Q_{\text{R},1}+Q_{\text{P},1}+Q_{\text{C},1}\\
            = & \langle\bm{A}_\text{D}^{(i)}+\bm{A}_\text{R}^{(i)}+\bm{A}_\text{P}^{(i)}+\bm{A}_\text{C}^{(i)},\nabla\mathcal{L}(\bm{A}^{(i)})\rangle + a(T_\text{R}+T_\text{P}+T_\text{C}).
        \end{split}
    \end{equation}
    Note that
    \begin{equation}
        \begin{split}
            & \bm{A}_\text{D}^{(i)}+\bm{A}_\text{R}^{(i)}+\bm{A}_\text{P}^{(i)}+\bm{A}_\text{C}^{(i)}\\
            = & 3\bm{A}^{(i)} - \lb\bm{C}^{(i)}~\bm{R}^{(i)}\rb\bm{O}_1^{(i)\top}\bm{D}^*\bm{O}_2^{(i)}
            \lb\bm{C}^{(i)}~\bm{P}^{(i)}\rb^\top-\lb\bm{C}^{(i)}~\bm{R}^{(i)}\rb\bm{D}^{(i)}\bm{O}_2^{(i)\top}\lb\bm{C}^*~\bm{P}^*\rb^\top\\
            - & \lb\bm{C}^*~\bm{P}^*\rb\bm{O}_1^{(i)}\bm{D}^{(i)}\lb\bm{C}^{(i)}~\bm{P}^{(i)}\rb^\top\\
            = & \bm{A}^{(i)}-\bm{A}^*+\bm{H}^{(i)},
        \end{split}
    \end{equation}
    where
    \begin{equation}
        \begin{split}
            \bm{H}^{(i)}=&\lb\bm{C}^{(i)}~\bm{R}^{(i)}\rb\bm{D}^{(i)}\lb\bm{C}^{(i)}~\bm{P}^{(i)}\rb^\top-\lb\bm{C}^{*}~\bm{R}^{*}\rb\bm{O}_1^{(i)}\bm{D}^{(i)}\lb\bm{C}^{(i)}~\bm{P}^{(i)}\rb^\top\\
            +&\lb\bm{C}^{(i)}~\bm{R}^{(i)}\rb\bm{D}^{(i)}\lb\bm{C}^{(i)}~\bm{P}^{(i)}\rb^\top-\lb\bm{C}^{(i)}~\bm{R}^{(i)}\rb\bm{D}^{(i)}\bm{O}_2^{(i)\top}\lb\bm{C}^{*}~\bm{P}^{*}\rb^\top\\
            +&\lb\bm{C}^{*}~\bm{R}^{*}\rb\bm{D}^{*}\lb\bm{C}^{*}~\bm{P}^{*}\rb^\top-\lb\bm{C}^{(i)}~\bm{R}^{(i)}\rb\bm{O}_1^{(i)\top}\bm{D}^{*}\bm{O}_2^{(i)}\lb\bm{C}^{(i)}~\bm{P}^{(i)}\rb^\top.
        \end{split}
    \end{equation}

    By Lemma \ref{lemma:perturb}, since $\|\lb\bm{C}^{(i)}~\bm{R}^{(i)}\rb\|_\text{op}\leq 1.01b$, $\|\lb\bm{C}^{(i)}~\bm{P}^{(i)}\rb\|_\text{op}\leq 1.01b$, $\|\bm{D}^{(i)}\|_\text{op}\leq 1.01\sigma_1b^{-2}$, $b=\sigma_1^{1/3}$ and $E^{(i)}\leq \sigma_1^{2/3}$, we can derive an upper bound for $\bm{H}^{(i)}$,
    \begin{equation}
        \begin{split}
            \|\bm{H}^{(i)}\|_\text{F} & \leq 1.01\sigma_1b^{-2}E^{(i)}+2(1.01b)E^{(i)} + (E^{(i)})^{3/2}\\
            & \leq (4\sigma_1^{1/3}+\sqrt{E^{(i)}})E^{(i)} \leq 5\sigma_1^{1/3}E^{(i)}.
        \end{split}
    \end{equation}

    By the $\alpha$-RSC and $\beta$-RSS conditions, the first term on the right hand side of \eqref{eq:Q_bound} can be bounded as
    \begin{equation}
        \begin{split}
            & \langle\bm{A}^{(i)}-\bm{A}^{*}+\bm{H}^{(i)},\nabla\mathcal{L}(\bm{A}^{(i)})\rangle\\
            = & \langle\bm{A}^{(i)}-\bm{A}^{*}+\bm{H}^{(i)},\nabla\mathcal{L}(\bm{A}^{*})\rangle + 
            \langle\bm{A}^{(i)}-\bm{A}^{*},\nabla\mathcal{L}(\bm{A}^{(i)})-\nabla\mathcal{L}(\bm{A}^{*})\rangle\\
            + & \langle\bm{H}^{(i)},\nabla\mathcal{L}(\bm{A}^{(i)})-\nabla\mathcal{L}(\bm{A}^{*})\rangle\\
            \geq & \frac{\alpha}{2}\|\bm{A}^{(i)}-\bm{A}^*\|_\text{F}^2+\frac{1}{2\beta}\|\nabla\mathcal{L}(\bm{A}^{(i)})-\nabla\mathcal{L}(\bm{A}^*)\|_\text{F}^2-\|\bm{H}^{(i)}\|_\text{F}\|\nabla\mathcal{L}(\bm{A}^{(i)})-\nabla\mathcal{L}(\bm{A}^*)\|_\text{F}\\
            - & |\langle\bm{A}^{(i)}-\bm{A}^{*}+\bm{H}^{(i)},\nabla\mathcal{L}(\bm{A}^{*})\rangle|.
        \end{split}
    \end{equation}
    In addition, we have that for any $c_1>0$
    \begin{equation}
        \begin{split}
            & \|\bm{H}^{(i)}\|_\text{F}\|\nabla\mathcal{L}(\bm{A}^{(i)})-\nabla\mathcal{L}(\bm{A}^*)\|_\text{F}\\
            \leq & \frac{1}{4\beta}\|\nabla\mathcal{L}(\bm{A}^{(i)})-\nabla\mathcal{L}(\bm{A}^*)\|_\text{F}^2+\beta\|\bm{H}^{(i)}\|_\text{F}^2\\
            \leq & \frac{1}{4\beta}\|\nabla\mathcal{L}(\bm{A}^{(i)})-\nabla\mathcal{L}(\bm{A}^*)\|_\text{F}^2+\beta\left(25\sigma_1^{2/3}E^{(i)}\right)E^{(i)}\\
            \leq & \frac{\beta}{2}\|\nabla\mathcal{L}(\bm{A}^{(i)})-\nabla\mathcal{L}(\bm{A}^*)\|_\text{F}^2+\frac{C\alpha\sigma_1^{4/3}}{\kappa^2}E^{(i)},
        \end{split}
    \end{equation}
    and
    \begin{equation}
        \begin{split}
            & |\langle\bm{A}^{(i)}-\bm{A}^*+\bm{H}^{(i)},\nabla\mathcal{L}(\bm{A}^{*})\rangle|\\
            \leq & |\langle\bm{A}^{(i)}-\lb\bm{C}^*~\bm{P}^*\rb\bm{O}_1^{(i)}\bm{D}^{(i)}\lb\bm{C}^{(i)}~\bm{P}^{(i)}\rb^\top,\nabla\mathcal{L}(\bm{A}^*)\rangle|\\
            + & |\langle\bm{A}^{(i)}-\lb\bm{C}^{(i)}~\bm{R}^{(i)}\rb\bm{D}^{(i)}\bm{O}_2^{(i)\top}\lb\bm{C}^*~\bm{P}^*\rb^\top,\nabla\mathcal{L}(\bm{A}^*)\rangle|\\
            + & |\langle\bm{A}^{(i)}-\lb\bm{C}^{(i)}~\bm{R}^{(i)}\rb\bm{O}_1^{(i)\top}\bm{D}^*\bm{O}_2^{(i)}
            \lb\bm{C}^{(i)}~\bm{P}^{(i)}\rb^\top,\nabla\mathcal{L}(\bm{A}^*)\rangle|\\
            \leq & \xi(r,d)\left(\|\bm{D}^{(i)}\|_\text{op}\cdot\|\lb\bm{C}^{(i)}~\bm{P}^{(i)}\rb\|_\text{op}\cdot\|\lb\bm{C}^{(i)}~\bm{R}^{(i)}\rb-\lb\bm{C}^{*}~\bm{R}^{*}\rb\bm{O}_1^{(i)}\|_\text{F}\right)\\
            + & \xi(r,d)\left(\|\bm{D}^{(i)}\|_\text{op}\cdot\|\lb\bm{C}^{(i)}~\bm{R}^{(i)}\rb\|_\text{op}\cdot\|\lb\bm{C}^{(i)}~\bm{P}^{(i)}\rb-\lb\bm{C}^{*}~\bm{P}^{*}\rb\bm{O}_2^{(i)}\|_\text{F}\right)\\
            + & \xi(r,d)\left(\|\lb\bm{C}^{(i)}~\bm{R}^{(i)}\rb\|_\text{op}\cdot\|\lb\bm{C}^{(i)}~\bm{P}^{(i)}\rb\|_\text{op}\cdot\|\bm{D}^{(i)}-\bm{O}_1^{(i)\top}\bm{D}^*\bm{O}_2^{(i)}\|_\text{F}\right)\\
            \leq & \left[1.01^2b^2+2\times1.01^2\sigma_1/b\right]\xi(r,d)(2E^{(i)})^{1/2}\\
            \leq & 4\sigma_1^{2/3}\xi(r,d)(2E^{(i)})^{1/2}\\
            \leq & 32c_2\sigma_1^{4/3}E^{(i)}+\frac{1}{4c_2}\xi^2(r,d),
        \end{split}
    \end{equation}
    for any $c_2>0$. Combining these inequalities, we have
    \begin{equation}
        \begin{split}
            &\langle\bm{A}^{(i)}-\bm{A}^{*}+\bm{H}^{(i)},\nabla\mathcal{L}(\bm{A}^{(i)})\rangle\\
            \geq & \frac{\alpha}{2}\|\bm{A}^{(i)}-\bm{A}^*\|_\text{F}^2+\frac{1}{4\beta}\|\nabla\mathcal{L}(\bm{A}^{(i)})-\nabla\mathcal{L}(\bm{A}^*)\|_\text{F}^2-\frac{1}{4c_2}\xi^2(r,d)\\
            -&\left(32c_2\sigma_1^{4/3}+\frac{C\alpha\sigma_1^{4/3}}{\kappa^2}\right)E^{(i)}.
        \end{split}
    \end{equation}
    Applying Lemma \ref{lemma:errorbound} with $b=\sigma_1^{1/3}$, we can obtain an upper bound for $E^{(i)}$,
    \begin{equation}
        \begin{split}
            &E^{(i)}\leq (4\sigma_1^{-4/3}+136\sigma_1^{2/3}\sigma_{r}^{-2})\|\bm{A}^{(i)}-\bm{A}^*\|_\text{F}^2\\
            +&28b^{-2}\left(\|\lb\bm{C}^{(i)}~\bm{R}^{(i)}\rb^\top\lb\bm{C}^{(i)}~\bm{R}^{(i)}\rb-b^2\bm{I}_r\|_\text{F}^2 + \|\lb\bm{C}^{(i)}~\bm{R}^{(i)}\rb^\top\lb\bm{C}^{(i)}~\bm{R}^{(i)}\rb-b^2\bm{I}_r\|_\text{F}^2\right)\\
            \leq & 140\sigma_1^{2/3}\sigma_{r}^{-2}\|\bm{A}^{(i)}-\bm{A}^*\|_\text{F}^2\\
            +&28\sigma_1^{-2/3}\left(\|\lb\bm{C}^{(i)}~\bm{R}^{(i)}\rb^\top\lb\bm{C}^{(i)}~\bm{R}^{(i)}\rb-b^2\bm{I}_r\|_\text{F}^2 + \|\lb\bm{C}^{(i)}~\bm{R}^{(i)}\rb^\top\lb\bm{C}^{(i)}~\bm{R}^{(i)}\rb-b^2\bm{I}_r\|_\text{F}^2\right)\\
            \leq & 140\sigma_1^{-4/3}\kappa^2\|\bm{A}^{(i)}-\bm{A}^*\|_\text{F}^2\\
            + & 28\sigma_1^{-2/3}\left(\|\lb\bm{C}^{(i)}~\bm{R}^{(i)}\rb^\top\lb\bm{C}^{(i)}~\bm{R}^{(i)}\rb-b^2\bm{I}_r\|_\text{F}^2 + \|\lb\bm{C}^{(i)}~\bm{R}^{(i)}\rb^\top\lb\bm{C}^{(i)}~\bm{R}^{(i)}\rb-b^2\bm{I}_r\|_\text{F}^2\right).
        \end{split}
    \end{equation}

    For the second term on the right-hand side of \eqref{eq:Q_bound}, note that 
    \begin{equation}
        \begin{split}
            & T_\text{R}+T_\text{P}+T_\text{C}\\
            = & \langle\lb\bm{C}^{(i)}~\bm{R}^{(i)}\rb-\lb\bm{C}^{*}~\bm{R}^{*}\rb\bm{O}_1^{(i)},\lb\bm{C}^{(i)}~\bm{R}^{(i)}\rb(\lb\bm{C}^{(i)}~\bm{R}^{(i)}\rb^\top\lb\bm{C}^{(i)}~\bm{R}^{(i)}\rb-b^2\bm{I}_r)\rangle\\
            + & \langle\lb\bm{C}^{(i)}~\bm{P}^{(i)}\rb-\lb\bm{C}^{*}~\bm{P}^{*}\rb\bm{O}_1^{(i)},\lb\bm{C}^{(i)}~\bm{P}^{(i)}\rb(\lb\bm{C}^{(i)}~\bm{P}^{(i)}\rb^\top\lb\bm{C}^{(i)}~\bm{P}^{(i)}\rb-b^2\bm{I}_r)\rangle.
        \end{split}
    \end{equation}
    Denote $\bm{U}^{(i)}=\lb\bm{C}^{(i)}~\bm{P}^{(i)}\rb$ and $\bm{U}^{*}=\lb\bm{C}^{*}~\bm{P}^{*}\rb$. Note that
    \begin{equation}
        \begin{split}
            & \left\langle\bm{U}^{(i)}-\bm{U}^*\bm{O}_1^{(i)},\bm{U}^{(i)}(\bm{U}^{(i)\top}\bm{U}^{(i)}-b^2\bm{I}_r)\right\rangle\\
            = & \left\langle\bm{U}^{(i)\top}\bm{U}^{(i)}-\bm{U}^{(i)\top}\bm{U}^*\bm{O}_1^{(i)},\bm{U}^{(i)\top}\bm{U}^{(i)}-b^2\bm{I}_r\right\rangle\\
            = & \frac{1}{2}\left\langle\bm{U}^{(i)\top}\bm{U}^{(i)}-\bm{U}^{*\top}\bm{U}^*,\bm{U}^{(i)\top}\bm{U}^{(i)}-b^2\bm{I}_r\right\rangle\\
            + & \frac{1}{2}\left\langle\bm{U}^{*\top}\bm{U}^{*}-2\bm{U}^{(i)\top}\bm{U}^*\bm{O}_1^{(i)}+\bm{U}^{(i)\top}\bm{U}^{(i)},\bm{U}^{(i)\top}\bm{U}^{(i)}-b^2\bm{I}_r\right\rangle\\
            = & \frac{1}{2}\|\bm{U}^{(i)\top}\bm{U}^{(i)}-b^2\bm{I}_r\|_\text{F}^2 + \frac{1}{2}\left\langle\bm{U}^{(i)\top}\left(\bm{U}^{(i)}-\bm{U}^{*}\bm{O}_1^{(i)}\right),\bm{U}^{(i)\top}\bm{U}^{(i)}-b^2\bm{I}_r\right\rangle\\
            + & \frac{1}{2}\left\langle\bm{U}^{*\top}\bm{U}^*-\bm{U}^{(i)\top}\bm{U}^*\bm{O}_1^{(i)},\bm{U}^{(i)\top}\bm{U}^{(i)}-b^2\bm{I}_r\right\rangle.
        \end{split}
    \end{equation}
    In addition, by the fact that $\bm{U}^{*\top}\bm{U}^*=b^2\bm{I}_r$ and $\bm{O}_1^{(i)\top}\bm{O}_1^{(i)}=\bm{I}_r$, we have
    \begin{equation}
        \begin{split}
            & \left\langle \bm{U}^{*\top}\bm{U}^*-\bm{U}^{(i)\top}\bm{U}^*\bm{O}_1^{(i)},\bm{U}^{(i)\top}\bm{U}^{(i)}-b^2\bm{I}_r \right\rangle\\
            = & \left\langle \bm{U}^{*\top}\bm{U}^*-\bm{O}_1^{(i)\top}\bm{U}^{*\top}\bm{U}^{(i)},\bm{U}^{(i)\top}\bm{U}^{(i)}-b^2\bm{I}_r \right\rangle\\
            = & \left\langle \bm{O}_1^{(i)\top}\bm{U}^{*\top}\bm{U}^*\bm{O}_1^{(i)}-\bm{O}_1^{(i)\top}\bm{U}^{*\top}\bm{U}^{(i)},\bm{U}^{(i)\top}\bm{U}^{(i)}-b^2\bm{I}_r \right\rangle\\
            = & \left\langle(\bm{U}^*\bm{O}_1^{(i)})^\top(\bm{U}^*\bm{O}_1^{(i)}-\bm{U}^{(i)}),\bm{U}^{(i)\top}\bm{U}^{(i)}-b^2\bm{I}_r\right\rangle
        \end{split}
    \end{equation}
    and
    \begin{equation}
        \begin{split}
            & \left\langle\bm{U}^{(i)}-\bm{U}^{*}\bm{O}_1^{(i)},\bm{U}^{(i)}(\bm{U}^{(i)\top}\bm{U}^{(i)}-b^2\bm{I}_r)\right\rangle \\
            = & \frac{1}{2}\|\bm{U}^{(i)\top}\bm{U}^{(i)}-b^2\bm{I}_r\|_\text{F}^2\\
            + & \frac{1}{2}\left\langle\left(\bm{U}^{*}\bm{O}_1^{(i)}-\bm{U}^{(i)}\right)^\top\left(\bm{U}^{*}\bm{O}_1^{(i)}-\bm{U}^{(i)}\right),\bm{U}^{(i)\top}\bm{U}^{(i)}-b^2\bm{I}_r\right\rangle\\
            \geq & \frac{1}{2}\|\bm{U}^{(i)\top}\bm{U}^{(i)}-b^2\bm{I}_r\|_\text{F}^2 - \frac{1}{2}\|\bm{U}^{*}\bm{O}_1^{(i)}-\bm{U}^{(i)}\|_\text{F}^2\cdot\|\bm{U}^{(i)\top}\bm{U}^{(i)}-b^2\bm{I}_r\|_\text{F}\\
            \geq & \frac{1}{2}\|\bm{U}^{(i)\top}\bm{U}^{(i)}-b^2\bm{I}_r\|_\text{F}^2 - \frac{1}{4}\|\bm{U}^{*}\bm{O}_1^{(i)}-\bm{U}^{(i)}\|_\text{F}^4-\frac{1}{4}\|\bm{U}^{(i)\top}\bm{U}^{(i)}-b^2\bm{I}_{r-d}\|_\text{F}^2\\
            \geq & \frac{1}{4}\|\bm{U}^{(i)\top}\bm{U}^{(i)}-b^2\bm{I}_r\|_\text{F}^2 - \frac{1}{4}E^{(i)}\|\bm{U}^{*}\bm{O}_1^{(i)}-\bm{U}^{(i)}\|_\text{F}^2.
        \end{split}
    \end{equation}
    Therefore, we have
    \begin{equation}
        \begin{split}
            & T_\text{R}+T_\text{P}+T_\text{C}\\
            \geq & \frac{1}{4}\|\lb\bm{C}^{(i)}~\bm{R}^{(i)}\rb^\top\lb\bm{C}^{(i)}~\bm{R}^{(i)}\rb-b^2\bm{I}_r\|_\text{F}^2 + \frac{1}{4}\|\lb\bm{C}^{(i)}~\bm{P}^{(i)}\rb^\top\lb\bm{C}^{(i)}~\bm{P}^{(i)}\rb-b^2\bm{I}_r\|_\text{F}^2\\
            - & \frac{1}{4}E^{(i)}\|\lb\bm{C}^{(i)}~\bm{R}^{(i)}\rb-\lb\bm{C}^{*}~\bm{R}^{*}\rb\bm{O}_1\|_\text{F}^2 - \frac{1}{4}E^{(i)}\|\lb\bm{C}^{(i)}~\bm{P}^{(i)}\rb-\lb\bm{C}^{*}~\bm{P}^{*}\rb\bm{O}_2\|_\text{F}^2.
        \end{split}
    \end{equation}

    Combining these inequalities, and  since $b=\sigma_1^{1/3}$,
    \begin{equation}
        \begin{split}
            & Q_{\text{D},1} + Q_{\text{R},1} + Q_{\text{P},1} + Q_{\text{C},1}\\
            \geq & \frac{\alpha}{2}\|\bm{A}^{(i)}-\bm{A}^*\|_\text{F}^2+\frac{1}{4\beta}\|\nabla\mathcal{L}(\bm{A}^{(i)})-\nabla\mathcal{L}(\bm{A}^*)\|_\text{F}^2-\frac{1}{4c_2}\xi^2(r,d)\\
            -&\left[32c_2\sigma_1^{4/3}+\frac{C\alpha\sigma_1^{4/3}}{\kappa^2}\right]E^{(i)}\\
            + & \frac{1}{4}\|\lb\bm{C}^{(i)}~\bm{R}^{(i)}\rb^\top\lb\bm{C}^{(i)}~\bm{R}^{(i)}\rb-b^2\bm{I}_r\|_\text{F}^2 + \frac{1}{4}\|\lb\bm{C}^{(i)}~\bm{P}^{(i)}\rb^\top\lb\bm{C}^{(i)}~\bm{P}^{(i)}\rb-b^2\bm{I}_r\|_\text{F}^2\\
            - & \frac{a}{4}E^{(i)}\|\lb\bm{C}^{(i)}~\bm{R}^{(i)}\rb-\lb\bm{C}^{*}~\bm{R}^{*}\rb\bm{O}_1\|_\text{F}^2 - \frac{a}{4}E^{(i)}\|\lb\bm{C}^{(i)}~\bm{P}^{(i)}\rb-\lb\bm{C}^{*}~\bm{P}^{*}\rb\bm{O}_2\|_\text{F}^2\\
            \geq & \frac{\alpha}{2}\Bigg[\|\bm{A}^{(i)}-\bm{A}^*\|_\text{F}^2 + \frac{\sigma_1^{2/3}}{5\kappa^2}\left(\|\lb\bm{C}^{(i)}~\bm{R}^{(i)}\rb^\top\lb\bm{C}^{(i)}~\bm{R}^{(i)}\rb-b^2\bm{I}_r\|_\text{F}^2 \right.\\
            + &\left. \|\lb\bm{C}^{(i)}~\bm{P}^{(i)}\rb^\top\lb\bm{C}^{(i)}~\bm{P}^{(i)}\rb-b^2\bm{I}_r\|_\text{F}^2\right)\Bigg]\\
            + & \frac{1}{4\beta}\|\nabla\mathcal{L}(\bm{A}^{(i)})-\nabla\mathcal{L}(\bm{A}^*)\|_\text{F}^2-\frac{1}{4c_2}\xi^2(r,d)\\
            - &  \left(32c_2\sigma_1^{4/3}+\frac{C\alpha\sigma_1^{4/3}}{\kappa^2}\right)E^{(i)}\\
            + & \left(\frac{a}{4}-\frac{\alpha\sigma_1^{2/3}}{10\kappa^2}\right)\left(\|\lb\bm{C}^{(i)}~\bm{R}^{(i)}\rb^\top\lb\bm{C}^{(i)}~\bm{R}^{(i)}\rb-b^2\bm{I}_r\|_\text{F}^2 + \|\lb\bm{C}^{(i)}~\bm{P}^{(i)}\rb^\top\lb\bm{C}^{(i)}~\bm{P}^{(i)}\rb-b^2\bm{I}_r\|_\text{F}^2).\right.
        \end{split}
    \end{equation}
    Letting $a=0.8\alpha\sigma_1^{2/3}\kappa^{-2}$ and since $E^{(i)}\leq\sigma_1^{2/3}$
    \begin{equation}
        \begin{split}
            & Q_{\text{D},1} + Q_{\text{R},1} + Q_{\text{P},1} + Q_{\text{C},1}\\
            \geq & \left(\frac{\alpha\sigma_1^{4/3}}{140\kappa^2}-32c_2\sigma_1^{4/3}-\frac{C\alpha\sigma_1^{4/3}}{\kappa^2}\right)E^{(i)}\\
            + & \frac{1}{4\beta}\|\nabla\mathcal{L}(\bm{A}^{(i)})-\nabla\mathcal{L}(\bm{A}^*)\|_\text{F}^2 - \frac{1}{4c_2}\xi^2(r,d)\\
            + & \frac{a}{8}\left(\|\lb\bm{C}^{(i)}~\bm{R}^{(i)}\rb^\top\lb\bm{C}^{(i)}~\bm{R}^{(i)}\rb-b^2\bm{I}_r\|_\text{F}^2 + \|\lb\bm{C}^{(i)}~\bm{P}^{(i)}\rb^\top\lb\bm{C}^{(i)}~\bm{P}^{(i)}\rb-b^2\bm{I}_r\|_\text{F}^2\right).
        \end{split}
    \end{equation}~\\

    \noindent\textit{Step 4.} (Convergence analysis of $E^{(i)}$)

    In the following, we combine all the results in the previous steps to establish the error bound for $E^{(i)}$ and $\|\bm{A}-\bm{A}^*\|_\text{F}$. Plugging in $b=\sigma_1^{1/3}$ and $a=1.6\alpha\sigma_1^{2/3}\kappa^{-2}$ to $Q_{\text{D},2}$, $Q_{\text{R},2}$, $Q_{\text{P},2}$ and $Q_{\text{C},2}$, we have
    \begin{equation}
        \begin{split}
            & Q_{\text{D},2} + Q_{\text{R},2} + Q_{\text{P},2} + Q_{\text{C},2} \\
            \leq & 92\sigma_1^{4/3}[\xi^2(r,d)+\|\nabla\mathcal{L}(\bm{A}^{(i)})-\nabla\mathcal{L}(\bm{A}^*)\|_\text{F}^2]\\
            + & \frac{52\alpha^2\sigma_1^{2}}{\kappa^4}\left(\|\lb\bm{C}^{(i)}~\bm{R}^{(i)}\rb^\top\lb\bm{C}^{(i)}~\bm{R}^{(i)}\rb-b^2\bm{I}_r\|_\text{F}^2 + \|\lb\bm{C}^{(i)}~\bm{P}^{(i)}\rb^\top\lb\bm{C}^{(i)}~\bm{P}^{(i)}\rb-b^2\bm{I}_r\|_\text{F}^2\right).
        \end{split}
    \end{equation}
    
    Combining the upper bound for $Q_{\text{D},2} + Q_{\text{R},2} + Q_{\text{P},2} + Q_{\text{C},2}$ and the lower bound for $Q_{\text{D},1} + Q_{\text{R},1} + Q_{\text{P},1} + Q_{\text{C},1}$, we have
    \begin{equation}\label{eq:Ebound}
        \begin{split}
            E^{(i+1)} & \leq \left(1-2\eta\left(\frac{\alpha\sigma_1^{4/3}}{140\kappa^2}-32c_2\sigma_1^{4/3}-\frac{C\alpha\sigma_1^{4/3}}{\kappa^2}\right)\right)E^{(i)}\\
            & + \left(92\sigma_1^{4/3}\eta^2-\frac{\eta}{2\beta}\right)\|\nabla\mathcal{L}(\bm{A}^{(i)})-\nabla\mathcal{L}(\bm{A}^*)\|_\text{F}^2\\
            & + \left(\frac{\eta}{2c_2}+92\sigma_1^{4/3}\eta^2\right)\xi^2(r,d)\\
            & - \left(\frac{\eta\alpha\sigma_1^{2/3}}{5\kappa^2}-\frac{52\eta^2\alpha^2\sigma_1^2}{\kappa^4}\right)(\|\bm{U}^{(i)\top}\bm{U}^{(i)}-b^2\bm{I}_{r}\|_\text{F}^2 + \|\bm{V}^{(i)\top}\bm{V}^{(i)}-b^2\bm{I}_{r}\|_\text{F}^2).
        \end{split}
    \end{equation}
    Letting $c_2=C\alpha\kappa^{-2}$, $\eta=\eta_0\beta^{-1}\sigma_1^{-4/3}$ with $\eta_0\leq 1/260$, since $E^{(i)}\leq C\sigma_1^{2/3}$, the coefficients of the second, third, and fourth term in \eqref{eq:Ebound} are
    \begin{equation}
        92\sigma_1^{4/3}\eta^2-0.5\eta\beta^{-1} = (92\eta_0-0.5)\eta_0\sigma_1^{-4/3}\beta^{-2}\leq 0,
    \end{equation}
    \begin{equation}
        \frac{\eta}{2c_2}+92\sigma_1^{4/3}\eta^2\leq \eta_0\alpha^{-1}\beta^{-1}\sigma_1^{-4/3}\kappa^2+92\sigma_1^{-4/3}\beta^{-2}\leq C\alpha^{-1}\beta^{-1}\sigma_1^{-4/3}\kappa^2,
    \end{equation}
    and
    \begin{equation}
        \frac{\eta\alpha\sigma_1^{2/3}}{5\kappa^2}-\frac{52\eta^2\alpha^2\sigma_1^2}{\kappa^4}=0.2\eta_0\alpha\beta^{-1}\sigma_1^{-2/3}\kappa^{-2}(1-260\eta_0\alpha\beta^{-1}\kappa^{-2})\geq0,
    \end{equation}
    as $\alpha\beta^{-1}\leq 1$, $\kappa^{-2}\leq 1$ and $\eta_0\leq 1/260$.
    Therefore, we can derive the following recursive inequality
    \begin{equation}\label{eq:recursive}
        \begin{split}
            & E^{(i+1)} \leq \left(1-C\eta_0\alpha\beta^{-1}\kappa^{-2}\right)E^{(i)} + C\kappa^2\alpha^{-2}\sigma_1^{-4/3}\xi^2(r,d).
        \end{split}
    \end{equation}
    By induction, we have that for any $i=1,2,\dots$,
    \begin{equation}
        E^{(i)} \leq (1-C\eta_0\alpha\beta^{-1}\kappa^{-2})^iE^{(0)} + C\kappa^2\alpha^{-2}\sigma_1^{-4/3}\xi^2(r,d).
    \end{equation}
    For the error bound of $\|\bm{A}^{(i)}-\bm{A}^*\|_\text{F}$, by Lemma \ref{lemma:errorbound},
    \begin{equation}
        \begin{split}
            & \|\bm{A}^{(i)}-\bm{A}^*\|_\text{F}^2\leq C\sigma_1^{4/3}E^{(i)} \\
            \leq & C\sigma_1^{4/3}(1-C\eta_0\alpha\beta^{-1}\kappa^{-2})^iE^{(0)} + C\kappa^2\alpha^{-2} \xi^2(r,d)\\
            \leq & C\kappa^2(1-C\eta_0\alpha\beta^{-1}\kappa^{-2})^i\|\bm{A}^{(0)}-\bm{A}^*\|_\text{F}^2 + C\kappa^2\alpha^{-2}\xi^2(r,d).
        \end{split}
    \end{equation}~\\

    \noindent\textit{Step 5.} (Verification of conditions)
    
    Finally, we show that conditions $E^{(i)}\leq C\sigma_1^{2/3}\alpha\beta^{-1}\kappa^{-2}$ and \eqref{eq:UVDopbound} hold. 
    
    Since $\lb\bm{C}^{(0)}~\bm{R}^{(0)}\rb^\top\lb\bm{C}^{(0)}~\bm{R}^{(0)}\rb=\bm{I}_r$ and $\lb\bm{C}^{(0)}~\bm{P}^{(0)}\rb^\top\lb\bm{C}^{(0)}~\bm{P}^{(0)}\rb=\bm{I}_r$, by Lemma \ref{lemma:errorbound} and initialization bound $\|\bm{A}^{(0)}-\bm{A}^*\|_\text{F}\leq C\sigma_{r}\alpha^{1/2}\beta^{-1/2}$, we have
    \begin{equation}
        E^{(0)}  \leq (C\sigma_1^{-4/3}+C\sigma_1^{2/3}\sigma_{r}^{-2})\|\bm{A}^{(0)}-\bm{A}^*\|_\text{F}^2\leq C  \sigma_1^{-4/3}\kappa^2\|\bm{A}^{(0)}-\bm{A}^*\|_\text{F}^2\leq C\sigma_1^{2/3}\alpha\beta^{-1}\kappa^{-2}.
    \end{equation}
    Based on the recursive relationship in \eqref{eq:recursive}, by induction it is easy to check that $E^{(i)}\leq C\sigma_1^{2/3}\alpha\beta^{-1}\kappa^{-2}$ for all $i\geq 1$. In other words, as $\alpha\beta^{-1}\leq 1$ and $\kappa^{-2}\leq 1$, we have $E^{(i)}\leq Cb^2$ for all $i\geq 1$, which further implies that
    \begin{equation}
        \begin{split}
            & \|\lb\bm{C}^{(i)}~\bm{R}^{(i)}\rb\|_\text{op} \leq \|\lb\bm{C}^{*}~\bm{R}^{*}\rb\bm{O}_1^{(i)}\|_\text{op} + \|\lb\bm{C}^{(i)}~\bm{R}^{(i)}\rb-\lb\bm{C}^{*}~\bm{R}^{*}\rb\bm{O}_1^{(i)}\|_\text{op}\\
            \leq & b + \|\lb\bm{C}^{(i)}~\bm{R}^{(i)}\rb-\lb\bm{C}^{*}~\bm{R}^{*}\rb\bm{O}_1^{(i)}\|_\text{F} \leq (1+\sqrt{C})b\leq (1+c_b)b,
        \end{split}
    \end{equation}
    and 
    \begin{equation}
        \begin{split}
            & \|\bm{D}^{(i)}\|_\text{op} \leq \|\bm{O}_1^{(i)\top}\bm{D}^*\bm{O}_2^{(i)}\|_\text{op} + \|\bm{D}^{(i)}-\bm{O}_1^{(i)\top}\bm{D}^*\bm{O}_2^{(i)}\|_\text{op}\\
            \leq & \sigma_1b^{-2} + \|\bm{D}^{(i)}-\bm{O}_1^{(i)\top}\bm{D}^*\bm{O}_2^{(i)}\|_\text{F} \leq (1+c_b)\sigma_1b^{-2},
        \end{split}
    \end{equation}
    which completes the deterministic analysis.

\end{proof}

\subsection{Auxiliary lemmas}
    
    The first lemma follows from Lemma E.3 in \citet{han2020optimal} with the tensor order changed from 3 to 2.
    
    \begin{lemma}\label{lemma:perturb}
        Suppose that $\bm{A}^*=\lb\bm{C}^*~\bm{R}^*\rb\bm{D}^*\lb\bm{C}^*~\bm{P}^*\rb^\top$, $\bm{A}=\lb\bm{C}~\bm{R}\rb\bm{D}\lb\bm{C}~\bm{P}\rb^\top$ with $\bm{D}^*,\bm{D}\in\mathbb{R}^{r\times r}$, $\bm{C}^*,\bm{C}\in\mathbb{R}^{p\times d}$, $\bm{R},\bm{R}^*,\bm{P},\bm{P}^*\in\mathbb{R}^{p\times (r-d)}$, $\bm{O}_c\in\mathbb{O}^{d\times d}$, $\bm{O}_r,\bm{O}_p\in\mathbb{O}^{(r-d)\times (r-d)}$, $\bm{O}_1=\textup{diag}(\bm{O}_c,\bm{O}_r)$, and $\bm{O}_2=\textup{diag}(\bm{O}_c,\bm{O}_p)$. Let
        \begin{equation}
            \begin{split}
                &\bm{A}_1 = \lb\bm{C}^*~\bm{R}^*\rb\bm{O}_1\bm{D}\lb\bm{C}~\bm{P}\rb^\top,~~\bm{A}_2 = \lb\bm{C}~\bm{R}\rb\bm{D}\bm{O}_2^\top\lb\bm{C}^*~\bm{P}^*\rb^\top,\\
                &\bm{H}_1=\lb\bm{C}^*~\bm{R}^*\rb-\lb\bm{C}~\bm{R}\rb\bm{O}_1^\top,~~\bm{H}_2=\lb\bm{C}^*~\bm{P}^*\rb-\lb\bm{C}~\bm{P}\rb\bm{O}_2^\top,\\
                &\bm{A}_{d}=\lb\bm{C}~\bm{R}\rb\bm{O}_1^\top\bm{D}^*\bm{O}_2\lb\bm{C}~\bm{P}\rb^{\top},~~\bm{H}_{d}=\bm{D}^*-\bm{O}_1\bm{D}\bm{O}_2^\top.
            \end{split}
        \end{equation}
        Then, defining  
        \begin{equation}
            \bm{H} = \bm{A}^*-\bm{A}_{d} - (\bm{A}_1-\bm{A}) - (\bm{A}_2-\bm{A}),
        \end{equation}
        we have 
        \begin{equation}
            \|\bm{H}\|_\textup{F} \leq B_2B_3 + 2B_1B_3 + B_3^{3/2},
        \end{equation}
        where
        \begin{equation}
            \begin{split}
                B_1&:=\max\{\|[\bm{C},\bm{R}]\|_\textup{op},\|[\bm{C}^*,\bm{R}^*]\|_\textup{op},\|[\bm{C},\bm{P}]\|_\textup{op},\|[\bm{C}^*,\bm{P}^*]\|_\textup{op}\},\\
                B_2&:=\max\{\|\bm{D}\|_\textup{op},\|\bm{D}^*\|_\textup{op}\},\\
                B_3&:=\max\{\|\bm{H}_{d}\|_\textup{F}^2,\|\bm{H}_1\|_\textup{F}^2,\|\bm{H}_2\|_\textup{F}^2\}.
            \end{split}
        \end{equation}
    \end{lemma}

    \begin{proof}
        
        Since $\bm{D}^*=\bm{O}_1\bm{D}\bm{O}_2^\top+\bm{H}_{d}$, we have
        \begin{equation}
            \bm{A}^*=\underbrace{\lb\bm{C}^*~\bm{R}^*\rb\bm{O}_1\bm{D}\bm{O}_2^\top\lb\bm{C}^*~\bm{P}^*\rb^\top}_{T_1} + \underbrace{\lb\bm{C}^*~\bm{R}^*\rb\bm{H}_{\bm{D}}\lb\bm{C}^*~\bm{P}^*\rb^{\top}}_{T_2}.
        \end{equation}
        For $T_1$, since $\lb\bm{C}^*~\bm{R}^*\rb=\lb\bm{C}~\bm{R}\rb\bm{O}_1^\top + \bm{H}_1$ and $\lb\bm{C}^*~\bm{P}^*\rb=\lb\bm{C}~\bm{P}\rb\bm{O}_2^\top + \bm{H}_2$, we have
        \begin{equation}
            \begin{split}
                T_1 & = (\lb\bm{C}~\bm{R}\rb+\bm{H}_1\bm{O}_1)\bm{D}(\lb\bm{C}~\bm{P}\rb+\bm{H}_2\bm{O}_2)^\top\\
                & = \lb\bm{C}~\bm{R}\rb\bm{D}\lb\bm{C}~\bm{P}\rb^\top + \bm{H}_1\bm{O}_1\bm{D}\lb\bm{C}~\bm{P}\rb^\top + \lb\bm{C}~\bm{R}\rb\bm{D}\bm{O}_2^\top\bm{H}_2^\top + \bm{H}^{(1)}_\varepsilon \\
                & = \lb\bm{C}~\bm{R}\rb\bm{D}\lb\bm{C},\bm{P}\rb^\top + (\lb\bm{C}^*~\bm{R}^*\rb\bm{O}_1-\lb\bm{C}~\bm{R}\rb)\bm{D}\lb\bm{C}~\bm{P}\rb^\top\\ & + \lb\bm{C}~\bm{R}\rb\bm{D}(\lb\bm{C}^*~\bm{P}^*\rb\bm{O}_2-\lb\bm{C}~\bm{P}\rb)^\top + \bm{H}^{(1)}_\varepsilon \\
                & = \lb\bm{C}^*~\bm{R}^*\rb\bm{O}_1\bm{D}\lb\bm{C}~\bm{P}\rb^\top + \lb\bm{C}~\bm{R}\rb\bm{D}\bm{O}_2^\top\lb\bm{C}^*~\bm{P}^*\rb^\top - \lb\bm{C}~\bm{R}\rb\bm{D}\lb\bm{C}~\bm{P}\rb^\top + \bm{H}^{(1)}_\varepsilon,
            \end{split}
        \end{equation}
        where $\bm{H}^{(1)}_\varepsilon=\bm{H}_1\bm{O}_1\bm{D}\bm{O}_2^\top\bm{H}_2^\top$.
        For $T_2$, we have
        \begin{equation}
            \begin{split}
                T_2 & =(\bm{H}_1+\lb\bm{C}~\bm{R}\rb\bm{O}_1^\top)\bm{H}_{d}(\bm{H}_2+\lb\bm{C}~\bm{P}\rb\bm{O}_2^\top)^\top\\
                & = \lb\bm{C}~\bm{R}\rb\bm{O}_1^\top\bm{D}^*\bm{O}_2\lb\bm{C}~\bm{P}]^\top - \lb\bm{C}~\bm{R}\rb\bm{D}\lb\bm{C}~\bm{P}\rb^\top + \bm{H}_\varepsilon^{(2)},
            \end{split}
        \end{equation}
        where $\bm{H}_\varepsilon^{(2)}=\bm{H}_1\bm{H}_{d}\bm{H}_2^\top+\lb\bm{C}~\bm{R}\rb\bm{O}_1^\top\bm{H}_{d}\bm{H}_2^\top+\bm{H}_1\bm{H}_{d}\bm{O}_2\lb\bm{C}~\bm{P}\rb^\top$. Hence, $\bm{H}=\bm{H}_\varepsilon^{(1)}+\bm{H}_\varepsilon^{(2)}$ and
        \begin{equation}
            \begin{split}
                \|\bm{H}\|_\text{F} & \leq \|\bm{H}_1\bm{O}_1\bm{D}\bm{O}_2^\top\bm{H}_2^\top\|_\text{F} + \|\lb\bm{C}~\bm{R}\rb\bm{O}_1^\top\bm{H}_{d}\bm{H}_2^\top\|_\text{F}\\
                & + \|\bm{H}_1\bm{H}_{d}\bm{O}_2\lb\bm{C}~\bm{P}\rb^\top\|_\text{F} + \|\bm{H}_1\bm{H}_{d}\bm{H}_2^\top\|_\text{F}\\
                & \leq B_2B_3 + 2B_1B_3 + B_3^{3/2}.
            \end{split}
        \end{equation}
        
    \end{proof}
    
    The following lemma shares similar ideas and techniques as those of Lemma E.2 in \citet{han2020optimal} with the common subspace structure included.
    \begin{lemma}\label{lemma:errorbound}
        Suppose that $\bm{A}^*=\lb\bm{C}^*~\bm{R}^*\rb\bm{D}^*\lb\bm{C}^*~\bm{P}^*\rb^{\top}$, $\lb\bm{C}^*~\bm{R}^*\rb^{\top}\lb\bm{C}^*~\bm{R}^{*}\rb=\bm{I}_r$, $\lb\bm{C}^*~\bm{P}^*\rb^{\top}\lb\bm{C}^*~\bm{P}^{*}\rb=\bm{I}_r$, $\sigma_1=\|\bm{A}^*\|_\textup{op}$, and $\sigma_{r}=\sigma_{r}(\bm{A}^*)$. Let $\bm{A}=[\bm{C},\bm{R}]\bm{D}[\bm{C},\bm{P}]^{\top}$ with $\|\lb\bm{C}~\bm{R}\rb\|_\textup{op}\leq (1+c_b)b$, $\|\lb\bm{C}~\bm{P}\rb\|_\textup{op}\leq (1+c_b)b$ and $\|\bm{D}\|_\textup{op}\leq(1+c_b)\sigma_1/b^2$ for some constant $c_b>0$. Define
        \begin{equation}
            \begin{split}
                E:=&\min_{\substack{\bm{O}_c\in\mathbb{O}^{d\times d}\\\bm{O}_r,\bm{O}_p\in\mathbb{O}^{(r-d)\times(r-d)}}}\big(\|\lb\bm{C}^{(i)}~\bm{R}^{(i)}\rb-\lb\bm{C}^*~\bm{R}^*\rb\textup{diag}(\bm{O}_c,\bm{O}_r)\|_\textup{F}^2+\|\lb\bm{C}^{(i)}~\bm{P}^{(i)}\rb\\
                &-\lb\bm{C}^*~\bm{P}^*\rb\textup{diag}(\bm{O}_c,\bm{O}_p)\|_\textup{F}^2+ \|\bm{D}^{(i)}-\textup{diag}(\bm{O}_c,\bm{O}_r)^\top\bm{D}^*\textup{diag}(\bm{O}_c,\bm{O}_p)\|_\textup{F}^2\big).
            \end{split}
        \end{equation}
        Then, we have
        \begin{equation}
            \begin{split}
                E & \leq \left(4b^{-4}+\frac{8b^2}{\sigma_{r^*}^2}C_b\right)\|\bm{A}-\bm{A}^*\|_\textup{F}^2\\
                + & 2b^{-2}C_b\left(\|\lb\bm{C}~\bm{R}\rb^\top\lb\bm{C}~\bm{R}\rb-b^2\bm{I}_r\|_\textup{F}^2+\|\lb\bm{C}~\bm{P}\rb^\top\lb\bm{C}~\bm{P}\rb-b^2\bm{I}_r\|_\textup{F}^2\right),
            \end{split}
        \end{equation}
        and
        \begin{equation}
            \|\bm{A}-\bm{A}^*\|_\textup{F}^2 \leq 3b^4[1+4\sigma_1^2b^{-6}(1+c_b)^4]E,
        \end{equation}
        where $C_b=1+4\sigma_1^2b^{-6}((1+c_b)^4+(1+c_b)^2(2+c_b)^2/2)$.
    \end{lemma}

    \begin{proof}
        Denote $\bm{O}_1=\textup{diag}(\bm{O}_c,\bm{O}_r)$ and $\bm{O}_2=\textup{diag}(\bm{O}_c,\bm{O}_p)$.
        Note that $\|\bm{O}_1\bm{D}\bm{O}_2^\top-\bm{D}^*\|_\text{F}=b^{-2}\|\lb\bm{C}^*~\bm{R}^*\rb\bm{O}_1\bm{D}\bm{O}_2^\top\lb\bm{C}^*~\bm{P}^*\rb^{\top}-\lb\bm{C}^*~\bm{R}^*\rb\bm{D}^*\lb\bm{C}^*~\bm{P}^*\rb^{\top}\|_\text{F}$. We have the decomposition
        \begin{equation}
            \begin{split}
                &\lb\bm{C}^*~\bm{R}^*\rb\bm{O}_1\bm{D}\bm{O}_2^\top\lb\bm{C}^*~\bm{P}^*\rb^{\top}-\lb\bm{C}^*~\bm{R}^*\rb\bm{D}^*\lb\bm{C}^*~\bm{P}^*\rb^{\top}\\
                =&(\lb\bm{C}~\bm{R}\rb+\lb\bm{C}^*~\bm{R}^*\rb\bm{O}_1-\lb\bm{C}~\bm{R}\rb)\bm{D}(\lb\bm{C}~\bm{P}\rb+\lb\bm{C}^*~\bm{P}^*\rb\bm{O}_2-\lb\bm{C}~\bm{P}\rb)^\top
                -\lb\bm{C}^*~\bm{R}^*\rb\bm{D}^*\lb\bm{C}^*~\bm{P}^*\rb^{\top}\\
                =&(\lb\bm{C}~\bm{R}\rb\bm{D}\lb\bm{C}~\bm{P}\rb^\top-\lb\bm{C}^*~\bm{R}^*\rb\bm{D}^*\lb\bm{C}^*~\bm{P}^*\rb^{\top})+(\lb\bm{C}^*~\bm{R}^*\rb\bm{O}_1-\lb\bm{C}~\bm{R}\rb)\bm{D}\lb\bm{C}~\bm{P}\rb^\top\\
                +&\lb\bm{C}~\bm{R}\rb\bm{D}(\lb\bm{C}^*~\bm{P}^*\rb\bm{O}_2-\lb\bm{C}~\bm{P}\rb)^\top+(\lb\bm{C}^*~\bm{R}^*\rb\bm{O}_1-\lb\bm{C}~\bm{R}\rb)\bm{D}(\lb\bm{C}^*~\bm{P}^*\rb\bm{O}_2-\lb\bm{C}~\bm{P}\rb)^\top.
            \end{split}
        \end{equation}
        By mean inequality,
        \begin{equation}
            \begin{split}
                & \|\bm{D}-\bm{O}_1^\top\bm{D}^*\bm{O}_2\|_\text{F}^2 = \|\bm{O}_1\bm{D}\bm{O}_2^\top-\bm{D}^*\|_\text{F}^2\\
                \leq & 4b^{-4}\|\bm{A}-\bm{A}^*\|_\text{F}^2+4b^{-4}\|\lb\bm{C}^*~\bm{R}^*\rb\bm{O}_1-\lb\bm{C}~\bm{R}\rb\|_\text{F}^2\cdot\|\bm{D}\|_\text{op}^2\cdot\|\lb\bm{C}~\bm{P}\rb\|_\text{op}^2\\
                + & 4b^{-4}\|\lb\bm{C}^*~\bm{P}^*\rb\bm{O}_2-\lb\bm{C}~\bm{P}\rb\|_\text{F}^2\cdot\|\bm{D}\|_\text{op}^2\cdot\|\lb\bm{C}~\bm{R}\rb\|_\text{op}^2\\
                +& 2b^{-4}\|\lb\bm{C}^*~\bm{R}^*\rb\bm{O}_1-\lb\bm{C}~\bm{R}\rb\|_\text{F}^2\cdot\|\lb\bm{C}^*~\bm{P}^*\rb\bm{O}_2-\lb\bm{C}~\bm{P}\rb\|_\text{op}^2\cdot\|\bm{D}\|_\text{op}^2\\
                + & 2b^{-4} \|\lb\bm{C}^*~\bm{R}^*\rb\bm{O}_1-\lb\bm{C}~\bm{R}\rb\|_\text{op}^2\cdot\|\lb\bm{C}^*~\bm{P}^*\rb\bm{O}_2-\lb\bm{C}~\bm{P}\rb\|_\text{F}^2\cdot\|\bm{D}\|_\text{op}^2 \\
                \leq &  4b^{-4}\|\bm{A}-\bm{A}^*\|_\text{F}^2 +  4b^{-4}\left((1+c_b)^4\sigma_1^2b^{-2}+(1+c_b)^2(2+c_b)^2\sigma_1^2b^{-2}/2\right)\\
                &\left(\|\lb\bm{C}^*~\bm{R}^*\rb\bm{O}_1-\lb\bm{C}~\bm{R}\rb\|_\text{F}^2 + \|\lb\bm{C}^*~\bm{P}^*\rb\bm{O}_2-\lb\bm{C}~\bm{P}\rb\|_\text{F}^2\right).
            \end{split}
        \end{equation}
        Hence, it follows that
        \begin{equation}
            \begin{split}
                &E\\ = & \min_{\bm{O}_i,i=1,2}\left(\|\lb\bm{C}^*~\bm{R}^*\rb\bm{O}_1-\lb\bm{C}~\bm{R}\rb\|_\text{F}^2+\|\lb\bm{C}^*~\bm{P}^*\rb\bm{O}_2-\lb\bm{C}~\bm{P}\rb\|_\text{F}^2+\|\bm{D}-\bm{O}_1^\top\bm{D}^*\bm{O}_2\|_\text{F}^2\right)\\
                \leq & 4b^{-4}\|\bm{A}-\bm{A}^*\|_\text{F}^2 + C_b\min_{\bm{O}_i,i=1,2}\left\{\|\lb\bm{C}~\bm{R}\rb-\lb\bm{C}^*~\bm{R}^*\rb\bm{O}_1\|_\text{F}^2+\|\lb\bm{C}~\bm{P}\rb-\lb\bm{C}^*~\bm{P}^*\rb\bm{O}_1\|_\text{F}^2\right\},
            \end{split}
        \end{equation}
        where $C_b=1+4\sigma^2b^{-6}((1+c_b)^4+(1+c_b)^2(2+c_b)^2/2)$.
        
        Let $\widetilde{\bm{U}}\widetilde{\bm{D}}\widetilde{\bm{V}}^\top$ be the SVD of $\lb\bm{C}~\bm{R}\rb$. Then, we have
        \begin{equation}
            \begin{split}
                & \min_{\bm{O}_1}\|\lb\bm{C}~\bm{R}\rb-\lb\bm{C}^*~\bm{R}^*\rb\bm{O}_1\|_\text{F}^2\\ = & \min_{\bm{O}_1}\|\lb\bm{C}~\bm{R}\rb-b\widetilde{\bm{U}}\widetilde{\bm{V}}^\top+b\widetilde{\bm{U}}\widetilde{\bm{V}}^\top-\lb\bm{C}^*~\bm{R}^*\rb\bm{O}_1\|_\text{F}^2\\
                \leq & 2\|\widetilde{\bm{U}}\widetilde{\bm{D}}\widetilde{\bm{V}}^\top-b\widetilde{\bm{U}}\widetilde{\bm{V}}^\top\|_\text{F}^2+2\min_{\bm{O}_1}\|b\widetilde{\bm{U}}\widetilde{\bm{V}}^\top-\lb\bm{C}^*~\bm{R}^*\rb\bm{O}_1\|_\text{F}^2\\
                = & 2\|\widetilde{\bm{D}}-b\bm{I}_r\|_\text{F}^2+2\min_{\bm{O}_1}\|b\widetilde{\bm{U}}-\lb\bm{C}^*~\bm{R}^*\rb\bm{O}_1\|_\text{F}^2.
            \end{split}
        \end{equation}
        Similarly to Lemma E.2 in \citet{han2020optimal}, we have
        \begin{equation}
            \|\widetilde{\bm{D}}-b\bm{I}_r\|_\text{F}^2
            \leq  b^{-2}\|\lb\bm{C}~\bm{R}\rb^\top\lb\bm{C}~\bm{R}\rb-b^2\bm{I}_r\|_\text{F}^2.
        \end{equation}
        Let $\widetilde{\bm{U}}_\perp$ be the perpendicular orthonormal matrix of $\widetilde{\bm{U}}$. As $\widetilde{\bm{U}}$ and $\lb\bm{C}^*~\bm{R}^*\rb/b$ are orthonormal matrices spanning left singular subspaces of $\bm{A}$ and $\bm{A}^*$, we have
        \begin{equation}
            \begin{split}
                \|\bm{A}-\bm{A}^*\|_\text{F}^2 & \geq \|\widetilde{\bm{U}}_\perp^\top(\bm{A}-\bm{A}^*)\|_\text{F}^2=\|\widetilde{\bm{U}}_\perp^\top\bm{A}^*\|_\text{F}^2\\
                & =  \|\widetilde{\bm{U}}_\perp^\top(\lb\bm{C}^*~\bm{R}^*\rb/b)(\lb\bm{C}^*~\bm{R}^*\rb/b)^\top\bm{A}^*\|_\text{F}^2\\
                & \geq \sigma_{r}^2\|\widetilde{\bm{U}}_\perp^\top(\lb\bm{C}^*~\bm{R}^*\rb/b)\|_\text{F}^2.
            \end{split}
        \end{equation}
        By Lemma 1 in \citet{cai2018rate},
        \begin{equation}
            \min_{\bm{O}_1}\|b\widetilde{\bm{U}}-\lb\bm{C}^*~\bm{R}^*\rb\bm{O}_1\|_\text{F}^2
            \leq 2b^2\|\widetilde{\bm{U}}_\perp^\top\lb\bm{C}^*~\bm{R}^*\rb\|_\text{F}^2\leq2b^2\frac{\|\bm{A}-\bm{A}^*\|_\text{F}^2}{\sigma_{r}^2}.
        \end{equation}
        These imply that
        \begin{equation}
            \begin{split}
                E & \leq \left(4b^{-4}+\frac{8b^2}{\sigma_1^2}c_b\right)\|\bm{A}-\bm{A}^*\|_\text{F}^2\\
                & + 2b^{-2}c_b\left(\|\lb\bm{C}~\bm{R}\rb^\top\lb\bm{C}~\bm{R}\rb-b^2\bm{I}_r\|_\text{F}^2+\|\lb\bm{C}~\bm{P}\rb^\top\lb\bm{C}~\bm{P}\rb-b^2\bm{I}_r\|_\text{F}^2\right).
            \end{split}
        \end{equation}
        
        For the second inequality, denote the optimal rotation matrices by
        \begin{equation}
            \begin{split}
                (\bm{O}_1,\bm{O}_2)=\argmin_{\substack{\bm{Q}_k,k=1,2}}&\big\{\|\lb\bm{C}~\bm{R}\rb-\lb\bm{C}^*~\bm{R}^*\rb\bm{Q}_1\|_\text{F}^2\\
                &+\|\lb\bm{C}~\bm{P}\rb-\lb\bm{C}^*~\bm{P}^*\rb\bm{Q}_2\|_\text{F}^2+\|\bm{D}-\bm{Q}_1^\top\bm{D}^*\bm{Q}_2\|_\text{F}^2\big\}.
            \end{split}
        \end{equation}
        Let $\bm{H}_{d}=\bm{D}^*-\bm{O}_1^\top\bm{D}\bm{O}_2$,  $\bm{H}_1=\lb\bm{C}^*~\bm{R}^*\rb-\lb\bm{C}~\bm{R}\rb\bm{O}_1^\top$, and $\bm{H}_2=\lb\bm{C}^*~\bm{P}^*\rb-\lb\bm{C}~\bm{P}\rb\bm{O}_2^\top$. Then,
        \begin{equation}
            \bm{A}^*=(\bm{H}_1+\lb\bm{C}~\bm{R}\rb\bm{O}_1^\top)(\bm{H}_{d}+\bm{O}_1^\top\bm{D}\bm{O}_2)(\bm{H}_2+\lb\bm{C}~\bm{P}\rb\bm{O}_2^\top)^\top
        \end{equation}
        and
        \begin{equation}
            \begin{split}
                & \|\bm{A}^*-\bm{A}\|_\text{F} \\ \leq & \|\lb\bm{C}^*~\bm{R}^*\rb\bm{H}_{d}\lb\bm{C}^*~\bm{P}^*\rb^{\top}\|_\text{F} +
                \|\bm{H}_1\bm{O}_1^\top\bm{D}\bm{O}_2\bm{H}_2^\top\|_\text{F}\\
                + & \|\bm{H}_1\bm{O}_1^\top\bm{D}\lb\bm{C}~\bm{P}\rb^\top\|_\text{F} + \|\lb\bm{C}~\bm{R}\rb\bm{D}\bm{O}_2\bm{H}_2^\top\|_\text{F}\\
                \leq & \|\lb\bm{C}^*~\bm{R}^*\rb\|_\text{op}\cdot\|\lb\bm{C}^*~\bm{P}^*\rb\|_\text{op}\cdot\|\bm{H}_{d}\|_\text{F} + \frac{1}{2}\|\bm{D}\|_\text{op}\cdot\|\bm{H}_1\|_\text{op}\cdot\|\bm{H}_2\|_\text{F}\\
                + &  \frac{1}{2}\|\bm{D}\|_\text{op}\cdot\|\bm{H}_1\|_\text{F}\cdot\|\bm{H}_2\|_\text{op} + \|\bm{H}_1\|_\text{F}\cdot\|\bm{D}\|_\text{op}\cdot\|\lb\bm{C}~\bm{P}\rb\|_\text{op} + \|\lb\bm{C}~\bm{R}\rb\|_\text{op}\cdot\|\bm{D}\|_\text{op}\cdot\|\bm{H}_2\|_\text{F} \\
                \leq & b^2\|\bm{H}_{d}\|_\text{F} + \sigma_1b^{-1}[(1+c_b)^2+(1+c_b)(2+c_b)/2](\|\bm{H}_1\|_\text{F}+\|\bm{H}_2\|_\text{F})\\
                \leq & b^2\|\bm{H}_{d}\|_\text{F}+2\sigma_1b^{-1}(1+c_b)^2(\|\bm{H}_1\|_\text{F}+\|\bm{H}_2\|_\text{F}).
            \end{split}
        \end{equation}
        Hence,
        \begin{equation}
            \|\bm{A}-\bm{A}^*\|_\text{F}^2 \leq 3\left[b^4\|\bm{H}_{d}\|_\text{F}^2+4\sigma_1^2b^{-2}(1+c_b)^4(\|\bm{H}_1\|_\text{F}^2 + \|\bm{H}_2\|_\text{F}^2)\right].
        \end{equation}
        
    \end{proof}

\section{Statistical convergence analysis of gradient descent}\label{append:stat}

In this appendix, we present the stochastic properties of the time series data.

\begin{proof}[Proof of Theorem \ref{thm:stat}]
        
        The proof of Theorem \ref{thm:stat} consists of two steps. In the first step, we show that the RSC, RSS and deviation bound conditions defined in the deterministic computational convergence analysis hold with high probability, and proofs of these conditions are presented in Appendix \ref{sec:B.1}. Given these regularity conditions, it suffices to show the statistical properties of the initial values, which will be discussed in Appendix \ref{sec:B.2}.
        
        By Lemmas \ref{lemma:RSC} and \ref{lemma:deviation}, with probability at least $1-2\exp[-CM_2^2\min(\tau^{-2},\tau^{-4})T]-\exp(-Cp)$, the empirical loss function $\mathcal{L}(\cdot)$ satisfies the RSC-$\alpha_\text{RSC}$ and RSS-$\beta_\text{RSS}$ conditions, and
        \begin{equation}
            \xi(r,d)\lesssim\tau^2M_1\sqrt{\frac{d_\textup{CS}(p,r,d)}{T}}.
        \end{equation}
        
        By Theorem \ref{thm:gd}, we have that, for all $i=1,2,\dots$,
        \begin{equation}
            \begin{split}
                & \|\bm{A}^{(i)}-\bm{A}^*\|_\text{F}^2 \\
                \lesssim & \kappa^2(1-C\eta_0\alpha\beta^{-1}\kappa^{-2})^i\|\bm{A}^{(0)}-\bm{A}^*\|_\text{F}^2 + \kappa^2\alpha^{-1}\beta^{-1}\xi^2(r,d).
            \end{split}
        \end{equation}
        Hence, when
        \begin{equation}
            I\gtrsim\frac{\log(\kappa^2\alpha^{-1}\beta^{-1}\xi^2(r,d))-\log(\|\bm{A}^{(0)}-\bm{A}^*\|_\text{F}^2)}{\log(1-C\eta_0\alpha\beta^{-1}\kappa^{-2})},
        \end{equation}
        the statistical error will absorb the optimization error, so
        \begin{equation}
            \|\bm{A}^{(I)}-\bm{A}^*\|_\text{F}^2\lesssim \kappa^2\alpha^{-1}\beta^{-1}\xi^2(r,d).
        \end{equation}
    
        Moreover, by Lemma \ref{lemma:A0_init}, with probability at least $1-2\exp[-CM_2^2\min(\tau^{-2},\tau^{-4})T]-\exp(-Cp)$,
        \begin{equation}
            \|\bm{A}^{(0)}-\bm{A}^*\|_\text{F}\lesssim \sigma_1^{2/3}\kappa^2 g_{\min}^{-2}\alpha^{-1}\tau^2M_1\sqrt{d_\text{RR}(p,r)/T}.
        \end{equation} 
        Combining these results, we have that when
        \begin{equation}
            I \gtrsim \frac{\log(\kappa^{-2}\sigma_1^{-4/3}g_{\min}^2\alpha_\text{RSC}\beta_\text{RSS}^{-1}[d_\text{CS}(p,r,d)/d_\text{RR}(p,r)])}{\log(1-C\eta_0\alpha_\text{RSC}\beta_\text{RSS}^{-1}\kappa^{-2})},
        \end{equation}
        with probability at least $1-4\exp[-CM_2^2\min(\tau^{-2},\tau^{-4})T]-2\exp(-Cp)$,
        \begin{equation}
            \|\bm{A}^{(I)}-\bm{A}^*\|_\text{F}\lesssim \kappa\alpha_\text{RSC}^{-1}\tau^2M_1\sqrt{\frac{d_\text{CS}(p,r,d)}{T}}.
        \end{equation}
    
\end{proof}

\subsection{Proofs of RSC, RSS and deviation bound}\label{sec:B.1}

    We first prove the restricted strong convexity (RSC) and restricted strong smoothness (RSS) conditions. For the least squares loss function $\mathcal{L}(\bm{A})={(2T)}^{-1}\|\bm{Y}-\bm{A}\bm{X}\|_\text{F}^2$, it is easy to check that for any $\bm{A},\bm{A}'\in\mathbb{R}^{p\times p}$,
    \begin{equation}
        \mathcal{L}(\bm{A})-\mathcal{L}(\bm{A}')-\langle\nabla\mathcal{L}(\bm{A}'),\bm{A}-\bm{A}'\rangle = \frac{1}{2T}\|(\bm{A}-\bm{A}')\bm{X}\|_\text{F}^2=\frac{1}{2T}\sum_{t=0}^{T-1}\|(\bm{A}-\bm{A}')\bm{y}_t\|_2^2.
    \end{equation}
    
    \begin{lemma}\label{lemma:RSC}
        Assume the conditions in Theorem \ref{thm:stat} hold. Suppose that $T\gtrsim M_2^{-2}\max(\tau^4,\tau^2) p$. For any rank-$2r$ matrix $\bm{\Delta}\in\mathbb{R}^{p\times p}$, with probability at least $1-2\exp[-CM_2^2\min(\tau^{-4},\tau^{-2})T]$,
        \begin{equation}
            \alpha_\textup{RSC}\|\bm{\Delta}\|_\textup{F}^2 \leq \frac{1}{T}\sum_{t=0}^{T-1}\|\bm{\Delta}\bm{y}_t\|_2^2\leq \beta_\textup{RSS}\|\bm{\Delta}\|_\textup{F}^2,
        \end{equation}
        where $\alpha_\textup{RSC}=\lambda_{\min}(\bm{\Sigma}_{\bbm{\varepsilon}})/(2\mu_{\max}(\mathcal{A}))$ and $\beta_\textup{RSS}=(3\lambda_{\max}(\bm{\Sigma}_{\bbm{\varepsilon}}))/(2\mu_{\min}(\mathcal{A}))$.
    \end{lemma}

    \begin{proof}[Proof of Lemma \ref{lemma:RSC}]
        
        For any $\bm{M}\in\mathbb{R}^{m\times p}$, denote $R_T(\bm{M})=\sum_{t=0}^{T-1}\|\bm{M}\bm{y}_t\|_2^2$. Note that $R_T(\bm{\Delta})\geq\mathbb{E}R_T(\bm{\Delta})-\sup_{\bm{\Delta}}|R_T(\bm{\Delta})-\mathbb{E}R_T(\bm{\Delta})|$. 
        
        Based on the moving average representation of VAR(1), we can rewrite $\bm{y}_t$ as a VMA($\infty$) process,
        \begin{equation}
            \bm{y}_t=\bbm{\varepsilon}_t+\bm{A}\bbm{\varepsilon}_{t-1}+\bm{A}^2\bbm{\varepsilon}_{t-2}+\bm{A}^3\bbm{\varepsilon}_{t-3}+\cdots
        \end{equation}
        Let $\bm{z}=(\bm{y}_{T-1}^\top,\bm{y}_{T-2}^\top,\dots,\bm{y}_0^\top)^\top$, $\bbm{\varepsilon}=(\bbm{\varepsilon}_{T-1}^\top,\bbm{\varepsilon}_{T-2}^\top,\dots,\bbm{\varepsilon}_0,\dots)^\top$, and $\bbm{\zeta}=(\bbm{\zeta}_{T-1}^\top,\bbm{\zeta}_{T-2}^\top,\dots,\bbm{\zeta}_0,\dots)^\top$.
        Note that $\bm{z}=\widetilde{\bm{A}}\bbm{\varepsilon}$, where $\widetilde{\bm{A}}$ is defined as
        \begin{equation}
            \widetilde{\bm{A}} = \begin{bmatrix}
                \bm{I}_p & \bm{A} & \bm{A}^2 & \bm{A}^3 & \dots \\
                \bm{O} & \bm{I}_p & \bm{A} & \bm{A}^2 & \dots\\
                \vdots & \vdots & \vdots & \vdots & \ddots\\
            \end{bmatrix}.
        \end{equation}
        Then, we have
        \begin{equation}
            \begin{split}
                &R_T(\bm{\Delta})=\bm{z}^\top(\bm{I}_T\otimes\bm{\Delta}^\top\bm{\Delta})\bm{z}=\bbm{\varepsilon}^\top\widetilde{\bm{A}}^\top(\bm{I}_T\otimes\bm{\Delta}^\top\bm{\Delta})\widetilde{\bm{A}}\bbm{\varepsilon}\\
                =&\bbm{\zeta}^\top\widetilde{\bm{\Sigma}}_{\bbm{\varepsilon}}\widetilde{\bm{A}}^\top(\bm{I}_T\otimes \bm{\Delta}^\top\bm{\Delta})\widetilde{\bm{A}}\widetilde{\bm{\Sigma}}_{\bbm{\varepsilon}}\bbm{\zeta}:=\bbm{\zeta}^\top\bm{\Sigma_{\bm{\Delta}}}\bbm{\zeta},
            \end{split}
        \end{equation}
        where
        \begin{equation}
            \widetilde{\bm{\Sigma}}_{\bbm{\varepsilon}}=\begin{bmatrix}
                \bm{\Sigma}_{\bbm{\varepsilon}}^{1/2} & \bm{O} & \bm{O} & \cdots\\
                \bm{O} & \bm{\Sigma}_{\bbm{\varepsilon}}^{1/2} & \bm{O} & \cdots\\
                \bm{O} & \bm{O} & \bm{\Sigma}_{\bbm{\varepsilon}}^{1/2} & \cdots\\
                \vdots & \vdots & \vdots & \ddots
            \end{bmatrix}.
        \end{equation}
        Thus, $\mathbb{E}R_T(\bm{\Delta})=\|(\bm{I}_T\otimes\bm{\Delta})\widetilde{\bm{A}}\widetilde{\bm{\Sigma}}_{\bbm{\varepsilon}}\|_\text{F}^2 \geq T\|\bm{\Delta}\|_\text{F}^2\lambda_{\min}(\bm{\Sigma}_{\bbm{\varepsilon}})\lambda_{\min}(\widetilde{\bm{A}}\widetilde{\bm{A}}^\top)$.
        
        As $\|\bm{\Delta}\|_\text{F}=1$, by the sub-multiplicative property of the Frobenius norm and operator norm, we have
        \begin{equation}
            \|\bm{\Sigma}_{\bm{\Delta}}\|_\text{F}^2\leq T\lambda_{\max}^2(\bm{\Sigma}_{\bbm{\varepsilon}})\lambda_{\max}^2(\widetilde{\bm{A}}\widetilde{\bm{A}}^\top)
        \end{equation}
        and
        \begin{equation}
            \|\bm{\Sigma}_{\bm{\Delta}}\|_\text{op}\leq \lambda_{\max}(\bm{\Sigma}_{\bbm{\varepsilon}})\lambda_{\max}(\widetilde{\bm{A}}\widetilde{\bm{A}}^\top).
        \end{equation}
        For any $\bm{v}\in\mathbb{S}^{p-1}$ and any $t>0$, by Hanson-Wright inequality,
        \begin{equation}
            \begin{split}
                & \mathbb{P}[|R_T(\bm{v}^\top)-\mathbb{E}R_T(\bm{v}^\top)|\geq t]\\
                \leq & 2\exp\left(-\min\left(\frac{t^2}{\tau^4 T\lambda^2_{\max}(\bm{\Sigma}_{\bbm{\varepsilon}})\lambda^2_{\max}(\widetilde{\bm{A}}\widetilde{\bm{A}}^\top)},\frac{t}{\tau^2\lambda_{\max}^2(\bm{\Sigma}_{\bbm{\varepsilon}})\lambda^2_{\max}(\widetilde{\bm{A}}\widetilde{\bm{A}}^\top)}\right)\right).
            \end{split}
        \end{equation}
        Considering an $\epsilon$-covering net of $\mathbb{S}^{p-1}$, by Lemma \ref{lemma:covering}, we can easily construct the union bound for $T\gtrsim p$,
        \begin{equation}
            \begin{split}
                & \mathbb{P}\left[\sup_{\bm{v}\in\mathbb{S}^{p-1}}|R_T(\bm{v}^\top)-\mathbb{E}R_n(\bm{v}^\top)|\geq t\right]\\
                \leq & C\exp\left(p-\min\left(\frac{t^2}{\tau^4 T\lambda^2_{\max}(\bm{\Sigma}_{\bbm{\varepsilon}})\lambda^2_{\max}(\widetilde{\bm{A}}\widetilde{\bm{A}}^\top)},\frac{t}{\tau^2\lambda_{\max}^2(\bm{\Sigma}_{\bbm{\varepsilon}})\lambda^2_{\max}(\widetilde{\bm{A}}\widetilde{\bm{A}}^\top)}\right)\right).
            \end{split}
        \end{equation}
        Letting $t=T\lambda_{\min}(\bm{\Sigma}_{\bbm{\varepsilon}})\lambda_{\min}(\widetilde{\bm{A}}\widetilde{\bm{A}}^\top)/2$, for $T\gtrsim M_2^{-2}\max(\tau^4,\tau^2)p$, we have
        \begin{equation}\label{eq:deviation1}
            \begin{split}
                & \mathbb{P}\left[\sup_{\bm{v}\in\mathbb{S}^{p-1}}|R_n(\bm{v}^\top)-\mathbb{E}R_n(\bm{v}^\top)|\geq n\lambda_{\min}(\bm{\Sigma}_{\bbm{\varepsilon}})\lambda_{\min}(\widetilde{\bm{A}}\widetilde{\bm{A}}^\top)/2\right]\\
                \leq & 2\exp\left[-CM_2^2\min(\tau^{-4},\tau^{-2})T\right],
            \end{split}
        \end{equation}
        where $M_2=[\lambda_{\min}(\bm{\Sigma}_{\bbm{\varepsilon}})\lambda_{\min}(\widetilde{\bm{A}}\widetilde{\bm{A}}^\top)]/[\lambda_{\max}(\bm{\Sigma}_{\bbm{\varepsilon}})\lambda_{\max}(\widetilde{\bm{A}}\widetilde{\bm{A}}^\top)]$.
        
        Therefore, with probability at least $1-2\exp[-CM_2^2\min(\tau^{-4},\tau^{-2})T]$,
        \begin{equation}
            R_T(\bm{\Delta})\geq\frac{T}{2}\lambda_{\min}(\bm{\Sigma}_{\bbm{\varepsilon}})\lambda_{\min}(\widetilde{\bm{A}}\widetilde{\bm{A}}^\top)\|\bm{\Delta}\|_\text{F}^2.
        \end{equation}
    
        Similarly, $R_T(\bm{\Delta})\leq\mathbb{E}R_T(\bm{\Delta})+\sup_{\bm{\Delta}}|R_T(\bm{\Delta})-\mathbb{E}R_T(\bm{\Delta})|$ and $\mathbb{E}R_T(\bm{\Delta})\leq T\|\bm{\Delta}\|_\text{F}^2\cdot\lambda_{\max}(\bm{\Sigma}_{\bbm{\varepsilon}})\lambda_{\max}(\widetilde{\bm{A}}\widetilde{\bm{A}}^\top)$. Additionally, the upper bound in \eqref{eq:deviation1} can easily be expanded to $T\lambda_{\max}(\bm{\Sigma}_{\bbm{\varepsilon}})\lambda_{\max}(\widetilde{\bm{A}}\widetilde{\bm{A}}^\top)/2$, and the upper bound of $R_T(\bm{\Delta})$ follows.

        Finally, since $\widetilde{\bm{A}}$ is related to the VMA($\infty$) process, by the spectral measure of ARMA process discussed in \citet{basu2015regularized}, we may replace $\lambda_{\max}(\widetilde{\bm{A}}\widetilde{\bm{A}}^\top)$ and $\lambda_{\min}(\widetilde{\bm{A}}\widetilde{\bm{A}}^\top)$ with $1/\mu_{\min}(\mathcal{A})$ and $1/\mu_{\max}(\mathcal{A})$, respectively.

    \end{proof}

    We next prove the deviation bound for $\xi(r,d)$. For the least squares loss function $\mathcal{L}(\bm{A})=(2T)^{-1}\|\bm{Y}-\bm{A}\bm{X}\|_\text{F}^2$, it is clear that
    \begin{equation}
        \nabla\mathcal{L}(\bm{A}^*)=\frac{1}{T}\sum_{t=1}^T\bbm{\varepsilon}_t\bm{y}_{t-1}^\top.
    \end{equation}
    
    \begin{lemma}\label{lemma:deviation}
        Assume conditions in Theorem \ref{thm:stat} hold. If $T\gtrsim M_2^{-2}\max(\tau^4,\tau^2)p$, then, with probability at least $1-\exp(-Cp)$,
        \begin{equation}
            \xi(r,d):=\sup_{\substack{\textup{\bf{[}}\bm{C}~\bm{R}\textup{\bf{]}},\textup{\bf{[}}\bm{C}~\bm{P}\textup{\bf{]}}\in\mathbb{O}^{p\times r}\\ \bm{D}\in\mathbb{R}^{r\times r}}}\left\langle \nabla\mathcal{L}(\bm{A}^*),\lb\bm{C}~\bm{R}\rb\bm{D}\lb\bm{C}~\bm{P}\rb^\top\right\rangle\lesssim\tau^2M_1\sqrt{\frac{d_\textup{CS}(p,r,d)}{T}}
        \end{equation}
        where $M_1=\lambda_{\max}(\bm{\Sigma}_{\bbm{\varepsilon}})/\mu_{\min}^{1/2}(\mathcal{A})$.
    \end{lemma}

    \begin{proof}[Proof of Lemma \ref{lemma:deviation}]
        
        Denote $\mathcal{W}(r,d;p)=\{\bm{W}\in\mathbb{R}^{p\times p}:\bm{W}=\lb\bm{C}~\bm{R}\rb\bm{D}\lb\bm{C}~\bm{P}\rb^\top,~\bm{C}\in\mathbb{R}^{p\times d},~\bm{R},\bm{P}\in\mathbb{R}^{p\times(r-d)},~\text{and}~\|\bm{W}\|_\text{F}=1\}$. By definition,
        \begin{equation}\label{eq:stat_error_xi}
            \xi(r,d) = \sup_{\bm{W}\in\mathcal{W}(r,d;p)}\left\langle\frac{1}{T}\sum_{t=1}^T\bbm{\varepsilon}_t\bm{y}_{t-1}^\top,\bm{W}\right\rangle.
        \end{equation}
    
        First, we consider an $\epsilon$-net $\overline{\mathcal{W}}(r,d;p)$ for $\mathcal{W}(r,d;p)$. For any matrix $\bm{W}\in\mathcal{W}(r,d;p)$, there exists a matrix $\overline{\bm{W}}\in\overline{\mathcal{W}}(r,d;p)$ such that $\|\bm{W}-\overline{\bm{W}}\|_\text{F}\leq \epsilon$. Obviously, $\bm{\Delta}=\bm{W}-\overline{\bm{W}}$ is a rank-$2r$ matrix with common dimension $2d$. Based on the SVD of $\bm{\Delta}$, we can split the first $2r$ pairs of left and right singular vectors into two equal-size groups such that the dimension of left and right singular vectors is $r-d$ in each group. By the splitting of SVD, we can write $\bm{\Delta}=\bm{\Delta}_1+\bm{\Delta}_2$, where both $\bm{\Delta}_1$ and $\bm{\Delta}_2$ are rank-$r$ matrix with common dimension $d$ and $\langle\bm{\Delta}_1,\bm{\Delta}_2\rangle=0$.
        
        By Cauchy's inequality, as $\|\bm{\Delta}\|_\text{F}^2=\|\bm{\Delta}_1\|_\text{F}^2+\|\bm{\Delta}_2\|_\text{F}^2$, we have $\|\bm{\Delta}_1\|_\text{F}+\|\bm{\Delta}_2\|_\text{F}\leq \sqrt{2}\|\bm{\Delta}\|_\text{F}\leq \sqrt{2}\epsilon$. Moreover, since $\bm{\Delta}_i/\|\bm{\Delta}_i\|_\text{F}\in\mathcal{W}(r,d;p)$,
        \begin{equation}
            \begin{split}
                \xi(r,d) & \leq \max_{\bm{W}\in\overline{\mathcal{W}}(r,d;p)}\left\langle\frac{1}{T}\sum_{t=1}^T\bbm{\varepsilon}_t\bm{y}_{t-1}^\top,\bm{W}\right\rangle + \sum_{i=1}^2\left\langle\frac{1}{T}\sum_{t=1}^T\bbm{\varepsilon}_t\bm{y}_{t-1}^\top,\bm{\Delta}_i/\|\bm{\Delta}\|_\text{F}\right\rangle\|\bm{\Delta}_i\|_\text{F}\\
                & \leq \max_{\bm{W}\in\overline{\mathcal{W}}(r,d;p)}\left\langle\frac{1}{T}\sum_{t=1}^T\bbm{\varepsilon}_t\bm{y}_{t-1}^\top,\bm{W}\right\rangle + \sqrt{2}\epsilon\xi(r,d),
            \end{split}
        \end{equation}
        which implies that
        \begin{equation}
            \xi(r,d) \leq (1-\sqrt{2}\epsilon)^{-1}\max_{\bm{W}\in\overline{\mathcal{W}}(r,d;p)}\left\langle\frac{1}{T}\sum_{t=1}^T\bbm{\varepsilon}_t\bm{y}_{t-1}^\top,\bm{W}\right\rangle.
        \end{equation}
    
        Next, for any fixed $\bm{W}\in\mathbb{R}^{p\times p}$ such that $\|\bm{W}\|_\text{F}=1$, $\langle\bbm{\varepsilon}_t\bm{y}_{t-1}^\top,\bm{W}\rangle=\langle\bbm{\varepsilon}_t,\bm{W}\bm{y}_{t-1}\rangle$, and we denote $S_t(\bm{W})=\sum_{s=1}^t\langle\bbm{\varepsilon}_s,\bm{W}\bm{y}_{s-1}\rangle$ and $R_t(\bm{W})=\sum_{s=0}^{t-1}\|\bm{W}\bm{y}_s\|_2^2$, for $1\leq t\leq T$. Similar to \citet{wang2021high}, by the standard Chernoff bound, for any $z_1>0$ and $z_2>0$,
        \begin{equation}
            \mathbb{P}[\{S_T(\bm{W})\geq z_1\}\cap\{R_T(\bm{W})\leq z_2\}]\leq \exp\left(-\frac{z_1^2}{2\tau^2\lambda_{\max}(\bm{\Sigma}_{\bbm{\varepsilon}}z_2)}\right).
        \end{equation}
        Similar to Lemma \ref{lemma:RSC}, with probability at least $1-2\exp[-CM_2^2\min(\tau^{-4},\tau^{-2})T]$,
        \begin{equation}
            R_T(\bm{W})\geq \frac{T}{2}\lambda_{\min}(\bm{\Sigma}_{\bbm{\varepsilon}})\lambda_{\min}(\widetilde{\bm{A}}\widetilde{\bm{A}}^\top).
        \end{equation}
        Therefore, for any $x>0$,
        \begin{equation}
            \begin{split}
                & \mathbb{P}\left[\sup_{\bm{W}\in\mathcal{W}(1;p)}\left\langle\frac{1}{n}\sum_{t=1}^n\bbm{\varepsilon}_t\bm{y}_{t-1}^\top,\bm{W}\right\rangle\geq x\right]\\
                \leq & \mathbb{P}\left[\max_{\overline{\bm{W}}\in\overline{\mathcal{W}}(1;p)}\left\langle\frac{1}{n}\sum_{t=1}^n\bbm{\varepsilon}_t\bm{y}_{t-1}^\top,\bm{W}\right\rangle\geq (1-\sqrt{2}\epsilon)x\right]\\
                \leq & |\overline{\mathcal{W}}(r,d;p)|\cdot\mathbb{P}\left[\left\langle\frac{1}{n}\sum_{t=1}^n\bbm{\varepsilon}_t\bm{y}_{t-1}^\top,\bm{W}\right\rangle\geq (1-\sqrt{2}\epsilon)x\right]\\
                \leq & |\overline{\mathcal{W}}(r,d;p)|\cdot\Big{\{}\mathbb{P}[\{S_T(\bm{W})\geq T(1-\sqrt{2}\epsilon)x\}\cap\{R_T(\bm{W})\leq C\tau^2T\lambda_{\max}(\bm{\Sigma}_{\bbm{\varepsilon}})\lambda_{\max}(\widetilde{\bm{A}}\widetilde{\bm{A}}^\top)\}] \\
                + & \mathbb{P}[R_T(\bm{W}) >  C\tau^2T\lambda_{\max}(\bm{\Sigma}_{\bbm{\varepsilon}})\lambda_{\max}(\widetilde{\bm{A}}\widetilde{\bm{A}}^\top)]\Big{\}}\\
                \leq & |\overline{\mathcal{W}}(r,d;p)|\cdot\Bigg\{\exp\left[-\frac{CTx^2}{\tau^4\lambda_{\max}^2(\bm{\Sigma}_{\bbm{\varepsilon}})\lambda_{\max}(\widetilde{\bm{A}}\widetilde{\bm{A}}^\top)}\right] + 2\exp[-CM_2^2\min(\tau^{-4},\tau^{-2})T]\Bigg\}.
            \end{split}
        \end{equation}
        
        By Lemma \ref{lemma:cs_covering}, $|\overline{\mathcal{W}}(r,d;p)|\leq (24/\epsilon)^{p(2r-d)+r^2}$. Thus, if we take $\epsilon=0.1$ and $x=C\tau^2\lambda_{\max}(\bm{\Sigma}_{\bbm{\varepsilon}})\lambda_{\max}(\widetilde{\bm{A}}\widetilde{\bm{A}}^\top)\sqrt{[p(2r-d)+r^2]/T}$, when $T\gtrsim M_2^{-2}\max(\tau^4,\tau^2)p$, we have
        \begin{equation}
            \mathbb{P}\left[\sup_{\bm{W}\in\mathcal{W}(r,d;p)}\left\langle\frac{1}{T}\sum_{t=1}^T\bbm{\varepsilon}_t\bm{y}_{t-1}^\top,\bm{W}\right\rangle\gtrsim\tau^2\lambda_{\max}(\bm{\Sigma}_{\bbm{\varepsilon}})\lambda_{\max}^{1/2}(\widetilde{\bm{A}}\widetilde{\bm{A}}^\top)\sqrt{\frac{d_\textup{CS}(p,r,d)}{T}}\right]\leq \exp(-Cp).
        \end{equation}
        Finally, by the spectral measure of ARMA processes, we can replace $\lambda_{\max}(\widetilde{\bm{A}}\widetilde{\bm{A}}^\top)$ with $1/\mu_{\min}(\mathcal{A})$.
        
    \end{proof}

\subsection{Properties of initial value}\label{sec:B.2}

We present some statistical properties of the initial value of the gradient descent algorithm, the reduced-rank estimator and the corresponding $\bm{C}$, $\bm{R}$, $\bm{P}$, and $\bm{D}$. Consider the reduced-rank estimator
\begin{equation}
    \widetilde{\bm{A}}_\textup{RR}(r)=\argmin_{\textup{rank}(\bm{A})\leq r}\frac{1}{2T}\|\bm{Y}-\bm{A}\bm{X}\|_\text{F}^2=\widehat{\bm{H}}\widehat{\bm{H}}^\top\bm{Y}\bm{X}^{-1}(\bm{X}\bm{X}^\top)^{-1}
\end{equation}
where $\widehat{\bm{H}}\in\mathbb{O}^{p\times r}$ contains the leading eigenvectors of $\bm{Y}\bm{X}^\top(\bm{X}\bm{X}^\top)^{-1}\bm{X}\bm{Y}^\top$.

\begin{lemma}\label{lemma:rr_initial}
    Under Assumptions \ref{asmp:1}--\ref{asmp:3}, with probability at least $1-2\exp[-CM_2^2\min(\tau^{-2},\tau^{-4})T]-\exp(-Cp)$,
    \begin{equation}
        \|\widetilde{\bm{A}}_\textup{RR}(r)-\bm{A}^*\|_\textup{F}\lesssim\alpha_{\textup{RSC}}^{-1}\tau^2M_1\sqrt{\frac{d_\textup{RR}(p,r)}{T}}.
    \end{equation}
\end{lemma}

\begin{proof}
    
    Denote $\bm{\Delta}=\widetilde{\bm{A}}_\textup{RR}(r)-\bm{A}^*$, then by the optimality of the reduced-rank estimator
    \begin{equation}
        \begin{split}
            & \frac{1}{2T}\sum_{t=1}^T\|\bm{y}_t-\widetilde{\bm{A}}_\text{RR}\bm{y}_{t-1}\|_2^2 \leq \frac{1}{2T}\sum_{t=1}^T\|\bm{y}_t-\bm{A}^*\bm{y}_{t-1}\|_2^2\\
            \Rightarrow & \frac{1}{2T}\sum_{t=1}^T\|\bm{\Delta}\bm{y}_{t-1}\|_2^2\leq\frac{1}{T}\sum_{t=1}^T\langle\bbm{\varepsilon}_t,\bm{\Delta}\bm{y}_{t-1}\rangle\\
            \Rightarrow & \frac{1}{2T}\sum_{t=1}^T\|\bm{\Delta}\bm{y}_{t-1}\|_2^2\leq\frac{1}{T}\sum_{t=1}^T\langle\bbm{\varepsilon}_t\bm{y}_{t-1}^\top,\bm{\Delta}\rangle.
        \end{split}
    \end{equation} 
    Since the rank of both $\widetilde{\bm{A}}_\textup{RR}$ and $\bm{A}^*$ is $r$, $\bm{\Delta}$ is at most rank $2r$. Denote the set of matrices $\mathcal{W}(r;p)=\{\bm{W}\in\mathbb{R}^{p\times p}:\|\bm{W}\|_\text{F}=1,~\text{rank}(\bm{W})\leq r\}$. Then, we have
    \begin{equation}
        \frac{1}{2T}\sum_{t=1}^T\|\bm{\Delta}\bm{y}_{t-1}\|_2^2 \leq \|\bm{\Delta}\|_\text{F}\sup_{\bm{W}\in\mathcal{W}(r;p)}\left\langle\frac{1}{T}\sum_{t=1}^T\bbm{\varepsilon}_t\bm{y}_{t-1}^\top,\bm{W}\right\rangle.
    \end{equation}

    By Lemma \ref{lemma:RSC}, with probability at least $1-2\exp[-CM_2^2\min(\tau^{-4},\tau^{-2})T]$,
    \begin{equation}
        \frac{1}{2T}\sum_{t=1}^T\|\bm{\Delta}\bm{y}_{t-1}\|_2^2\geq \frac{\alpha_\textup{RSC}}{2}\|\bm{\Delta}\|_\text{F}^2,
    \end{equation}
    and thus,
    \begin{equation}
        \|\bm{\Delta}\|_\text{F}\leq \frac{2}{\alpha_{\textup{RSC}}}\sup_{\bm{W}\in\mathcal{W}(r;p)}\left\langle\frac{1}{T}\sum_{t=1}^T\bbm{\varepsilon}_t\bm{y}_{t-1}^\top,\bm{W}\right\rangle.
    \end{equation}

    By Lemma \ref{lemma:deviation}, with probability at least $1-\exp(-Cp)$,
    \begin{equation}
        \xi(r,0)=\sup_{\bm{W}\in\mathcal{W}(r;p)}\left\langle\frac{1}{T}\sum_{t=1}^T\bbm{\varepsilon}_t\bm{y}_{t-1}^\top,\bm{W}\right\rangle\lesssim\tau^2M_1\sqrt{\frac{d_\textup{RR}(p,r)}{T}}.
    \end{equation}
    
    Combining these results, we have that with probability at least $1-2\exp[-CM_2^2\min(\tau^{-4},\tau^{-2})T]-\exp(-Cp)$,
    \begin{equation}
        \|\bm{\Delta}\|_\text{F} \lesssim \alpha_{\textup{RSC}}^{-1}\tau^2M_1\sqrt{\frac{d_\textup{RR}(p,r)}{T}}.
    \end{equation}
    
\end{proof}

Next, we derive the estimation error rate for the resulting initial estimator $\bm{C}$, $\bm{R}$, $\bm{P}$, and $\bm{D}$.

\begin{lemma}\label{lemma:A0_init}
    Suppose that Assumptions \ref{asmp:1}--\ref{asmp:3} hold. If 
    $T\gtrsim\max(\tau^4,\tau^2)M_2^{-2}p$, then 
    \begin{equation}
        \|\bm{A}^{(0)}-\bm{A}^*\|_\textup{F}\lesssim\sigma_1^{2/3}\kappa^2g_{\min}^{-2}\alpha_\textup{RSC}^{-1}\tau^2M_1\sqrt{\frac{d_\textup{RR}(p,r)}{T}}
    \end{equation}
    with probability at least $1-2\exp[-CM_2^2\min(\tau^{-2},\tau^{-4})T]-\exp(-Cp)$.
\end{lemma}

\begin{proof}
    
    Throughout this proof, we assume that the true values $\bm{C}^*$, $\bm{R}^*$ and $\bm{P}^*$ satisfy that $\lb\bm{C}^*~\bm{R}^*\rb^\top\lb\bm{C}^*~\bm{R}^*\rb=\bm{I}_r$ and $\lb\bm{C}^*~\bm{P}^*\rb^\top\lb\bm{C}^*~\bm{P}^*\rb=\bm{I}_r$.
    
    We begin with the rate of $\widetilde{\bm{R}}$ and $\widetilde{\bm{P}}$. Based on Lemma \ref{lemma:rr_initial}, we have that $\|\widetilde{\bm{A}}_\text{RR}-\bm{A}^*\|_\text{F}=o_p(1)$. By Lemma \ref{lemma:svd_perturb}, we have that
    \begin{equation}
        \|\sin\Theta(\widetilde{\bm{U}},\bm{U}^*)\|_\text{F}\lesssim\frac{\sigma_1\|\widetilde{\bm{A}}_\text{RR}-\bm{A}^*\|_\text{F}}{\sigma_r^2}\asymp\kappa\sigma_r^{-1}\|\widetilde{\bm{A}}_\text{RR}-\bm{A}^*\|_\text{F}
    \end{equation}
    and
    \begin{equation}
        \|\sin\Theta(\widetilde{\bm{V}},\bm{V}^*)\|_\text{F}\lesssim\frac{\sigma_1\|\widetilde{\bm{A}}_\text{RR}-\bm{A}^*\|_\text{F}}{\sigma_r^2}\asymp\kappa\sigma_r^{-1}\|\widetilde{\bm{A}}_\text{RR}-\bm{A}^*\|_\text{F}.
    \end{equation}

    By triangle inequality,
    \begin{equation}
        \begin{split}
            &\|\widetilde{\bm{U}}\widetilde{\bm{U}}^\top(\bm{I}_p-\widetilde{\bm{V}}\widetilde{\bm{V}}^\top)-\widetilde{\bm{U}}^*\widetilde{\bm{U}}^{*\top}(\bm{I}_p-\bm{V}^*\bm{V}^{*\top})\|_\text{F}\\
            \leq & \|(\widetilde{\bm{U}}\widetilde{\bm{U}}^\top-\bm{U}^*\bm{U}^{*\top})(\bm{I}_p-\widetilde{\bm{V}}\widetilde{\bm{V}}^\top)\|_\text{F}+\|\widetilde{\bm{U}}^*\widetilde{\bm{U}}^{*\top}(\bm{V}^*\bm{V}^{*\top}-\widetilde{\bm{V}}\widetilde{\bm{V}}^\top)\|_\text{F}\\
            \leq & \|\bm{I}_p-\widetilde{\bm{V}}\widetilde{\bm{V}}^\top\|_\text{op}\|\widetilde{\bm{U}}\widetilde{\bm{U}}^\top-\bm{U}^*\bm{U}^{*\top}\|_\text{F}+\|\widetilde{\bm{U}}^*\widetilde{\bm{U}}^{*\top}\|_\text{op}\|\bm{V}^*\bm{V}^{*\top}-\widetilde{\bm{V}}\widetilde{\bm{V}}^\top\|_\text{F}\\
            \leq & \sqrt{2}(\|\sin\Theta(\widetilde{\bm{U}},\bm{U}^*)\|_\text{F}+\|\sin\Theta(\widetilde{\bm{V}},\bm{V}^*)\|_\text{F})\\
            \asymp & \kappa\sigma_r^{-1}\|\widetilde{\bm{A}}_\text{RR}-\bm{A}^*\|_\text{F}.
        \end{split}
    \end{equation}
    
    By Assumption \ref{asmp:3}, $\sigma_{r-d}(\widetilde{\bm{U}}^*\widetilde{\bm{U}}^{*\top}(\bm{I}_p-\bm{V}^*\bm{V}^{*\top}))=\sigma_{r-d}(\widetilde{\bm{U}}^{*\top}\bm{V}^*_\perp)\geq g_{\min}$. By Lemma \ref{lemma:svd_perturb},
    \begin{equation}
        \begin{split}
            \|\sin\Theta(\widetilde{\bm{R}},\bm{R}^*)\|_\text{F}&\lesssim\frac{\|\widetilde{\bm{U}}\widetilde{\bm{U}}^\top(\bm{I}_p-\widetilde{\bm{V}}\widetilde{\bm{V}}^\top)-\widetilde{\bm{U}}^*\widetilde{\bm{U}}^{*\top}(\bm{I}_p-\bm{V}^*\bm{V}^{*\top})\|_\text{F}}{g_{\min}^2}\\
            &\asymp\kappa\sigma_r^{-1}g_{\min}^{-2}\|\widetilde{\bm{A}}_\text{RR}-\bm{A}^*\|_\text{F}.
        \end{split}
    \end{equation}
    Similarly, we also have
    \begin{equation}
        \|\sin\Theta(\widetilde{\bm{P}},\bm{P}^*)\|_\text{F}\lesssim\kappa\sigma_r^{-1}g_{\min}^{-2}\|\widetilde{\bm{A}}_\text{RR}-\bm{A}^*\|_\text{F}.
    \end{equation}
    
    For $\widetilde{\bm{C}}$, by triangle inequality,
    \begin{equation}
        \begin{split}
            &\|(\bm{I}_p-\bm{R}^*\bm{R}^{*\top})(\bm{I}_p-\bm{P}^*\bm{P}^{*\top})(\bm{U}^*\bm{U}^{*\top}+\bm{V}^*\bm{V}^{*\top})(\bm{I}_p-\bm{R}^*\bm{R}^{*\top})(\bm{I}_p-\bm{P}^*\bm{P}^{*\top})\\
            &-(\bm{I}_p-\widetilde{\bm{R}}\widetilde{\bm{R}}^\top)(\bm{I}_p-\widetilde{\bm{P}}\widetilde{\bm{P}}^\top)(\widetilde{\bm{U}}\widetilde{\bm{U}}^\top+\widetilde{\bm{V}}\widetilde{\bm{V}}^\top)(\bm{I}_p-\widetilde{\bm{R}}\widetilde{\bm{R}}^\top)(\bm{I}_p-\widetilde{\bm{P}}\widetilde{\bm{P}}^\top)\|_\text{F}\\
            \leq & \|(\widetilde{\bm{R}}\widetilde{\bm{R}}^\top-\bm{R}^*\bm{R}^{*\top})(\bm{I}_p-\bm{P}^*\bm{P}^{*\top})(\bm{U}^*\bm{U}^{*\top}+\bm{V}^*\bm{V}^{*\top})(\bm{I}_p-\bm{R}^*\bm{R}^{*\top})(\bm{I}_p-\bm{P}^*\bm{P}^{*\top})\|_\text{F}\\
            + & \|(\bm{I}_p-\widetilde{\bm{R}}\widetilde{\bm{R}}^{\top})(\widetilde{\bm{P}}\widetilde{\bm{P}}^\top-\bm{P}^*\bm{P}^{*\top})(\bm{U}^*\bm{U}^{*\top}+\bm{V}^*\bm{V}^{*\top})(\bm{I}_p-\bm{R}^*\bm{R}^{*\top})(\bm{I}_p-\bm{P}^*\bm{P}^{*\top})\|_\text{F}\\
            + & \|(\bm{I}_p-\widetilde{\bm{R}}\widetilde{\bm{R}}^\top)(\bm{I}_p-\widetilde{\bm{P}}\widetilde{\bm{P}}^\top)(\bm{U}^*\bm{U}^{*\top}+\bm{V}^*\bm{V}^{*\top}-\widetilde{\bm{U}}\widetilde{\bm{U}}^\top-\widetilde{\bm{V}}\widetilde{\bm{V}}^\top)(\bm{I}_p-\bm{R}^*\bm{R}^{*\top})\\
            &(\bm{I}_p-\bm{P}^*\bm{P}^{*\top})\|_\text{F}\\
            + & \|(\bm{I}_p-\widetilde{\bm{R}}\widetilde{\bm{R}}^\top)(\bm{I}_p-\widetilde{\bm{P}}\widetilde{\bm{P}}^\top)(\widetilde{\bm{U}}\widetilde{\bm{U}}^\top+\widetilde{\bm{V}}\widetilde{\bm{V}}^\top)(\bm{R}^*\bm{R}^{*\top}-\widetilde{\bm{R}}\widetilde{\bm{R}}^\top)(\bm{I}_p-\bm{P}^*\bm{P}^{*\top})\|_\text{F}\\
            + & \|(\bm{I}_p-\widetilde{\bm{R}}\widetilde{\bm{R}}^\top)(\bm{I}_p-\widetilde{\bm{P}}\widetilde{\bm{P}}^\top)(\widetilde{\bm{U}}\widetilde{\bm{U}}^\top+\widetilde{\bm{V}}\widetilde{\bm{V}}^\top)(\bm{I}_p-\widetilde{\bm{R}}\widetilde{\bm{R}}^\top)(\bm{P}^*\bm{P}^{*\top}-\widetilde{\bm{P}}\widetilde{\bm{P}}^\top)\|_\text{F}\\
            \lesssim & \|\sin\Theta(\widetilde{\bm{R}},\bm{R}^*)\|_\text{F} + \|\sin\Theta(\widetilde{\bm{P}},\bm{P}^*)\|_\text{F} + \|\sin\Theta(\widetilde{\bm{U}},\bm{U}^*)\|_\text{F} + \|\sin\Theta(\widetilde{\bm{V}},\bm{V}^*)\|_\text{F}\\
            \lesssim & \kappa\sigma_r^{-1}g_{\min}^{-2}\|\widetilde{\bm{A}}_\text{RR}-\bm{A}^*\|_\text{F}.
        \end{split}
    \end{equation}

    By Lemma \ref{lemma:eigen_perturb}, since
    \begin{equation}
        (\bm{I}_p-\bm{R}^*\bm{R}^{*\top})(\bm{I}_p-\bm{P}^*\bm{P}^{*\top})(\bm{U}^*\bm{U}^{*\top}+\bm{V}^*\bm{V}^{*\top})(\bm{I}_p-\bm{R}^*\bm{R}^{*\top})(\bm{I}_p-\bm{P}^*\bm{P}^{*\top})=2\bm{C}^*\bm{C^{*\top}},
    \end{equation}
    we have that
    \begin{equation}
        \|\sin\Theta(\widetilde{\bm{C}},\bm{C}^*)\|_\text{F}\lesssim\kappa\sigma_r^{-1}g_{\min}^{-2}\|\widetilde{\bm{A}}_\text{RR}-\bm{A}^*\|_\text{F}.
    \end{equation}
    
        Denote 
    \begin{equation}
        \begin{split}
            \widetilde{\bm{O}}_c&=\argmin_{\bm{O}\in\mathbb{O}^{p\times d}}\|\widetilde{\bm{C}}\bm{O}-\bm{C}^*\|_\text{F}\\
            \widetilde{\bm{O}}_r&=\argmin_{\bm{O}\in\mathbb{O}^{p\times (r-d)}}\|\widetilde{\bm{R}}\bm{O}-\bm{R}^*\|_\text{F}\\
            \widetilde{\bm{O}}_p&=\argmin_{\bm{O}\in\mathbb{O}^{p\times (r-d)}}\|\widetilde{\bm{P}}\bm{O}-\bm{P}^*\|_\text{F},
        \end{split}
    \end{equation}
    and let $\widetilde{\bm{O}}_1=\text{diag}(\widetilde{\bm{O}}_c,\widetilde{\bm{O}}_r)$
    and $\widetilde{\bm{O}}_2=\text{diag}(\widetilde{\bm{O}}_c,\widetilde{\bm{O}}_p)$.
    For $\widetilde{\bm{D}}$ and $\bm{D}^*$,
    \begin{equation}
        \begin{split}
            & \|\widetilde{\bm{O}}_1^\top\lb\widetilde{\bm{C}}~\widetilde{\bm{R}}\rb^\top\widetilde{\bm{A}}_\text{RR}\lb\widetilde{\bm{C}}~\widetilde{\bm{P}}\rb\widetilde{\bm{O}}_2-\lb\bm{C}^*~\bm{R}^*\rb^\top\bm{A}^*\lb\bm{C}^*~\bm{P}^*\rb\|_\text{F} \\
            \leq & \|\lb\widetilde{\bm{C}}~\widetilde{\bm{R}}\rb\widetilde{\bm{O}}_1-\lb\bm{C}^*~\bm{R}^*\rb\|_\text{F}\cdot\|\widetilde{\bm{A}}_\text{RR}\|_\text{op}\cdot\|\lb\widetilde{\bm{C}}~\widetilde{\bm{P}}\rb\|_\text{op}\\
            + & \|\lb\bm{C}^*~\bm{R}^*\rb\|_\text{op}\cdot\|\widetilde{\bm{A}}_\text{RR}-\bm{A}^*\|_\text{F}\cdot\|\lb\widetilde{\bm{C}}~\widetilde{\bm{P}}\rb\|_\text{op}\\
            + & \|\lb\bm{C}^*~\bm{R}^*\rb\|_\text{op}\cdot\|\bm{A}^*\|_\text{op}\cdot\|\lb\widetilde{\bm{C}}~\widetilde{\bm{P}}\rb\widetilde{\bm{O}}_2-\lb\bm{C}^*~\bm{P}^*\rb\|_\text{F}\\
            \lesssim & \kappa^2g_{\min}^{-2}\|\widetilde{\bm{A}}_\text{RR}-\bm{A}^*\|_\text{F}.
        \end{split}
    \end{equation}
    In summary, when $b\asymp\sigma_1^{1/3}$, we have that
    \begin{equation}
        \begin{split}
            \|\bm{C}^{(0)}\widetilde{\bm{O}}_c-b\bm{C}^*\|_\text{F}=b\|\widetilde{\bm{C}}\widetilde{\bm{O}}_c-\bm{C}^*\|_\text{F} & \lesssim \sigma_1^{4/3}\sigma_r^{-2}g_{\min}^{-2}\|\widetilde{\bm{A}}_\text{RR}-\bm{A}^*\|_\text{F},\\
            \|\bm{R}^{(0)}\widetilde{\bm{O}}_r-b\bm{R}^*\|_\text{F}=b\|\widetilde{\bm{R}}\widetilde{\bm{O}}_r-\bm{R}^*\|_\text{F} & \lesssim \sigma_1^{4/3}\sigma_r^{-2}g_{\min}^{-2}\|\widetilde{\bm{A}}_\text{RR}-\bm{A}^*\|_\text{F},\\
            \|\bm{P}^{(0)}\widetilde{\bm{O}}_r-b\bm{P}^*\|_\text{F}=b\|\widetilde{\bm{P}}\widetilde{\bm{O}}_p-\bm{P}^*\|_\text{F} & \lesssim \sigma_1^{4/3}\sigma_r^{-2}g_{\min}^{-2}\|\widetilde{\bm{A}}_\text{RR}-\bm{A}^*\|_\text{F},\\
            \text{and}~\|\widetilde{\bm{O}}_1^\top\bm{D}^{(0)}\widetilde{\bm{O}}_2-b^{-2}\bm{D}^*\|_\text{F}=b^{-2}\|\widetilde{\bm{O}}_1^\top\widetilde{\bm{D}}\widetilde{\bm{O}}_2-\bm{D}^*\|_\text{F} & \lesssim \sigma_1^{4/3}\sigma_r^{-2}g_{\min}^{-2}\|\widetilde{\bm{A}}_\text{RR}-\bm{A}^*\|_\text{F}.
        \end{split}
    \end{equation}
    
    By Lemmas \ref{lemma:errorbound} and \ref{lemma:rr_initial}, with probability approaching one,
    \begin{equation}
        \|\bm{A}^{(0)}-\bm{A}^*\|_\text{F}\lesssim \sigma_1^{2/3}\kappa^2g_{\min}^{-2}\|\widetilde{\bm{A}}_\text{RR}-\bm{A}^*\|_\text{F}\lesssim\sigma_1^{2/3}\kappa^2g_{\min}^{-2}\alpha_{\textup{RSC}}^{-1}\tau^2M_1\sqrt{\frac{d_\text{RR}(p,r)}{T}}.
    \end{equation}
\end{proof}

\subsection{Auxiliary lemmas}

We first present a deviation bound inequality for the quadratic term $R_T(\bm{M})=\sum_{t=0}^{T-1}\|\bm{M}\bm{y}_t\|_2^2$. This is Lemma 6 in \citet{wang2021high}

\begin{lemma}\label{lemma:quadratic}
    
For any $\bm{M}\in\mathbb{R}^{p\times p}$ such that $\|\bm{M}\|_\textup{F}=1$ and any $t>0$,
    \begin{equation}
        \begin{split}
            & \mathbb{P}[|R_T(\bm{M})-\mathbb{E}R_T(\bm{M})|\geq t]\\
            \leq & 2\exp\left(-\min\left(\frac{t^2}{\tau^4T\lambda_{\max}^2(\bm{\Sigma}_{\bbm{\varepsilon}})\lambda_{\max}^2(\widetilde{\bm{A}}\widetilde{\bm{A}}^\top)},\frac{t}{\tau^2\lambda_{\max}^2(\bm{\Sigma}_{\bbm{\varepsilon}})\lambda^2_{\max}(\widetilde{\bm{A}}\widetilde{\bm{A}}^\top)}\right)\right),
        \end{split}
    \end{equation}
    where $\widetilde{\bm{A}}$ is defined as
    \begin{equation}
        \widetilde{\bm{A}} = \begin{bmatrix}
            \bm{I}_p & \bm{A}^* & \bm{A}^{*2} & \bm{A}^{*3} & \dots & \bm{A}^{*(T-1)} & \dots\\
            \bm{O} & \bm{I}_p & \bm{A}^* & \bm{A}^{*2} & \dots & \bm{A}^{*(T-2)} & \dots\\
            \vdots & \vdots & \vdots & \vdots & \ddots & \vdots & \dots \\
            \bm{O} & \bm{O} & \bm{O} & \bm{O} & \dots & \bm{I}_p & \dots
        \end{bmatrix}.
    \end{equation}
\end{lemma}

The following lemma is the covering number of $\mathcal{W}(r,d;p)$. The proof  essentially follows that of the Lemma 3.1 in \citet{candes2011tight}.

\begin{lemma}\label{lemma:cs_covering}
    
    Let $\overline{\mathcal{W}}(r,d;p)$ be an $\epsilon$-net of $\mathcal{W}(r,d;p)$, where $\epsilon\in(0,1]$. Then
    \begin{equation}
        |\overline{\mathcal{W}}(r,d;p)| \leq \left(\frac{24}{\epsilon}\right)^{p(2r-d)+r^2}.
    \end{equation}
\end{lemma}

\begin{proof}
    
    For any $\bm{W}=\lb\bm{C}~\bm{R}\rb\bm{D}\lb\bm{C}~\bm{P}\rb^\top \in \mathcal{W}(r,d;p)$, where $\lb\bm{C}~\bm{R}\rb,\lb\bm{C}~\bm{P}\rb\in\mathbb{O}^{p\times r}$ and $\bm{D}\in\mathbb{R}^{r\times r}$, we construct an $\epsilon$-net for $\bm{W}$ by covering the set of $\bm{C}$, $\bm{R}$, $\bm{P}$, and $\bm{D}$. 
    
    By Lemma \ref{lemma:covering}, we take $\overline{\mathbb{D}}$ to be an $\epsilon/8$-net for $\bm{D}$ with $|\overline{\mathbb{D}}|\leq (24/\epsilon)^{r^2}$.
    
    Next, to cover $\mathbb{O}^{p\times r}$, we consider the $\|\cdot\|_{2,\infty}$ norm, defined as
    \begin{equation}
        \|\bm{X}\|_{2,\infty}=\max_{i}\|\bm{X}_i\|_2,
    \end{equation}
    where $\bm{X}_i$ is the $i$-th column of $\bm{X}$. Let $\mathbb{Q}^{p\times r}=\{\bm{X}\in\mathbb{R}^{p\times r}:\|\bm{X}\|_{2,\infty}\leq 1\}$. It can be easily checked that $\mathbb{O}^{p\times r}\subset\mathbb{Q}^{p\times r}$, and thus an $\epsilon/8$-net $\overline{\mathbb{O}}^{p\times r}$ for $\mathbb{O}^{p\times r}$ obeying $|\overline{\mathbb{O}}^{p\times r}|\leq (24/\epsilon)^{pr}$.
    
    Denote $\overline{\mathcal{W}}=\{\overline{\bm{D}}\in\overline{\mathbb{D}},\overline{\bm{C}}\in\overline{\mathbb{O}}^{p\times d},\overline{\bm{R}}\in\overline{\mathbb{O}}^{p\times (r-d)},\overline{\bm{P}}\in\overline{\mathbb{O}}^{p\times (r-d)}\}$ and we have
    \begin{equation}
        |\overline{\mathcal{W}}| \leq |\overline{\mathbb{D}}|\times |\overline{\mathbb{O}}^{p\times d}| \times |\overline{\mathbb{O}}^{p\times (r-d)}|^2= \left(\frac{24}{\epsilon}\right)^{p(2r-d)+r^2}.
    \end{equation}
    It suffices to show that for any $\bm{W}\in\mathcal{W}(r,d;p)$, there exists a $\overline{\bm{W}}\in\overline{\mathcal{W}}$ such that $\|\bm{W}-\overline{\bm{W}}\|_\text{F}\leq \epsilon$.
    
    For any fixed $\bm{W}\in\mathcal{W}(r,d;p)$, decompose it as $\bm{W}=\lb\bm{C}~\bm{R}\rb\bm{D}\lb\bm{C}~\bm{P}\rb^\top$. Then, there exist $\overline{\bm{W}}=\lb\overline{\bm{C}}~\overline{\bm{R}}\rb\overline{\bm{D}}\lb\overline{\bm{C}}~\overline{\bm{P}}\rb^\top$ satisfying that $\|\overline{\bm{C}}-\bm{C}\|_{2,\infty}\leq \epsilon/8$, $\|\overline{\bm{R}}-\bm{R}\|_{2,\infty}\leq \epsilon/8$, $\|\overline{\bm{P}}-\bm{P}\|_{2,\infty}\leq \epsilon/8$, and $\|\overline{\bm{D}}-\bm{D}\|_\text{F}\leq \epsilon/8$.
    This gives
    \begin{equation}
        \begin{split}
            &\|\bm{W}-\overline{\bm{W}}\|_\text{F}\\
            \leq & \|(\lb\bm{C}~\bm{R}\rb-\lb\overline{\bm{C}}~\overline{\bm{R}}\rb)\bm{D}\lb\bm{C}~\bm{P}\rb^\top\|_\text{F} + \|\lb\overline{\bm{C}}~\overline{\bm{R}}\rb(\bm{D}-\overline{\bm{D}})\lb\bm{C}~\bm{P}\rb^\top\|_\text{F}\\
            + & \|\lb\overline{\bm{C}}~\overline{\bm{R}}\rb\overline{\bm{D}}(\lb\bm{C}~\bm{P}\rb-\lb\overline{\bm{C}}~\overline{\bm{P}}\rb)^\top\|_\text{F}\\
            \leq & \|\bm{D}\|_\text{F}\cdot\|\lb\bm{C}~\bm{P}\rb\|_\text{op}\cdot\|\lb\bm{C}-\overline{\bm{C}}~\bm{R}-\overline{\bm{R}}\rb\|_{2,\infty}\\
            + & \|\lb\overline{\bm{C}}~\overline{\bm{R}}\rb\|_\text{op}\cdot\|\bm{D}-\overline{\bm{D}}\|_\text{F}\cdot\|\lb\bm{C}~\bm{P}\rb\|_\text{op}\\
            + & \|\overline{\bm{D}}\|_\text{F}\cdot\|\lb\overline{\bm{C}}~\overline{\bm{R}}\rb\|_\text{op}\cdot\|\lb\bm{C}-\overline{\bm{C}}~\bm{P}-\overline{\bm{P}}\rb\|_{2,\infty}\\
            \leq & \frac{\epsilon}{4}+\frac{\epsilon}{4}+\frac{\epsilon}{2}=\epsilon.
        \end{split}
    \end{equation}
    
\end{proof}

The next lemma is the covering number of the $p$-dimensional unit sphere, which follows directly from Corollary 4.2.13 of \citet{vershynin2018high}.

\begin{lemma}\label{lemma:covering}
    
    Let $\mathcal{N}$ be an $\epsilon$-net of the unit sphere $\mathbb{S}^{p-1}$, where $\epsilon\in(0,1]$. Then,
    \begin{equation}
        |\mathcal{N}| \leq \left(\frac{3}{\epsilon}\right)^p.
    \end{equation}
    
\end{lemma}

The following two lemmas are variants of the Davis-Kahan theorem for eigenvector perturbation for symmetric matrices and singular vector perturbation for generic matrices. These results are Theorems 2 and 4 in \citet{yu2015useful}. To make the proof self-contained, they are presented below.

\begin{lemma}\label{lemma:eigen_perturb}
    Let $\bm{\Sigma}$, $\widehat{\bm{\Sigma}}\in\mathbb{R}^{p\times p}$ be symmetric, with eigenvalues $\lambda_1\geq \dots\geq \lambda_p$ and $\widehat{\lambda}_1\geq\dots\geq\widehat{\lambda}_p$, respectively. Fix $1\leq r\leq s\leq p$ and assume that $\min(\lambda_{r-1}-\lambda_r,\lambda_{s}-\lambda_{s+1})>0$, where $\lambda_0:=\infty$ and $\lambda_{p+1}:=-\infty$. Let $d:=s-r+1$, and let $\bm{V}=\lb\bm{v}_r~\bm{v}_{r+1}~\dots~\bm{v}_s\rb\in\mathbb{O}^{p\times d}$ and $\widehat{\bm{V}}=\lb\widehat{\bm{v}}_r~\widehat{\bm{v}}_{r+1}~\dots~\widehat{\bm{v}}_s\rb\in\mathbb{O}^{p\times d}$ contain the eigenvectors corresponding to the eigenvalues. Then
    \begin{equation}
        \|\sin\Theta(\widehat{\bm{V}},\bm{V})\|_\textup{F}\leq\frac{2\min(d^{1/2}\|\widehat{\bm{\Sigma}}-\bm{\Sigma}\|_\textup{op},\|\widehat{\bm{\Sigma}}-\bm{\Sigma}\|_\text{F})}{\min(\lambda_{r-1}-\lambda_r,\lambda_s-\lambda_{s-1})}.
    \end{equation}
\end{lemma}

\begin{lemma}\label{lemma:svd_perturb}
    Let $\bm{A}$, $\widehat{\bm{A}}\in\mathbb{R}^{p\times q}$ have singular values $\sigma_1\geq\dots\geq\sigma_{\min(p,q)}$ and $\widehat{\sigma}_1\geq\dots\geq\widehat{\sigma}_{\min(p,q)}$, respectively. Fix $1\leq r\leq s\leq \textup{rank}(\bm{A})$ and assume that $\min(\sigma_{r-1}^2-\sigma_r^2,\sigma_s^2-\sigma_{s-1}^2)>0$, where $\sigma_0^2:=\infty$ and $\sigma_{\textup{rank}(\bm{A})+1}^2=-\infty$. Let $d:=s-r+1$, and let $\bm{V}=\lb\bm{v}_r~\bm{v}_{r+1}~\dots~\bm{v}_s\rb\in\mathbb{O}^{p\times d}$ and $\widehat{\bm{V}}=\lb\widehat{\bm{v}}_r~\widehat{\bm{v}}_{r+1}~\dots~\widehat{\bm{v}}_s\rb\in\mathbb{O}^{p\times d}$ contain the right singular vectors. Then,
    \begin{equation}
        \|\sin\Theta(\widehat{\bm{V}},\bm{V})\|_\textup{F}\leq \frac{2(2\sigma_1+\|\widehat{\bm{A}}-\bm{A}\|_\textup{op})\min(d^{1/2}\|\widehat{\bm{A}}-\bm{A}\|_\textup{op},\|\widehat{\bm{A}}-\bm{A}\|_\textup{F})}{\min(\sigma_{r-1}^2-\sigma_r^2,\sigma_{s}^2-\sigma_{s-1}^2)}.
    \end{equation}
\end{lemma}

\section{Determination of rank and common dimension}\label{append:selection}

We start from the proof of rank selection consistency in Theorem \ref{thm:rank_selection}.

\begin{proof}[Proof of Theorem \ref{thm:rank_selection}]
    
    Following the proof of Lemma \ref{lemma:rr_initial}, if $T\gtrsim\max(\tau^4,\tau^2)M_2^{-2}p$, then with probability approaching one,
    \begin{equation}
        \|\widetilde{\bm{A}}_\text{RR}(\bar{r})-\bm{A}^*\|_\text{F}\lesssim \alpha_{\textup{RSC}}^{-1}\tau^2M_1\sqrt{\frac{d_\text{RR}(p,\bar{r})}{T}}.
    \end{equation}

    Obviously, $\text{rank}(\widetilde{\bm{A}}_\text{RR}(\bar{r})-\bm{A}^*)\leq \bar{r}+r$. By definition,
    \begin{equation}
        \|\widetilde{\bm{A}}_\text{RR}(\bar{r})-\bm{A}^*\|_\text{F}^2=\sum_{j=1}^{\bar{r}+r}\sigma_j^2(\widetilde{\bm{A}}_\text{RR}(\bar{r})-\bm{A}^*).
    \end{equation}
    By Mirsky's singular value inequality \citep{mirsky1960symmetric},
    \begin{equation}
        \sum_{j=1}^{\bar{r}+r}[\sigma_j(\widetilde{\bm{A}}_\text{RR}(\bar{r}))-\sigma_j(\bm{A}^*)]^2\leq\sum_{j=1}^{\bar{r}+r}\sigma_j^2(\widetilde{\bm{A}}_\text{RR}(\bar{r})-\bm{A}^*)=\|\widetilde{\bm{A}}_\text{RR}(\bar{r})-\bm{A}^*\|_\text{F}^2.
    \end{equation}
    As the $\ell_\infty$ norm of any vector is smaller than the $\ell_2$ norm, it follows the same upper bound
    \begin{equation}
        \max_{1\leq j\leq \bar{r}+r}\left|\sigma_j(\widetilde{\bm{A}}_\text{RR}(\bar{r}))-\sigma_j(\bm{A}^*)\right|\leq \left\{\sum_{j=1}^{\bar{r}+r}\sigma_j^2(\widetilde{\bm{A}}_\text{RR}(\bar{r})-\bm{A}^*)\right\}^{1/2}=\|\widetilde{\bm{A}}_\text{RR}(\bar{r})-\bm{A}^*\|_\text{F}.
    \end{equation}
    
    For any $1\leq j\leq \bar{r}$, note that $\sigma_j(\widetilde{\bm{A}}_\text{RR}(\bar{r}))+s(p,T)=\sigma_j(\bm{A}^*)+[\sigma_j(\widetilde{\bm{A}}_\text{RR}(\bar{r}))-\sigma_j(\bm{A}^*)]+s(p,T)$. For $j>r$, $\sigma_j(\bm{A}^*)=0$ and $|\sigma_j(\widetilde{\bm{A}}_\text{RR}(\bar{r}))-\sigma_j(\bm{A}^*)|=o_p(s(p,T))$, provided that $\alpha_{\textup{RSC}}^{-1}\tau^2M_1\sqrt{d_\text{RR}(p,\bar{r})/T}=o(s(p,T))$. Hence, $s(p,T)$ is the dominating term in $\sigma_j(\widetilde{\bm{A}}_\text{RR}(\bar{r}))+s(p,T)$, when $j>r$.
    When $j\leq r$, as $T\to\infty$, $s(p,T)/\sigma_r(\bm{A}^*)\to0$ and $\sigma_j(\widetilde{\bm{A}}_\text{RR}(p,\bar{r}))+s(p,T)\to\sigma_j(\bm{A}^*)$.
    
    Hence, for $j>r$, as $T\to\infty$,
    \begin{equation}
        \frac{\sigma_{j+1}(\widetilde{\bm{A}}_\text{RR}(\bar{r}))+s(p,T)}{\sigma_{j}(\widetilde{\bm{A}}_\text{RR}(\bar{r}))+s(p,T)}\to\frac{s(p,T)}{s(p,T)}=1.
    \end{equation}
    For $j<r$,
    \begin{equation}
        \frac{\sigma_{j+1}(\widetilde{\bm{A}}_\text{RR}(\bar{r}))+s(p,T)}{\sigma_{j}(\widetilde{\bm{A}}_\text{RR}(\bar{r}))+s(p,T)}\to\frac{\sigma_{j+1}(\bm{A}^*)}{\sigma_{j}(\bm{A}^*)}.
    \end{equation}
    For $j=r$,
    \begin{equation}
        \frac{\sigma_{j+1}(\widetilde{\bm{A}}_\text{RR}(\bar{r}))+s(p,T)}{\sigma_{j}(\widetilde{\bm{A}}_\text{RR}(\bar{r}))+s(p,T)}\to\frac{s(p,T)}{\sigma_r(\bm{A}^*)}=o\left(\min_{1\leq j\leq r-1}\frac{\sigma_{j+1}(\bm{A}^*)}{\sigma_j(\bm{A}^*)}\right).
    \end{equation}
    
\end{proof}

Next, we prove the common dimension selection consistency of the BIC.

\begin{proof}[Proof of Theorem \ref{thm:d_consistency}]
    
    To show the consistency of common dimension selection via BIC, it suffices to show that
    \begin{equation}
        \min_{0\leq k<d}\text{BIC}(k)-\text{BIC}(d)>0~~\text{and}~~\min_{d< k\leq r}\text{BIC}(k)-\text{BIC}(d)>0.
    \end{equation}
    Consider the under-parameterized case $k<d$ first. Note that
    \begin{equation}
        \begin{split}
            & \|\bm{Y}-\widehat{\bm{A}}(d)\bm{X}\|_\text{F}^2=\|\bm{E}-(\widehat{\bm{A}}(d)-\bm{A}^*)\bm{X}\|_\text{F}^2\\
            =&\|\bm{E}\|_\text{F}^2+\|(\widehat{\bm{A}}(d)-\bm{A}^*)\bm{X}\|_\text{F}^2+2\langle\bm{E}\bm{X}^\top,\widehat{\bm{A}}(d)-\bm{A}^*\rangle\\
            = & T\text{tr}(\bm{\Sigma}_{\bbm{\varepsilon}}) + Cd_{\text{CS}}(p,r,d) + o_p(T) = T\text{tr}(\bm{\Sigma}_{\bbm{\varepsilon}})+o_p(T)
        \end{split}
    \end{equation}
    and
    \begin{equation}
        \begin{split}
            & \|\bm{Y}-\widehat{\bm{A}}(k)\bm{X}\|_\text{F}^2=\|\bm{E}-(\widehat{\bm{A}}(k)-\bm{A}^*)\bm{X}\|_\text{F}^2\\
            =&\|\bm{E}\|_\text{F}^2+\|(\widehat{\bm{A}}(k)-\bm{A}^*)\bm{X}\|_\text{F}^2+2\langle\bm{E}\bm{X}^\top,\widehat{\bm{A}}(k)-\bm{A}^*\rangle\\
            \geq & T\text{tr}(\bm{\Sigma}_{\bbm{\varepsilon}}) + CTg_{\min}^2 + o_p(T).
        \end{split}
    \end{equation}

    Thus, by $\log(1+x)\asymp x$ for $x\to0$, we have
    \begin{equation}
        \begin{split}
            & \min_{0\leq k<d}\text{BIC}(k)-\text{BIC}(d)\\
            \asymp & pT\log\left(1+\frac{\|\bm{Y}-\widehat{\bm{A}}(k)\bm{X}\|_\text{F}^2-\|\bm{Y}-\widehat{\bm{A}}(d)\bm{X}\|_\text{F}^2}{\|\bm{Y}-\widehat{\bm{A}}(d)\bm{X}\|_\text{F}^2}\right)\\
            \geq & CTg_{\min}^2-Cpd\log(T)+o_p(T)
        \end{split}
    \end{equation}
    and it follows that $\mathbb{P}(\min_{0\leq k<d}\text{BIC}(k)-\text{BIC}(d)>0)\to1$, as $T\to\infty$, provided that $\log(T)pg_{\min}^{-2}/T\to0$.
    
    For the cases $k>d$, using the similar arguments in the proof of Theorem \ref{thm:stat}, we can show that $\|\widehat{\bm{A}}(k)-\bm{A}^*\|_\text{F}\asymp\sqrt{d_\text{CS}(p,r,k)/T}$. It follows that
    \begin{equation}
        \begin{split}
            & \min_{d< k\leq r}\text{BIC}(k)-\text{BIC}(d) \geq O_p(p)+p(k-d)\log(T)
        \end{split}
    \end{equation}
    which implies that $\min_{d< k\leq r}\text{BIC}(k)-\text{BIC}(d)>0$, since $\log(T)\to\infty$ as $T\to\infty$.
    
\end{proof}

\section{Supplementary materials for VAR($\ell$) models}\label{append:VAR_L}

Appendix \ref{append:VAR_L} presents the supplementary materials of modeling, estimation and theory for the VAR($\ell$). It begins with some preliminaries of tensor notation and tensor operation.

\subsection{Some basics of tensor algebra}

We follow the notations in \citet{kolda2009tensor} to denote tensors of order three or higher by Euler script boldface letters, e.g., $\cm{A}$. For a generic $d$-th order tensor $\cm{A}\in\mathbb{R}^{p_1\times\cdots\times p_d}$, denote its elements by $\cm{A}(i_1,i_2,\dots,i_d)$ and unfolding of $\cm{A}$ along the $n$-mode by $\cm{A}_{(n)}$, where the columns of $\cm{A}_{(n)}$ are the $n$-mode vectors of $\cm{A}$, for $n=1,\dots,d$. The Frobenius norm of a tensor $\cm{A}$ is defined as $\|\cm{A}\|_\text{F}=\sqrt{\sum_{i_1}\cdots\sum_{i_d}\cm{A}(i_1,\dots,i_d)^2}$. The mode-$n$ multiplication $\times_n$ of a tensor $\cm{A}\in\mathbb{R}^{p_1\times\cdots\times p_d}$ and a matrix $\bm{B}\in\mathbb{R}^{q_n\times p_n}$ is defined as
\begin{equation}
    (\cm{A}\times_n\bm{B})(i_1,\dots,j_n\dots,i_d)=\sum_{i_n=1}^{p_n}\cm{A}(i_1,\dots,i_n,\dots,i_d)\bm{B}(j_n,i_n),
\end{equation}
for $n=1,\dots,d$, respectively.

The tensor ranks considered in this paper are defined as the matrix ranks of the  unfoldings of $\cm{A}$ along all modes, namely $\text{rank}_i(\cm{A})=\text{rank}(\cm{A}_{(i)})$, for $i=1,\dots,d$. If the tensor ranks of $\cm{A}$ are $r_1,\dots,r_d$, where $1\leq r_i\leq p_i$, there exists a tensor $\cm{G}\in\mathbb{R}^{r_1\times\cdots\times r_d}$ and matrices $\bm{U}_i\in\mathbb{R}^{p_i\times r_i}$, such that
\begin{equation}
    \cm{A}=\cm{G}\times_1\bm{U}_1\times_2\bm{U}_2\cdots\times_d\bm{U}_d,
\end{equation}
which is well known as Tucker decomposition \citep{tucker1966some}. With the Tucker decomposition, the $n$-mode unfolding of $\cm{A}$ can be written as
\begin{equation}
    \cm{A}_{(n)} = \bm{U}_n\cm{G}_{(n)}(\bm{U}_d\otimes\cdots\otimes\bm{U}_{n+1}\otimes\bm{U}_{n-1}\otimes\cdots\otimes\bm{U}_1)^\top,
\end{equation}
where $\otimes$ denotes the Kronecker product for matrices.

\subsection{Algorithm, rank selection and common dimension selection for VAR($\ell$) models}\label{append:D.3}

A gradient descent algorithm (Algorithm \ref{alg:GD_2}) is proposed for the estimation of VAR($\ell$) model.
Note that the loss function with respect to the parameter tensor $\cm{A}$ is $\mathcal{L}(\cm{A})=(2T)^{-1}\sum_{t=1}^T\|\bm{y}_t-\cm{A}_{(1)}\bm{x}_t\|_2^2$, and its gradient has the form of $\nabla\mathcal{L}(\cm{A})=T^{-1}\sum_{t=1}^T(\cm{A}_{(1)}\bm{x}_t-\bm{y}_t)\circ\bm{X}_t$, where $\circ$ denotes the tensor outer product, and $\bm{X}_t=\lb\bm{y}_{t-1}\dots\bm{y}_{t-\ell}\rb\in\mathbb{R}^{p\times \ell}$.
The partial derivatives are listed below,
\begin{equation}
    \begin{split}
        \nabla_{\bm{C}}\mathcal{L} & = \nabla\mathcal{L}(\cm{A})_{(1)}\left[(\bm{L}\otimes\bm{C})(\cm{G}_{11})_{(1)}^\top+(\bm{L}\otimes\bm{P})(\cm{G}_{12})_{(1)}^\top\right]\\
        & + \nabla\mathcal{L}(\cm{A})_{(2)}\left[(\bm{L}\otimes\bm{C})(\cm{G}_{11})_{(2)}^\top+(\bm{L}\otimes\bm{R})(\cm{G}_{21})_{(2)}^\top\right],\\
        \nabla_{\bm{R}}\mathcal{L} & = \nabla\mathcal{L}(\cm{A})_{(1)}[(\bm{L}\otimes\bm{C})(\cm{G}_{21})_{(1)}^\top+(\bm{L}\otimes\bm{P})(\cm{G}_{22})_{(1)}^\top],\\
        \nabla_{\bm{P}}\mathcal{L} & = \nabla\mathcal{L}(\cm{A})_{(2)}[(\bm{L}\otimes\bm{C})(\cm{G}_{12})_{(2)}^\top+(\bm{L}\otimes\bm{R})(\cm{G}_{22})_{(2)}^\top],\\
        \nabla_{\bm{L}}\mathcal{L} & = \nabla\mathcal{L}(\cm{A})_{(3)}(\lb\bm{C}~\bm{P}\rb\otimes\lb\bm{C}~\bm{R}\rb)\cm{G}_{(3)}^\top,\\
        \text{and}~\nabla_{\scalebox{0.7}{\cm{G}}}\mathcal{L} & = \nabla\mathcal{L}(\cm{A})\times_1\lb\bm{C}~\bm{R}\rb^\top\times_2\lb\bm{C}~\bm{P}\rb^\top\times_3\bm{L}^\top,
    \end{split}
\end{equation}
where $\otimes$ is the Kronecker product of matrices.

\begin{algorithm}[!htp]
    \caption{Gradient descent algorithm for VAR($\ell$) model with known $r_1$, $r_2$, $r_3$ and $d$}
    \label{alg:GD_2}
    1: \textbf{Input}: $\bm{Y}$, $\bm{X}$, $\eta$, $I$, $\bm{C}^{(0)}$, $\bm{R}^{(0)}$, $\bm{P}^{(0)}$ ,$\bm{L}^{(0)}$ and $\cm{G}^{(0)}$\\[-0.3em]
    2: \textbf{for} $i=0,\dots,I-1$\\[-0.3em]
    3: \hspace*{0.7cm} $\bm{C}^{(i+1)}=\bm{C}^{(i)}-\eta\nabla_{\bm{C}}\mathcal{L}^{(i)}-\eta a\big[2\bm{C}^{(i)}(\bm{C}^{(i)\top}\bm{C}^{(i)}-b^2\bm{I}_d)+\bm{R}^{(i)}\bm{R}^{(i)\top}\bm{C}^{(i)}+\bm{P}^{(i)}\bm{P}^{(i)\top}\bm{C}^{(i)}\big]$\\[-0.3em]
    4: \hspace*{0.7cm} $\bm{R}^{(i+1)}=\bm{R}^{(i)}-\eta\nabla_{\bm{R}}\mathcal{L}^{(i)}-\eta a\big[\bm{R}^{(i)}(\bm{R}^{(i)\top}\bm{R}^{(i)}-b^2\bm{I}_{r_1-d})+\bm{C}^{(i)}\bm{C}^{(i)\top}\bm{R}^{(i)}\big]$\\[-0.3em]
    5: \hspace*{0.7cm} $\bm{P}^{(i+1)}=\bm{P}^{(i)}-\eta\nabla_{\bm{P}}\mathcal{L}^{(i)}-\eta a\big[\bm{P}^{(i)}(\bm{P}^{(i)\top}\bm{P}^{(i)}-b^2\bm{I}_{r_2-d})+\bm{C}^{(i)}\bm{C}^{(i)\top}\bm{P}^{(i)}\big]$\\[-0.3em]
    6: \hspace*{0.7cm} $\bm{L}^{(i+1)}=\bm{L}^{(i)}-\eta\nabla_{\bm{L}}\mathcal{L}^{(i)}-\eta a\bm{L}^{(i)}(\bm{L}^{(i)\top}\bm{L}^{(i)}-b^2\bm{I}_{r_3})$\\[-0.3em]
    7: \hspace*{0.7cm}
    $\cm{G}^{(i+1)}=\cm{G}^{(i)}-\eta\nabla_{\scalebox{0.7}{\cm{G}}}\mathcal{L}^{(i)}$\\[-0.5em]
    8: \textbf{end for}\\[-0.3em]
    9: \textbf{Return}: $\cm{A}^{(I)}=\cm{G}^{(I)}\times_1\lb\bm{C}^{(I)}~\bm{R}^{(I)}\rb\times_2\lb\bm{C}^{(I)}~\bm{P}^{(I)}\rb\times_3\bm{L}^{(I)}$
\end{algorithm}

For the initialization of Algorithm \ref{alg:GD_2}, we first consider the rank-constrained estimator, 
\begin{equation}\label{eq:tensor-init2}
    \cm{\widetilde{A}}_\text{RR}(r_1,r_2,r_3)=\argmin_{\text{rank}(\scalebox{0.7}{\cm{A}}_{(i)})=r_i,1\leq i\leq 3}\frac{1}{2T}\|\bm{y}_t-\cm{A}_{(1)}\bm{x}_t\|_\text{F}^2,
\end{equation}
and then conduct the HOSVD: $\cm{\widetilde{A}}_\text{RR}(r_1,r_2,r_3)=\cm{\widehat{G}}\times_1\widetilde{\bm{U}}_1\times_2\widetilde{\bm{U}}_2\times_3\widetilde{\bm{U}}_3$, where $\widetilde{\bm{U}}_i$ is the top $r_i$ left singular vectors of its mode-$i$ matricization for each $1\leq i\leq 3$. 
By applying the method in Section \ref{sec:3.2} to $\widetilde{\bm{U}}_1$ and $\widetilde{\bm{U}}_2$, we can obtain $\widetilde{\bm{C}}$, $\widetilde{\bm{R}}$, and $\widetilde{\bm{P}}$, and the initialization can then be set to $\bm{C}^{(0)}=b\widetilde{\bm{C}}$, $\bm{R}^{(0)}=b\widetilde{\bm{R}}$, $\bm{P}^{(0)}=b\widetilde{\bm{P}}$, $\bm{L}^{(0)}=b\widetilde{\bm{U}}_3$, and $\cm{G}^{(0)}=\cm{\widetilde{A}}_\text{RR}(r_1,r_2,r_3)\times_1\lb\bm{C}^{(0)}~\bm{R}^{(0)}\rb^\top\times_2\lb\bm{C}^{(0)}~\bm{P}^{(0)}\rb^\top\times_3\bm{L}^{(0)\top}$.

To select the tensor ranks of $(r_1,r_2,r_3)$, we first set their pre-specified upper bounds,  $(\bar{r}_1,\bar{r}_2,\bar{r}_3)$, where each $\bar{r}_i$ is greater than $r_i$ but much smaller than $p$. The rank-constrained estimator $\cm{\widetilde{A}}_\text{RR}(\bar{r}_1,\bar{r}_2,\bar{r}_3)$ can then be calculated according to \eqref{eq:tensor-init2}.
As a result, the tensor ranks can be selected by the ridge-type ratio method,
\begin{equation}
    \widehat{r}_i=\argmin_{1\leq j\leq \bar{r}_i-1}\frac{\sigma_{j+1}(\cm{\widetilde{A}}_\text{RR}(\bar{r}_1,\bar{r}_2,\bar{r}_3)_{(i)})+s(p,T)}{\sigma_{j}(\cm{\widetilde{A}}_\text{RR}(\bar{r}_1,\bar{r}_2,\bar{r}_3)_{(i)})+s(p,T)}, \hspace{5mm}1\leq i\leq 3.
\end{equation}

Denote by $\cm{\widehat{A}}(r_1,r_2,r_3,d)$ the estimated parameter tensor from Algorithm \ref{alg:GD_2} with the tensor ranks of $(r_1,r_2,r_3)$ and common dimension $d$. The BIC of models \eqref{eq:VAR_ell} and \eqref{eq:common_tensor_decomp2} can be defined as
\begin{equation*}
    \text{BIC}(r_1,r_2,r_3,d)=Tp\log\left(\sum_{t=1}^T\|\bm{y}_t-\cm{\widehat{A}}(r_1,r_2,r_3,d)\bm{x}_t\|_2^2\right)+d_\text{CS}(p,\ell,r_1,r_2,r_3,d)\log(T),
\end{equation*}
and we can choose the common dimension by
$   
\widehat{d}=\argmin_{0\leq d\leq \min(r_1,r_2)}\text{BIC}(r_1,r_2,r_3,d).
$

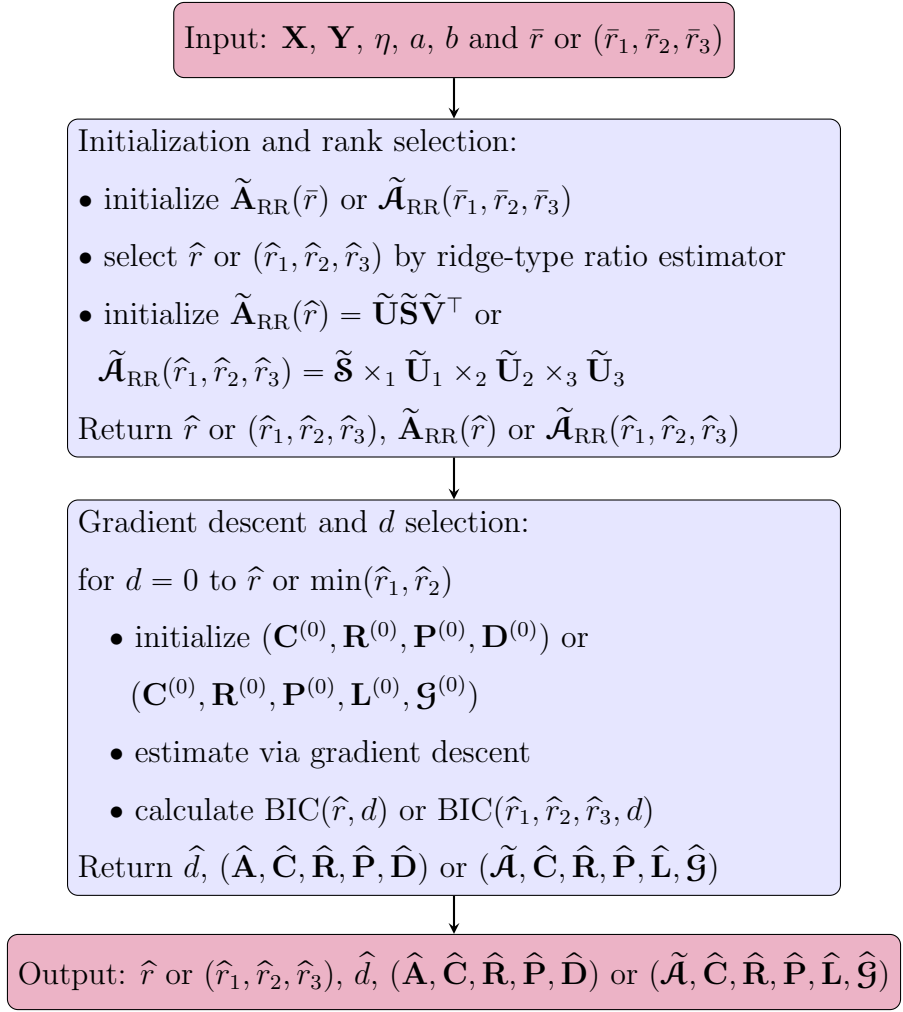
\begin{figure}[t]
    \begin{center}
        \begin{tikzpicture}[node distance = 3cm]
        \node(input)[startstop]{Input: $\bm{X}$, $\bm{Y}$,
            $\eta$, $a$, $b$ and $\bar{r}$ or ($\bar{r}_1,\bar{r}_2,\bar{r}_3$)};
        \node(init)[process2, below of = input, yshift = -0.3cm]{Initialization and rank selection:\\
            $\bullet$ initialize $\widetilde{\bm{A}}_\text{RR}(\bar{r})$ or $\cm{\widetilde{A}}_\text{RR}(\bar{r}_1,\bar{r}_2,\bar{r}_3)$\\
            $\bullet$ select $\widehat{r}$ or $(\widehat{r}_1,\widehat{r}_2,\widehat{r}_3)$ by ridge-type ratio estimator\\
            $\bullet$ initialize $\widetilde{\bm{A}}_\text{RR}(\widehat{r})=\widetilde{\bm{U}}\widetilde{\bm{S}}\widetilde{\bm{V}}^\top$ or\\ ~~$\cm{\widetilde{A}}_\text{RR}(\widehat{r}_1,\widehat{r}_2,\widehat{r}_3)=\cm{\widetilde{S}}\times_1\widetilde{\bm{U}}_1\times_2\widetilde{\bm{U}}_2\times_3\widetilde{\bm{U}}_3$\\
            Return $\widehat{r}$ or $(\widehat{r}_1,\widehat{r}_2,\widehat{r}_3)$, $\widetilde{\bm{A}}_\text{RR}(\widehat{r})$ or $\cm{\widetilde{A}}_\text{RR}(\widehat{r}_1,\widehat{r}_2,\widehat{r}_3)$};
        \node(gd)[process, below of = input, yshift = -5.75cm]{Gradient descent and $d$ selection:\\
            for $d=0$ to $\widehat{r}$ or $\min(\widehat{r}_1,\widehat{r}_2)$\\
            ~~~$\bullet$ initialize $(\bm{C}^{(0)},\bm{R}^{(0)}, \bm{P}^{(0)},\bm{D}^{(0)})$ or\\ ~~~~~$(\bm{C}^{(0)},\bm{R}^{(0)}, \bm{P}^{(0)},\bm{L}^{(0)},\cm{G}^{(0)})$\\
            ~~~$\bullet$ estimate via gradient descent\\
            ~~~$\bullet$ calculate $\text{BIC}(\widehat{r},d)$ or $\text{BIC}(\widehat{r}_1,\widehat{r}_2,\widehat{r}_3,d)$\\
            Return  $\widehat{d}$,  $(\widehat{\bm{A}},\widehat{\bm{C}},\widehat{\bm{R}},\widehat{\bm{P}},\widehat{\bm{D}})$ or $(\cm{\widetilde{A}},\widehat{\bm{C}},\widehat{\bm{R}},\widehat{\bm{P}},\widehat{\bm{L}},\cm{\widehat{G}})$};
        \node(output)[startstop, below of = input, yshift = -9.4cm]{Output: $\widehat{r}$ or $(\widehat{r}_1,\widehat{r}_2,\widehat{r}_3)$, $\widehat{d}$, $(\widehat{\bm{A}},\widehat{\bm{C}},\widehat{\bm{R}},\widehat{\bm{P}},\widehat{\bm{D}})$ or $ (\cm{\widetilde{A}},\widehat{\bm{C}},\widehat{\bm{R}},\widehat{\bm{P}},\widehat{\bm{L}},\cm{\widehat{G}})$};
        \draw [arrow] (input) -- (init);
        \draw [arrow] (init) -- (gd);
        \draw [arrow] (gd) -- (output);
        \end{tikzpicture}
    \end{center}
    \vspace{-0.5cm}
    \caption{Flowchart of the proposed estimation procedure.}
    \label{fig:summary}
\end{figure}

\subsection{Theoretical results for VAR($\ell$) models}

\begin{proof}[Proof of Theorem \ref{thm:VAR_L}]
    
    The proof consists of two parts. The first part is an  extension of Theorem \ref{thm:gd} for computational convergence analysis of the gradient descent iterates, provided that some regularity conditions are satisfied. The second part is the statistical analysis to show that these conditions do hold with high probability.
    
    First, similarly to Theorem \ref{thm:gd}, for the iterate at the step $i$, define the estimation error of $(\bm{C},\bm{R},\bm{P},\bm{L},\cm{G})$ up to the optimal rotations
    \begin{equation}
        \begin{split}
            E^{(i)} = & \min_{\substack{\bm{O}_c\in\mathbb{O}^{d\times d},\bm{O}_r\in\mathbb{O}^{(r_1-d)\times(r_1-d)},\\ \bm{O}_p\in\mathbb{O}^{(r_2-d)\times(r_2-d)},\bm{O}_l\in\mathbb{O}^{r_3\times r_3}}}  \Big\{\|\bm{C}^{(i)}-\bm{C}^*\bm{O}_c\|_\text{F}^2 + \|\bm{R}^{(i)}-\bm{R}^*\bm{O}_r\|_\text{F}^2+\|\bm{P}^{(i)}-\bm{P}^*\bm{O}_p\|_\text{F}^2\\
            & + \|\bm{L}^{(i)}-\bm{L}^*\bm{O}_l\|_\text{F}^2+\|\cm{G}^{(i)}-\cm{G}^*\times_1\text{diag}(\bm{O}_c,\bm{O}_r)\times_2\text{diag}(\bm{O}_c,\bm{O}_p)\times\bm{O}_l\|_\text{F}^2\Big\}
        \end{split}
    \end{equation}
    and the corresponding optimal rotations as $(\bm{O}_c^{(i)},\bm{O}_r^{(i)},\bm{O}_p^{(i)},\bm{O}_3^{(i)})$. Denote $\bm{O}_1^{(i)}=\text{diag}(\bm{O}_c^{(i)},\bm{O}_r^{(i)})$ and $\bm{O}_2^{(i)}=\text{diag}(\bm{O}_c^{(i)},\bm{O}_p^{(i)})$.
    
    In a similar fashion, the $\alpha$-RSC and $\beta$-RSS conditions are also assumed for the loss function $\mathcal{L}(\cm{A})$, and for the given sample, define
    \begin{equation}
        \xi(r_1,r_2,r_3,d) = \sup_{\substack{\lb\bm{C}~\bm{R}\rb\in\mathbb{O}^{p\times r_1},\lb\bm{C}~\bm{P}\rb\in\mathbb{O}^{p\times r_2},\\ \bm{L}\in\mathbb{O}^{\ell\times r_3},\scalebox{0.7}{\cm{G}}\in\mathbb{R}^{r_1\times r_2\times r_3},\|\scalebox{0.7}{\cm{G}}\|_\text{F}=1}}\big\langle\mathcal{L}(\cm{A}^*),\cm{G}\times_1\lb\bm{C}~\bm{R}\rb\times_2\lb\bm{C}~\bm{P}\rb\times_3\bm{L}\big\rangle.
    \end{equation}
    For simplicity, assume $b=\bar{\sigma}^{1/4}$ and $a=C\alpha\alpha_1^{3/4}\kappa^{-2}$. For any $i=0,1,2\dots$, we assume that
    \begin{equation}
        \|\lb\bm{C}^{(i)}~\bm{R}^{(i)}\rb\|_\text{op}\leq 1.1b,~\|\lb\bm{C}^{(i)}~\bm{P}^{(i)}\rb\|_\text{op}\leq 1.1b,~\|\bm{L}^{(i)}\|_\text{op}\leq 1.1b, \max_{i=1,2,3}\|\cm{G}_{(i)}\|_\text{op}\leq \frac{1.1\bar{\sigma}}{b^3},
    \end{equation}
    and $E^{(i)}\leq C\bar{\sigma}^{1/2}\alpha\beta^{-1}\kappa^{-2}$. In addition, for the initial value $\cm{A}^{(0)}=\cm{G}^{(0)}\times_1\lb\bm{C}^{(0)}~\bm{R}^{(0)}\rb\times_2\lb\bm{C}^{(0)}~\bm{P}^{(0)}\rb\times_3\bm{L}^{(0)}$, we assume that $\|\cm{A}^{(0)}-\cm{A}^*\|_\text{F}\lesssim \underline{\sigma}$.
    
    Based on these conditions, we have that
    \begin{equation}
        \begin{split}
            E^{(i+1)} \leq & E^{(i)} - 2\eta(Q_{\text{G},1}+Q_{\text{C},1}+Q_{\text{R},1}+Q_{\text{P},1}+Q_{\text{L},1})\\
            & +  \eta^2(Q_{\text{G},2}+Q_{\text{C},2}+Q_{\text{R},2}+Q_{\text{P},2}+Q_{\text{L},2}),
        \end{split}
    \end{equation}
    where the terms $Q_{\cdot,1}$ and $Q_{\cdot,2}$ are defined similarly as those in the proof of Theorem \ref{thm:gd}. The upper bound for $Q_{\text{G},2}+Q_{\text{C},2}+Q_{\text{R},2}+Q_{\text{P},2}+Q_{\text{L},2}$ and the lower bound for $Q_{\text{G},1}+Q_{\text{C},1}+Q_{\text{R},1}+Q_{\text{P},1}+Q_{\text{L},1}$ can be derived. Extending these results from the matrix case to the 3rd order tensor case hinges on the framework in the Theorem 3.1 of \citet{han2020optimal}, which leads to
    \begin{equation}
        \begin{split}
            E^{(i)} & \leq \left(1-C\eta_0\alpha\beta^{-1}\kappa^{-2}\right)E^{(i)} + C\kappa^2\alpha^{-2}\bar{\sigma}^{-3/2}\xi^2(r_1,r_2,r_3,d)
        \end{split}
    \end{equation}
    when $\eta=\eta_0\beta^{-1}\bar{\sigma}^{-3/2}$ for some $\eta_0<1/280$.
    
    By induction, we have that for any $t=1,2,\dots$,
    \begin{equation}
        E^{(i)}\leq (1-C\eta_0\alpha\beta^{-1}\kappa^{-2})^{(i)}E^{(i)} + C\kappa^2\alpha^{-2}\bar{\sigma}^{-3/2}\xi^2(r_1,r_2,r_3,d).
    \end{equation}
    For the error bound of $\|\cm{A}^{(i)}-\cm{A}^*\|_\text{F}^2$, by Lemma \ref{lemma:pertub_tensor},
    \begin{equation}
        \begin{split}
            & \|\cm{A}^{(i)}-\cm{A}^*\|_\text{F}^2\leq C\bar{\sigma}^{3/2}E^{(i)}\\
            \leq & C\bar{\sigma}^{3/2}(1-C\eta_0\alpha\beta^{-1}\kappa^{-2})^iE^{(0)}+C\kappa^2\alpha^{-2}\xi^2(r_1,r_2,r_3,d)\\
            \leq & C\kappa^2(1-C\eta_0\alpha\beta^{-1}\kappa^{-2})^i\|\cm{A}^{(0)}-\cm{A}^*\|_\text{F}^2+C\kappa^2\alpha^{-2}\xi^2(r_1,r_2,r_3,d).
        \end{split}
    \end{equation}
    The conditions for iterates $(\cm{G}^{(i)},\bm{C}^{(i)},\bm{R}^{(i)},\bm{P}^{(i)},\bm{L}^{(i)})$ can be verified similarly as in Theorem \ref{thm:gd}.

    Second, in the statistical analysis, we need to verify the upper bound  for the initial value $\cm{A}^{(0)}$, the deviation bound for $\xi(r_1,r_2,r_3,d)$, $\alpha$-RSC condition and $\beta$-RSS conditions. For the low-rank estimator $\widetilde{\cm{A}}_\text{RR}(r_1,r_2,r_3)$, according to Theorem 1 of \citet{wang2020compact}, with probability at least $1-\exp[-C(r_1r_2r_3+pr_1+pr_2+\ell r_3)]-\exp(-CM_2^2\min(\tau^{-4},\tau^{-2})T)$,
    \begin{equation}
        \|\widetilde{\cm{A}}_\text{RR}(r_1,r_2,r_3)-\cm{A}^*\|_\text{F}\lesssim \alpha^{-1}\tau^2M_1\sqrt{\frac{d_\text{RR}(p,\ell,r_1,r_2,r_3)}{T}}.
    \end{equation}
    Then, by the same techniques in the proof of Lemma \ref{lemma:A0_init}, we can obtain that
    \begin{equation}
        \|\cm{A}^{(0)}-\cm{A}^*\|_\text{F}\lesssim\bar{\sigma}^{3/4}\kappa^2g_{\min}^{-2}\tau^2M_1\sqrt{\frac{d_\text{RR}(p,\ell,r_1,r_2,r_3)}{T}}\lesssim \underline{\sigma}.
    \end{equation}
    The $\alpha_{\textup{RSC}}$-RSC and $\beta_\text{RSS}$-RSS conditions and the deviation bound for $\xi(r_1,r_2,r_3,d)$ can be proved to hold with high probability approaching one, in the same manner as Lemmas \ref{lemma:RSC} and \ref{lemma:deviation}. For brevity, the detailed proofs are omitted.
    
\end{proof}

The following lemma is a straightforward extension of Lemma E.2 in \citet{han2020optimal} with common subgroup structures included, so the proof is omitted for simplicity.

\begin{lemma}\label{lemma:pertub_tensor}
    Suppose that $\cm{A}^*=\cm{G}^*\times_1\lb\bm{C}^*~\bm{R}^*\rb\times_2\lb\bm{C}^*~\bm{P}^*\rb\times_3\bm{L}^*$, $\lb\bm{C}^*~\bm{R}^*\rb^{\top}\lb\bm{C}^*~\bm{R}^{*}\rb=\bm{I}_{r_1}$, $\lb\bm{C}^*~\bm{P}^*\rb^{\top}\lb\bm{C}^*~\bm{P}^{*}\rb=\bm{I}_{r_2}$,
    $\bm{L}^{*\top}\bm{L}^*=\bm{I}_{r_3}$, $\bar{\sigma}=\max_{1\leq i\leq 3}\|\cm{A}^*_{(i)}\|_\textup{op}$, and $\underline{\sigma}=\min_{1\leq i\leq 3}\sigma_{r_i}(\cm{A}_{(i)}^*)$. Let $\cm{A}=\cm{G}\times_1\lb\bm{C}~\bm{R}\rb\times_2\lb\bm{C}~\bm{P}\rb\times_3\bm{L}$ with $\|\lb\bm{C}~\bm{R}\rb\|_\textup{op}\leq (1+c_b)b$, $\|\lb\bm{C}~\bm{P}\rb\|_\textup{op}\leq (1+c_b)b$, $\|\bm{L}\|_\textup{op}\leq (1+c_b)b$ and $\max_{1\leq i\leq 3}\|\cm{G}_{(i)}\|_\textup{op}\leq(1+c_b)\bar{\sigma}/b^3$ for some constant $c_b>0$. Define
    \begin{equation}
        \begin{split}
            E:=&\min_{\bm{O}_1,\bm{O}_2,\bm{O}_3}\big(\|\lb\bm{C}^{(i)}~\bm{R}^{(i)}\rb-\lb\bm{C}^*~\bm{R}^*\rb\bm{O}_1\|_\textup{F}^2+\|\lb\bm{C}^{(i)}~\bm{P}^{(i)}\rb-\lb\bm{C}^*~\bm{P}^*\rb\bm{O}_2\|_\textup{F}^2 \\
            & + \|\bm{L}^{(i)}-\bm{L}^*\bm{O}_3\|_\textup{F}^2+ \|\cm{G}^{(i)}-\cm{G}^*\times_1\bm{O}_1\times_2\bm{O}_2\times_3\bm{O}_3\|_\textup{F}^2\big).
        \end{split}
    \end{equation}
    Then, we have
    \begin{equation}
        \begin{split}
            E & \leq \left(7b^{-6}+\frac{12b^2}{\underline{\sigma}^2}C_b\right)\|\cm{A}-\cm{A}^*\|_\textup{F}^2 + 2b^{-2}C_b\big(\|\lb\bm{C}~\bm{R}\rb^\top\lb\bm{C}~\bm{R}\rb-b^2\bm{I}_{r_1}\|_\textup{F}^2\\
            &+\|\lb\bm{C}~\bm{R}\rb^\top\lb\bm{C}~\bm{P}\rb-b^2\bm{I}_{r_2}\|_\textup{F}^2+\|\bm{L}^\top\bm{L}-b^2\bm{I}_{r_3}\|\big),
        \end{split}
    \end{equation}
    and
    \begin{equation}
        \|\cm{A}-\cm{A}^*\|_\textup{F}^2 \leq 4b^6[1+\bar{\sigma}^2b^{-8}(3+2c_b)^2(1+c_b)^4]E,
    \end{equation}
    where $C_b=1+7\bar{\sigma}^2b^{-8}[(1+c_b)^4+(1+c_b)^4(2+c_b)^2]$.
\end{lemma}

\section{VAR for diverging eigenvalue effects}
\label{append:diverging}

\subsection{Factor analysis of reduced-rank VAR process}

In this section, we conduct factor analysis for the reduced-rank VAR process with diverging eigenvalue effects. For the VAR(1) model $\bm{y}_t=\bm{A}\bm{y}_{t-1}+\bbm{\varepsilon}_t$, the covariance matrix of $\bm{y}_t$ satisfies
\begin{equation}
    \bm{\Sigma}_{\bm{y}}=\bm{A}\bm{\Sigma}_{\bm{y}}\bm{A}^\top+\bm{\Sigma}_{\bbm{\varepsilon}},
\end{equation}
where $\bm{\Sigma}_{\bm{y}}=\text{var}(\bm{y}_t)$ and $\bm{\Sigma}_{\bbm{\varepsilon}}=\text{var}(\bbm{\varepsilon}_t)$.
When the covariance matrix $\bm{\Sigma}_{\bm{y}}$ has some diverging eigenvalues, the singular values of $\bm{A}$ and/or the eigenvalues of $\bm{\Sigma}_{\bbm{\varepsilon}}$ may also be diverging.
If $\bm{A}$ has diverging singular values, the dynamic autoregressive part has the diverging effect, whereas if the leading eigenvalues of $\bm{\Sigma}_{\bbm{\varepsilon}}$ are spiky, the white noise part has the diverging eigenvalue effect.

Note that the stationarity of the VAR process requires that all eigenvalues of $\bm{A}$ are strictly smaller than one in terms of absolute value. However, this condition can hold with some diverging singular values. For example, consider a $p\times p$ coefficient matrix
\begin{equation}
    \bm{A}=0.9\bm{1}_p(1,0,0,0,\cdots,0)^\top=\begin{pmatrix}
    0.9 & 0 & 0 & \cdots & 0 \\
    0.9 & 0 & 0 & \cdots & 0 \\
    0.9 & 0 & 0 & \cdots & 0 \\
    \vdots & \vdots & \vdots & \ddots & \vdots\\
    0.9 & 0 & 0 & \cdots & 0  
    \end{pmatrix}
\end{equation}
and it can be verified that $\sigma_1(\bm{A})=0.9\sqrt{p}$ and the nonzero eigenvalue of $\bm{A}$ is 0.9.

The reduced-rank VAR(1) process of rank $r$ can be formulated to the factor model in \citet{gao2021modeling},
\begin{equation}\label{eq:VAR_GTFM}
    \begin{split}
        \bm{y}_t & = \bm{U}\bm{S}\bm{V}^\top\bm{y}_{t-1}+\bbm{\varepsilon}_t = \bm{U}(\bm{S}\bm{V}^\top\bm{y}_{t-1}+\bm{U}^\top\bbm{\varepsilon}_t)+\bm{U}_\perp(\bm{U}_\perp^\top\bbm{\varepsilon}_t)\equiv \bm{U}\bm{f}_t+\bm{U}_\perp\bbm{\varepsilon}_{2t}
    \end{split}
\end{equation}
where $\bm{f}_t$ is the $r$-dimensional dynamic factor and $\bbm{\varepsilon}_{2t}$ is the $(p-r)$-dimensional white noise. To analyze the properties of this factor model, we first characterize the strength of the dynamic factor and the white noise term.


For simplicity, we assume that all nonzero singular values in $\bm{S}$ are diverging with a rate of $p^{\delta_s}$ with some $\delta_s\in[0,1/2]$. If $\delta_s=0$, all singular values are bounded.

In the factor model in \eqref{eq:VAR_GTFM}, $\bbm{\varepsilon}_t$ is split into the factor component and white noise component. Hence, it is essential to characterize the signal strength in $\mathcal{M}(\bm{U})$ and $\mathcal{M}(\bm{V})$.
Assume that the first $r$ eigenvalues of $\bm{U}^\top\bm{\Sigma}_{\bbm{\varepsilon}}\bm{U}$ scale as $p^{\delta_u}$ with some $\delta_u\in[0,1]$. Assume that the first $K$ eigenvalues of $\bm{U}_\perp^\top\bm{\Sigma}_{\bbm{\varepsilon}}\bm{U}_\perp$ scale as $p^{\delta_u'}:=p^{1-\delta_2}$ with some $\delta_u'\in(0,1]$ and $\delta_2=1-\delta_u'$. For the predictor factors, assume that the diverging eigenvalues of $\bm{V}^\top\bm{\Sigma}_{\bm{y}}\bm{V}^\top$ scale as $p^{\delta_v}$.


Since
$\bm{\Sigma}_{\bm{y}}=\bm{A}\bm{\Sigma}_{\bm{y}}\bm{A}^\top+\bm{\Sigma}_{\bbm{\varepsilon}}=\bm{U}\bm{S}\bm{V}^\top\bm{\Sigma}_{\bm{y}}\bm{V}\bm{S}\bm{U}^\top+\bm{\Sigma}_{\bbm{\varepsilon}}$,
we have
$$\lambda_{1}(\bm{U}^\top\bm{\Sigma}_{\bm{y}}\bm{U})\asymp\cdots\lambda_{r}(\bm{U}^\top\bm{\Sigma}_{\bm{y}}\bm{U})\asymp p^{\delta_u}+p^{2\delta_s+\delta_v},$$
and
$$\lambda_{1}(\bm{V}^\top\bm{\Sigma}_{\bm{y}}\bm{V})\asymp\cdots\asymp\lambda_{r}(\bm{V}^\top\bm{\Sigma}_{\bm{y}}\bm{V})\asymp p^{\delta_v}.$$
Since $\bm{\Sigma}_{\bm{f}}=\bm{S}\bm{V}^\top\bm{\Sigma}_{\bm{y}}\bm{V}\bm{S}+\bm{U}^\top\bm{\Sigma}_{\bbm{\varepsilon}}\bm{U}$, we have $\lambda_{1}(\bm{\Sigma}_{\bm{f}})\asymp\cdots\asymp\lambda_{r}(\bm{\Sigma}_{\bm{f}}) \asymp p^{(2\delta_s+\delta_v)\vee\delta_u}:=p^{1-\delta_1}$ with $\delta_1=1-((2\delta_s+\delta_v)\vee\delta_u)\in[0,1]$. 

Moreover, denote the standardized variable $\bm{x}_t=\bm{\Sigma}_{\bm{f}}^{-1/2}\bm{f}_t$ and $\bm{e}_t=\bm{\Sigma}_{\bbm{\varepsilon}_2}^{-1/2}\bbm{\varepsilon}_{2t}$. For any $k\geq1$,
\begin{equation}
    \mathbb{E}[\bm{x}_t\bm{e}_{t-k}^\top]=\bm{\Sigma}_{\bm{f}}^{-1/2}\mathbb{E}[\bm{S}\bm{V}^\top\bm{y}_{t-1}\bbm{\varepsilon}_{2,t-k}]\bm{\Sigma}_{\bbm{\varepsilon}_2}^{-1/2}.
\end{equation}
Since $\bm{y}_{t-1}=\bbm{\varepsilon}_{t-1}+\bm{A}\bbm{\varepsilon}_{t-2}+\bm{A}^2\bbm{\varepsilon}_{t-3}+\cdots$, we have
\begin{equation}
    \mathbb{E}[\bm{x}_t\bm{e}_{t-1}^\top]=\bm{\Sigma}_{\bm{f}}^{-1/2}\mathbb{E}[\bm{S}\bm{V}^\top\bbm{\varepsilon}_{t-1}\bbm{\varepsilon}_{2,t-1}]\bm{\Sigma}_{\bbm{\varepsilon}_2}^{-1/2}
\end{equation}
and $\|\mathbb{E}[\bm{x}_t\bm{e}_{t-1}^\top]\|_\text{op}\asymp p^{\delta_s+\delta_v/2-(1-\delta_1)/2}:=p^{\delta_3}$.


For any orthonormal matrices $\bm{H}_1,\bm{H}_2\in\mathbb{O}^{p\times r}$, consider the discrepancy measure 
$$D(\bm{H}_1,\bm{H}_2)=\sqrt{1-\frac{1}{r}\text{tr}(\bm{H}_1\bm{H}_1^\top\bm{H}_2\bm{H}_2^\top)}.$$

By the similar arguments as Theorem 3 in \citet{gao2021modeling}, we have the following lemma for the loading matrix estimation error upper bounds. 

\begin{lemma}\label{lemma:factor_upper_bound1}
    Under the conditions in Theorem \ref{thm:diverging}, if $p^{\delta_1\vee\delta_2}T^{-1}=o(1)$ and $\delta_1\leq\delta_2$, with probability at least $1-\exp(-Cp)$,
    \begin{equation}
        D(\widehat{\bm{U}},\bm{U}) \lesssim\sqrt{\frac{p^{\delta_1}}{T}},~~D(\widehat{\bm{U}}_\perp,\bm{U}_\perp)\lesssim\sqrt{\frac{p^{\delta_1}}{T}},~~\text{and}~~
        D(\widehat{\bm{U}}_\perp\widehat{\bm{K}},\bm{U}_\perp\bm{K}) \lesssim \sqrt{\frac{p^{2\delta_2-3\delta_1}}{T}}.
    \end{equation}
\end{lemma}

Moreover, for any $1\leq t\leq T$, by the triangle inequality,
\begin{equation}
    \begin{split}
        \|(\widehat{\bm{T}}_U-\bm{T}_U)\bm{y}_t\|_2
        \leq&\|\widehat{\bm{U}}\widehat{\bm{U}}^\top\bm{y}_t-\bm{U}\bm{U}^\top\bm{y}_t\|_2+
        \|\widehat{\bm{U}}_\perp\widehat{\bm{K}}_\perp\widehat{\bm{K}}_\perp^\top\widehat{\bm{U}}_\perp^\top\bm{y}_t-\bm{U}_\perp\bm{K}_\perp\bm{K}_\perp^\top\bm{U}_\perp^\top\bm{y}_t\|_2\\
        & + \|\widehat{\bm{U}}_\perp\widehat{\bm{K}}(\widehat{\bm{\Lambda}}_{\bbm{\varepsilon}_2}^{K})^{-1/2}\widehat{\bm{K}}^\top\widehat{\bm{U}}_\perp^\top\bm{y}_t-\bm{U}_\perp\bm{K}(\bm{\Lambda}_{\bbm{\varepsilon}_2}^{K})^{-1/2}\bm{K}^\top\bm{U}_\perp^\top\bm{y}_t\|_2.
    \end{split}
\end{equation}
Hence, following Theorem 4 of \citet{gao2021modeling}, we have the following lemma for the estimation upper bound for the transformed data.

\begin{lemma}\label{lemma:factor_upper_bound2}
    Under conditions in Lemma \ref{lemma:factor_upper_bound1}, with probability at least $1-\exp(-Cp)$,
    \begin{equation}
        \begin{split}
            \max_{1\leq t\leq T}\|(\widehat{\bm{T}}_U-\bm{T}_U)\bm{y}_t\|_2\lesssim & p^{(1-\delta_1)/2}\log(T)D(\widehat{\bm{U}},\bm{U})\\ &+p^{1/2-\delta_2}\log(T)D(\widehat{\bm{U}}_\perp\widehat{\bm{K}},\bm{U}_\perp\bm{K})\\
            \lesssim & \log(T)\left(\sqrt{\frac{p}{T}}+\sqrt{\frac{p^{1-3\delta_1}}{T}}\right)\asymp \sqrt{\frac{\log^2(T)p}{T}}.
        \end{split}
    \end{equation}
\end{lemma}

\subsection{Theoretical analysis of the transformed time series with the true transformation}

Suppose that we have the true transformation $\bm{T}_U=\bm{U}\bm{U}^\top+\bm{U}_\perp\bm{T}_K\bm{U}_\perp^\top$ and the transformed data $\widebar{\bm{y}}_t=\bm{T}_U\bm{y}_t$. Consider the model
$$\widebar{\bm{y}}_t=\bm{A}\bm{y}_{t-1}+\widebar{\bbm{\varepsilon}}_t$$
and the loss function $\widebar{\mathcal{L}}(\bm{A})=(2T)^{-1}\|\widebar{\bm{Y}}-\bm{A}\bm{X}\|_\text{F}^2$, where $\widebar{\bm{Y}}=[\widebar{\bm{y}}_1,\cdots,\widebar{\bm{y}}_T]$ is the transformed data and $\bm{X}$ is the original design matrix.



We start with the analysis of RSC and RSS with diverging eigenvalue effects.
\begin{lemma}\label{lemma:RSC2}
    Assume the conditions in Theorem \ref{thm:diverging} hold. For $\bm{A}$ and $\bm{A}^*$ such that $\|\bm{A}-\bm{A}^*\|_\textup{F}\asymp\sqrt{p^{1+\delta_u-\delta_v}/T}$, denote $\bm{\Delta}=\bm{A}-\bm{A}^*$, with probability at least $1-2\exp[-C\min(\tau^{-4},\tau^{-2})T]$,
    $$\alpha'_\textup{RSC}\|\bm{\Delta}\|_\textup{F}^2\leq\frac{1}{T}\sum_{t=0}^{T-1}\|\bm{\Delta}\bm{y}_t\|_2^2\leq\beta'_\textup{RSS}\|\bm{\Delta}\|_\textup{F}^2,$$
    where the dependency measurement quantities are defined as $\alpha'_\textup{RSC}=\lambda_{\min}(\bm{V}^{*\top}\bm{\Sigma}_{\bm{y}}\bm{V}^*)/2$ and $\beta'_\textup{RSS}=3\lambda_{\max}(\bm{V}^{*\top}\bm{\Sigma}_{\bm{y}}\bm{V}^*)/2$.
\end{lemma}

\begin{proof}[Proof of Lemma \ref{lemma:RSC2}]

Consider the SVD $\bm{A}=\bm{U}\bm{S}\bm{V}^\top$ and $\bm{A}^*=\bm{U}^*\bm{S}^*\bm{V}^{*\top}$. By Lemma \ref{lemma:svd_perturb}, $\|\bm{V}^\top\bm{V}_\perp^*\|_\text{F}=\|\sin\Theta(\bm{V},\bm{V}^*)\|_\text{F}\leq \sigma_r^{-1}(\bm{A}^*)\|\bm{A}-\bm{A}^*\|_\text{F}=p^{-\delta_s}\|\bm{A}-\bm{A}^*\|_\text{F}\asymp\sqrt{p^{1+\delta_u-\delta_v-2\delta_s}/T}$, which diverges to zero as $p$ goes to infinity.

Starting from the RSS,
\begin{equation*}
    \begin{split}
        & \left\|(\bm{U}\bm{S}\bm{V}^\top-\bm{U}^*\bm{S}^*\bm{V}^{*\top})\bm{y}_{t}\right\|^2_2 = \left\|(\bm{U}\bm{S}\bm{V}^\top-\bm{U}^*\bm{S}^*\bm{V}^{*\top})(\bm{V}^*\bm{V}^{*\top}+\bm{V}_\perp^*\bm{V}_\perp^{*\top})\bm{y}_{t}\right\|^2_2 \\
        \leq & 2\left\|(\bm{U}\bm{S}\bm{V}^\top-\bm{U}^*\bm{S}^*\bm{V}^{*\top})\bm{V}^*\bm{V}^{*\top}\bm{y}_{t}\right\|^2_2 + 2\left\|(\bm{U}\bm{S}\bm{V}^\top-\bm{U}^*\bm{S}^*\bm{V}^{*\top})\bm{V}_\perp^*\bm{V}_\perp^{*\top}\bm{y}_{t}\right\|^2_2\\
        \leq & 4\underbrace{\left\|(\bm{U}\bm{S}-\bm{U}^*\bm{S}^*)(\bm{V}^{*\top}\bm{y}_{t})\right\|_2^2}_{T_{1t}} + 4\underbrace{\left\|(\bm{U}\bm{S}-\bm{U}^*\bm{S}^*)(\bm{V}^\top\bm{V}_\perp^{*})(\bm{V}_\perp^{*\top}\bm{y}_{t})\right\|_2^2}_{T_{2t}}
    \end{split}
\end{equation*}

As in Lemma \ref{lemma:RSC}, for $T_{1t}$ we have for any $\bm{\Delta}\in\mathbb{R}^{p\times r}$, with probability at least $1-2\exp[-C\min(\tau^{-4},\tau^{-2})T]$
\begin{equation}
    \alpha'_\text{RSC}\|\bm{\Delta}\|_\text{F}^2\leq\frac{1}{T}\sum_{t=0}^{T-1}\|\bm{\Delta}(\bm{V}^{*\top}\bm{y}_t)\|_2^2\leq\beta'_\text{RSS}\|\bm{\Delta}\|_\text{F}^2,
\end{equation}
where $\alpha'_\text{RSC}=\lambda_{\min}(\bm{V}^{*\top}\bm{\Sigma}_{\bm{y}}\bm{V}^*)/2\asymp p^{\delta_v}$ and $\beta'_\text{RSS}=3\lambda_{\max}(\bm{V}^{*\top}\bm{\Sigma}_{\bm{y}}\bm{V}^*)/2\asymp p^{\delta_v}$.

Similarly, for $T_{2t}$, for any $\bm{\Delta}\in\mathbb{R}^{p\times r}$, with probability at least $1-2\exp[-C\min(\tau^{-4},\tau^{-2})T]$,
\begin{equation}
    \frac{1}{T}\sum_{t=0}^{T-1}\|\bm{\Delta}(\bm{V}^\top\bm{V}_\perp^*)(\bm{V}_\perp^{*\top}\bm{y}_t)\|_2^2\asymp \frac{p^{1+\delta_u-\delta_v-2\delta_s}}{T}\|\bm{\Delta}\|_\text{F}^2,
\end{equation}
and it is a much smaller than the sum of $T_{1t}$ given that $T\gtrsim p^{1+\delta_u-2\delta_v-2\delta_s}$, which is confirmed by the sample size requirement of Theorem \ref{thm:diverging}.

\end{proof}

Based on the transformed data $\widebar{\bm{y}}_t=\bm{A}\bm{y}_{t-1}+\widebar{\bbm{\varepsilon}}_t$, the statistical error becomes
\begin{equation}
    \begin{split}
        \widebar{\xi}(r,d)&=\sup_{\substack{\textup{\bf{[}}\bm{C}~\bm{R}\textup{\bf{]}},\textup{\bf{[}}\bm{C}~\bm{P}\textup{\bf{]}}\in\mathbb{O}^{p\times r},\\ \bm{D}\in\mathbb{R}^{r\times r},\|\bm{D}\|_\textup{F}=1}}\left\langle\nabla\mathcal{L}(\bm{A}^*),\textup{\bf{[}}\bm{C}~\bm{R}\textup{\bf{]}}\bm{D}\textup{\bf{[}}\bm{C}~\bm{P}\textup{\bf{]}}^\top\right\rangle\\
        &=\sup_{\substack{\textup{\bf{[}}\bm{C}~\bm{R}\textup{\bf{]}},\textup{\bf{[}}\bm{C}~\bm{P}\textup{\bf{]}}\in\mathbb{O}^{p\times r},\\ \bm{D}\in\mathbb{R}^{r\times r},\|\bm{D}\|_\textup{F}=1}}\frac{1}{T}\sum_{t=1}^T\left\langle\widebar{\bbm{\varepsilon}}_t\bm{y}_{t-1}^\top,\textup{\bf{[}}\bm{C}~\bm{R}\textup{\bf{]}}\bm{D}\textup{\bf{[}}\bm{C}~\bm{P}\textup{\bf{]}}^\top\right\rangle
    \end{split}
\end{equation}
Similarly to Lemma \ref{lemma:deviation}, we have the following lemma for the transformed deviation bound.

\begin{lemma}\label{lemma:deviation2}
    Assume conditions in Theorem \ref{thm:diverging} hold. If $T\gtrsim\max(\tau^2,\tau^4)p$, then, with probability at least $1-\exp(-Cp)$,
    $$\widebar{\xi}(r,d)\lesssim \tau^2 M_1'\sqrt{\frac{d_\textup{CS}(p,r,d)}{T}},$$
    where $M_1'=p^{(\delta_u+\delta_v)/2}$.
\end{lemma}
We can prove Lemma \ref{lemma:deviation2} by using the fact $\lambda_{\max}(\bm{\Sigma}_{\widebar{\bbm{\varepsilon}}})\asymp p^{\delta_u}$, Lemma \ref{lemma:RSC2}, and the same $\epsilon$-net construction method as in the proof of Lemma \ref{lemma:deviation}. The detailed proof is also omitted.

\subsection{Analysis of transformed VAR process with estimated transformation}

Finally, we consider the estimation error analysis for the estimated data $\widetilde{\bm{y}}_t$ with the estimated transformation matrix.
Denote the estimated transformation matrix by $\widehat{\bm{T}}_U=\widehat{\bm{U}}\widehat{\bm{U}}^\top+\widehat{\bm{U}}_\perp\widehat{\bm{T}}_K\widehat{\bm{U}}_\perp^\top$ and the corresponding data $\widetilde{\bm{y}}_t=\widehat{\bm{T}}_U\bm{y}_t=(\widehat{\bm{T}}_U-\bm{T}_U)\bm{y}_t+\bm{T}_U(\bm{A}^*\bm{x}_t+\bbm{\varepsilon}_t)=(\widehat{\bm{T}}_U-\bm{T}_U)\bm{y}_t+\bm{A}^*\bm{x}_t+\widebar{\bbm{\varepsilon}}_t$.
Denote the loss function with the estimated transformation as
\begin{equation}
    \widetilde{\mathcal{L}}(\bm{A})=\frac{1}{2T}\|\widetilde{\bm{Y}}-\bm{A}\bm{X}\|_\text{F}^2,
\end{equation}
where $\widetilde{\bm{Y}}=[\widetilde{\bm{y}}_1,\cdots,\widetilde{\bm{y}}_T]$ is the response matrix with the estimated transformation.

The deviation bound for the estimated data is defined as
\begin{equation}
    \begin{split}
        \widetilde{\xi}(r,d)&=\sup_{\substack{\textup{\bf{[}}\bm{C}~\bm{R}\textup{\bf{]}},\textup{\bf{[}}\bm{C}~\bm{P}\textup{\bf{]}}\in\mathbb{O}^{p\times r},\\ \bm{D}\in\mathbb{R}^{r\times r},\|\bm{D}\|_\textup{F}=1}}\left\langle\nabla\widetilde{\mathcal{L}}(\bm{A}^*),\textup{\bf{[}}\bm{C}~\bm{R}\textup{\bf{]}}\bm{D}\textup{\bf{[}}\bm{C}~\bm{P}\textup{\bf{]}}^\top\right\rangle\\
        &=\sup_{\substack{\textup{\bf{[}}\bm{C}~\bm{R}\textup{\bf{]}},\textup{\bf{[}}\bm{C}~\bm{P}\textup{\bf{]}}\in\mathbb{O}^{p\times r},\\ \bm{D}\in\mathbb{R}^{r\times r},\|\bm{D}\|_\textup{F}=1}}\frac{1}{T}\sum_{t=1}^T\left\langle[\widebar{\bbm{\varepsilon}}_t+(\widehat{\bm{T}}_U-\bm{T}_U)\bm{y}_t]\bm{y}_{t-1}^\top, \textup{\bf{[}}\bm{C}~\bm{R}\textup{\bf{]}}\bm{D}\textup{\bf{[}}\bm{C}~\bm{P}\textup{\bf{]}}^\top \right\rangle\\
        &\leq \widebar{\xi}(r,d) + \sup_{\substack{\textup{\bf{[}}\bm{C}~\bm{R}\textup{\bf{]}},\textup{\bf{[}}\bm{C}~\bm{P}\textup{\bf{]}}\in\mathbb{O}^{p\times r},\\ \bm{D}\in\mathbb{R}^{r\times r},\|\bm{D}\|_\textup{F}=1}}\frac{1}{T}\sum_{t=1}^T\langle(\widehat{\bm{T}}_U-\bm{T}_U)\bm{y}_t\bm{y}_{t-1}^\top, \textup{\bf{[}}\bm{C}~\bm{R}\textup{\bf{]}}\bm{D}\textup{\bf{[}}\bm{C}~\bm{P}\textup{\bf{]}}^\top \rangle:=T_1+T_2.
    \end{split}
\end{equation}

The upper bound of $T_1$ is established in Lemma \ref{lemma:deviation2}. In addition, by Lemma \ref{lemma:factor_upper_bound2}, with probability at least $1-\exp(-Cp)$,
$$\max_{1\leq t\leq T}\|(\widehat{\bm{T}}_U-\bm{T}_U)\bm{y}_t\|_2\lesssim\tau^2\sqrt{r}\sqrt{\frac{\log(T)^2p}{T}}:=B.$$
Then, conditioning on this upper bound, for any $\bm{M}\in\mathbb{R}^{p\times p}$ such that $\|\bm{M}\|_\text{F}=1$,
\begin{equation}
    \langle(\widehat{\bm{T}}_U-\bm{T}_U)\bm{y}_t\bm{y}_{t-1}^\top,\bm{M}\rangle\leq\|(\widehat{\bm{T}}_U-\bm{T}_U)\bm{y}_t\|_2\|\bm{M}\bm{y}_{t-1}\|_2\leq B\|\bm{M}\bm{y}_{t-1}\|_2,
\end{equation}
and by Cauchy's inequality and Lemma \ref{lemma:RSC2}, with probability approaching one,
\begin{equation}
    \begin{split}
        T_2 & \leq \frac{B}{T}\sup_{\substack{\textup{\bf{[}}\bm{C}~\bm{R}\textup{\bf{]}},\textup{\bf{[}}\bm{C}~\bm{P}\textup{\bf{]}}\in\mathbb{O}^{p\times r},\\ \bm{D}\in\mathbb{R}^{r\times r},\|\bm{D}\|_\textup{F}=1}}\sum_{t=1}^T\|\textup{\bf{[}}\bm{C}~\bm{R}\textup{\bf{]}}\bm{D}\textup{\bf{[}}\bm{C}~\bm{P}\textup{\bf{]}}^\top\bm{y}_{t-1}\|_2\\
        & \leq B\sup_{\substack{\textup{\bf{[}}\bm{C}~\bm{R}\textup{\bf{]}},\textup{\bf{[}}\bm{C}~\bm{P}\textup{\bf{]}}\in\mathbb{O}^{p\times r},\\ \bm{D}\in\mathbb{R}^{r\times r},\|\bm{D}\|_\textup{F}=1}}\sqrt{\frac{1}{T}\sum_{t=1}^T\|\textup{\bf{[}}\bm{C}~\bm{R}\textup{\bf{]}}\bm{D}\textup{\bf{[}}\bm{C}~\bm{P}\textup{\bf{]}}^\top\bm{y}_{t-1}\|_2^2}\\
        & \leq B\sqrt{\beta_\text{RSS}'}\asymp\tau^2\sqrt{r}\sqrt{\frac{\log^2(T)p^{1+\delta_v}}{T}}
    \end{split}
\end{equation}

Combining the upper bounds for $T_1$ and $T_2$, we have that with probability approaching one,
$$\widetilde{\xi}(r,d)\lesssim\tau^2p^{\delta^*}\sqrt{\frac{d_\text{CS}(p,r,d)}{T}},$$
where $\delta^*=(\delta_u+\delta_v)/2$.

Finally, we conclude this appendix by providing the proof of Theorem \ref{thm:diverging}.

\begin{proof}[Proof of Theorem \ref{thm:diverging}]

    According to Theorem \ref{thm:gd}, when we input the transformed response $\widetilde{\bm{Y}}$ and the original predictor, after $I$-th iteration with $I\gtrsim\log(p^{\delta_s/3}g_{\min})/\log(1-C\eta_0\alpha_\text{RSC}\beta_\text{RSS}^{-1})$,
    $$\|\bm{A}^{(I)}-\bm{A}\|_\text{F}\lesssim(\alpha'_\text{RSC})^{-1}\widetilde{\xi}(r,d)^2,$$
    where $\widetilde{\xi}(r,d)$ is defined above, given the $\alpha'_\text{RSC}$-RSC and $\beta'_\text{RSS}$-RSS conditions satisfied.

    By Lemma \ref{lemma:RSC2}, the RSC and RSS conditions hold with probability at least $1-2\exp[-C(M_2')^2\min(\tau^{-4},\tau^{-2})T]$. Conditioning on RSC and RSS conditions, by Lemmas \ref{lemma:factor_upper_bound1}, \ref{lemma:factor_upper_bound2} and \ref{lemma:deviation2},
    $$\widetilde{\xi}(r,d)\lesssim\tau^2p^{(\delta_u+\delta_v)/2}\sqrt{\frac{d_\text{CS}(p,r,d)}{T}},$$
    with probability at least $1-C\exp(-Cp)$. Combining these two results, as $\alpha'_\text{RSC}\asymp p^{\delta_v}$, we can obtain the desired result.

\end{proof}

\section{Supplementary materials for sparse estimation}\label{append:sparse}

\subsection{Proof of Theorem \ref{thm:sparse}}

\begin{proof}
The proof of Theorem \ref{thm:sparse} consists of two steps. In the first step, we establish the computational convergence result as an extension of Theorem \ref{thm:gd}. In the second step, a statistical convergence analysis is given to verify the sparsity-constrained restricted strong convexity and smoothness conditions.\\

\noindent \textit{Step 1.} (Computational convergence analysis)\\

First, we state some essential notations and conditions. Similarly to Theorem \ref{thm:gd}, define the combined estimation errors of $(\bm{C}^{(i)},\bm{R}^{(i)},\bm{P}^{(i)},\bm{D}^{(i)})$ up to the optimal rotations
    \begin{equation}
        \begin{split}
            E^{(i)}= \min_{\substack{\bm{O}_c\in\mathbb{O}^{d\times d}\\\bm{O}_r,\bm{O}_p\in\mathbb{O}^{(r-d)\times(r-d)}}}&\Big\{\|\bm{C}^{(i)}-\bm{C}^*\bm{O}_c\|_\text{F}^2 + \|\bm{R}^{(i)}-\bm{R}^*\bm{O}_r\|_\text{F}^2+\|\bm{P}^{(i)}-\bm{P}^*\bm{O}_p\|_\text{F}^2\\
            &+\left\|\bm{D}^{(i)}-\text{diag}(\bm{O}_c,\bm{O}_r)^\top\bm{D}^*\text{diag}(\bm{O}_c,\bm{O}_p)\right\|_\textup{F}^2\Big\}
        \end{split}
    \end{equation}
and the estimation errors of $(\widetilde{\bm{C}}^{(i)},\widetilde{\bm{R}}^{(i)},\widetilde{\bm{P}}^{(i)},\bm{D}^{(i)})$ before the hard thresholding as
\begin{equation}
        \begin{split}
            \widetilde{E}^{(i)}= \min_{\substack{\bm{O}_c\in\mathbb{O}^{d\times d}\\\bm{O}_r,\bm{O}_p\in\mathbb{O}^{(r-d)\times(r-d)}}}&\Big\{\|\widetilde{\bm{C}}^{(i)}-\bm{C}^*\bm{O}_c\|_\text{F}^2 + \|\widetilde{\bm{R}}^{(i)}-\bm{R}^*\bm{O}_r\|_\text{F}^2+\|\widetilde{\bm{P}}^{(i)}-\bm{P}^*\bm{O}_p\|_\text{F}^2\\
            &+\left\|\bm{D}^{(i)}-\text{diag}(\bm{O}_c,\bm{O}_r)^\top\bm{D}^*\text{diag}(\bm{O}_c,\bm{O}_p)\right\|_\textup{F}^2\Big\}.
        \end{split}
    \end{equation}

We assume that for the given sample size $T$, $\mathcal{L}$ is restricted strongly convex with parameter $\alpha$ and restricted strongly smooth with parameter $\beta$, such that for any rank-$r$ matrices $\bm{A}$ and $\bm{A}'$ where $\bm{A},\bm{A}^\top\in\mathbb{S}(p,p,s_c+s_r)$, and $\bm{A}',(\bm{A}')^{\top}\in\mathbb{S}(p,p,s_c+s_p)$,
\begin{equation}\label{eq:sparse_RSC_RSS}
    \frac{\alpha}{2}\|\bm{A}-\bm{A}'\|_\text{F}^2\leq \mathcal{L}(\bm{A})-\mathcal{L}(\bm{A}')-\langle\nabla\mathcal{L}(\bm{A}'),\bm{A}-\bm{A}'\rangle\leq\frac{\beta}{2}\|\bm{A}-\bm{A}'\|_\text{F}^2.
\end{equation}
In addition, for $\bm{s}=(s_c,s_r,s_p)$, we denote
\begin{equation}
    \xi(r,d,\bm{s})=\sup_{\substack{\bm{D}\in\mathbb{R}^{r\times r},\lb\bm{C}~\bm{R}\rb,\lb\bm{C}~\bm{P}\rb\in\mathbb{O}^{p\times r}\\ \bm{C}\in\mathbb{S}(p,d,s_c),\bm{R}\in\mathbb{S}(p,r-d,s_r),\bm{P}\in\mathbb{S}(p,r-d,s_p)}}\left\langle\nabla\mathcal{L}(\bm{A}^*),\lb\bm{C}~\bm{R}\rb\bm{D}\lb\bm{C}~\bm{P}^\top\rb\right\rangle.
\end{equation}

Moreover, similarly to the proof of Theorem \ref{thm:gd}, we assume $b=\sigma_1^{1/3}$ and $a=C\alpha\sigma_1^{2/3}\kappa^{-2}$. For any $i=0,1,2,\dots$, we assume that
\begin{equation}
    \|\lb\bm{C}^{(i)}~\bm{R}^{(i)}\rb\|_\text{op}\leq 1.1b, \|\lb\bm{C}^{(i)}~\bm{P}^{(i)}\rb\|_\text{op}\leq 1.1b,~\text{and}~\|\bm{D}^{(i)}\|_\text{op}\leq\frac{1.1\sigma_1}{b^2}.
\end{equation}
In addition, we assume that for any $i=0,1,2,\dots$, $E^{(i)}\leq C\sigma_1^{2/3}\alpha\beta^{-1}\kappa^{-2}$.

As $\bm{C}^{(i)}=\text{HT}(\widetilde{\bm{C}}^{(i)},s_c)$, $\bm{R}^{(i)}=\text{HT}(\widetilde{\bm{R}}^{(i)},s_r)$, $\bm{P}^{(i)}=\text{HT}(\widetilde{\bm{P}}^{(i)},s_c)$, $\bm{C}^*\bm{O}_c\in\mathbb{S}(p,d,s_c^*)$, $\bm{R}^*\bm{O}_r\in\mathbb{S}(p,r-d,s_r^*)$, and $\bm{P}^*\bm{O}_p\in\mathbb{S}(p,r-d,s_p^*)$, by Lemma \ref{lemma:hard_thresholding}, 
\begin{equation}\label{eq:HT_upper}
    E^{(i)}\leq (1+2\gamma^{-1/2})\widetilde{E}^{(i)}.
\end{equation}
Following the same arguments in the proof of Theorem \ref{thm:gd}, as in \eqref{eq:recursive}, we have the recursive arguments
\begin{equation}
    \widetilde{E}^{(i+1)} \leq \left(1-C\eta_0\alpha\beta^{-1}\kappa^{-2}\right)E^{(i)} + C\kappa^2\alpha^{-2}\sigma_1^{-4/3}\xi^2(r,d,\bm{s}),
\end{equation}
where the statistical error is replaced with $\xi(r,d,\bm{s})$. 

Combining \eqref{eq:HT_upper}, we have
\begin{equation}
    E^{(i+1)} \leq (1+2\gamma^{-1/2})(1-C\eta_0\alpha\beta^{-1}\kappa^{-2})E^{(i)} + C\kappa^2\alpha^{-2}\sigma_1^{-4/3}\xi^2(r,d,\bm{s}).
\end{equation}
When $\gamma\gtrsim\alpha^{-2}\beta^2\kappa^4$, it implies that
\begin{equation}
    E^{(i+1)} \leq (1-C\eta_0\alpha^2\beta^{-2}\kappa^{-4})E^{(i)} + C\kappa^2\alpha^{-2}\sigma_1^{-4/3}\xi^2(r,d,\bm{s}).
\end{equation}

By induction, we have that for any $i=1,2,\dots$,
\begin{equation}
    E^{(i)}\leq (1-C\eta_0\alpha^2\beta^{-2}\kappa^{-4})^iE^{(0)} + C\kappa^2\alpha^{-2}\sigma_1^{-4/3}\xi^2(r,d,\bm{s}).
\end{equation}
Furthermore, by Lemma \ref{lemma:errorbound},
\begin{equation}
    \|\bm{A}^{(i)}-\bm{A}^*\|_\text{F}^2\leq C\kappa^2(1-C\eta_0\alpha^2\beta^{-2}\kappa^{-4})^i\|\bm{A}^{(0)}-\bm{A}^*\|_\text{F}^2+C\kappa^2\alpha^{-2}\xi^2(r,d,\bm{s}).
\end{equation}

Lastly, similarly to Step 5 in the proof of Theorem \ref{thm:gd}, we can verify the essential conditions, except for RSC and RSS conditions, hold recursively, which concludes the computational convergence analysis.\\

\noindent \textit{Step 2.} (Statistical convergence analysis)

First, we establish the initialization error bound. By Lemma \ref{lemma:sparse_lasso}, with probability approaching one, the estimation error of the initial estimator $\widetilde{\bm{A}}_{L_1}$ is
\begin{equation}
    \|\widetilde{\bm{A}}_{L_1}-\bm{A}^*\|_\textup{F}\lesssim \alpha_\textup{RSC}^{-1}\tau^2M_1\sqrt{\frac{S^*\log(p)}{T}},
\end{equation}
where $S^*=(s_c^*+s_r^*)(s_c^*+s_p^*)$.

By Lemma \ref{lemma:hard_thresholding} and the same arguments in the proof of Lemma \ref{lemma:A0_init}, with probability at least $1-2\exp[-CM_2^2\min(\tau^{-2},\tau^{-4})T]-\exp[-C\log(p)]$,
\begin{equation}\label{eq:sparse_init}
    \|\bm{A}^{(0)}-\bm{A}^*\|_\textup{F}\lesssim\sigma_1^{2/3}\kappa^2g_{\min}^{-2}\|\widetilde{\bm{A}}_{L_1}-\bm{A}^*\|_\text{F}\lesssim \sigma_1^{2/3}\kappa^2g_{\min}^{-2}\alpha_{\textup{RSC}}^{-1}\tau^2M_1\sqrt{\frac{S^*\log(p)}{T}}.
\end{equation}

By the computational convergence results in Step 1, we have that, for all $i=1,2,\dots$,
\begin{equation}
    \|\bm{A}^{(i)}-\bm{A}^*\|_\textup{F}^2\lesssim\kappa^2(1-C\eta_0\alpha^2\beta^{-2}\kappa^{-4})^i\|\bm{A}^{(0)}-\bm{A}^*\|_\textup{F}^2+\kappa^2\alpha^{-2}\xi^2(r,d,\bm{s}).
\end{equation}
Combining the initialization upper bound in \eqref{eq:sparse_init}, the $\alpha_\textup{RSC}$-RSC and $\beta_\textup{RSS}$-RSS conditions in Lemma \ref{lemma:sparse_RSC_RSS}, and the deviation bound in Lemma \ref{lemma:sparse_deviation}, we have that after $I$-th iteration with
\begin{equation}
    I\gtrsim\frac{\log(\kappa^{-1}\sigma_1^{-1/3}g_{\min})}{\log(1-C\eta_0\alpha_\textup{RSC}^2\beta_\text{RSS}^{-2}\kappa^{-4})}
\end{equation}
with probability at least $1-2\exp[-CM_2^2\min(\tau^{-4},\tau^{-2})T]-2\exp[-C\log(p)]$, it holds that
\begin{equation}
    \|\bm{A}^{(I)}-\bm{A}^*\|_\textup{F}\lesssim \kappa\alpha_{\textup{RSC}}^{-1}\tau^2M_1\sqrt{\frac{sr+r^2+s\min[\log(p),\log(ep/s)]}{T}}
\end{equation}
where $s=\max(s_c+s_r,s_c+s_p)$.

\end{proof}

\subsection{Auxiliary lemmas}

The following lemma establishes an upper bound for the hard thresholding operation. This technique is Lemma 3.3 in \citet{li2016nonconvex}. To make the proof self-contained, it is presented here.

\begin{lemma}\label{lemma:hard_thresholding}
    Let $\textup{HT}(\cdot,k):\mathbb{R}^d\to\mathbb{R}^d$ be a hard thresholding operator that keeps the largest $k$ entries setting other entries to zero. For $k>k^*$, $\bbm{\theta}^*\in\mathbb{R}^{d}$ such that $\|\bbm{\theta}^*\|_0\leq k^*$, and any vector $\bbm{\theta}\in\mathbb{R}^d$, we have
    \begin{equation}
        \|\textup{HT}(\bbm{\theta},k)-\bbm{\theta}^*\|_2^2\leq \left(1+\frac{2\sqrt{k^*}}{\sqrt{k-k^*}}\right)\|\bbm{\theta}-\bbm{\theta}^*\|_2^2.
    \end{equation}
\end{lemma}

Next, we prove the restricted strong convexity (RSC) and restricted strong smoothness (RSS) conditions. For the least squares loss function $\mathcal{L}(\bm{A})=(2T)^{-1}\|\bm{Y}-\bm{A}\bm{X}\|_\text{F}^2$, it is easy to check that for any $\bm{A},\bm{A}'\in\mathbb{R}^{p\times p}$,
\begin{equation}
    \mathcal{L}(\bm{A})-\mathcal{L}(\bm{A}')-\langle\nabla\mathcal{L}(\bm{A}'),\bm{A}-\bm{A}'\rangle=\frac{1}{2T}\sum_{t=0}^{T-1}\|(\bm{A}-\bm{A}')\bm{y}_t\|_2^2.
\end{equation}

\begin{lemma}\label{lemma:sparse_RSC_RSS}
    For $\bm{A}=\lb\bm{C}~\bm{R}\rb\bm{D}\lb\bm{C}~\bm{P}\rb^\top$ and $\bm{A}'=\lb\bm{C}'~\bm{R}'\rb\bm{D}'\lb\bm{C}'~\bm{P}'\rb^\top$, where $\bm{C},\bm{C}'\in\mathbb{S}(p,d,s_c)$, $\bm{R},\bm{R}'\in\mathbb{S}(p,r-d,s_r)$, $\bm{P},\bm{P}'\in\mathbb{S}(p,r-d,s_p)$, let  $\bm{\Delta}=\bm{A}-\bm{A}'$. Under conditions in Theorem \ref{thm:sparse}, if $T\gtrsim M_2^{-2}\max(\tau^4,\tau^2)\{sr+r^2+s\min[\log(p),\log(ep/s)]\}$, then with probability at least $1-2\exp[-CM_2^2\min(\tau^{-4},\tau^{-2})T)]$
    \begin{equation}
        \alpha_\textup{RSC}\|\bm{\Delta}\|_\textup{F}^2\leq\frac{1}{T}\sum_{t=0}^{T-1}\|\bm{\Delta}\bm{y}_t\|_2^2\leq \beta_\textup{RSS}\|\bm{\Delta}\|_\textup{F}^2,
    \end{equation}
    where $s=\max(s_c+s_r,s_c+s_p)$, $\alpha_\textup{RSC}=\lambda_{\min}(\bm{\Sigma}_{\bbm{\varepsilon}})/(2\mu_{\max}(\mathcal{A}))$, and $\beta_\textup{RSS}=(3\lambda_{\max}(\bm{\Sigma}_{\bbm{\varepsilon}}))/(2\mu_{\min}(\mathcal{A}))$.
\end{lemma}

\begin{proof}[Proof of Lemma \ref{lemma:sparse_RSC_RSS}]

Note that $\bm{\Delta}=\bm{A}-\bm{A}^*$ admits the matrix decomposition
\begin{equation}\label{eq:Delta_decomposition}
    \bm{\Delta}=\lb\bm{\Delta}_c~\bm{\Delta}_r\rb\bm{\Delta}_d\lb\bm{\Delta}_c~\bm{\Delta}_p\rb^\top
\end{equation}
where $\bm{\Delta}_c\in\mathbb{S}(p,2d,2s_c)$, $\bm{\Delta}_r\in\mathbb{S}(p,2(r-d),2s_r)$, $\bm{\Delta}_p\in\mathbb{S}(p,2(r-d),2s_p)$, and $\bm{\Delta}_d\in\mathbb{R}^{2r\times 2r}$. 

For simplicity in presentation, we consider that $\|\bm{\Delta}\|_\textup{F}=1$.
By Lemma \ref{lemma:quadratic}, for any matrix $\bm{M}\in\mathbb{R}^{p\times p}$ such that $\|\bm{M}\|_\textup{F}=1$,
\begin{equation}
    \begin{split}
        & \mathbb{P}[|R_T(\bm{M})-\mathbb{E}R_T(\bm{M})|\geq t]\\
        \leq & 2\exp\left(-\min\left(\frac{t^2}{\tau^4T\lambda_{\max}^2(\bm{\Sigma}_{\bbm{\varepsilon}})\lambda_{\max}^2(\widetilde{\bm{A}}\widetilde{\bm{A}}^\top)},\frac{t}{\tau^2\lambda_{\max}^2(\bm{\Sigma}_{\bbm{\varepsilon}})\lambda^2_{\max}(\widetilde{\bm{A}}\widetilde{\bm{A}}^\top)}\right)\right),
    \end{split}
\end{equation}
where $\widetilde{\bm{A}}$ is defined as
\begin{equation}
    \widetilde{\bm{A}} = \begin{bmatrix}
        \bm{I}_p & \bm{A}^* & \bm{A}^{*2} & \bm{A}^{*3} & \dots & \bm{A}^{*(T-1)} & \dots\\
        \bm{O} & \bm{I}_p & \bm{A}^* & \bm{A}^{*2} & \dots & \bm{A}^{*(T-2)} & \dots\\
        \vdots & \vdots & \vdots & \vdots & \ddots & \vdots & \dots \\
        \bm{O} & \bm{O} & \bm{O} & \bm{O} & \dots & \bm{I}_p & \dots
    \end{bmatrix}.
\end{equation}

It remains to find the covering number for the set of $\bm{\Delta}$ admitting matrix decomposition as in \eqref{eq:Delta_decomposition}. Based on the row-wise sparsity structure on $\lb\bm{\Delta}_s~\bm{\Delta}_r\rb$ and $\lb\bm{\Delta}_s~\bm{\Delta}_p\rb$, the nonzero entries in $\bm{\Delta}$ can be summarized into a submatrix $\widetilde{\bm{\Delta}}\in\mathbb{R}^{2(s_c+s_r)\times2(s_c+s_p)}$ that is of rank at most $2r$.

For any given row-wise and column-wise sparsity index sets $\mathcal{S}_r$ and $\mathcal{S}_c$ for $\bm{\Delta}$, following Lemma \ref{lemma:cs_covering}, the $\epsilon$-covering number of the low-rank set is $(24/\epsilon)^{4(2s_c+s_r+s_p)r+4r^2}$. Denote $s=\max(s_c+s_r,s_c+s_p)$. As $\binom{p}{s}\leq \min[p^s,(ep/s)^s]$, for any constant $\epsilon<1$, taking a union bound on the $\epsilon$-covering net for the set of $\widetilde{\bm{\Delta}}$ and another union bound on $\mathcal{S}_r$ and $\mathcal{S}_c$, we have
\begin{equation}
    \begin{split}
        \mathbb{P}&\left[\sup_{\bm{\Delta}}|R_T(\bm{\Delta})-\mathbb{E}R_T(\bm{\Delta})|\geq t\right]
        \leq C\exp\Bigg(2s\min[\log(p),\log(ep/s)]+2sr+r^2\\
        &-\min\left(\frac{t^2}{\tau^4T\lambda_{\max}^2(\bm{\Sigma}_{\bbm{\varepsilon}})\lambda_{\max}^2(\widetilde{\bm{A}}\widetilde{\bm{A}}^\top)},\frac{t}{\tau^2\lambda_{\max}^2(\bm{\Sigma}_{\bbm{\varepsilon}})\lambda^2_{\max}(\widetilde{\bm{A}}\widetilde{\bm{A}}^\top)}\right)\Bigg).
    \end{split}
\end{equation}
Letting $t=T\lambda_{\min}(\bm{\Sigma}_{\bbm{\varepsilon}})\lambda_{\min}(\widetilde{\bm{A}}\widetilde{\bm{A}}^\top)/2$, for $T\gtrsim M_2^{-2}\max(\tau^4,\tau^2)\{sr+r^2+s\min[\log(p),\log(ep/s)]\}$, we have
\begin{equation}
    \begin{split}
        &\mathbb{P}\left[\sup_{\bm{\Delta}}|R_T(\bm{\Delta})-\mathbb{E}R_T(\bm{\Delta})|\geq T\lambda_{\min}(\bm{\Sigma_\varepsilon})\lambda_{\min}(\widetilde{\bm{A}}\widetilde{\bm{A}}^\top)/2\lambda\right]\\
        2&\exp[-CM_2^2\min(\tau^{-4},\tau^{-2})T)].
    \end{split}
\end{equation}
Therefore, with probability at least $1-2\exp[-CM_2^2\min(\tau^{-4},\tau^{-2})T)]$, $\alpha_\textup{RSC}\leq T^{-1}R_T(\bm{\Delta})\leq \beta_\textup{RSS}$. Finally, we may replace $\lambda_{\max}(\widetilde{\bm{A}}\widetilde{\bm{A}}^\top)$ and $\lambda_{\min}(\widetilde{\bm{A}}\widetilde{\bm{A}}^\top)$ with $1/\mu_{\min}(\mathcal{A})$ and $1/\mu_{\max}(\mathcal{A})$, as in the proof of Lemma \ref{lemma:RSC}.

\end{proof}

The following lemma is the deviation bound for the sparsity-constrained set. The proof of this lemma essentially follows that of Lemma \ref{lemma:deviation}, and the covering number of the sparsity-constrained low-rank set has been established in Lemma \ref{lemma:sparse_RSC_RSS}. The proof is hence omitted for simplicity.

\begin{lemma}\label{lemma:sparse_deviation}

    If $T\gtrsim M_2^{-2}\max(\tau^4,\tau^2)\{sr+r^2+s\min[\log(p),\log(ep/s)]\}$, under the  conditions in Theorem \ref{thm:sparse}, with probability at least $1-2\exp[-C\log(p)]$,
    \begin{equation}
        \xi(r,d,\bm{s})\leq \tau^2M_1\sqrt{\frac{sr+r^2+s\min[\log(p),\log(ep/s)]}{T}},
    \end{equation}
    where $s=\max(s_c+s_r,s_c+s_p)$.

\end{lemma}

The last lemma is the estimation error of $L_1$ regularized estimator \citep{basu2015regularized} in which the Gaussian distribution condition is relaxed to the sub-Gaussian condition in Assumption \ref{asmp:2}.

\begin{lemma}\label{lemma:sparse_lasso}
    
    If $T\gtrsim M_2^{-2}\max(\tau^4,\tau^2)S^*\log(p)$ and $\lambda\asymp\tau^2M_1\sqrt{S^*\log(p)/T}$, under the conditions in Theorem \ref{thm:sparse}, with probability at least $1-2\exp[-C\log(p)]$,
    \begin{equation}
        \|\widetilde{\bm{A}}_{L_1}-\bm{A}^*\|_\textup{F}\lesssim \alpha_\textup{RSC}^{-1}\tau^2M_1\sqrt{\frac{S^*\log(p)}{T}},
    \end{equation}
    where $S^*=(s_c^*+s_r^*)(s_c^*+s_p^*)$.
    
\end{lemma}

\section{Additional simulation results}
\label{append:simulation}

We also conduct a simulation experiment to compare the proposed model with the dynamic factor model due to their close relationship in Section \ref{sec:2.1}. Three data generating processes are used with the dimension $p=50$.
\begin{itemize}
    \item DGP1: Dynamic factor model with $r=3$, i.e. 
    $\bm{y}_t=\bm{\Lambda}\bm{f}_t+\bbm{\varepsilon}_t$ and $\bm{f}_t=\bm{B}\bm{f}_{t-1}+\bbm{\xi}_t$, where loading matrix $\bm{\Lambda}\in\mathbb{O}^{50\times 3}$ is orthonormal, $\bm{B}=\text{diag}\{0.8,0.8,0.8\}\in\mathbb{R}^{3\times 3}$, $\{\bbm{\varepsilon}_t\}$ are $i.i.d.$ with $N(\bm{0}_{50},0.5\bm{I}_{50})$, and $\{\bbm{\xi}_t\}$ are $i.i.d.$ with $N(\bm{0}_{3},\bm{I}_3)$.
    \item DGP2: The same as in DGP1 except that $\bbm{\varepsilon}_t=\bm{\Gamma}\bm{e}_t$, where $\bm{\Gamma}\in\mathbb{O}^{50\times 47}$, $\bm{\Gamma}^\top\bm{\Lambda}=\bm{0}_{47\times 3}$, and $\{\bm{e}_t\}$ are $i.i.d.$ with $N(\bm{0}_{47},0.5\bm{I}_{47})$.
    \item DGP3: The proposed model in \eqref{eq:equivalent_form} with rank $r=3$ and common dimension of $d=1$, i.e. $\bm{y}_t=\lb\bm{C}~\bm{R}\rb\bm{D}\lb\bm{C}~\bm{P}\rb^\top\bm{y}_{t-1}+\bbm{\varepsilon}_t$, where $\bm{C}\in\mathbb{O}^{50\times 1}$, $\bm{R},\bm{P}\in\mathbb{O}^{50\times 2}$, $\bm{C}^\top\bm{R}=\bm{C}^\top\bm{P}=\bm{0}_{1\times 2}$, $\bm{D}=\text{diag}\{1,1,1\}\in\mathbb{R}^{3\times 3}$, and $\{\bbm{\varepsilon}_t\}$ are $i.i.d.$ with $N(\bm{0}_{50},\bm{I}_{50})$.
\end{itemize}

Note that DGP1 is a standard dynamic factor model, and the series $\{\bm{y}_t\}$ admits a VARMA(1,1) form; see also \cite{wang2019high}.
As discussed in Section \ref{sec:2.1}, DGP2 is a special dynamic factor model \citep{gao2021modeling}, and it is equivalent to our model with the common dimension $d=3$, i.e., the response and predictor spaces are identical.
In DGP3, we have $d<r$ making it fundamentally different from dynamic factor models.

We apply two different modeling frameworks, denoted by VAR-CS and DFM-VAR for convenience, to all the above three data generating processes.
The VAR-CS framework is a VAR(1) model with a common subspace, i.e., the proposed model in \eqref{eq:equivalent_form}, and we adopt the proposed methods for estimation, where the upper bound of ranks $\bar{r}$ and hyper-parameter $s(p,T)$ are set as those in the first experiment.
For the DFM-VAR framework, we estimate $\widehat{\bm{\Lambda}}$ via eigen-decomposition of autocovariance matrices and obtain the estimated factors $\widehat{\bm{f}}_t$.
Then a VAR(1) model is assumed for the factors, and the ordinary least squares method is used in estimation.
It is noteworthy that the VAR-CS framework is misspecified for DGP1, while the DFM-VAR framework is misspecified for DGP3.

From Section \ref{sec:2.1}, the proposed model also admits a form of factor models, and the corresponding factor space is exactly the response factor space.
This motivates us first to compare the estimated factor space and response factor space from the DFM-VAR and VAR-CS frameworks, respectively, and their estimation errors can be defined as $\|\widehat{\bm{\Lambda}}\widehat{\bm{\Lambda}}^\top-\bm{\Lambda}\bm{\Lambda}^\top\|_\text{F}$ and $\|\lb\widehat{\bm{C}}~\widehat{\bm{R}}\rb\lb\widehat{\bm{C}}~\widehat{\bm{R}}\rb^\top-\bm{\Lambda}\bm{\Lambda}^\top\|_\text{F}$, respectively.
The median and quartiles of errors are plotted in the upper panel of Figure \ref{fig:sim3}, and the VAR-CS framework has a worse performance for DGP1, but a better performance for DGP3. This is due to the fact that DGP1 is a VARMA(1,1) model in nature, while a VAR(1) model is used in the VAR-CS framework. Moreover, although DGP3 admits a form of static factor models, its low-dimensional structure is totally different from the VAR model used in the DFM-VAR framework.
Finally, for DGP2, the two modeling frameworks are both correctly specified, while our framework has a slightly, but uniformly, better performance.

We next evaluate the prediction performance of the two modeling frameworks. 
The prediction errors are measured by $\|\widehat{\bm{y}}_{T+1}-\mathbb{E}[\bm{y}_{T+1}|\mathcal{F}_T]\|_2$, and the one-step-ahead prediction of our framework is $\widehat{\bm{y}}_{T+1}=\widehat{\bm{A}}\bm{y}_{T}$, where $\widehat{\bm{A}}$ is the fitted parameter matrix.
For the DFM-VAR framework, the prediction is defined as $\widehat{\bm{y}}_{T+1}=\widehat{\bm{\Lambda}}\widehat{\bm{f}}_{T+1}$, while $\widehat{\bm{f}}_{T+1}$ is predicted by the fitted low-dimensional VAR(1) model.
The lower panel of Figure \ref{fig:sim3} plots the median and quartiles of prediction errors from 500 replications.
Interestingly, for both DGP1 and DGP2, our modeling framework has almost the same prediction errors as those of DFM-VAR, even though model misspecification occurs in DGP1 for our modeling framework.
More importantly, for DGP3, the low-dimensional response and predictor spaces are distinct, and hence the dynamic factor modeling has much worse performance in prediction.
From the results, it is advantageous to use the proposed model in high-dimensional time series forecasting, which is the main task in the literature.

\begin{figure}[!htp]
    \begin{center}
        \includegraphics[width=0.7\textwidth]{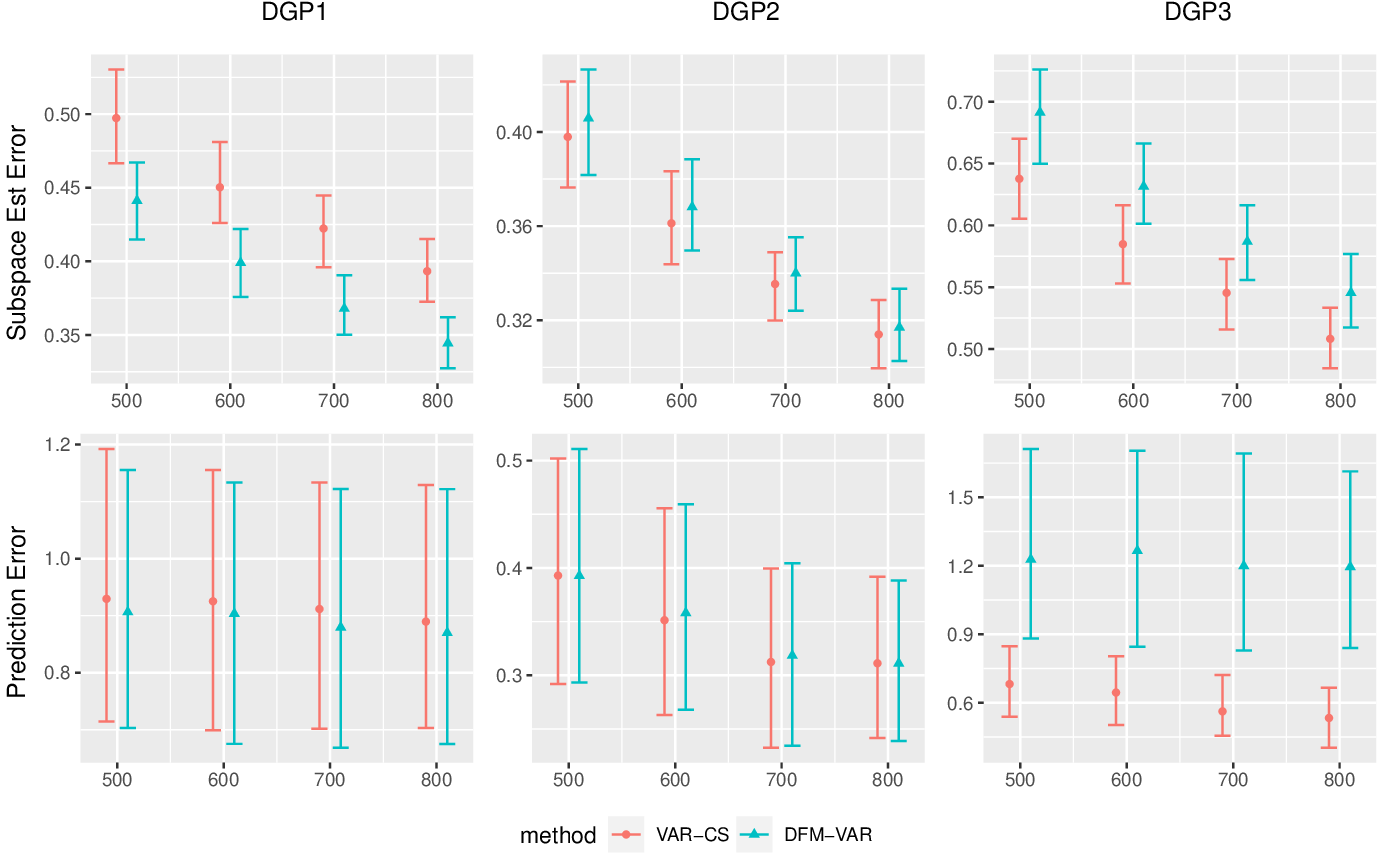}
        \vspace{-0.5cm}
    \end{center}
    \caption{\small{Estimation errors of factor spaces (upper panel) and prediction errors (lower panel) from the proposed methodology (VAR-CS) and dynamic factor modeling (DFM-VAR).}}
    \label{fig:sim3}
\end{figure}

\section{Information of macroeconomic dataset}\label{append:real_data}

\setcounter{table}{0}
\renewcommand{\thetable}{E.\arabic{table}}

The information of forty macroeconomic variables is given in the Table \ref{tbl:macro}. All variables are transformed to be stationary with codes given in column T, and except for financial variables, all variables are subject to seasonal adjustments. The US macroeconomic data set is originally from \citet{stock2009forecasting}, and these forty economic variables are selected from \citet{koop2013forecasting}.

\begin{landscape}
    \begin{table}[]
        \small
        \centering
        \caption{Forty quarterly macroeconomic variables belonging to 8 categories. Category code (C) represents: 1 = GDP and its decomposition, 2 = national association of purchasing managers (NAPM) indices, 3 = industrial production, 4 = housing, 5 = money, credit, interest rates, 6 = employment, 7 = prices and wages, 8 = others. Variables are seasonally adjusted except for those in category 5. All variables are transformed to stationarity with the following transformation codes (T): 1 = no transformation, 2 = first difference, 3 = second difference, 4 = log, 5 = first difference of logged variables, 6 = second difference of logged variables.}
        \label{tbl:macro}
        \renewcommand{\arraystretch}{1.2}
        \begin{tabular}{@{}llllllll@{}}
            \toprule
            Short name&C&T&Description&Short name&C&T&Description\\
            \midrule
            GDP251&1&5& Real GDP, quantity index (2000=100)            &FM2&5&6& Money stock: M2 (bil\$)                        \\[-2ex]
            GDP252&1&5& Real personal cons exp, quantity index       &FMRNBA&5&3& Depository inst reserves: nonborrowed (mil\$)  \\[-2ex]
            GDP253&1&5& Real personal cons exp: durable goods&FMRRA&5&6& Depository inst reserves: total (mil\$)        \\[-2ex]
            GDP256&1&5& Real gross private domestic investment&FSPIN&5&5& S\&P's common stock price index: industrials   \\[-2ex]
            GDP263&1&5& Real exports&FYFF&5&2& Interest rate: federal funds (\% per annum)    \\[-2ex]
            GDP264&1&5& Real imports&FYGT10&5&2& Interest rate: US treasury const. mat., 10-yr  \\[-2ex]
            GDP265&1&5& Real govt cons expenditures \&  gross investment&SEYGT10&5&1& Spread btwn 10-yr and 3-mth T-bill rates\\[-2ex]
            GDP270&1&5& Real final sales to domestic purchasers&CES002&6&5& Employees, nonfarm: total private              \\[-2ex]
            PMCP&2&1& NAPM commodity price index (\%)&LBMNU&6&5& Hrs of all persons: nonfarm business sector\\[-2ex]
            PMDEL&2&1& NAPM vendor deliveries index (\%)&LBOUT&6&5& Output per hr: all persons, business sec\\[-2ex]
            PMI&2&1& Purchasing managers' index&LHEL&6&2& Index of help-wanted ads in newspapers\\[-2ex]
            PMNO&2&1& NAPM new orders index (\%)&LHUR&6&2& Unemp. rate: All workers, 16 and over (\%)     \\[-2ex]
            PMNV&2&1& NAPM inventories index (\%)&CES275R&7&5& Real avg hrly earnings, non-farm prod. workers \\[-2ex]
            PMP&2&1& NAPM production index (\%)&CPIAUCSL&7&6& CPI all items                                  \\[-2ex]
            IPS10&3&5& Industrial production index: total             &GDP273&7&6& Personal consumption exp.: price index         \\[-2ex]
            UTL11&3&1& Capacity utilization: manufacturing (SIC)      &GDP276&7&6& Housing price index\\[-2ex]
            HSFR&4&4& Housing starts: Total (thousands)              &PSCCOMR&7&5& Real spot market price index: all commodities  \\[-2ex]
            BUSLOANS&5&6& Comm. and industrial loans at all comm. Banks&PWFSA&7&6& Producer price index: finished goods           \\[-2ex]
            CCINRV&5&6& Consumer credit outstanding: nonrevolving&EXRUS&8&5& US effective exchange rate: index number       \\[-2ex]
            FM1&5&6& Money stock: M1 (bil\$)                        &HHSNTN&8&2& Univ of Mich index of consumer expectations\\[-1ex]
            \bottomrule
        \end{tabular}
    \end{table}
\end{landscape}

\end{appendix}

\end{document}